\documentclass[EES, biber]{nowfnt}

\usepackage[utf8]{inputenc}

\title{Power Grid Topology Control}

\renewcommand{\maketitle}{}

\subtitle{~}

\maintitleauthorlist{
Tong Han \\
Nanyang Technological University \\
than@ncepu.edu.cn
\and 
Yan Xu \\
Nanyang Technological University \\
xuyan@ntu.edu.sg
\and
David J. Hill \\
Monash University \\
University of Sydney \\
DavidJ.Hill@monash.edu

}

\usepackage{multirow}
\usepackage{mathtools}

\usepackage{tabularx}

\usepackage{amssymb}
\usepackage{makecell}

\usepackage{booktabs}
\usepackage{amsmath}
\usepackage{amsthm}

\DeclareMathOperator*{\argmax}{arg\,max}
\DeclareMathOperator*{\argmin}{arg\,min}
\DeclareMathOperator*{\arginf}{arg\,inf}

\DeclareMathOperator*{\diag}{diag}

\usepackage{amsfonts}
\usepackage{bm}
\usepackage{bbm}
\DeclareMathOperator{\Tr}{Tr}

\usepackage{tikz}
\usepackage{pgf}
\usetikzlibrary{arrows,automata}
\usetikzlibrary{shapes}
\tikzstyle{data44}=[rectangle split,rectangle split parts=2,draw,text centered]

\newtheorem{assumption}{Assumption}
\newtheorem{problem}{Problem} 
\newtheorem{condition}{Condition} 

\newtheorem{criterion}{Criterion}

\usepackage[linesnumbered,ruled,vlined]{algorithm2e}
\SetKwInput{KwInput}{Input}                
\SetKwInput{KwOutput}{Output}   
\let\oldnl\nl
\newcommand{\nonl}{\renewcommand{\nl}{\let\nl\oldnl}}

\usepackage{etoolbox}

\makeatletter
\patchcmd{\@algocf@start}
  {-1.5em}
  {0pt}
  {}{}
\makeatother

\newcommand{\HR}{^{\circ \frac{1}{2}}}
\newcommand{\T}{^{\scriptscriptstyle\rm T}}
\newcommand{\B}{^{\scriptscriptstyle\rm min}}  
\newcommand{\U}{^{\scriptscriptstyle\rm max}}
\newcommand{\D}{^{\diamond}}

\newcommand{\A}{^{\scriptscriptstyle\rm A}}
\newcommand{\PW}{^{\circ 2}}

\usepackage{stmaryrd}
\usepackage{subcaption}

\usepackage{acro}
\acsetup{
    first-style = short,
    make-links = false
    }

\DeclareAcronym{fnn}{
  short=FNN,
  long=Feedforward Neural Network,
  short-plural=FNNs,
  long-plural=Feedforward Neural Networks,
}

\DeclareAcronym{mlf}{
  short=MLF,
  long=Multiple Lyapunov Function,
  short-plural=MLFs,
  long-plural=Multiple Lyapunov Functions,
}

\DeclareAcronym{mmg}{
  short=MMG,
  long=Multi-microgrid,
  short-plural=MMGs,
  long-plural=Multi-microgrids,
}

\DeclareAcronym{pcc}{
  short=PCC,
  long=Point of Common Coupling,
  short-plural=PCCs,
  long-plural=Points of Common Coupling,
}

\DeclareAcronym{ssa}{
  short=SSA,
  long=Stability and Stabilizability Assessment,
  short-plural=SSAs,
  long-plural=Stability and Stabilizability Assessments,
}

\DeclareAcronym{roa}{
  short=ROA,
  long=Region of Attraction,
  short-plural=ROAs,
  long-plural=Regions of Attraction,
}

\DeclareAcronym{smt}{
  short=SMT,
  long=Satisfiability Modulo Theories,
  short-plural=SMTs,
  long-plural=Satisfiability Modulo Theories,
}

\DeclareAcronym{ici}{
  short=ICI,
  long=Intentional Controlled Islanding,
  short-plural=ICIs,
  long-plural=Intentional Controlled Islandings,
}

\DeclareAcronym{dnr}{
  short=DNR,
  long=Distribution Network Reconfiguration,
  short-plural=DNRs,
  long-plural=Distribution Network Reconfigurations,
}

\DeclareAcronym{ots}{
  short=OTS,
  long=Optimal Transmission Switching,
  short-plural=OTSs,
  long-plural=Optimal Transmission Switching Problems,
}

\DeclareAcronym{minlp}{
  short=MINLP,
  long=Mixed-Integer Nonlinear Programming,
  short-plural=MINLPs,
  long-plural=Mixed-Integer Nonlinear Programming Problems,
}

\DeclareAcronym{avc}{
  short=AVC,
  long=Auxiliary Control Variable,
  short-plural=AVCs,
  long-plural=Auxiliary Control Variables,
}

\DeclareAcronym{svc}{
  short=SVC,
  long=Static VAR Compensator,
  short-plural=SVCs,
  long-plural=Static VAR Compensators,
}

\DeclareAcronym{statcom}{
  short=STATCOM,
  long=Static Synchronous Compensator,
  short-plural=STATCOMs,
  long-plural=Static Synchronous Compensators,
}

\DeclareAcronym{tcsc}{
  short=TCSC,
  long=Thyristor-Controlled Series Compensator,
  short-plural=TCSCs,
  long-plural=Thyristor-Controlled Series Compensators,
}

\DeclareAcronym{dvc}{
  short=DVC,
  long=Dynamic VAR Compensator,
  short-plural=DVCs,
  long-plural=Dynamic VAR Compensators,
}

\DeclareAcronym{ess}{
  short=ESS,
  long=Energy Storage System,
  short-plural=ESSs,
  long-plural=Energy Storage Systems,
}

\DeclareAcronym{sg}{
  short=SG,
  long=Synchronous Generator,
  short-plural=SGs,
  long-plural=Synchronous Generators,
}

\DeclareAcronym{cig}{
  short=CIG,
  long=Converter-Interfaced Generator,
  short-plural=CIGs,
  long-plural=Converter-Interfaced Generators,
}

\DeclareAcronym{acv}{
  short=ACV,
  long=Auxiliary Control Variable,
  short-plural=ACVs,
  long-plural=Auxiliary Control Variables,
}

\DeclareAcronym{oltc}{
  short=OLTC,
  long=On-Load Tap Changer,
  short-plural=OLTCs,
  long-plural=On-Load Tap Changers,
}

\DeclareAcronym{misocp}{
  short=MISOCP,
  long=Mixed-Integer Second-Order Cone Programming,
  short-plural=MISOCPs,
  long-plural=Mixed-Integer Second-Order Cone Programming Problems,
}

\DeclareAcronym{nlp}{
  short=NLP,
  long=Nonlinear Programming,
  short-plural=NLPs,
  long-plural=Nonlinear Programming Problems,
}

\DeclareAcronym{vre}{
  short=VRE,
  long=Variable Renewable Energy,
  short-plural=VREs,
  long-plural=Variable Renewable Energy Sources,
}

\DeclareAcronym{kkt}{
  short=KKT,
  long=Karush-Kuhn-Tucker,
  short-plural=KKT conditions,
  long-plural=Karush-Kuhn-Tucker Conditions,
}

\DeclareAcronym{ccg}{
  short=CCG,
  long=Column-and-Constraint Generation,
  short-plural=CCGs,
  long-plural=Column-and-Constraint Generation Methods,
}

\DeclareAcronym{ggnn}{
  short=GGNN,
  long=Gated Graph Neural Network,
  short-plural=GGNNs,
  long-plural=Gated Graph Neural Networks,
}

\DeclareAcronym{opf}{
  short=OPF,
  long=Optimal Power Flow,
  short-plural=OPFs,
  long-plural=Optimal Power Flow Problems,
}

\DeclareAcronym{pjm}{
  short=PJM,
  long=Pennsylvania-New Jersey-Maryland,
  short-plural=PJM systems,
  long-plural=PJM Interconnections,
}

\DeclareAcronym{ems}{
  short=EMS,
  long=Energy Management System,
  short-plural=EMSs,
  long-plural=Energy Management Systems,
}

\DeclareAcronym{milp}{
  short=MILP,
  long=Mixed-Integer Linear Programming,
  short-plural=MILPs,
  long-plural=Mixed-Integer Linear Programming Problems,
}

\DeclareAcronym{mip}{
  short=MIP,
  long=Mixed-Integer Programming,
  short-plural=MIPs,
  long-plural=Mixed-Integer Programming Problems,
}

\DeclareAcronym{tso}{
  short=TSO,
  long=Transmission System Operator,
  short-plural=TSOs,
  long-plural=Transmission System Operators,
}

\usepackage{enumitem}
\usepackage{threeparttable}
\usepackage{longtable}  

\usepackage{float} 

\usepackage{mwe}

\author[1]{Han, Tong}
\author[1]{Xu, Yan}
\author[2,3]{Hill, David J.}

\affil[1]{Center for Power Engineering, School of Electrical and Electronic Engineering, Nanyang Technological University, Singapore 639798, Singapore} 
\affil[2]{Department of Electrical and Computer Systems Engineering, Monash University, Clayton, VIC 3800, Australia}
\affil[3]{School of Electrical and Computer Engineering, University of Sydney, Sydney, NSW 2006, Australia}

\articledatabox{\nowfntstandardcitation}

\begin{document}

\makeabstracttitle

\begin{abstract}
    Power grids are facing major challenges from growing renewable integration and worsening climate impacts. While flexibility on both the demand and generation sides has been widely explored to address these challenges, network-side flexibility, especially in network topology, remains highly underutilized. Advances in communication, power electronics, and circuit breakers have made network topology increasingly controllable. However, leveraging this topological flexibility poses substantial challenges, primarily due to the inherent non-convexity and hybrid dynamics in associated optimization and control problems. 
    This monograph surveys the development of power grid topology control in both early and recent years. 
    It begins by discussing the fundamental topological constraints involved in topology control problems. Subsequently, it introduces steady-state topology control for transmission and distribution networks separately, covering fundamentals, a state-of-the-art review, and representative recent advances. Additionally, the network topology transition problem, which addresses the implementation of optimal topology solutions and has garnered increasing attention in recent years, is further modeled and analyzed. Beyond utilizing the flexibility of steady-state network topology, controlling network topology during transients can also contribute to system stabilization. Traditional approaches, such as intentional controlled islanding for transmission networks, as well as recently
    developed topology control methods for microgrid stabilization, exemplify this concept. Finally, a summary of this monograph is provided. 
\end{abstract}

\tableofcontents   

\printacronyms

\graphicspath{{chapter_1/Figs/}}
\chapter{Introduction}\label{chapter-1}

Network topology, determining the physical interconnections among electrical components, serves as the foundational structure of power grids. 
To address the challenges posed by growing renewable integration and worsening climate impacts, flexibility on both the demand and generation sides of the power grid has been extensively explored. 
However, network-side flexibility, particularly in terms of network topology, remains significantly underutilized and lacks comprehensive understanding. 
From a systemic perspective, this monograph offers an in-depth and holistic introduction to topology control of power grids.

\section{Overview of Power Grid Topology Control}

Dispatch and control of power grids can be divided into three categories according to the nature of the adopted flexibility resources: the generation-based, load-based, and network-based, as shown in Fig. \ref{fig-b1-2}. Different from the former two categories which utilize the flexibility of generation and load, the network-based approach mainly leverages the flexibility of network structure, including network topology and branch impedances. 
The fundamental mechanism of network topology control lies in the ability to alter power flow and system dynamics by switching transmission or feeder lines. By reconfiguring the network topology, it is possible to influence the operational state of the power grid and thereby enhance specific performance metrics. Compared to other categories of power grid control, topology control offers two distinctive advantages. 

\begin{figure}[h]
	\centering
	\includegraphics[width=1\linewidth]{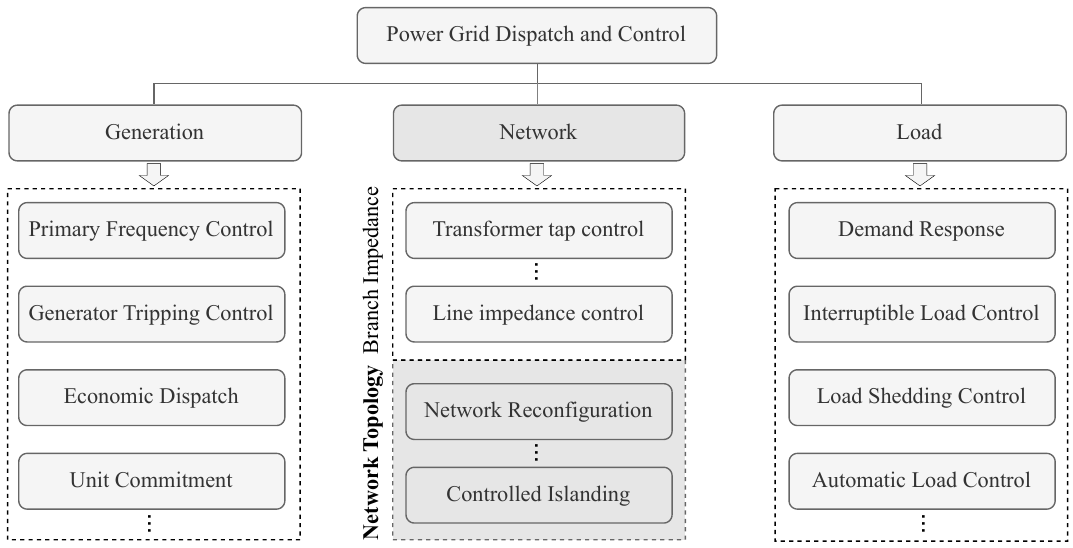}
	\caption{Classification of power grid dispatch and control.}
	\label{fig-b1-2}
\end{figure}

First, network topology, as the foundational structure of power grids, plays a critical role in system operations due to its extensive and global influence on system performance \citep{4-999-173, 4-b1-60, 4-1621, 4-b1-61}. Since transmission and feeder lines span the entire network, changes in topology can significantly affect power flow \citep{4-1484}, stability \citep{4-664}, and operational efficiency \citep{4-b1-62}. In contrast, control actions based on generation or load are inherently more localized. Moreover, control of branch impedances, such as that achieved through devices like Thyristor-Controlled Series Compensators (\ac{tcsc}s) or On-Load Tap Changer (\ac{oltc}) transformers, is limited to a small subset of branches with controllable impedances or transformers.

Second, from a cost perspective, adjusting generation or load often incurs substantial operational costs or requires ancillary service compensation \citep{4-b1-63, 4-b1-64}. Similarly, branch impedances regulation typically necessitates equipment upgrades. In contrast, topology control relies exclusively on existing infrastructure, namely, remotely controllable circuit breakers that are already integral components of modern power systems. These devices not only support low-cost operational adjustments but also exhibit long service lives, often capable of withstanding tens of thousands of switching operations \citep{4-b1-65}. Therefore, topology control is a highly economical solution for enhancing power grid performance.

\begin{figure}[h]
	\centering
	\includegraphics[width=1\linewidth]{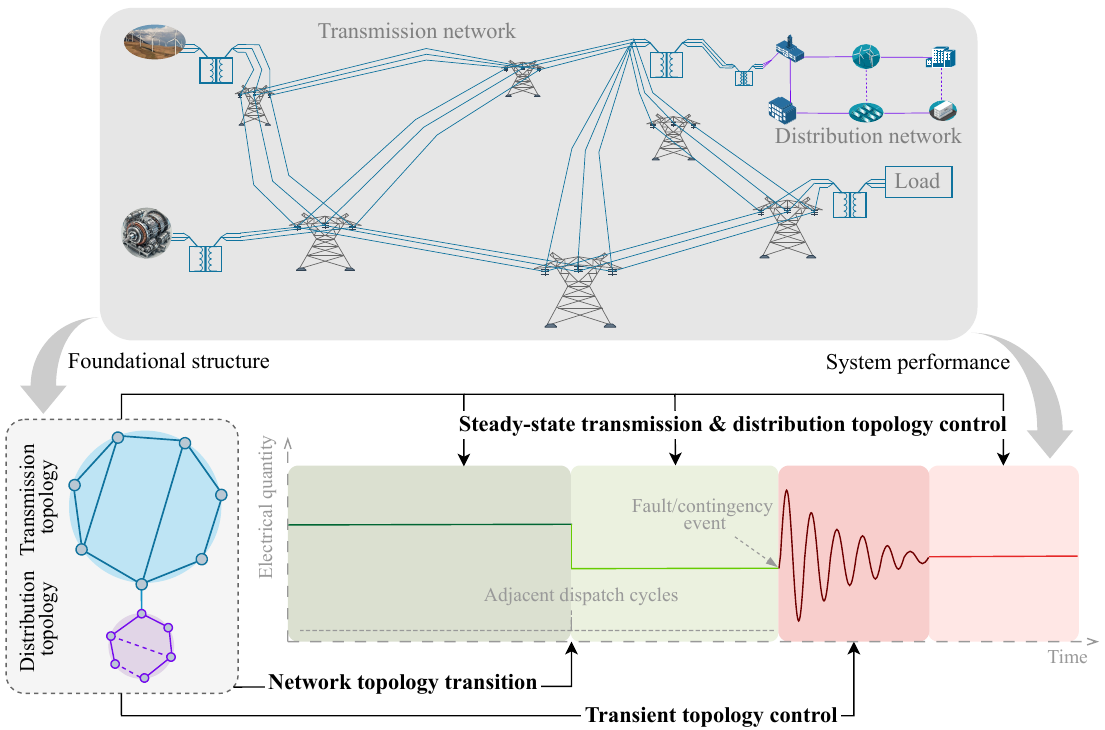}
	\caption{Overview of power grid topology control.}
	\label{fig-b1-1-1}
\end{figure}

Power grid topology control has various forms and is referred to by different terms in the existing literature, such as distribution network reconfiguration, optimal transmission switching, network topology optimization, corrective topology control, line switching, intentional controlled islanding, stabilizing switching, etc. Essentially, as shown in Fig. \ref{fig-b1-1-1}, according to the system state in which the network topology is adjusted or the state that the adjustment is intended to affect, power grid topology control can be classified into two categories: \textit{steady-state topology control} and \textit{transient topology control}. Steady-state topology control alters the network topology between steady states, typically to enhance power flow-related performance, or to improve dynamic performance, particularly system stability, in a preventive or corrective manner. For instance, corrective topology control after a contingency to ensure $N-1$ security of transmission networks and distribution network reconfiguration for loss reduction are both steady-state topology control. 
In contrast, transient topology control alters the network topology during the transient period following a disturbance, with the primary objective of stabilizing the system. For instance, intentional controlled islanding of transmission networks falls under the category of transient topology control. 
Moreover, the practical implementation of steady-state topology control raises another critical issue: how to perform the transition between network topologies across two adjacent dispatch cycles, where each dispatch cycle denotes a steady-state operational interval during which the network topology remains fixed and the system is operated under an optimized dispatch. This \textit{network topology transition} problem becomes particularly important when steady-state topology control is executed at high frequencies to maximally exploit the potential of topological flexibility.

Topology control will face increasing demands in future power grids, while the network topology itself will become increasingly controllable. 
Future power grids will confront significant operational challenges. For instance, high renewable penetration introduces substantial uncertainties into system operation and gives rise to new dynamic issues due to the increasing dominance of inverters within the grid \citep{4-990, 4-1302}. Additionally, climate change intensifies grid challenges by increasing extreme weather events that disrupt generation and damage infrastructure \citep{4-b1-66}. Given the notable impacts of network topology on both steady-state and dynamic system performance, power grid topology control has substantial importance for developing future power grids that are stronger, more resilient, and affordable, thereby enabling more effective handling of the emerging challenges. 
Meanwhile, network topology will be more controllable, driven by advancements in communications technology, power electronics, and reliable circuit breakers \citep{4-b1-68, 4-b1-67}. This provides the physical foundation for not only more frequent steady-state topology control but also more effective transient topology control. 

Motivated by the importance and potential of power grid topology control, this monograph aims to provide a comprehensive introduction by presenting the fundamentals and surveying both the authors’ recent work (e.g., \cite{4-995-ea}, \cite{4-1355}, \cite{4-1371}, \cite{4-1372}, \cite{4-1308}, and \cite{4-1582}) and relevant studies from researchers all over the world. 
The content spans from steady-state topology control to transient topology control,  and also addresses fundamental network topology constraints and the emerging issue of network topology transition.

\section{Organization of This Monograph}

The next seven chapters introduce in turn the preliminaries of this monograph, formulations of common network topology constraints, steady-state topology control, network topology transition, and transient topology control. The following outlines the contents of each chapter: 

Chapter \ref{chapter-2} introduces the mathematical notations used throughout the monograph, along with some basic concepts from graph theory. 

Chapter \ref{chapter-3} introduces the mathematical formulation of common network topology constraints in topology control problems. In particular, a recently developed electric flow-based network connectedness formulation, along with a formulation addressing network connectedness under contingencies, is discussed in more detail. 

Chapter \ref{chapter-4} discusses steady-state transmission topology control, covering associated fundamentals, a comprehensive state-of-the-art review of existing methodologies, and a recent advance in steady-state transmission topology control under multiple uncertainties. 

Chapter \ref{chapter-5} discusses steady-state distribution topology control, covering the basic model and a classical solution method, a comprehensive state-of-the-art review of existing solution methodologies, and a recently developed hybrid learning-heuristic solution paradigm. 

Chapter \ref{chapter-6} focuses on the network topology transition problem which arises when implementing steady-state transmission or distribution topology control. In particular, the bumpless topology transition method developed recently is discussed. 

Chapter \ref{chapter-7} introduces transient topology control, covering the basic form, i.e., Intentional Controlled Islanding (\ac{ici}) for transmission networks, a state-of-the-art review of various forms of transient topology control, and a recent advance in transient topology control designed for Multi-microgrid (\ac{mmg}) system stabilization. 

Chapter \ref{chapter-8} presents a summary of this monograph and outlines prospects for future directions in power grid topology control. 
\graphicspath{{chapter_2/Figs/}}
\chapter{Preliminaries}\label{chapter-2}

This chapter introduces the common notations used throughout this monograph and some basic concepts from graph theory relevant to the development of power grid topology control.

\section{List of Notations}

The notations that are commonly used by all the subsequent chapters in the monograph are defined in this chapter. The notations defined in each other chapter are applicable only within that chapter. Frequently used symbols are introduced at the beginning of the chapter, while others are defined as needed throughout the text.

\subsection{Mathematical Notations}\label{sec-2-1}

\setlist[description]{nosep, labelindent=0pt,style=multiline,leftmargin=1.85cm}

The following list summarizes the mathematical notations. Bold lowercase letters (e.g., $\bm{x}$) denote vectors, bold uppercase letters (e.g., $\bm{X}$) denote matrices, and calligraphic letters (e.g., $\mathcal{X}$) denote sets.

\begin{longtable}{|p{0.11\textwidth} p{0.81\textwidth}|}
    \toprule
    \multicolumn{2}{|c|}{Vector} \\
    \midrule
    \endfirsthead
    \endhead
    \bottomrule
    \endlastfoot

    $x_i$ &                   $i$-th entry or entry associated with the object $i$ of $\bm{x}$ \\
    $\bm{x}\A$ &              entry-wise absolute value of $\bm{x}$ \\
    $\bm{x}\T$ &              transpose of $\bm{x}$ \\
    $\bm{x}\D$ &              diagonal matrix with entries of $\bm{x}$ on its main diagonal \\
    $\bm{x}\PW$ &             entry-wise power of $\bm{x}$ \\
    $\bm{x}\HR$ &             entry-wise square root of $\bm{x}$  \\
    $\Vert \bm{x} \Vert_1$ &  1-norm of $\bm{x}$  \\
    $\Vert \bm{x} \Vert_2$ &  2-norm of $\bm{x}$ \\
    $\Vert \bm{x} \Vert_{2, \bm{w}}$ &  weighted 2-norm of $\bm{x}$ with $\bm{w}$ being the weight vector  \\
    $\Vert \bm{x} \Vert_{-2 \infty}$ &  second smallest absolute value of entries of $\bm{x}$ \\
    $\bm{0}$ &                          all-zeros vector of appropriate dimensions \\
    $\bm{0}_n$ &                        $n$-dimensional all-zeros vector \\
    $\bm{1}$ &                          all-ones vector of appropriate dimensions \\
    $\bm{1}_n$  &                       $n$-dimensional all-ones vector   
\end{longtable}

\vspace{-18pt}
\begin{longtable}{|p{0.11\textwidth} p{0.81\textwidth}|}
    \toprule
    \multicolumn{2}{|c|}{Matrix} \\
    \midrule
    \endfirsthead
    \endhead
    \bottomrule
    \endlastfoot

    $X_{ij}$ &                        $(i,j)$-th entry or entry associated with objects $(i, j)$ of $\bm{X}$ \\
    $\bm{X}\A$ &                      entry-wise absolute value of $\bm{X}$  \\
    $\bm{X}\T$ &                      transpose of $\bm{X}$ \\
    $\bm{X}\D$ &                      vector formed by the main diagonal of square matrix $\bm{X}$  \\
    $\Tr(\bm{X})$ &                   trace of $\bm{X}$ if $\bm{X}$ is a square matrix \\
    $\bm{I}_n$ &                      $n \times n$ identity matrix \\
    $\bm{J}$ &                        all-ones matrix of appropriate dimensions \\
    $\bm{J}_{n}$ &                    $n \times n$ all-ones matrix \\
    $\bm{J}_{n \times m}$ &           $n \times m$ all-ones matrix \\
    $\bm{X} \circ \bm{Y}$ &           element-wise product of $\bm{X}$ and $\bm{Y}$ of the same dimension  
\end{longtable}

\vspace{-18pt}
\begin{longtable}{|p{0.11\textwidth} p{0.81\textwidth}|}
    \toprule
    \multicolumn{2}{|c|}{Set} \\
    \midrule
    \endfirsthead
    \endhead
    \bottomrule
    \endlastfoot

    $\mathbb{R}$ &    set of real numbers \\
    $\mathbb{B}$ &    set of binary numbers \\
    $\mathbb{C}$ &    set of complex numbers \\
    $\mathbb{R}^n$ &  $n$-dimensional vector space over $\mathbb{R}$ \\
    $\mathbb{B}^n$ &  $n$-dimensional vector space over $\mathbb{B}$ \\
    $\mathbb{C}^n$ &  $n$-dimensional vector space over $\mathbb{C}$ \\
    $\mathbb{R}^{n \times m}$ &  $n \times m$ matrix space over $\mathbb{R}$ \\
    $\mathbb{B}^{n \times m}$ &  $n \times m$ matrix space over $\mathbb{B}$ \\ 
    $\mathbb{C}^{n \times m}$ &  $n \times m$ matrix space over $\mathbb{C}$ \\
    $\llbracket m, n \rrbracket$ &  integer set $\{ m, m+1, \cdots, n - 1, n \}$ \\
    $\mathcal{X} \triangle \mathcal{Y}$ &    symmetric difference of $\mathcal{X}$ and $\mathcal{Y}$ \\
    $\mathcal{Y} \backslash \mathcal{X}$ &   set difference of $\mathcal{Y}$ and $\mathcal{X}$ \\ 
    $|\mathcal{X}|$ &  cardinality of $\mathcal{X}$    
\end{longtable}

\subsection{Power grid notations} 

The following list collects notations associated with power grids.

\begin{longtable}{|p{0.11\textwidth} p{0.81\textwidth}|}
    \toprule
    \multicolumn{2}{|c|}{Component} \\
    \midrule
    \endfirsthead
    \endhead
    \bottomrule
    \endlastfoot

    $\mathcal{V}$ &          set of buses  \\
    $\mathcal{E}$ &          set of branches (including lines and transformers) \\ 
    $\mathcal{E}_{\rm us}$ & set of unswitchable branches for topology control \\
    $n_{\rm n}$ &            number of buses \\   
    $n_{\rm e}$ &            number of branches \\
    $n_{\rm g}$ &            number of generators \\
    $n_{\rm d}$ &            number of loads \\
    $n_{\rm gs}$ &           number of conventional \ac{sg}s \\
    $n_{\rm gv}$ &           number of \ac{vre}-based generators   
\end{longtable}

\vspace{-18pt}
\begin{longtable}{|p{0.11\textwidth} p{0.81\textwidth}|}
    \toprule
    \multicolumn{2}{|c|}{Network topology and admittance} \\
    \midrule
    \endfirsthead
    \endhead
    \bottomrule
    \endlastfoot

    $\bm{z}$ &   $\mathbb{B}^{n_{\rm e}}$, vector parameterizing the network topology, with entries 1 or 0 indicating branches are switched on or off \\
    $\bm{g}_{\rm b} \!+\! j \bm{b}_{\rm b}$ &   $\mathbb{C}^{n_{\rm e}}$, admittance of branches \\
\end{longtable}

\vspace{-18pt}
\begin{longtable}{|p{0.05\textwidth} p{0.87\textwidth}|}
    \toprule
    \multicolumn{2}{|c|}{State and power} \\
    \midrule
    \endfirsthead
    \endhead
    \bottomrule
    \endlastfoot

    $\bm{v}$ &  $\mathbb{R}^{n_{\rm n}}$, voltage magnitudes of buses \\
    $\bm{\theta}$ & $\mathbb{R}^{n_{\rm n}}$, voltage phase angles of buses \\ 
    $\bm{p}_{\rm g}$ & $\mathbb{R}^{n_{\rm g}}$, active power outputs of generators \\
    $\bm{q}_{\rm g}$ & $\mathbb{R}^{n_{\rm g}}$, reactive power outputs of generators \\
    $\bm{p}_{\rm d}$ & $\mathbb{R}^{n_{\rm d}}$, active powers drawn by loads \\
    $\bm{q}_{\rm d}$ & $\mathbb{R}^{n_{\rm d}}$, reactive powers drawn by loads \\
    $\bm{p}_{\rm fb}$ & $\mathbb{R}^{n_{\rm e}}$, active powers flowing into branches at starting buses \\
    $\bm{p}_{\rm tb}$ & $\mathbb{R}^{n_{\rm e}}$, active powers flowing into branches at end buses \\
    $\bm{p}_{\rm b} $ & $\mathbb{R}^{n_{\rm e}}$, branch active powers \\
    $\bm{q}_{\rm fb}$ & $\mathbb{R}^{n_{\rm e}}$, reactive power injections into branches at starting buses\\
    $\bm{q}_{\rm tb}$ & $\mathbb{R}^{n_{\rm e}}$, reactive power injections into branches at end buses \\
    $\bm{q}_{\rm b} $ & $\mathbb{R}^{n_{\rm e}}$, branch reactive powers  \\
    $\bm{s}_{\rm gv}$ & $\mathbb{R}^{n_{\rm gv}}$, apparent powers injected by \ac{vre}-based generators \\
    $\bm{s}_{\rm b} $ & $\mathbb{R}^{n_{\rm e}} $, apparent powers of branches \\
    $\bm{p}_{\rm gs}$ & $\mathbb{R}^{n_{\rm gs}}$, subvector of $\bm{p}_{\rm g}$ associated with conventional \ac{sg}s \\
    $\bm{p}_{\rm gv}$ & $\mathbb{R}^{n_{\rm gv}}$, subvector of $\bm{p}_{\rm g}$ associated with \ac{vre}-based generators \\
    $\bm{q}_{\rm gs}$ & $\mathbb{R}^{n_{\rm gs}}$, subvector of $\bm{q}_{\rm g}$ associated with conventional \ac{sg}s \\
    $\bm{q}_{\rm gv}$ & $\mathbb{R}^{n_{\rm gv}}$, subvector of $\bm{q}_{\rm g}$ associated with \ac{vre}-based generators \\
    $\bm{\phi}_{\rm gv}$ & $\mathbb{R}^{n_{\rm gv}}$, power factors of \ac{vre}-based generators  
\end{longtable}

\vspace{-18pt}
\begin{longtable}{|p{0.05\textwidth} p{0.87\textwidth}|}
    \toprule
    \multicolumn{2}{|c|}{Contingency} \\
    \midrule
    \endfirsthead
    \endhead
    \bottomrule
    \endlastfoot

    $\bm{o}$ &  $\mathbb{B}^{n_{\rm g} + n_{\rm e}}$, vector parameterizing $N-k$ contingencies, with entries of 1 and 0 indicating normal and failed component states \\
    $\bm{o}_{\rm g}$ &   $\mathbb{B}^{n_{\rm g}}$, subvector of $\bm{o}$ associated with generator \\
    $\bm{o}_{\rm b}$ &   $\mathbb{B}^{n_{\rm e}}$, subvector of $\bm{o}$ associated with branches \\
    $\bm{o}_{\rm gs}$ &  $\mathbb{B}^{n_{\rm gs}}$, subvector of $\bm{o}_{\rm g}$ associated with conventional \ac{sg}s \\
    $\bm{o}_{\rm gv}$ &  $\mathbb{B}^{n_{\rm gv}}$, subvector of $\bm{o}_{\rm g}$ associated with \ac{vre}-based generators \\
    $\Omega$ &  support of $\bm{o}$ \\
    $O$ &  probability distribution of $\bm{o}$ \\
    $\mathcal{O}$ &  ambiguity set of $O$ 
\end{longtable}

\vspace{-18pt}
\begin{longtable}{|p{0.955\textwidth}|}
    \toprule
    \multicolumn{1}{|c|}{Bound} \\
    \midrule
    \endfirsthead
    \endhead
    \bottomrule
    \endlastfoot

    For any variable vector $\bm{x}$, $\bm{x}\B$ and $\bm{x}\U$ denote the vectors of lower and upper bounds of $\bm{x}$ that is operationally feasible, respectively. For instance, $\bm{v}\B$ and $\bm{v}\U$ are the vectors of the minimal and maximal feasible values of bus voltage magnitudes, respectively. In addition,  $\bm{\theta}_{\rm b}\B \in \mathbb{R}^{n_{\rm e}}$ and $\bm{\theta}_{\rm b}\U \in \mathbb{R}^{n_{\rm e}}$ denote the vectors of lower and upper bounds of branch phase angle differences, respectively. 
\end{longtable}

\vspace{-18pt}
\begin{longtable}{|p{0.05\textwidth} p{0.87\textwidth}|}
    \toprule
    \multicolumn{2}{|c|}{Incidence matrix} \\
    \midrule
    \endfirsthead
    \endhead
    \bottomrule
    \endlastfoot

    \multicolumn{2}{|p{0.955\textwidth}|}{The incidence matrix between any two sets $\mathcal{X}$ and $\mathcal{Y}$ is a $|\mathcal{X}| \times |\mathcal{Y}|$ matrix showing the incidence relation between them, which has one row for each element of $\mathcal{X}$ and one column for each element of $\mathcal{Y}$. The entry of the incidence matrix is 1 if its associated elements in $\mathcal{X}$ and $\mathcal{Y}$ are related and 0 if they are not.} \\ \midrule
    $\bm{E}_{\rm g}$ &   $\mathbb{R}^{n_{\rm n} \times n_{\rm g}}$, incidence matrix between $\mathcal{V}$ and the set of generators  \\ 
    $\bm{E}_{\rm us}$ &  $\mathbb{R}^{|\mathcal{E}_{\rm us}| \times n_{\rm e}}$, incidence matrix between $\mathcal{E}_{\rm us}$ and $\mathcal{E}$
\end{longtable}

\section{Basic Concepts from Graph Theory}\label{sec-2-2}

\begin{figure}[h]
	\centering
	\includegraphics[width=0.98\linewidth]{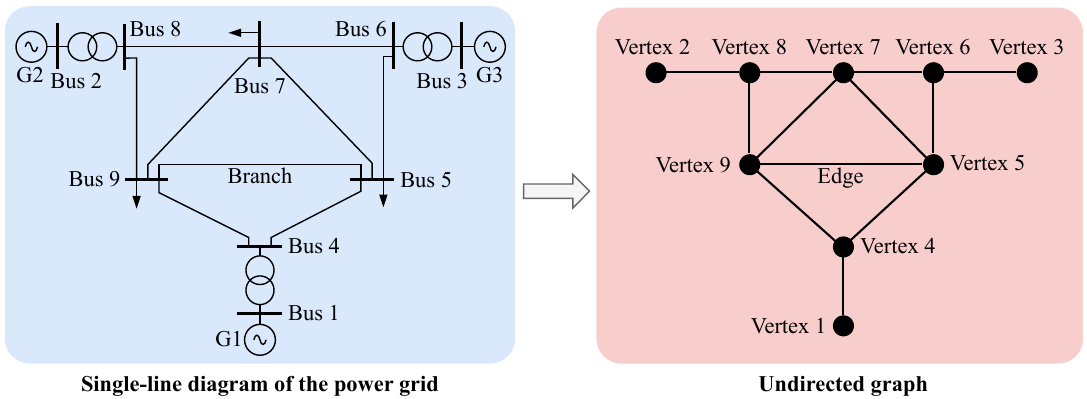}
	\caption{{Graph representation of the power grid.}}
    \label{fig-b1-2-1} 
\end{figure}

As illustrated in Fig. \ref{fig-b1-2-1}, the power grid can be topologically represented as an undirected graph $\mathcal{G}(\mathcal{V}, \mathcal{E})$, where $\mathcal{V}$ and $\mathcal{E}$ are the sets of buses (vertices) and branches (edges) that include lines and transformers, respectively. 
In this monograph, the terms \textit{bus} and \textit{vertex} will be used interchangeably, as will the terms \textit{branch} and \textit{edge}. Some basic concepts along with the associated notations are introduced as follows and also illustrated in Fig. \ref{fig-b1-2-2}:
\begin{figure}[h]
	\centering
	\includegraphics[width=0.98\linewidth]{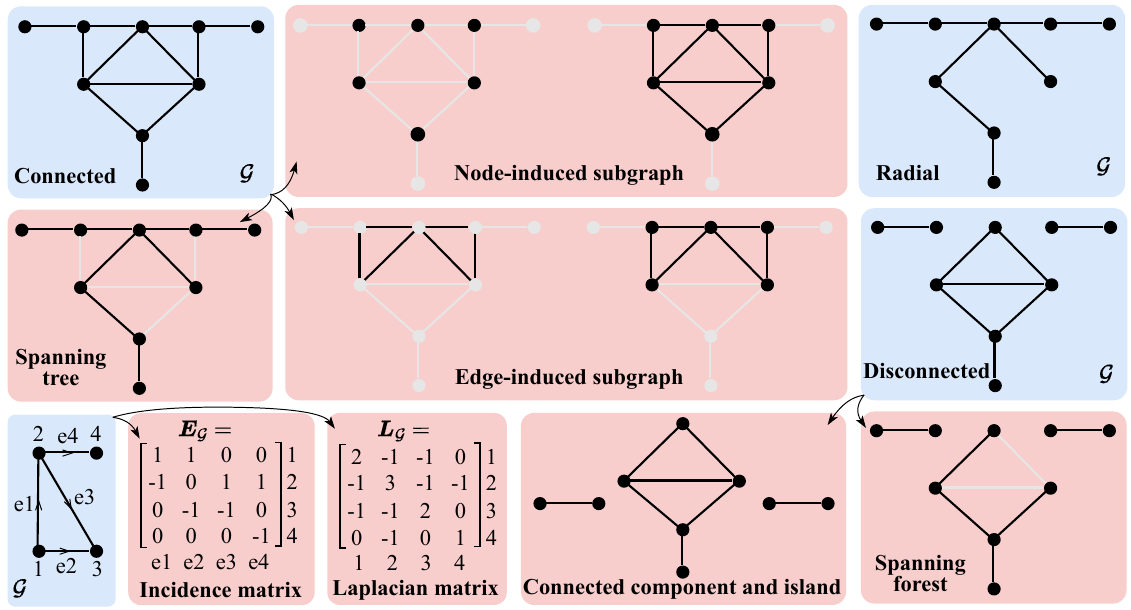}
	\caption{{Illustration of the basic concepts for graphs.}}
    \label{fig-b1-2-2} 
\end{figure}
\begin{itemize}
    \item \textit{Node-induced subgraph:} A node-induced subgraph of graph $\mathcal{G}(\mathcal{V}, \mathcal{E})$ refers to an arbitrary nonempty subset of $\mathcal{V}$ together with all the edges whose endpoints are both in this subset.
    \item \textit{Edge-induced subgraph:} An edge-induced subgraph of graph $\mathcal{G}(\mathcal{V}, \mathcal{E})$ is a subset of $\mathcal{E}$ together with all the nodes that are their endpoints.  
    \item \textit{Connectedness:} The undirected graph $\mathcal{G}(\mathcal{V}, \mathcal{E})$ is said to be connected if, for every pair of distinct vertices $u, v \in \mathcal{V}$, there exists a path in graph $\mathcal{G}$ that connects $u$ and $v$. If a graph is not connected, it is said to be disconnected. 
    \item \textit{Radiality:} The undirected graph $\mathcal{G}(\mathcal{V}, \mathcal{E})$ is said to be radial if it forms a tree; that is, it is connected and acyclic. 
    \item \textit{Connected component and island:} A connected component of graph $\mathcal{G}(\mathcal{V}, \mathcal{E})$ is a connected subgraph that is not part of any larger connected subgraph. Such a connected component is termed an \textit{island} in power system terminology. 
    \item \textit{Spanning tree:} A spanning tree of the undirected graph $\mathcal{G}(\mathcal{V}, \mathcal{E})$, where $\mathcal{G}$ is connected, is a minimal connected subgraph of graph $\mathcal{G}$ that includes all the vertices of the graph and forms a tree; that is, it is connected and acyclic.
    \item \textit{Spanning forest:} A spanning forest of the undirected graph $\mathcal{G}(\mathcal{V}, \mathcal{E})$ is a subgraph of graph $\mathcal{G}$ that  is a forest (i.e., contains no cycles), spans all vertices of graph $\mathcal{G}$, and in which each connected component is a spanning tree of a connected component of $\mathcal{G}$.
    \item \textit{Oriented incidence matrix:} With each edge of graph $\mathcal{G}(\mathcal{V}, \mathcal{E})$ assigned an arbitrary and fixed orientation, the oriented incidence matrix of $\mathcal{G}(\mathcal{V}, \mathcal{E})$, denoted as $\bm{E}_{\mathcal{G}} \in \mathbb{R}^{n_{\rm n} \times n_{\rm e}}$, is defined as 
    \begin{equation}
        {E}_{\mathcal{G}, ij} := 
        \left\{
            \begin{aligned}
                & 1   && \text{if bus $i$ is the starting bus of branch $j$} \\
                & -1  && \text{if bus $i$ is the end bus of branch $j$}  \\
                & 0   && \text{otherwise}
            \end{aligned}
        \right.
    \end{equation}
    with $i \in \mathcal{V}$ and $j \in \mathcal{E}$.
    \item \textit{Laplacian matrix:} The Laplacian matrix of graph $\mathcal{G}(\mathcal{V}, \mathcal{E})$, denoted as $\bm{L}_{\mathcal{G}} \in \mathbb{R}^{n_{\rm n} \times n_{\rm n}}$, is defined as
    \begin{equation}\!\!\!\!\!\!
        {L}_{\mathcal{G}, ij} \!:=\! 
        \left\{
            \begin{aligned}
                & \text{deg}(i)   && \text{if bus $i = j$} \\
                & -1       && \text{if bus $i \neq j$ and bus $i$ is adjacent to bus $j$}  \\
                & 0        && \text{otherwise}
            \end{aligned}
        \right.
    \end{equation}
    with $i, j \in \mathcal{V}$ and {deg}($i$) being the degree of bus $i$. 
    \item \textit{Other notations:} The matrices $\bm{E}_{\mathcal{G}, \rm f}$ and $\bm{E}_{\mathcal{G}, \rm t}$ are formed by replacing all entries equal to -1 and 1 in the matrix $\bm{E}_{\mathcal{G}}$ with 0, respectively.  The edge-induced subgraph of $\mathcal{G}(\mathcal{V}, \mathcal{E})$ associated with the set of edges in $\mathcal{E}$ whose corresponding entries in $\bm{z}$ are equal to 1 is denoted by $\mathcal{G}_{\bm{z}}$. Let $\bm{E}_{\mathcal{G}_{\bm{z}}}$ and $\bm{L}_{\mathcal{G}_{\bm{z}}}$ denote the oriented incidence matrix and Laplacian matrix of graph $\mathcal{G}_{\bm{z}}$, respectively. 
\end{itemize}

\graphicspath{{chapter_3/Figs/}}
\chapter{Network Topology Constraints}\label{chapter-3}

This chapter introduces the mathematical formulation of network topology constraints relevant to topology control problems. As illustrated in Fig. \ref{fig-b1-3-1}, four commonly encountered topological requirements are considered: connectedness, connectedness with contingencies, multi-island structure, and radiality.

\begin{figure}[h]
	\centering
	\includegraphics[width=1\linewidth]{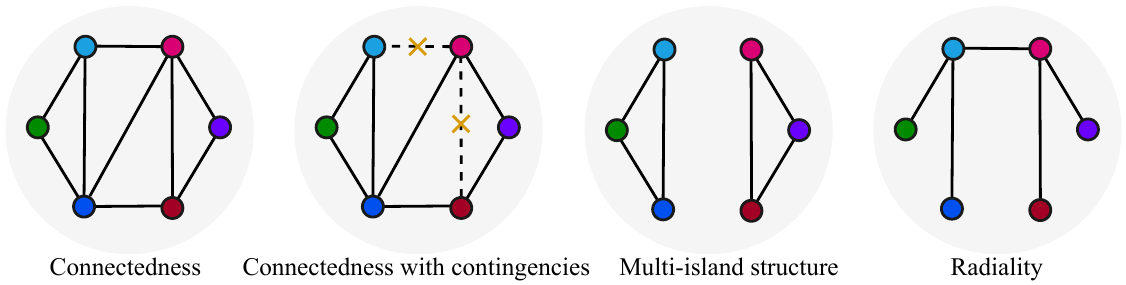}
	\caption{{Different requirements on network topology.}}
    \label{fig-b1-3-1} 
\end{figure}

\section{Connectedness of Network Topology}

Transmission networks are required to operate with a connected topology in the normal state to ensure system reliability and maintain synchronization across all generators. This section introduces three different formulations of network connectedness constraints. 

\subsection{Miller-Tucker-Zemli Formulation}

The Miller-Tucker-Zemli formulation of network topology connectedness for topology optimization problems is proposed by \cite{4-1283}. This formulation is based on the Miller-Tucker-Zemlin constraints, which were previously applied to connected facility location, traveling salesman, and steiner tree problems. 
Recall the underlying undirected graph of the power network, i.e., $\mathcal{G} = (\mathcal{V}, \mathcal{E})$. This graph can be transformed into a bidirected graph $(\mathcal{V}, \tilde{\mathcal{E}})$ by replacing the edge $(i, j) \in \mathcal{E}$ with two arcs (i.e., directed edges) $(i, j) \in \tilde{\mathcal{E}}$ and $(j, i) \in \tilde{\mathcal{E}}$. In this graph, an arbitrary vertex, denoted as $r \in \mathcal{V}$ is selected as the root vertex. 
Two binary variables $y_{ij}$ and $y_{ji}$ are further introduced for each edge $(i, j) \in \mathcal{E}$. 

A graph is connected if there is a path between any pair of vertices. Therefore, if arc $(j,k)$ is closed, at least one arc $(i,j)$ is connected to vertex $j$. 
This constraint can be formulated as:
\begin{equation}\label{eq-b1-3-1}
    \sum_{i \in \mathcal{V} \setminus \{k\}} y_{ij} \geq y_{jk}, \quad \forall j \in \mathcal{V} \setminus \{r\}, \forall k \in \mathcal{V}.
\end{equation}

Constraint (\ref{eq-b1-3-1}) is a necessary but not sufficient condition for connectedness. If there is a cycle in the graph, the constraint cannot guarantee the connectedness of the graph. To eliminate the cycles, a dummy voltage level variable $u_i$ for each vertex is introduced and the graph is viewed as a circuit with the root vertex grounded. According to the Kirchhoff’s loop law, there is no loop current if the loop does not have a source. This property can be used to eliminate the cycles. 

Assume the resistance of every line in the circuit network is one, and the current of every line is more than one per unit. If there is line $(i,j)$ with current flow from vertex $j$ to vertex $i$, $u_j - u_i \geq 1$. Otherwise, if there is no line between vertices $i$ and $j$, the voltage of the two vertices satisfies $u_j - u_i \geq -M + 1$, where $M$ is a large number. Therefore, the following constraints can be used to eliminate the cycles in graph $\mathcal{G}$:
\begin{subequations}\label{eq-b1-3-2}
    \begin{align}
        & My_{ij} + u_i \leq u_j + M - 1 \quad \forall (i, j) \in \tilde{\mathcal{E}}, i \in \mathcal{V}, j \in \mathcal{V} \\
        & u_i \geq 0 \quad \forall i \in \mathcal{V} \setminus \{r\} \\
        & u_r = 0.
    \end{align}
\end{subequations}

Moreover, it is practical to impose the following constraint on islanded buses that have both generators and loads:
\begin{equation}\label{eq-b1-3-3}
    \sum_{(i, j) \in \tilde{\mathcal{E}_i}} y_{ij} \geq 1, \quad \forall i \in \mathcal{V}_{\rm d}.
\end{equation}

Since the arcs that are used to deduce the connectedness constraints are bidirected, there are two variables $y_{ij}$ and $y_{ji}$ for a line $(i,j)$ while there is only one variable $z_{ij}$ for the line $(i,j)$ in topology control problems. Thus, it is necessary to add the following constraint to link the two variables:
\begin{equation}\label{eq-b1-3-4}
    z_{ij} = y_{ij} + y_{ji}, \quad \forall (i,j) \in \mathcal{E}.
\end{equation}
which also enforces only one direction $y_{ij}$ or $y_{ji}$ of the line $(i,j)$ that is selected. 
Finally, constraints (\ref{eq-b1-3-1})-(\ref{eq-b1-3-4}) form the Miller-Tucker-Zemli formulation of network topology connectedness. 

\subsection{Electrical-Flow-Based Formulation}

The electrical-flow-based formulation of network topology connectedness for topology optimization problems is proposed by \cite{4-995-ea}. The general idea is to construct an auxiliary electrical flow network with the same topology as $\mathcal{G}_{\bm{z}}$, and then the existence of its electrical flow is guaranteed if and only if graph $\mathcal{G}_{\bm{z}}$ is connected. Accordingly, determination of the connectedness of $\mathcal{G}_{\bm{z}}$ can be converted into that of the feasibility of vertex potential equations of the auxiliary electrical flow network. The latter condition can be directly embedded into topology optimization models.

First introduce an electrical flow network (allow multiple edges between two vertices) with the same topology as $\mathcal{G}_{\bm{z}}$, unit resistance for each edge, vertex potential denoted as $\bm{\vartheta} \in \mathbb{R}^{ n_{\rm n} }$, and vertex electrical flow injection given by $\bm{c}_{\rm nc} \in \mathbb{R}^{n_{\rm n}}$ \citep{4-p-0-1-1}. 
The Laplacian matrix links $\bm{\vartheta}$ and $\bm{c}_{\rm nc}$ as $\bm{L}_{\mathcal{G}_{\bm{z}}} \bm{\vartheta} = \bm{c}_{\rm nc} $. 
Suppose that graph $\mathcal{G}$ has $n_{\rm s}$ connected node-induced subgraphs and their node sets are collected by $\{ \mathcal{V}_i | i \in \llbracket 1, n_{\rm s} \rrbracket \}$. In such manner, $\{ \mathcal{V}_i | i \in \llbracket 1, n_{\rm s} \rrbracket \}$ includes the node sets of connected components of all possible $\mathcal{G}_{\bm{z}}$.

Without loss of generality, it is assumed that $\mathcal{V}_1 = \mathcal{V}$ is the node set of the node-induced subgraph equal to $\mathcal{G}$. 
Let $\bm{E}_{\rm ns} \in \mathbb{R}^{ n_{\rm s} \times n_{\rm n} } $ be the constant matrix satisfying $\forall  i \in  \llbracket 1, n_{\rm s} \rrbracket$, ${E}_{{\rm ns}, ij} = 1$ if $j \in \mathcal{V}_j$ and ${E}_{{\rm ns}, ij} = 0$ otherwise. Fig. \ref{fig-0-1-2} illustrates the set $\{ \mathcal{V}_i | i \in \llbracket 1, n_{\rm s} \rrbracket \}$ and corresponding matrix $\bm{E}_{\rm ns}$ for a 4-node graph. With the matrix $\bm{E}_{\rm ns}$, some balance properties of $\bm{c}_{\rm nc}$ are defined by Definition \ref{cond-0-1-1}.
 
\begin{figure}[h]
	\centering
	\includegraphics[width=0.95\linewidth]{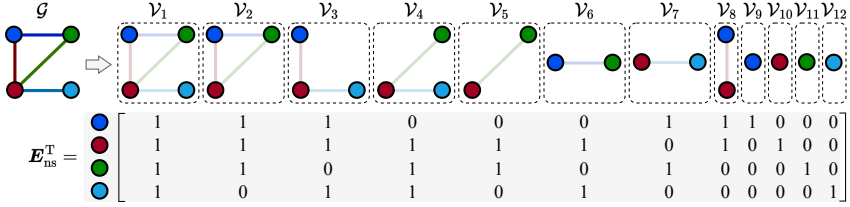}
	\caption{Illustration of set $\{ \mathcal{V}_i |i \in \llbracket 1, n_{\rm s} \rrbracket \}$ and matrix $\bm{E}_{\rm ns}$.}
	\label{fig-0-1-2}
\end{figure} 

\begin{definition}[Unique balance, unbalance]\label{cond-0-1-1}
  Let $\bm{b}_{\rm nc} = \bm{E}_{\rm ns} \bm{c}_{\rm nc} \in \mathbb{R}^{n_{\rm s}} $. Then $\bm{c}_{\rm nc}$ is multiply-balanced if $\forall i \in \llbracket 1, n_{\rm m} \rrbracket $, $b_{{\rm nc}, i} = 0$ and $\forall i \in \llbracket n_{\rm m} + 1, n_{\rm s} \rrbracket $, $b_{{\rm nc}, i} \neq 0$, with $1 \leq n_{\rm m} \leq n_{\rm s}$. If $n_{\rm m} = 1$, $\bm{c}_{\rm nc}$ is uniquely-balanced. Moreover, $\bm{c}_{\rm nc}$ is unbalanced if $\forall i \in \llbracket 1, n_{\rm s} \rrbracket$, $b_{{\rm nc}, i} \neq 0$.
\end{definition}

\begin{lemma}\label{lem-0-1-1}
  Given any uniquely-balanced vector $\bm{c}_{\rm nc}$, solutions of the vertex potential equation of the electrical flow network, i.e., $\bm{L}_{\mathcal{G}_{\bm{z}}} \bm{\vartheta} = \bm{c}_{\rm nc} $, exist iff graph $\mathcal{G}_{\bm{z}}$ is connected.
\end{lemma}

\begin{proof}
  Sufficiency. Suppose graph $\mathcal{G}_{\bm{z}}$ is connected. A necessary and sufficient condition for any solution(s) of $\bm{L}_{\mathcal{G}_{\bm{z}}} \bm{\vartheta} = \bm{c}_{\rm nc} $ to exist is that $\bm{W}_{\rm nc} \bm{c}_{\rm nc} = \bm{c}_{\rm nc}$ where $\bm{W}_{\rm nc} = \bm{L}_{\mathcal{G}_{\bm{z}}} \bm{L}_{\mathcal{G}_{\bm{z}}}^+$ with $\bm{L}_{\mathcal{G}_{\bm{z}}}^+$ being the Moore-Penrose pseudoinverse of $\bm{L}_{\mathcal{G}_{\bm{z}}}$. 

  Since $\mathcal{G}_{\bm{z}}$ is connected, by \cite[Lemma~3]{4-969}, $\bm{W}_{\rm nc}$ is given as $W_{{\rm nc}, ij} = - \frac{1}{ n_{\rm n} }$ for $i \neq j$ and $W_{{\rm nc}, ij} =  \sum_{ j'=1, j' \neq j }^{n_{\rm n}} W_{{\rm nc}, ij'}  = \frac{n_{\rm n} - 1}{ n_{\rm n} }$ for $i = j$. Recall that $\bm{c}_{\rm nc}$ is uniquely-balanced, thus we have $\bm{1}_{n_{\rm n}}\T \bm{c}_{\rm nc} = 0$. Therefore, $\forall i \in \llbracket 1, n_{\rm n} \rrbracket$, we have 
  \begin{equation}
      \begin{aligned}
        [\bm{W}_{\rm nc} \bm{c}_{\rm nc}]_i & = \sum_{j=1}^{n_{\rm n}} W_{{\rm nc}, ij} c_{{\rm nc}, j} = W_{{\rm nc}, ii} c_{{\rm nc}, i} + \sum_{j=1, j\neq i}^{n_{\rm n}} W_{{\rm nc},ij} c_{{\rm nc}, j} \\
        & =  \frac{n_{\rm n} - 1}{ n_{\rm n} } c_{{\rm nc}, i} + \frac{1}{ n_{\rm n} } c_{{\rm nc}, i} = c_{{\rm nc}, i},
      \end{aligned}
  \end{equation}
  i.e., $\bm{W}_{{\rm nc}} \bm{c}_{\rm nc} = \bm{c}_{\rm nc}$. Thus, solutions of $\bm{L}_{\mathcal{G}_{\bm{z}}} \bm{\vartheta} = \bm{c}_{\rm nc} $ exist.

  Necessity. Suppose that graph $\mathcal{G}_{\bm{z}}$ is disconnected and consists of $n_{\rm cc}$ connected components, each of size $n_k$, $\forall k \in \llbracket 1, n_{\rm cc} \rrbracket$. By rearranging the nodes of each connected component, we have $\bm{L}_{\mathcal{G}_{\bm{z}}} = \diag\{ \bm{L}_{\mathcal{G}_{\bm{z}}}^k | k \in \llbracket 1, n_{\rm cc} \rrbracket \} $ with $\bm{L}_{\mathcal{G}_{\bm{z}}}^k$ being the Laplacian matrix of the $k$-th connected component of $\mathcal{G}_{\bm{z}}$. Let $\bm{\vartheta}^k$ and $\bm{c}_{\rm nc}^k$ be subvectors of $\bm{\vartheta}$ and $\bm{c}_{\rm nc}$ corresponding to the $k$-th component, respectively; and $\bm{W}^k_{\rm nc} = \bm{L}_{\mathcal{G}_{\bm{z}}}^k\bm{L}_{\mathcal{G}_{\bm{z}}}^{k +} $. Consider the existence of solutions of $\bm{L}_{\mathcal{G}_{\bm{z}}}^k \bm{\vartheta}^k = \bm{c}_{\rm nc}^k$. Following the proof of sufficiency, but since $\bm{1}_{n_k}\T \bm{c}_{\rm nc}^k \neq 0 $ by its unique balance, it is trivial that $\bm{W}^k_{\rm nc} \bm{c}_{\rm nc}^k \neq\bm{c}_{\rm nc}^k$. Thus, $\forall k \llbracket 1, n_{\rm cc} \rrbracket$, solutions of $\bm{L}_{\mathcal{G}_{\bm{z}}}^k \bm{\vartheta}^k = \bm{c}_{\rm nc}^k$ do not exist and the same for $\bm{L}_{\mathcal{G}_{\bm{z}}} \bm{\vartheta} = \bm{c}_{\rm nc} $. 
\end{proof}

\begin{theorem}\label{theo-0-1-1}
  Graph $\mathcal{G}_{\bm{z}}$ is connected iff the following set of constraints is feasible: 
  \begin{subequations}\label{eq-0-1-1} 
    \begin{align}
      & - M (\bm{1} - \bm{z} ) \leq \bm{E}_{\mathcal{G}}\T \bm{\vartheta} - \bm{\rho}  \leq M (\bm{1} - \bm{z} )  \label{eq-0-1-1:0} \\  
      & - M \bm{z} \leq \bm{\rho} \leq M \bm{z}   \label{eq-0-1-1:2} \\ 
      & \bm{E}_{\mathcal{G}} \bm{\rho} = \bm{c}_{\rm nc}  \label{eq-0-1-1:3}
    \end{align}
  \end{subequations}
  where $\bm{\rho} \in \mathbb{R}^{ n_{\rm e} }$ is a vector of auxiliary variables, and $\bm{c}_{\rm nc}$ is any uniquely-balanced vector.
\end{theorem}

\begin{proof}
  The equivalence between $\bm{L}_{\mathcal{G}_{\bm{z}}} \bm{\vartheta} = \bm{c}_{\rm nc} $ and (\ref{eq-0-1-1}) is first proved. It holds that $\bm{L}_{\mathcal{G}_{\bm{z}}} \bm{\vartheta} = \bm{c}_{\rm nc} $ $\Leftrightarrow$ $\bm{E}_{\mathcal{G}_{\bm{z}}} \bm{E}_{\mathcal{G}_{\bm{z}}}\T \bm{\vartheta} = \bm{c}_{\rm nc} $ $\Leftrightarrow$  $\{ \bm{E}_{\mathcal{G}_{\bm{z}}}\T \bm{\vartheta} = \bm{\rho}, \bm{E}_{\mathcal{G}_{\bm{z}}} \bm{\rho} = \bm{c}_{\rm nc} \}$ $\Leftrightarrow$  
  \begin{subequations}\label{eq-0-1-2} 
    \begin{align}
      & \bm{z} \circ ( \bm{E}_{\mathcal{G}}\T \bm{\vartheta}) = \bm{\rho} \label{eq-0-1-2:0} \\ 
      & \bm{E}_{\mathcal{G}} (\bm{z} \circ \bm{\rho} ) = \bm{c}_{\rm nc}. \label{eq-0-1-2:1}
    \end{align}
  \end{subequations}
  Regarding the left-hand side of (\ref{eq-0-1-2:0}) as bilinear terms of $\bm{z}$ and $\bm{E}_{\mathcal{G}}\T \bm{\vartheta}$, (\ref{eq-0-1-1:0}) and (\ref{eq-0-1-1:2}) are just the McCormick envelopes of (\ref{eq-0-1-2:0}), which are exact given $\bm{z} \in \mathbb{B}^{n_{\rm e}}$ \citep{4-2200}. For the bilinear terms $\bm{z} \circ \bm{\rho} $ in (\ref{eq-0-1-2:1}), it holds that $\bm{z} \circ \bm{\rho} = \bm{\rho}$ deriving from their exact McCormick envelopes with upper and lower bounds of $\bm{\rho}$ given by (\ref{eq-0-1-1:2}). Thus, (\ref{eq-0-1-2:1}) and (\ref{eq-0-1-1:3}) are equivalent if (\ref{eq-0-1-1:2}) holds. Then the equivalence between $\bm{L}_{\mathcal{G}_{\bm{z}}} \bm{\vartheta} = \bm{c}_{\rm nc} $ and (\ref{eq-0-1-1}) can be concluded, which together with Lemma \ref{lem-0-1-1} gives Theorem \ref{theo-0-1-1}.
\end{proof}

Constraints (\ref{eq-0-1-1}) forms the electrical-flow-based formulation of network connectedness. 
According to Theorem \ref{theo-0-1-1}, by adding them with auxiliary variables $\bm{\vartheta} \in \mathbb{R}^{n_{\rm n}}$ and $\bm{\rho} \in \mathbb{R}^{n_{\rm e}} $ to the topology optimization models, network connectedness can be guaranteed. 

\subsection{Network-Flow-Based Formulation}

Based on the electrical-flow-based formulation (\ref{eq-0-1-1}), by taking a special case of vector $\bm{c}_{\rm nc}$ and ignoring constraint (\ref{eq-0-1-1:0}), the network-flow-based formulation of network connectedness is obtained \citep{4-1130}:
\begin{subequations}\label{eq-b1-3-0} 
    \begin{align}
        & \sum_{(j, i) \in \mathcal{E}} \rho_{ji} - \sum_{(i, j) \in \mathcal{E}} \rho_{ij} = c_{{\rm nc}, i}  \quad i \in \mathcal{V} \\
        & |\rho_{ij}| \leq   (n_{\rm n} - 1) z_{ij} \quad \forall (i,j) \in \mathcal{E} \\ 
        &  c_{{\rm nc}, i} = 1 \quad \forall i \in \mathcal{V} \backslash \{i^*\} \\
        & c_{{\rm nc}, i^*} = 1 - n_{\rm n}  
    \end{align}
\end{subequations}
where $i^*$ is an arbitrarily selected bus from $\mathcal{V}$; the entry of the vector $\bm{c}_{\rm nc}$ corresponding to bus $i$ is assigned the value $1 - n_{\rm n}$, while all other entries are assigned the value 1.

\section{Connectedness of Network Topology with Contingencies}\label{sec-b1-3-2}

For security-constraint topology optimization problems that account for system contingencies, the issue of network connectedness is further complicated by branch contingencies and corrective topology control actions. 
This section focuses on the criteria and formulation of network connectedness in such scenarios \citep{4-1355}. 

\subsection{Criteria of Network Connectedness with Contingencies}

Accounting for contingencies, the issue of network connectedness arises from three operating states, i.e., the normal state, the post-contingency state, and the post-control state. Fig. \ref{fig-0-2-1} illustrates the different topologies under these three states. Let $\tilde{\bm{z}} \in \mathbb{B}^{n_{\rm e}}$ be the counterpart of $\bm{z}$ for the post-contingency topology. 
Recall that $\bar{\bm{z}} \in \mathbb{B}^{n_{\rm e}}$ is the counterpart of $\bm{z}$ for the post-control topology. 
Hereafter, the vectors $\bm{z}$, $\tilde{\bm{z}}$ and $\bar{\bm{z}}$ are also used to refer to the corresponding topologies for brevity. 

\begin{figure}[h]
	\centering
	\includegraphics[width=0.7\linewidth]{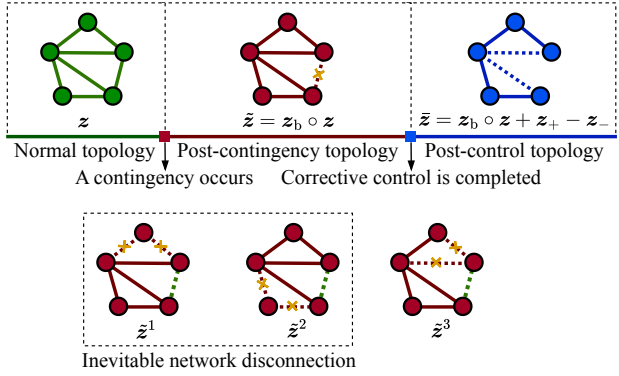} 
	\caption{Different topologies and inevitable network disconnection.}
	\label{fig-0-2-1}
\end{figure}

In addition, the concept \textit{inevitable network disconnection} is introduced and explained using Fig. \ref{fig-0-2-1}. Assume that the green graph is the network with all branches switched on, and the dashed green line shown separately is switched off under the normal state. Three different contingencies result in disconnected topologies $\tilde{\bm{z}}^1$, $\tilde{\bm{z}}^2$ and $\tilde{\bm{z}}^3$. Topologies $\tilde{\bm{z}}^1$ and $\tilde{\bm{z}}^2$ are considered \textit{inevitably disconnected}, as these network disconnections will occur even if the green line is switched on. This type of network disconnection is referred to as \textit{inevitable network disconnection}. In contrast, $\tilde{\bm{z}}^3$ is not inevitably disconnected, as $\tilde{\bm{z}}^3$ will be connected if the green line is switched on. 

The network connectedness constraints in the last section can be used to ensure network connectedness of $\bm{z}$, while for $\tilde{\bm{z}}$ and $\bar{\bm{z}}$, these constraints are inapplicable due to the following practical situations: 
\begin{itemize}
    \item For many transmission networks, an $N-1$ contingency on the network with all branches closed may cause network disconnection. 
    It is reasonable to specially treat such inevitable network disconnection, as it is caused by inherent local weak connection (e.g. between states in Australia where just have one interconnector) instead of line switching. 
    For topology optimization models with predefined contingency scenarios, identifying contingencies that inevitably disconnect the network can preserve the applicability of the previous connectedness constraints. However, robust and distributionally robust models are not scenario-based, rendering such pre-identification infeasible.
    \item In general, it is impossible to ensure network connectedness of $\tilde{\bm{z}}$ for all $N-k$ contingencies as $k$ get larger. Network connectedness can reasonably be ensured only when the number of fault components is below a threshold. Moreover, the network connectedness of $\bar{\bm{z}}$ cannot always be ensured, as the post-contingency topology may be disconnected. Therefore, the requirements for the network connectedness of $\bar{\bm{z}}$ should depend on the connectedness of the post-contingency topology.
\end{itemize}

Then, define \textit{$k_{\rm b}$-branch contingencies} as the contingencies whose number of fault branches is not more than $k_{\rm b}$ with $k_{\rm b} \in \mathbb{Z}^+$. 
The above analysis leads to the following two criteria for network connectedness with contingencies: 
\begin{criterion}\label{crit-0-2-1}
    $\tilde{\bm{z}}$ is connected for all $k_{\rm b}$-branch contingencies, ignoring inevitable network disconnections.
\end{criterion}
\begin{criterion}\label{crit-0-2-2}
    Corrective line switching should not further disconnect the network when $\tilde{\bm{z}}$ is connected or disconnected with only inevitable network disconnection.
\end{criterion}
\begin{remark}
    Criterion \ref{crit-0-2-2} indicates that when $\tilde{\bm{z}}$ is connected, $\bar{\bm{z}}$ should be connected; when $\tilde{\bm{z}}$ is disconnected with only inevitable network disconnection, $\bar{\bm{z}}$ should not create new network disconnection; and otherwise, $\bar{\bm{z}}$ can be connected or not. The principle of the above criteria is to preserve network connectedness after contingencies and corrective line switching as much as possible within reasonable limits.
\end{remark}

\subsection{Formulation of Network Connectedness with Contingencies}\label{sec-thesis-nc-3-2}

The formulation of the above two criteria can be derived based on the network connectedness constraints (\ref{eq-0-1-1}). To begin with, consider a parameterized region $\mathcal{C}(\phi, \bm{c}_{\rm sc}, \bm{d}) = $
\begin{equation}\label{eq-0-2-4} 
      \left\{ \!\!
     \bm{u} \!\in\! \mathbb{B}^{n_{\rm e}} \! \left|
    \begin{aligned} 
        & M (\bm{u}  - \bm{1} -  \phi \bm{1}) \leq \bm{E}_{\mathcal{G}}\T \bm{\vartheta} - \bm{\rho}  \leq M (\bm{1} - \bm{u} +   \phi \bm{1})  \\  
        & - M (\bm{u} + \phi \bm{1}) \leq \bm{\rho} \leq M (\bm{u} + \phi \bm{1})  \\
        & - \phi M \bm{1} \leq \bm{E}_{\mathcal{G}} \bm{\rho} - (\bm{c}_{\rm sc} + \bm{d}) \leq \phi M \bm{1}, \bm{\vartheta} \in \mathbb{R}^{n_{\rm n}}, \bm{\rho} \in \mathbb{R}^{n_{\rm e}}
    \end{aligned} 
    \right.
    \!\!\right\}
\end{equation}
with $\phi \in \mathbb{R}$, $\bm{d} \in \mathbb{R}^{n_{\rm n}}$, and $\bm{c}_{\rm sc}$ being an $n_{\rm n}$-dimensional constant uniquely-balanced vector (see Definition \ref{cond-0-1-1}). When $\phi=0$ and $\bm{d} = \bm{0}$, constraints in (\ref{eq-0-2-4}) reduce to constraints (\ref{eq-0-1-1}), such that $\mathcal{C}$ is the region of $\bm{u}$ whose associated topology is connected. When $\phi\geq1$ and $\bm{d} = \bm{0}$, constraints in (\ref{eq-0-2-4}) are all invalid. 
The general idea to formulate the two criteria is to design $\bm{c}_{\rm sc}$ particularly rather than just being uniquely-balanced, such that $\bm{d}$, satisfying certain conditions, can be used to identify and ensure more complex cases of network connectedness.

Let $\mathcal{W}(k_{\rm b}) = \{ (\mathcal{L}_i, \mathcal{N}_i)| i=1,2,...,n_{\rm w} \}$ be the set of all $n_{\rm w}$ pairs of ($\mathcal{L}_i, \mathcal{N}_i$) that satisfy the following three conditions: 
\begin{itemize}
    \item $\mathcal{L}_i \subset \mathcal{E}$, $\mathcal{N}_i \subset \mathcal{V}$, and $|\mathcal{L}_i| \leq k_{\rm b}$.
    \item Removing ${\mathcal{L}}_i$ from $\mathcal{G}(\mathcal{V}, \mathcal{E})$ causes network disconnection with $\mathcal{N}_i$ being the set of buses not in the main connected component. 
    \item There does not exist a nonempty subset $\mathcal{L}' \subset \mathcal{L}_i$ such that removing $\mathcal{L}_i \backslash \mathcal{L}'$ from $\mathcal{G}(\mathcal{V}, \mathcal{E})$ causes network disconnection with $\mathcal{N}_i$ being the set of buses not in the main connected component. 
\end{itemize}
For short, network disconnections in the second condition are referred to as $\mathcal{W}(k_{\rm b})$ \textit{network disconnections} 
and the associated graphs are $\mathcal{W}(k_{\rm b})$-\textit{disconnected}. 
Thereby, the set of $\mathcal{W}(k_{\rm b})$ network disconnections and that of all inevitable network disconnections in Criterion \ref{crit-0-2-1} are identical. Note that lines in $\cup_{i=1}^{n_{\rm w}} \mathcal{L}_i$ are assumed not to be switched off in $\bm{z}$ since such switching facilitates inevitable network disconnection.

Next, introduce matrix $\bm{E}_{\rm w} \in \mathbb{R}^{n_{\rm w} \times n_{\rm n} }$ satisfying that $\forall i = 1,2,...,n_{\rm w}$, ${E}_{{\rm w}, ij} = 1$ if $j \in \mathcal{N}_i$ and ${E}_{{\rm w}, ij} = 0$ otherwise; and matrix $\bm{H}_{\rm w}$ formed by deleting the common rows between $\bm{E}_{\rm ns}$ and $\bm{E}_{\rm w}$ and between $\bm{E}_{\rm ns}$ and $\bm{J} - \bm{E}_{\rm w}$, and the all-ones row, from $\bm{E}_{\rm ns}$. Denote by $n_{\rm u}$ the maximal number of connected components among all $\mathcal{W}(k_{\rm b})$ network disconnections, and $n_{\rm rh}$ the row dimension of $\bm{H}_{\rm w}$. Then a balance property of $\bm{c}_{\rm sc}$ is defined by Definition \ref{def-0-2-1}. 

\begin{definition}[$\mathcal{W}(k_{\rm b})$-unique balance]\label{def-0-2-1}
    Let $\bm{b}_{\rm sc} = \bm{E}_{\rm ns} \bm{c}_{\rm sc} \in \mathbb{R}^{n_{\rm s}}$. Then $\bm{c}_{\rm sc}$ is $\mathcal{W}(k_{\rm b})$-uniquely balanced if ${b}_{\rm sc, 1}=0$ and ${b}_{{\rm sc}, i} \neq0$ for $i= 2,3,..., n_{\rm s}$, namely that $\bm{c}_{\rm sc}$ is uniquely balanced; and $\exists r \!\in \mathbb{R}$, \text{s.t.}, $\Vert \bm{E}_{\rm w} \bm{c}_{\rm sc} \Vert_{\infty} \!\leq r$ and $\Vert \bm{H}_{\rm w} \bm{c}_{\rm sc} \Vert_{-\infty} \geq n_{\rm u} r$. 
\end{definition}

Unlike uniquely balanced $\bm{c}_{\rm sc}$ which has a trivial special case, it may be far from easy to find such a special case for $\mathcal{W}(k_{\rm b})$-uniquely balanced $\bm{c}_{\rm sc}$. An ad hoc approach is to solve the following Mixed-Integer Linear Programming (\ac{milp}) model:
\begin{subequations}\label{eq-0-2-c-milp}
    \begin{align}
        &\!\!\!\!\!\!\!\!\!\!\!\!\!\!\!\!\!\!\!\!\!\!\!\!\!\!\!\!\!\!\!\! \min_{\bm{c}_{\rm sc} \in \mathbb{R}^{n_{\rm n}}, \bm{b}_{\rm sc} \in \mathbb{R}^{n_{\rm s}}, \bm{\beta} \in \mathbb{B}^{n_{\rm s} - 1} \bm{\gamma} \in \mathbb{B}^{n_{\rm rh}}, r \in \mathbb{R} } ~~ r  \\
        \rm{s.t.} ~& \bm{b}_{\rm sc} = \bm{E}_{\rm ns} \bm{c}_{\rm sc} \\
        & b_{\rm sc, 1} = 0  \\
        &  -r \bm{1}  \leq \bm{E}_{\rm w} \bm{c}_{\rm sc}  \leq r \bm{1} \\ 
        & \epsilon -  M {\beta}_{i-1} \leq b_{{\rm sc}, i}  \leq -  \epsilon + M (1 - {\beta}_{i-1}), i \in \llbracket 2, n_{\rm s} \rrbracket \!\!\!\!\!\!\!\!\!\!\! \\
        &  n_{\rm u} r \bm{1} -  M \bm{\gamma} \leq  \bm{H}_{\rm w} \bm{c}_{\rm sc}  \leq - n_{\rm u} r \bm{1}  + M (1 -\bm{\gamma}) 
    \end{align}
\end{subequations}
where $\bm{\beta}$, $\bm{\gamma}$ and $r$ are auxiliary variables, and $\epsilon$ is a positive constant. This \ac{milp} model is derived from Definition \ref{def-0-2-1}, with the objective function set as $r$ that can be replaced by any objective function with a finite lower bound, and constraints in the form $|\cdot| > 0$ being linearized. Any feasible solution of (\ref{eq-0-2-c-milp}) yields a $\mathcal{W}(k_{\rm b})$-uniquely balanced $\bm{c}_{\rm sc}$.

With $\bm{c}_{\rm sc}$ being an $n_{\rm n}$-dimensional constant $\mathcal{W}(k_{\rm b})$-uniquely balanced vector, consider the following $\tilde{\bm{z}}$-parameterized LP model:
\begin{equation}\label{eq-0-2-5} 
         \min_{\bm{d}_+, \bm{d}_- \in \mathbb{R}^{n_{\rm n}}}  \bm{1}\T (\bm{d}_+ + \bm{d}_- )  ~\text{s.t.} ~  \mathcal{C}(0, \bm{c}_{\rm sc}, \bm{d}_+ - \bm{d}_-)  \owns \tilde{\bm{z}}, \bm{d}_+ \geq \bm{0}, \bm{d}_- \geq \bm{0} 
\end{equation}
Denote by $\mathcal{D}(\tilde{\bm{z}})$ the set of all optima of (\ref{eq-0-2-5}). Then network connectedness of $\tilde{\bm{z}}$ and the optima of (\ref{eq-0-2-5}) are associated by Proposition \ref{prop-0-2-1}.

\begin{proposition}\label{prop-0-2-1}
    $\forall (\bm{d}_+^*, \bm{d}_-^*) \in \mathcal{D}(\tilde{\bm{z}})$, (\uppercase\expandafter{\romannumeral1}) $\bm{1}\T (\bm{d}_+^* + \bm{d}_-^*) = 0$ iff $\tilde{\bm{z}}$ is connected; 
    (\uppercase\expandafter{\romannumeral2}) $\bm{1}\T (\bm{d}_+^* + \bm{d}_-^*) \geq 2 \Vert \bm{b}_{\rm sc} \Vert_{-2 \infty} > 0$ iff $\tilde{\bm{z}}$ is disconnected; 
    and (\uppercase\expandafter{\romannumeral3}) $\bm{1}\T (\bm{d}_+^* + \bm{d}_-^*) \leq n_{\rm u} r$ iff $\tilde{\bm{z}}$ is connected or $\mathcal{W}(k_{\rm b})$-disconnected. 
\end{proposition}

\begin{proof}   
    {(\uppercase\expandafter{\romannumeral1})}. (Sufficiency) By $\bm{d}_+^* \geq 0$, $\bm{d}_-^* \geq 0$, and $\bm{1}\T (\bm{d}_+^* + \bm{d}_-^*) = 0$, we have $\bm{d}_+^* = \bm{d}_-^* = \bm{0}$, and thus $\tilde{\bm{z}} \in \mathcal{C}(0, \bm{c}_{\rm sc}, \bm{0})$, i.e., $\tilde{\bm{z}}$ is connected. (Necessity) Ignoring the constraint $\tilde{\bm{z}} \in \mathcal{C}(\cdot)$ in (\ref{eq-0-2-5}), the unique optimum of (\ref{eq-0-2-5}) is $(\bm{d}_+^*, \bm{d}_-^*) = (\bm{0}, \bm{0})$. Since program (\ref{eq-0-2-5}) with $\tilde{\bm{z}}$ being connected is also feasible at $(\bm{d}_+^*, \bm{d}_-^*)$, $(\bm{d}_+^*, \bm{d}_-^*)$ is also the unique optimum of this program. Therefore, if $\tilde{\bm{z}}$ is connected, $\bm{1}\T (\bm{d}_+^* + \bm{d}_-^*) = 0$. 

    {(\uppercase\expandafter{\romannumeral2})}. (Sufficiency) By {(\uppercase\expandafter{\romannumeral1})}, $\bm{1}\T (\bm{d}_+^* + \bm{d}_-^*) \neq 0$ yields that $\tilde{\bm{z}}$ is disconnected. (Necessity) By the proof of Theorem \ref{theo-0-1-1}, the constraints in $\tilde{\bm{z}} \in \mathcal{C}(0, \bm{c}_{\rm sc}, \bm{d}_+ - \bm{d}_-)$ are equivalent to 
    \begin{equation}\label{eq-0-2-6}
        \bm{L}_{\mathcal{G}_{\tilde{\bm{z}}}} \bm{\vartheta} = \bm{c}_{\rm sc} + \bm{d}_+ - \bm{d}_-
    \end{equation}
    where $\bm{L}_{\mathcal{G}_{\tilde{\bm{z}}}}$ is the Laplacian matrix of graph $\mathcal{G}_{\tilde{\bm{z}}} (\mathcal{V}, \tilde{\mathcal{E}})$ with $\tilde{\mathcal{E}}$ being the set of branches in $\tilde{\bm{z}}$. Assume that graph $\mathcal{G}_{\tilde{\bm{z}}}(\mathcal{V}, \tilde{\mathcal{E}})$ contains $n_{\rm c}$ connected components, and let $[\bm{L}_{\mathcal{G}_{\tilde{\bm{z}}}}]_i$ be the principle submatrix of $\bm{L}_{\mathcal{G}_{\tilde{\bm{z}}}}$ associated with the $i$-th connected component, $[\bm{\vartheta}]_i$ be the subvector of $\bm{\vartheta}$ associated with the $i$-th connected component, and $[\bm{c}_{\rm sc}]_i$, $[\bm{d}_+]_i$, and $[\bm{d}_-]_i$ are analogous. 
    
    Then (\ref{eq-0-2-6}) can be reformulated as:
    \begin{equation}\label{eq-0-2-7}
         \begin{bmatrix}
            [\bm{L}_{\mathcal{G}_{\tilde{\bm{z}}}}]_1 &\cdots&\cdots& 0 \\
            \vdots& [\bm{L}_{\mathcal{G}_{\tilde{\bm{z}}}}]_2 && \vdots  \\
            \vdots && \ddots & \vdots \\
            0 &\cdots&\cdots& [\bm{L}_{\mathcal{G}_{\tilde{\bm{z}}}}]_{n_{\rm c}}
        \end{bmatrix} 
        \begin{bmatrix}
            [\bm{\vartheta}]_{1} \\
            [\bm{\vartheta}]_{2} \\
            \vdots\\
            [\bm{\vartheta}]_{n_{\rm c}}  
        \end{bmatrix}
        = 
        \begin{bmatrix}
            [\bm{c}_{\rm sc}]_1 + [\bm{d}_+]_1 - [\bm{d}_-]_1 \\
            [\bm{c}_{\rm sc}]_2 + [\bm{d}_+]_2 - [\bm{d}_-]_2 \\
            \vdots \\
            [\bm{c}_{\rm sc}]_{n_{\rm c}} + [\bm{d}_+]_{n_{\rm c}} - [\bm{d}_-]_{n_{\rm c}} 
        \end{bmatrix}
    \end{equation}
    which together with the definition of Laplacian matrices, gives
    \begin{equation}\label{eq-0-2-8}
        \bm{1}\T ( [\bm{c}_{\rm sc}]_i + [\bm{d}_+]_i - [\bm{d}_-]_i ) = 0, i = 1, 2, ..., n_{\rm c}
    \end{equation}
    Furthermore, for any connected component of $\mathcal{G}_{\tilde{\bm{z}}}(\mathcal{V}, \tilde{\mathcal{E}})$, there exists a node-induced subgraph of $\mathcal{G}(\mathcal{V}, \mathcal{E})$ whose node set is equal to that of the connected component and unequal to $\mathcal{V}$. By the fact that $\bm{c}_{\rm sc}$ is uniquely balanced, we have $\forall i=1,2,...,n_{\rm c}$, $\bm{1}\T [\bm{c}]_i = {b}_{{\rm sc}, j(i)}$ with $j(i) \in \llbracket 2, n_{\rm s} \rrbracket$, and thus 
    \begin{equation}\label{eq-0-2-9}
        \bm{1}\T ([\bm{d}_+]_i - [\bm{d}_-]_i ) = {b}_{{\rm sc}, j(i)} , i=1, 2, ..., n_{\rm c}
    \end{equation}
    Then (\ref{eq-0-2-5}) is equivalent to 
    \begin{subequations}\label{eq-0-2-10}
        \begin{align}
            \min_{\bm{d}_+ \in \mathbb{R}^{n_{\rm n}}, \bm{d}_- \in \mathbb{R}^{n_{\rm n}}} & \bm{1}\T (\bm{d}_+ + \bm{d}_- ) \label{eq-0-2-10:1} \\
            \text{s.t.} ~~~~~~& \bm{1}\T ([\bm{d}_+]_i - [\bm{d}_-]_i ) = {b}_{{\rm sc}, j(i)} , i=1, 2, ..., n_{\rm c} \label{eq-0-2-10:2} \\
            & \bm{d}_+ \geq \bm{0}, \bm{d}_- \geq \bm{0} \label{eq-0-2-10:3}
        \end{align}
    \end{subequations}
    To obtain the optimum of (\ref{eq-0-2-10}), we write its dual problem as
    \begin{subequations}\label{eq-0-2-11}
        \begin{align}
            \max_{\bm{y} \in \mathbb{R}^{n_{\rm c}}} & \sum_{i=1}^{n_{\rm c}} {y}_i {b}_{{\rm sc}, j(i)} ~~\text{s.t.}~ -1 \leq {y}_i \leq  1, i=1,2,..., n_{\rm c}
        \end{align}
    \end{subequations}
    Consider a solution of (\ref{eq-0-2-10}), denoted as $(\bm{d}_+^{\star}, \bm{d}_-^{\star})$, which satisfies that  $\forall i=1,2,..., n_{\rm c}$, if ${b}_{{\rm sc}, j(i)} > 0$, then $[\bm{d}_-^{\star}]_{i} = \bm{0}$ and $[\bm{d}_+^{\star}]_{i}$ contains only one nonzero element being ${b}_{{\rm sc}, j(i)}$; and if ${b}_{{\rm sc}, j(i)} < 0$, then $[\bm{d}_+^{\star}]_{i} = \bm{0}$ and $[\bm{d}_-^{\star}]_{i}$ contains only one nonzero element being $-{b}_{{\rm sc}, j(i)}$. Also consider a solution of (\ref{eq-0-2-11}), denoted as $\bm{y}^{\star}$, which satisfies that $\forall i=1,2,..., n_{\rm c}$, if ${b}_{{\rm sc}, j(i)} > 0$, then ${y}_i^{\star} = 1$, and if ${b}_{{\rm sc}, j(i)} < 0$, then ${y}_i^{\star} = -1$. Since $(\bm{d}_+^{\star}, \bm{d}_-^{\star})$ and $\bm{y}^{\star}$ are feasible solutions and $\bm{1}\T (\bm{d}_+^{\star} + \bm{d}_-^{\star} ) = \sum_{i=1}^{n_{\rm c}} {y}_i^{\star} {b}_{{\rm sc}, j(i)}$, by strong duality, $(\bm{d}_+^{\star}, \bm{d}_-^{\star})$ is a global optimum of (\ref{eq-0-2-10}). Therefore, $\bm{1}\T (\bm{d}_+^* + \bm{d}_-^*) = \bm{1}\T (\bm{d}_+^{\star} + \bm{d}_-^{\star}) = \sum_{i=1}^{n_{\rm c}} |{b}_{{\rm sc}, j(i)}| $, which together with $n_{\rm c} \geq 2$ and $j(i) \neq 1$, gives $\bm{1}\T (\bm{d}_+^* + \bm{d}_-^*) \geq 2 \Vert \bm{b}_{\rm sc} \Vert_{-2 \infty} > 0$.

    {(\uppercase\expandafter{\romannumeral3})} (Sufficiency) By {(\uppercase\expandafter{\romannumeral1})}, if $\bm{1}\T (\bm{d}_+^* + \bm{d}_-^*) = 0$, $\tilde{\bm{z}}$ is connected, and if $0 < \bm{1}\T (\bm{d}_+^* + \bm{d}_-^*) \leq n_{\rm u} r$, $\tilde{\bm{z}}$ is disconnected. Next, we prove that in the latter case, $\tilde{\bm{z}}$ is also $\mathcal{W}(k_{\rm b})$ disconnected by contradiction. Assume that $\tilde{\bm{z}}$ is not $\mathcal{W}(k_{\rm b})$ disconnected, then there must be at least one connected component associated with a row of $\bm{E}_{\rm ns}$ which is also in $\bm{H}_{\rm w}$. Following the proof of necessity of {(\uppercase\expandafter{\romannumeral2})}, it indicates that $\exists i \in \llbracket 1,n_{\rm c} \rrbracket$ s.t. $|{b}_{{\rm sc}, j(i)}| \geq \Vert \bm{H}_{\rm w} \bm{c}_{\rm sc} \Vert_{- \infty}$. Further by $\Vert \bm{H}_{\rm w} \bm{c}_{\rm sc} \Vert_{- \infty} \geq n_{\rm u} r$, $|{b}_{{\rm sc}, j(i)}| > 0$ with $i \in \llbracket 1, n_{\rm c} \rrbracket$, and $n_{\rm c} \geq 2$, we have $\bm{1}\T (\bm{d}_+^* + \bm{d}_-^*) = \sum_{i=1}^{n_{\rm c}} |{b}_{{\rm sc}, j(i)}| > n_{\rm u} r$, which contradicts $\bm{1}\T (\bm{d}_+^* + \bm{d}_-^*) \leq n_{\rm u} r$. This proves that $\tilde{\bm{z}}$ is $\mathcal{W}(k_{\rm b})$ disconnected, completing the proof of sufficiency.

    (Necessity) According to {(\uppercase\expandafter{\romannumeral1})}, if $\tilde{\bm{z}}$ is connected, $\bm{1}\T (\bm{d}_+^* + \bm{d}_-^*) = 0 < n_{\rm u} r$. When $\tilde{\bm{z}}$ is $\mathcal{W}(k_{\rm b})$ disconnected, $\bm{1}\T (\bm{d}_+^* + \bm{d}_-^*) = \sum_{i=1}^{n_{\rm c}} |{b}_{{\rm sc}, j(i)}|$, from the necessity proof of {(\uppercase\expandafter{\romannumeral2})}. The row of $\bm{E}_{\rm ns}$ associated with ${b}_{{\rm sc}, j(i)}$ is in $\bm{E}_{\rm w}$ or $\bm{J} - \bm{E}_{\rm w}$, for any $i \in \llbracket 1, n_{\rm c} \rrbracket$. Since $\bm{1}\T  \bm{c} = 0$, $\Vert (\bm{J} - \bm{E}_{\rm w}) \bm{c}_{\rm sc} \Vert_{\infty} = \Vert \bm{E}_{\rm w} \bm{c}_{\rm sc} \Vert_{\infty} \leq r$, which along with $n_{\rm c} \leq n_{\rm u}$, gives $\bm{1}\T (\bm{d}_+^* + \bm{d}_-^*) \leq n_{\rm u} r$.
\end{proof}

For Criterion \ref{crit-0-2-1}, introduce a variable $\phi_1 \in \mathbb{B}$ to indicate if a contingency is a $k_{\rm b}$-branch one. This variable satisfies the following inequality:
\begin{equation}\label{eq-0-2-12}
    1 - \phi_1 (k_{\rm b} + 1)  \leq  (n_{\rm e} - \bm{1}\T \bm{o}_{\rm b}) - k_{\rm b} \leq  (1 - \phi_1) (k_{\rm max}  - k_{\rm b} )
\end{equation}
where $k_{\rm max}$ is the maximal number of fault components. Here, $\phi_1 = 1$ if $n_{\rm e} - \bm{1}\T \bm{o}_{\rm b} \leq k_{\rm b}$, indicating that the contingency is a $k_{\rm b}$-branch one, and $\phi_1 = 0$ otherwise. According to Proposition \ref{prop-0-2-1}-{(\uppercase\expandafter{\romannumeral3})}, Criterion \ref{crit-0-2-1} can be formulated as 
\begin{equation}\label{eq-0-2-13}
    \bm{1}\T(\bm{d}_+^* + \bm{d}_-^*) \leq n_{\rm u} r + (1 - \phi_1) M
\end{equation}

For Criterion \ref{crit-0-2-2}, introduce variables $\phi_2, \phi_3 \in \mathbb{B}$ satisfying
\begin{subequations}\label{eq-0-2-14} 
    \begin{align}
        & M \phi_2 \geq 2 \Vert \bm{b}_{\rm sc} \Vert_{- 2 \infty}  - \bm{1}\T (\bm{d}_+^* + \bm{d}_-^*) \geq M (\phi_2 - 1)  \\ 
        & M \phi_3 \geq \bm{1}\T (\bm{d}_+^* + \bm{d}_-^*) - n_{\rm u} r \geq M (\phi_3 - 1) 
    \end{align}
\end{subequations}
Here, $\phi_2 = \phi_3 = 0$ iff $2 \Vert \bm{b}_{\rm sc} \Vert_{- 2 \infty}  \leq \bm{1}\T (\bm{d}_+^* + \bm{d}_-^*) \leq n_{\rm u} r$, i.e., $\tilde{\bm{z}}$ is $\mathcal{W}(k_{\rm b})$ disconnected, by Proposition \ref{prop-0-2-1}. Then Criterion \ref{crit-0-2-2} is formulated as
\begin{equation}\label{eq-0-2-15} 
     \bar{\bm{z}} \in \mathcal{C}(\frac{\bm{1}\T (\bm{d}_+^* + \bm{d}_-^*)}{2 \Vert \bm{b}_{\rm sc} \Vert_{-2 \infty}}, \bm{c}_{\rm sc}, \bm{0} ), \bar{\bm{z}} + \bm{1} - \bm{o}_{\rm b} \in \mathcal{C}(\phi_2 + \phi_3, \bm{c}_{\rm sc}, \bm{0} )  
\end{equation}

Constraints (\ref{eq-0-2-12})\text{-}(\ref{eq-0-2-15}) together with $(\bm{d}_+^*, \bm{d}_-^*) \in\arg (\ref{eq-0-2-5})$ formulate the network connectedness with contingencies. 
However, adding them into the topology optimization model will introduce additional structural complicity. Since problem (\ref{eq-0-2-5}) is a linear program, it can be replaced by its Karush-Kuhn-Tucker (\ac{kkt}) conditions, which with the \ac{kkt} complementarity condition reformulated in a mixed-integer form, is written as 
 \begin{equation}\label{eq-0-2-17}
         \bm{A} \bm{y} \leq \bm{w}(\tilde{\bm{z}}), \bm{A}\T \bm{\lambda} = \bm{h}, \bm{\lambda} \geq \bm{0}, \bm{w}(\tilde{\bm{z}}) - \bm{A} \bm{y} \leq M (\bm{1} - \bm{\xi}), \bm{\lambda} \leq M \bm{\xi}   
 \end{equation}
where $\bm{y} = [{\bm{d}_+^*}\T, {\bm{d}_-^*}\T, \bm{\vartheta}\T, \bm{\rho}\T ]\T$, $\bm{\lambda} \in \mathbb{R}^{4(n_{\rm n} + n_{\rm e})}$, $\bm{\xi} \in \mathbb{B}^{4(n_{\rm n} + n_{\rm e})}$, and $\bm{h}$, $\bm{w}(\tilde{\bm{z}})$ and $\bm{A}$ are coefficient matrices of the equivalent compact form of (\ref{eq-0-2-5}), i.e., $\min_{\bm{y} \in \mathbb{R}^{n_{\rm e} + 3 n_{\rm n}}} \bm{h}\T \bm{y} ~\text{s.t.}~ \bm{A} \bm{y} \leq \bm{w}(\tilde{\bm{z}})$. 

Accordingly, constraints (\ref{eq-0-2-12})-(\ref{eq-0-2-17}) constitute a tractable formulation of network connectedness with contingencies, ensuring compliance with Criterion \ref{crit-0-2-1} and Criterion \ref{crit-0-2-2}. This formulation is expressed in a mixed-integer linear form, thereby avoiding additional complexities in topology optimization problems.

\section{Multi-Island Structure of Network Topology}\label{sec-3-3}

For intentional controlled islanding of transmission networks, the objective is not to maintain a fully connected topology but to establish a multi-island structure, with each island containing a specific group of buses. To formulate it, the network is assumed to be partitioned into $n_{\rm is}$ islands, represented by the set $\mathcal{K} = \llbracket 1, n_{\rm is} \rrbracket$; let $\mathcal{V}_k^{\rm gen}$ denote the set of generator buses that are required to be located in island $k$.  
Then, the multi-island topology can be enforced through partitioning constraints, connectedness constraints, and co-location constraints. 

Firstly, the partitioning constraints ensure that the network is partitioned into $n_{\rm is}$ islands, expressed as follows:
\begin{subequations}\label{eq-b1-3-5}
    \begin{align}
        & W_{ij, h} \leq X_{i, h}, && \forall (i, j) \in \mathcal{E}, h \in \mathcal{K} \label{eq-b1-3-5:1} \\
        & W_{ij, h} \leq X_{j, h}, && \forall (i, j) \in \mathcal{E}, h \in \mathcal{K} \label{eq-b1-3-5:2} \\
        & z_{ij} = \sum_{h \in \mathcal{K}} W_{ij, h}, && \forall (i, j) \in \mathcal{E} \label{eq-b1-3-5:3} \\
        & \sum_{k \in \mathcal{K}} X_{i, k} = 1, && \forall i \in \mathcal{V}  \label{eq-b1-3-5:4}
    \end{align}
\end{subequations}
where binary variable $X_{i, h} \in \mathbb{B}$ indicates whether vertex $i$ belongs to island $h$; variable $W_{ij, h} \in [0, 1]$ indicates whether edge $(i, j)$ belongs to island $h$. 
Constraints (\ref{eq-b1-3-5:1}) and (\ref{eq-b1-3-5:2}) ensure that if both vertex $i$ and vertex $j$ belong to island $h$, edge $(i, j)$ must also belong to island $h$; otherwise, the edge does not belong to that island. 
Constraint (\ref{eq-b1-3-5:3}) ensures that edge $(i, j)$ is closed if it belongs to any island and remains open otherwise. 
Constraint (\ref{eq-b1-3-5:4}) enforces that each vertex belongs to exactly one island. 

Secondly, the following connectedness constraints enforce connectedness of each island:
\begin{subequations}\label{eq-b1-3-6}
    \begin{align}
        & \sum_{i=1}^j \frac{X_{i, h}}{n_{\rm n}} \leq Y_{j, h} \leq  \sum_{i=1}^j X_{i, h}, && \forall j \in \mathcal{V}, h \in \mathcal{K} \label{eq-b1-3-6:1}\\
        & Y_{j, h} \geq X_{j,h}, && j \in \mathcal{V}, h \in \mathcal{K} \label{eq-b1-3-6:2}\\
        & U_{j, h} = Y_{j,h} - Y_{j-1, h}, && j \!\in\! \mathcal{V} \setminus \{1\}, h \!\in\! \mathcal{K}   \label{eq-b1-3-6:3}\\
        &  U_{1,h} = Y_{1,h}, && h \in \mathcal{K}  \label{eq-b1-3-6:4} \\
        &  \sum\nolimits_{j \in \mathcal{V}} U_{j,h} = 1, && h \in \mathcal{K}  \label{eq-b1-3-6:5} \\
        & U_{j,h} \sum_{i \in \mathcal{V}} X_{i,h} \!-\! X_{j,h}  \!+\!\!\! \sum_{(v,j) \in \mathcal{E}} \!\! F_{vj,h} \!=\!\!\!\! \sum_{(j,v) \in \mathcal{E}} \!\! F_{jv,h}, && j \in \mathcal{V}, h \in \mathcal{K} \label{eq-b1-3-6:6} \\
        & 0 \leq F_{ij,h} \leq n_{\rm n} z_{ij}, && (i,j) \in \mathcal{E}   \label{eq-b1-3-6:7} 
    \end{align}
\end{subequations}
Here, constraints (\ref{eq-b1-3-6:1})-(\ref{eq-b1-3-6:5}) indicate a selection procedure that chooses the source vertex in each island for transmitting a unit of flow to the other vertices inside the same island. 
Variable $Y_{j, h} \in [0, 1]$ for defining variable $U_{j, h}$ indicates the source vertex in each island. The smallest index vertex is chosen in each island as the source vertex. 
Constraints (\ref{eq-b1-3-6:6}) indicate flow conservation of each vertex. Continuous variable $f_{jv, h} \in \mathbb{R}$ indicates the amount of flow of the edge $(j, v) \in \mathcal{E}$ for island $h$. 

Thirdly, the co-location constraints ensure that island $k$ contains all generator buses in $\mathcal{V}_{k}^{\rm gen}$, expressed as follows: 
\begin{equation}\label{eq-b1-3-7}
    X_{i, k} = 1, \quad \forall i \in \mathcal{V}_k^{\rm gen}, k \in \mathcal{K}
\end{equation}

Finally, constraints (\ref{eq-b1-3-5})-(\ref{eq-b1-3-7}) constitute the formulation of the multi-island structure of network topology.

\section{Radiality of Network Topology}

Unlike transmission networks, distribution networks are typically operated in a radial topology in order to simplify operation and protection. In a distribution network with a single substation, radiality requires that the underlying graph forms a spanning tree; while for a distribution network with multiple substations, radiality necessitates that the underlying graph forms a spanning forest, where each island contains exactly one substation. 
This section first introduces different radiality formulations of a distribution network with a single substation. The extension to multiple substations is then discussed. Note that all formulations exclude transfer nodes. 

\subsection{Network-Flow-Based Formulation}

The network-flow-based formulation enforces network radiality by ensuring that the underlying graph remains connected and consists of $n_{\rm n} - 1$ edges \citep{4-836}. The graph connectedness can be maintained by the network-flow-based model. Specifically, the network-flow-based radiality formulation is expressed as follows: 
\begin{subequations}\label{eq-b1-3-7}
    \begin{align}
        & \sum\nolimits_{(i, j)\in \mathcal{E}} z_{ij} = n_{\rm n} - 1 \\
        & \sum\nolimits_{(j, i) \in \mathcal{E}} \rho_{ji} - \sum\nolimits_{(i, j) \in \mathcal{E}} \rho_{ij} = c_{{\rm nc}, i}, && i \in \mathcal{V} \\
        & |\rho_{ij}| \leq   (n_{\rm n} - 1) z_{ij}, && \forall (i,j) \in \mathcal{E} \\ 
        & \rho_{ij}, \rho_{ji} \in \mathbb{R}, && \forall (i, j) \in \mathcal{E}
    \end{align}
\end{subequations}
where $c_{{\rm nc}, i_{\rm ss}} = 1 - n_{\rm n}$ with $i_{\rm ss}$ being the substation node, and $c_{{\rm nc}, i} = 1$, $\forall i \in \mathcal{V} \backslash \{i_{\rm ss}\}$.

\subsection{Parent-Child Relation-Based Formulation}

The parent-child relation-based formulation is based on the following characteristic of the spanning tree: every node except for the root has exactly one parent \citep{4-817}. This formulation is expressed as: 
\begin{subequations}\label{eq-b1-3-9}
    \begin{align}
        & \zeta_{ij} + \zeta_{ji} = z_{ij},  && \forall (i, j) \in \mathcal{E} \label{eq-b1-3-9:1}\\
        & \sum\nolimits_{(i, j) \in \mathcal{E} \lor (j, i) \in \mathcal{E}  } \zeta_{ij} = 1, && \forall i \in \mathcal{V} \backslash \{ i_{\rm ss} \} \label{eq-b1-3-9:2}\\
        & \zeta_{i_{\rm ss}j} = 0, && \forall (i_{\rm ss},j) \in \mathcal{E} \lor (j, i_{\rm ss}) \in \mathcal{E} \label{eq-b1-3-9:3}\\
        & \zeta_{ij} \in \mathbb{B}, \zeta_{ji} \in \mathbb{R}, && \forall (i, j) \in \mathcal{E}  \label{eq-b1-3-9:4}
    \end{align}
\end{subequations}
Eq. (\ref{eq-b1-3-9:1}) states that a line $(i, j)$ is in the spanning tree (i.e., $z_{ij} = 1$) if node $j$ is the parent of node $i$ (i.e., $\zeta_{ij} = 1$) or if node $i$ is the parent of node $j$ (i.e., $\zeta_{ji} = 1$). Eq. (\ref{eq-b1-3-9:2}) ensures that every node, except for the substation node, has exactly one parent. Eq. (\ref{eq-b1-3-9:3}) specifies that the substation node has no parents.

\subsection{Directed Multi-Commodity Flow-Based Formulation}

The directed multi-commodity flow-based formulation defines a fictitious commodity for each vertex $k \neq i_{\rm ss}$, and enforces 1 unit of commodity $k$ delivered from the substation node $i_{\rm ss}$ to vertex $k$. This formulation has extended flexibility than the other alternatives \citep{4-1871}. Specifically, this formulation is expressed as follows: 
\begin{subequations}\label{eq-b1-3-8}
    \begin{align}
        & \sum_{(j, i_{\rm ss}) \in \mathcal{E}} \rho_{ji_{\rm ss}}^k \!\!-\!\!\! \sum_{(i_{\rm ss}, j) \in \mathcal{E}} \rho_{i_{\rm ss}j}^k = -1, && \forall k \in \mathcal{V} \backslash \{ i_{\rm ss} \}  \label{eq-b1-3-8:1}\\
        & \sum_{(j, k) \in \mathcal{E}} \rho_{jk}^k - \sum_{(k, j) \in \mathcal{E}} \rho_{kj}^k = 1, && \forall k \in \mathcal{V} \backslash \{ i_{\rm ss} \} \label{eq-b1-3-8:2} \\
        & \sum_{(j, i) \in \mathcal{E}} \rho_{ji}^k - \sum_{(i, j) \in \mathcal{E}} \rho_{ij}^k = 0, && \forall k \in \mathcal{V} \backslash \{ i_{\rm ss} \}, i \in \mathcal{V} \backslash \{i_{\rm ss}, k\} \label{eq-b1-3-8:3}\\
        & 0 \leq \rho_{ij}^k \leq \lambda_{ij}, 0 \leq \rho_{ji}^k \leq \lambda_{ji}, && \forall k \in \mathcal{V} \backslash \{i_{\rm ss}\}, (i, j) \in \mathcal{E} \label{eq-b1-3-8:4} \\
        & \sum\nolimits_{(i,j) \in \mathcal{E}} (\lambda_{ij} + \lambda_{ji}) = n_{\rm n} - 1 && ~  \label{eq-b1-3-8:5} \\
        & \lambda_{ij} + \lambda_{ji} = z_{ij}, &&  \forall (i, j) \in \mathcal{E} \label{eq-b1-3-8:6} \\
        & \lambda_{ij} \in \mathbb{B}, \lambda_{ji} \in \mathbb{R}; \rho_{ij}^k, \rho_{ji}^k \in \mathbb{R}, && \forall (i, j) \in \mathcal{E}, k \in \mathcal{V} \backslash \{ i_{\rm ss} \} \label{eq-b1-3-8:7}
    \end{align}
\end{subequations}
where variable $\rho_{ij}^k$ represents the flow of commodity $k$ from vertex $i$ to vertex $j$, and variable $\lambda_{ij}$ indicates whether arc $(i, j)$ is included in the directed spanning tree. 
Eq. (\ref{eq-b1-3-8:1}) ensures that 1 unit of commodity $k$ flows out from the substation node $i_{\rm ss}$, while Eq. (\ref{eq-b1-3-8:2}) let that 1 unit of commodity $k$ flow into node $k$. Eq. (\ref{eq-b1-3-8:3}) prevents each commodity $k$ from flowing into other nodes. Eq. (\ref{eq-b1-3-8:4}) ensures that commodity flows only along arcs is included in the directed spanning tree defined by variables $\lambda_{ij}$. Eq. (\ref{eq-b1-3-8:5}) specifies the number of arcs in the spanning tree, while Eq. (\ref{eq-b1-3-8:6}) assigns the edge connection status $z_{ij}$ based on arc inclusion in the spanning tree. 

\subsection{Extension to Multiple Substations}

\begin{figure}[h]
	\centering
	\includegraphics[width=0.5\linewidth]{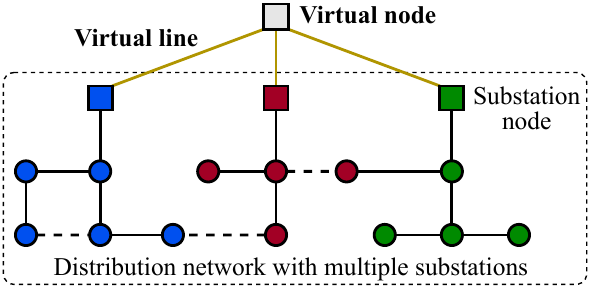}
	\caption{Illustration of handling distribution networks with multiple substations.}
    \label{fig-b1-3-2} 
\end{figure}

When the distribution network contains multiple substations, the underlying graph needs to form a spanning forest with each island containing exactly one substation. In this case, the aforementioned radiality formulations can still be applied by adding a \textit{virtual node} and certain \textit{virtual lines}. Specifically, as shown in Fig. \ref{fig-b1-3-2}, the virtual node is connected with each substation nodes via virtual lines. It is evident that the augmented graph is radial if and only if the original distribution network graph forms a spanning forest, where each tree (i.e., connected component) contains exactly one substation. Therefore, the radiality formulation for the augmented graph, with all virtual lines remaining closed, guarantees compliance with the topology requirements of distribution networks with multiple substations.

\graphicspath{{chapter_4/Figs/}}
\chapter{Steady-State Transmission Topology Control}\label{chapter-4}

This chapter presents the steady-state topology control of transmission networks. First, some fundamentals are introduced, including the phenomena of Braess's paradox in transmission networks and basic forms of steady-state transmission topology control. Next, a comprehensive state-of-the-art review of existing methodologies for steady-state transmission topology control is provided. Finally, a recent advance in steady-state transmission topology control under multiple uncertainties is discussed.

\section{Fundamentals}

\subsection{Braess's Paradox}\label{sec-b1-braess}

Braess's paradox is a counterintuitive phenomenon in road networks where adding one or more roads can slow down overall traffic flow over the network. Such analogous paradoxical behaviour has been observed in other network systems including power systems, especially transmission networks whose network topology is typically highly meshed \citep{4-34, 4-551, 4-1466, 4-1621}. We will refer to this as Braess's paradox for power networks whereby the improvement of certain system performance could be accomplished by removing specific transmission lines. 
Accordingly Braess's paradox suggests the counterintuitive possibility of enhancing system performance via steady-state topology control for transmission networks. 
In practice, most transmission networks normally operate with all transmission lines in service. However, steady-state topology control leverages this paradox to selectively switch off specific lines such that system performance metrics are optimized. 
The following are two examples of Braess's paradox in transmission networks.

\subsubsection{Economic dispatch performance}

\begin{figure}[h]
	\centering
	\includegraphics[width=1\linewidth]{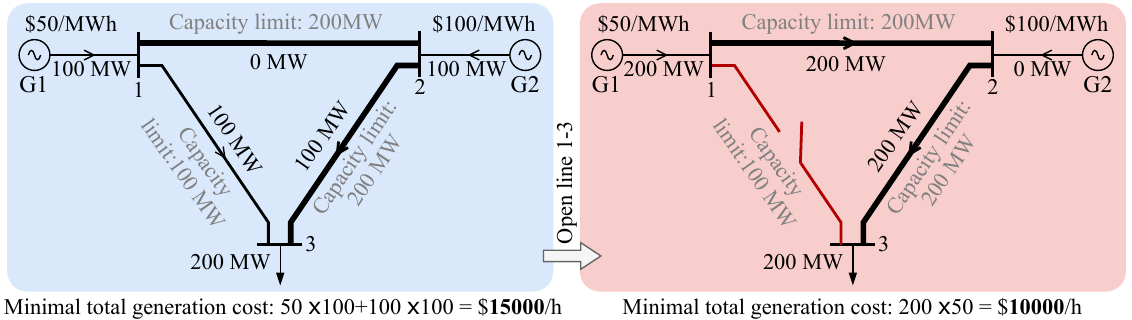}  
	\caption[xxx.]{ Braess's paradox for economic dispatch performance.}  
	\label{fig-b1-4-1}  
\end{figure}

Consider the three-bus transmission network shown in Fig. \ref{fig-b1-4-1}. The network contains one load of 200 MW, two generators G1 and G2 whose unit generation costs are \$50/MWh and \$100/MWh, respectively; and three transmission lines with identical impedance. The capacity limits of lines 1-2 and 2-3 are both 200 MW, that of line 1-3 is 100 MW. For simplicity, generator capacity limits and line resistances are neglected. 

As illustrated on the left side of Fig. \ref{fig-b1-4-1}, under the network topology with all transmission lines in service, the economic dispatch results in both generators G1 and G2 producing 100 MW, achieving a minimal total generation cost of \$15000/h. 
However, as shown on the right side of Fig. \ref{fig-b1-4-1}, if transmission line 1-3 is switched off, the economic dispatch yields 200 MM output from only generator G1, resulting in a minimal total generation cost of \$10000/h. 
Therefore, the best economic dispatch performance may not be achieved under the network topology with all lines in service, but can instead be achieved by switching off certain transmission lines.

\subsubsection{Stability} 

\begin{figure}[h]
	\centering
	\includegraphics[width=1\linewidth]{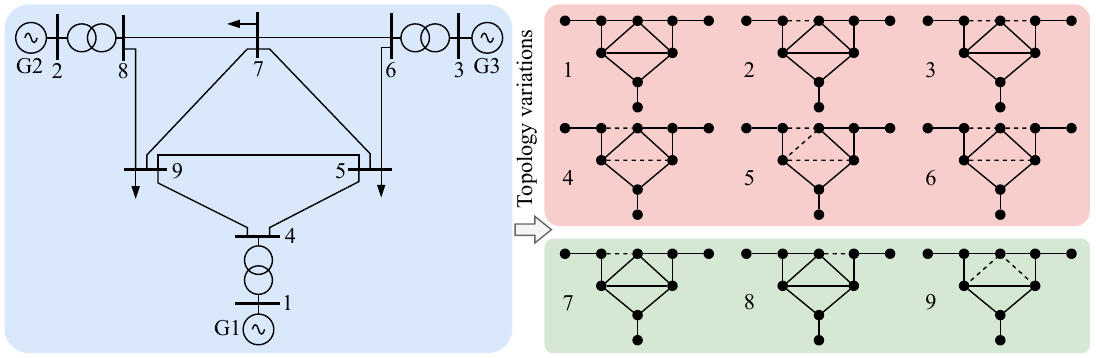}  
	\caption{Diagram of the modified IEEE 9-bus system and its topology variations.}  
	\label{fig-b1-4-2}  
\end{figure}

Consider the modified IEEE 9-bus system depicted in Fig. \ref{fig-b1-4-2}. The system dynamics considered are limited to generator dynamics, represented by the classical second-order model. Nine network topologies are examined: topology 1 with all transmission lines switched on, and topologies 2 to 9, each corresponding to a line-open scenario.

\begin{figure}[h]
	\centering
	\includegraphics[width=0.9\linewidth]{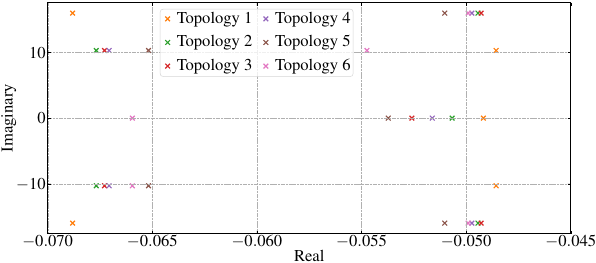}  
	\caption{Eigenvalue spectrum under various network topologies.}  
	\label{book-1-4-r1}  
\end{figure}

Fig. \ref{book-1-4-r1} shows the eigenvalue spectrum of the linearized system under network topologies 1 to 6, computed around their respective equilibrium point. 
Small-disturbance stability margin is measured by the largest eigenvalue real part. 
Compared to the system with topology 1 where all lines are switched on, topologies 2 to 6 exhibit eigenvalue spectrums with smaller largest real parts, indicating larger small-disturbance stability margins. This suggests that small-disturbance stability of the transmission network can be enhanced by switching off certain lines. 

For large-disturbance stability, consider the three-phase short-circuit fault occurring at bus 7 when $t=0.1$ s and cleared when $t=0.48$ s. Fig. \ref{book-1-4-r2} shows the time-domain responses of rotor angle $\delta_i$ of all generators for the system under topology 1 and topologies 7 to 9. The system exhibits transient instability under topology 1, whereas it maintains transient stability under topologies 7, 8, and 9. Thus, switching off certain transmission lines can improve the system's large-disturbance stability, at least under some disturbance scenarios. 

\begin{figure}[h]
	\centering
	\includegraphics[width=0.9\linewidth]{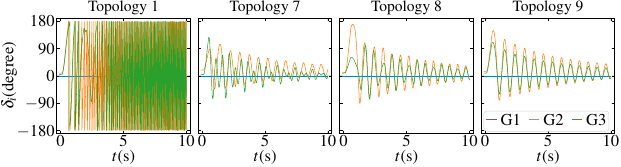}  
	\caption{Time-domain response of generator rotor angle $\delta_i$ for different topologies.}  
	\label{book-1-4-r2}  
\end{figure}

\subsection{Optimal Transmission Switching}\label{sec-4-1a}

\cite{4-62} first proposed the concept and model of Optimal Transmission Switching (\ac{ots}), which can be seen as one of the most basic forms of steady-state transmission topology control. This \ac{ots} model is formulated in matrix form as: 
\begin{subequations}\label{eq-0-1-0}
  \begin{align}
       \min_{\bm{z} \in \mathbb{B}^{n_{\rm e}}, \bm{p}_{\rm g} \in \mathbb{R}^{n_{\rm g}}}  ~& f_{\rm ots}(\bm{p}_{\rm g})   \label{eq-0-1-0:1}\\ 
     \text{s.t.~~~~~} & - (\bm{1} - \bm{z}) M \leq   \bm{b}_{\rm b}\D  \bm{E}_{\mathcal{G}}\T  \bm{\theta} - \bm{p}_{\rm b} \leq (\bm{1} - \bm{z}) M  \label{eq-0-1-0:6} \\ 
     & \bm{E}_{\rm g} \bm{p}_{\rm g} - \bm{E}_{\rm d} \bm{p}_{\rm d} = \bm{E}_{\mathcal{G}} \bm{p}_{\rm b}  \label{eq-0-1-0:4} \\
     &  \bm{\theta}_{\rm b}\B - (\bm{1} - \bm{z})M \leq \bm{E}_{\mathcal{G}}\T \bm{\theta} \leq \bm{\theta}_{\rm b}\U + (\bm{1} - \bm{z}) M   \label{eq-0-1-0:2} \\
     &  \bm{p}_{\rm b}\B \circ \bm{z} \leq \bm{p}_{\rm b} \leq \bm{p}_{\rm b}\U \circ \bm{z}  \label{eq-0-1-0:5} \\
     &  \bm{p}_{\rm g}\B \leq \bm{p}_{\rm g} \leq \bm{p}_{\rm g}\U  \label{eq-0-1-0:2-1} \\
     & \bm{E}_{\rm us} \bm{z} = \bm{1}   \label{eq-0-1-0:8}  
  \end{align}
\end{subequations}
where the objective function $f_{\rm ots}(\bm{p}_{\rm g})$ represents the total generation cost; 
constraint (\ref{eq-0-1-0:6}) enforces the DC power flow equations, incorporating the variable network topology parameterized by $\bm{z}$; constraint (\ref{eq-0-1-0:4}) ensures power balance at each bus; 
(\ref{eq-0-1-0:2}) to (\ref{eq-0-1-0:8}) impose operational limits on branch phase angle differences, branch power flows, generation outputs, and branch switching actions, respectively. 

\begin{table}[h]
    \centering
    \caption{Results of \ac{ots} for the IEEE 118-bus system.}
    \begin{tabular}{|cccc|}
        \hline\hline
        $n_{\rm sw}$  & Lines switched off & $f_{\rm ots}$ (\$/h) & Reduction rate  \\ \hline\hline
        0 & --- & 2898.57 & 0 \\
        1 & 8-30 & 2715.29 & 6.32\% \\
        2 & 8-30, 65-66   & 2652.60 & 8.49\% \\
        3 & 8-30, 65-66, 12-16   & 2575.97 & 11.13\% \\
        4 & 8-30, 65-66, 12-16, 2-12    & 2543.83 & 12.24\% \\
        5 & 8-30, 65-66, 12-16, 2-12, 11-12    & 2521.86 & 13.00\% \\
        $\infty$  &  42 lines (“x” symbol in Fig. \ref{fig-b1-4-3})   & 2357.06  & 18.68\%  \\
        \hline\hline
    \end{tabular}
    \label{table-b1-4-1}
\end{table}

Table \ref{table-b1-4-1} provides the solution results of the \ac{ots} model (\ref{eq-0-1-0}) on the IEEE 118-bus system. The congested operating condition data of the system from the PGLib-OPF networks is used \citep{4-1280}. To evaluate the effectiveness of \ac{ots}, the constraint that the number of lines switched off does not exceed $n_{\rm sw}$ is added to model (\ref{eq-0-1-0}). The model is solved for different values of $n_{\rm sw}$. When $n_{\rm sw}=0$, line switching is not considered, and model (\ref{eq-0-1-0}) reduces to the DC economic dispatch model; $n_{\rm sw}=\infty$ indicates that the number of lines switched off is not limited. 
According to Table \ref{table-b1-4-1}, \ac{ots} can reduce the total generation cost by up to 18.68\% with respect to the case without steady-state topology control, achieved by switching off 42 transmission lines as marked by symbol ``x'' in Fig. \ref{fig-b1-4-3}. However, switching off such a large portion of transmission lines may compromise system security and stability. A key observation is that a significant reduction in total generation cost can be achieved by disconnecting only a small subset of transmission lines. For instance, switching off only 5 transmission lines can reduce the total generation cost by 13.00\%. 

\begin{figure}[t]
	\centering
    \includegraphics[width=1\linewidth]{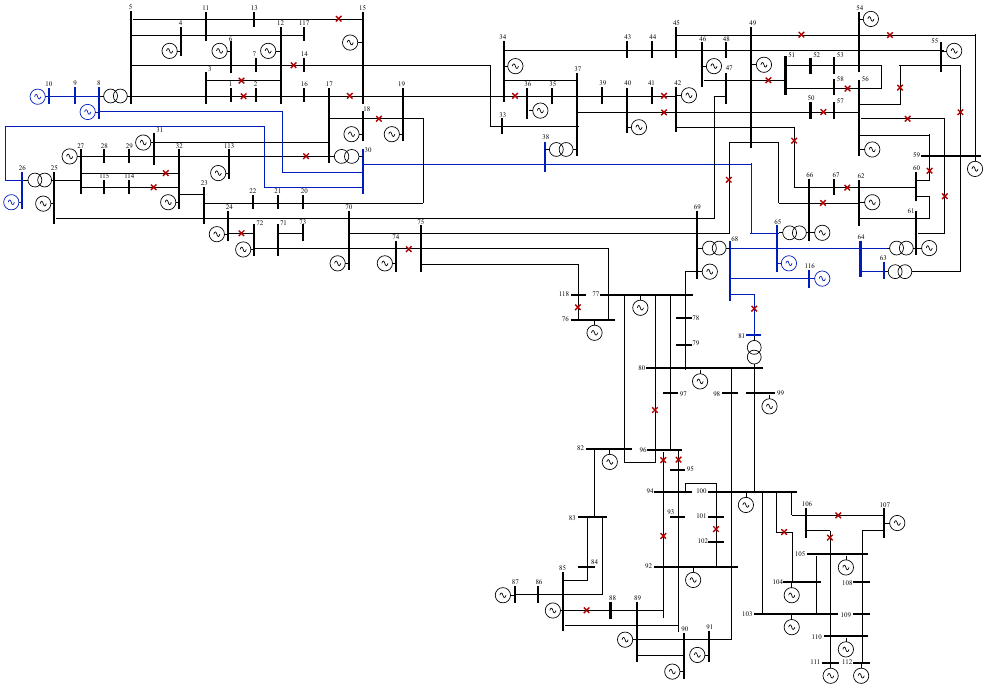}  
	\caption{Illustration of the \ac{ots} solution for the IEEE 118-bus system.}
	\label{fig-b1-4-3}
\end{figure}


\subsection{Substation-Level Transmission Topology Control}

In the above mentioned \ac{ots} formulation, the network is modeled by the bus-branch representation where each substation is reduced to a single bus and only the switching actions of transmission branches can be considered. 
However, certain substation arrangements enable the substation to be split into multiple separate busbars, a process referred to as bus splitting. 
Similar to line switching, bus splitting can also be utilized as a means of steady-state transmission topology control. This requires adopting the node-breaker representation rather than the bus-branch representation of the network, thereby enabling substation-level transmission topology control. 

\begin{figure}[h]
	\centering
    \includegraphics[width=1\linewidth]{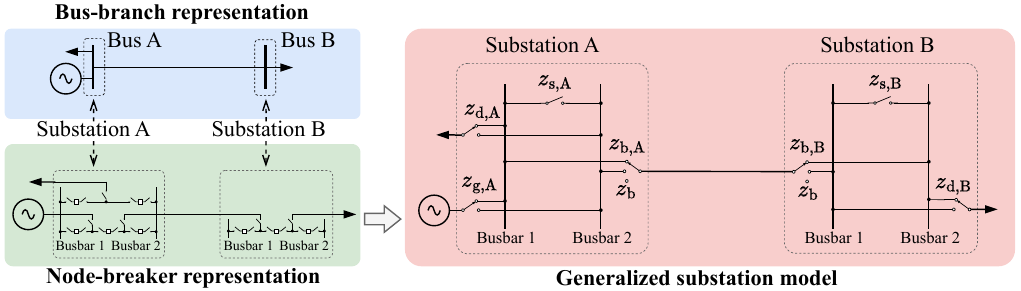}  
	\caption{Illustration of the breaker-and-a-half substation arrangement and the generalized substation model.}
	\label{fig-b1-4-4}
\end{figure}

Fig. \ref{fig-b1-4-4} illustrates the typical configuration of the breaker-and-a-half layout, which is a well-known and widely-used substation arrangement in power systems. The diagram depicts two substations interconnected by a transmission line. Each substation contains two busbars. 
Line switching, bus splitting and generator/load switching actions can be performed by switching the associated circuit breakers. 
The right side of Fig. \ref{fig-b1-4-4} shows the corresponding generalized substation model developed by \cite{4-1227}. In each substation, this model employs one zero-impedance line to allow bus splitting, and several fictitious switches to connect generator, load, line, etc to either of two busbars inside the substation. 
The following constraints are imposed on Substation A: 
\begin{subequations}\label{eq-b1-4-s1}
    \begin{align}
        & - M (1 - z_{\rm s, A}) \leq \theta_{\rm 1, A} -  \theta_{\rm 2, A}  \leq  M(1 - z_{\rm s, A}) \label{eq-b1-4-s1:1}\\
        & (1 - z_{\rm g, A}) p_{\rm g, A}^{\min} \leq p_{\rm g, A1} \leq  (1 - z_{\rm g, A}) p_{\rm g, A}^{\max} \label{eq-b1-4-s1:2}\\
        & z_{\rm g, A} \cdot p_{\rm g, A}^{\min} \leq p_{\rm g, A2} \leq  z_{\rm g, A} \cdot p_{\rm g, A}^{\max} \label{eq-b1-4-s1:3}\\
        & p_{\rm d, A1} = (1 - z_{\rm d, A}) p_{\rm d, A}, p_{\rm d, A2} =  z_{\rm d, A} \cdot p_{\rm d, A} \label{eq-b1-4-s1:4} 
    \end{align}
\end{subequations}
where $z_{\rm s, A}$ is the binary variable indicating if Busbar 1 and Busbar 2 in Substation A are operated in the split mode or the connected mode; 
$z_{\rm g, A}$ and $z_{\rm d, A}$ are the binary variables specifying the busbars to which the generator and the load in Substation A are connected, respectively. 
Constraint (\ref{eq-b1-4-s1:1}) indicates the conditional statement that if the zero-impedance line is closed, i.e., Busbar 1 and Busbar 2 in Substation A are operated in the connected mode and $z_{\rm s, A}=1$, the phase angles of Busbar 1 and Busbar 2 are equal; otherwise, this restriction is not required. 
Constraints (\ref{eq-b1-4-s1:2}) and (\ref{eq-b1-4-s1:3}) ensure that the generator can be connected to either Busbar A or Busbar B in Substation A. 
If $z_{\rm g, A} = 0$, the generator injects power to Busbus 1; otherwise, the generator delivers power to Busbar 2. 
Similarly, the load is connected to either of the two busbars, determined by binary variable $z_{\rm d, A} $ in constraint (\ref{eq-b1-4-s1:4}).

Analogously, the following constraints are imposed on Substation B: 
\begin{subequations}\label{eq-b1-4-s2}
    \begin{align}
        & - M (1 - z_{\rm s, B}) \leq \theta_{\rm 1, B} -  \theta_{\rm 2, B}  \leq  M(1 - z_{\rm s, B})  \\
        & p_{\rm d, B1} = (1 - z_{\rm d, B}) p_{\rm d, B}, p_{\rm d, B2} =  z_{\rm d, B} \cdot p_{\rm d, B}  
    \end{align}
\end{subequations}

Moreover, power flow through the transmission line between the two substations satisfies the following constraints:
\begin{subequations}\label{eq-b1-4-s3}
    \begin{align}
        & - (1 - z_{\rm b, A}) p_{\rm b}^{\max} \leq  p_{\rm b, A1} \leq (1 - z_{\rm b, A}) p_{\rm b}^{\max} \label{eq-b1-4-s3:1} \\
        & - (1 - z_{\rm b, B}) p_{\rm b}^{\max} \leq  p_{\rm b, B1} \leq (1 - z_{\rm b, B}) p_{\rm b}^{\max} \label{eq-b1-4-s3:2} \\
        & - z_{\rm b, A} p_{\rm b}^{\max} \leq  p_{\rm b, A2} \leq  z_{\rm b, A} p_{\rm b}^{\max}, - z_{\rm b, B} p_{\rm b}^{\max} \leq  p_{\rm b, B2} \leq  z_{\rm b, B} p_{\rm b}^{\max} \label{eq-b1-4-s3:3} \\
        & - z_{\rm b} p_{\rm b}^{\max} \leq  p_{\rm b, A1} \leq  z_{\rm b} p_{\rm b}^{\max}, - z_{\rm b} p_{\rm b}^{\max} \leq  p_{\rm b, B1} \leq  z_{\rm b} p_{\rm b}^{\max} \label{eq-b1-4-s3:4} \\ 
        & z_{\rm b} - z_{\rm b, A} \geq 0, z_{\rm b} - z_{\rm b, B} \geq 0  \label{eq-b1-4-s3:5} \\ 
        & p_{\rm b} = p_{\rm b, A1} + p_{\rm b, A2}, p_{\rm b} = p_{\rm b, B1} + p_{\rm b, B2}  \label{eq-b1-4-s3:6} \\
        & -M(1 - z_{\rm b}) \leq   b_{\rm b}( \theta_{\rm b, A} - \theta_{\rm b, B} ) - p_{\rm b} \leq  M(1 - z_{\rm b}) \label{eq-b1-4-s3:7} \\
        & - z_{\rm b, A} \theta^{\max} \leq \theta_{\rm b, A} - \theta_{\rm 1, A} \leq z_{\rm b, A} \theta^{\max}  \label{eq-b1-4-s3:8} \\
        & - (1 - z_{\rm b, A}) \theta^{\max} \leq \theta_{\rm b, A} - \theta_{\rm 2, A} \leq z_{\rm b, A} \theta^{\max}  \label{eq-b1-4-s3:9} \\
        & - z_{\rm b, B} \theta^{\max} \leq \theta_{\rm b, B} - \theta_{\rm 1, B} \leq z_{\rm b, B} \theta^{\max} \label{eq-b1-4-s3:10} \\
        & - (1 - z_{\rm b, B}) \theta^{\max} \leq \theta_{\rm b, B} - \theta_{\rm 2, B} \leq z_{\rm b, B} \theta^{\max}  \label{eq-b1-4-s3:11}
    \end{align}
\end{subequations}
where $z_{\rm b, A}$, $z_{\rm b, B}$ and $z_{\rm b}$ are binary variables that indicate which busbar the line is connected to and whether the transmission line is switched on or off. 
If $z_{\rm b} = 0$ that indicates the line is open, binary variables $z_{\rm b, A}$ and $z_{\rm b, B}$ are forced to be zero by constraints (\ref{eq-b1-4-s3:5}). 
Consequently, the power flows through the line from each busbar, denoted by $p_{\rm b, A1}$, $p_{\rm b, A2}$, $p_{\rm b, B1}$ and $p_{\rm b, B2}$, are also set to zero by constraints (\ref{eq-b1-4-s3:3}) and (\ref{eq-b1-4-s3:4}). 
If $z_{\rm b} = 1$ that indicates the line is closed, binary variables $z_{\rm b, A}$ and $z_{\rm b, B}$ are either one or zero given constraints (\ref{eq-b1-4-s3:5}). 
In this case, if $z_{\rm b, A} = 0$ that indicates the line is connected to Busbar 1, the power flow through the line from Busbar 2 in Substation A, denoted by $p_{\rm b, A2}$, is set to zero by the first constraint of (\ref{eq-b1-4-s3:3}). If $z_{\rm b, A} = 1$ that indicates the line is connected to Busbar 2, $p_{\rm b, A1}=0$ according to constraint (\ref{eq-b1-4-s3:1}). 
Constraints (\ref{eq-b1-4-s3:6}) represent that the power flow through the line, denoted by $p_{\rm b}$, equals the sum of the power flows from all busbars within the substation. 
When $z_{\rm b} = 1$, constraint (\ref{eq-b1-4-s3:7}) forms the DC power flow equation. 
Constraints (\ref{eq-b1-4-s3:8}) to (\ref{eq-b1-4-s3:11}) enforce that the phase angles at the two terminals of the line are respectively the phase angles of the busbars to which the line is connected.

\begin{figure}[h]
	\centering
    \includegraphics[width=1\linewidth]{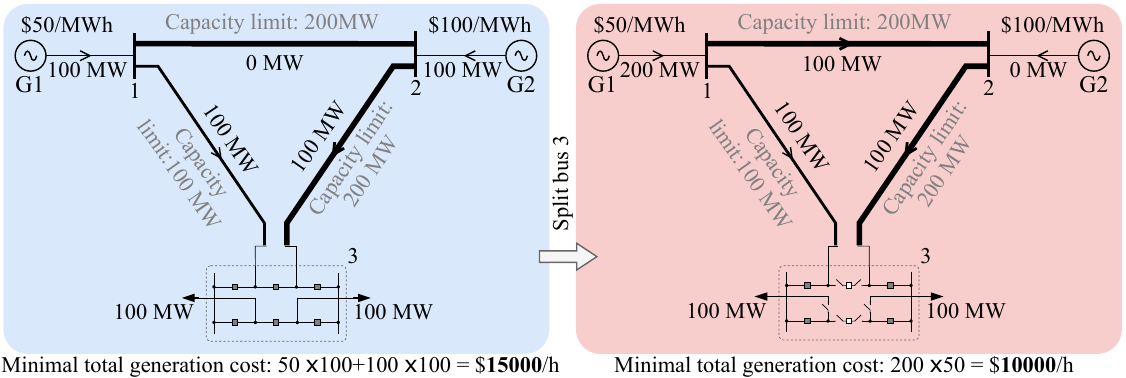}  
	\caption{Effectiveness of substation-level transmission topology control.}
	\label{fig-b1-4-5}
\end{figure}

Fig. \ref{fig-b1-4-5} demonstrates the effectiveness of substation-level transmission topology control on improving economic dispatch performance, using the same three-bus transmission network in Fig. \ref{fig-b1-4-1} of Section \ref{sec-b1-braess}. For simplicity, only the substation corresponding to bus 3 is represented by the node-breaker model. 
When the two busbars in the substation are operated in the connected mode, as illustrated on the left side of Fig. \ref{fig-b1-4-5}, the economic dispatch results in both generators G1 and G2 producing 100 MW, achieving a minimal total generation cost of \$15000/h. However, as shown on the right side of Fig. \ref{fig-b1-4-5}, if the two busbars corresponding to bus 3 is operated in the split mode, the economic dispatch yields 200 MM output from only generator G1, resulting in a minimal total generation cost of only \$10000/h. Therefore, the bus-splitting action at the substation level yields a cost reduction effect comparable to that of the line-switching action shown in Fig. \ref{fig-b1-4-1}. 

It is important to note that for substation-level transmission topology control problems, the combinatorial explosion of topological actions is particularly severe. \cite{4-1939} derived a formula for the action space size of substation-level transmission topology control problems, showing that even for the IEEE 14-bus system, the number of possible topologies can reach $3.9 \times 10^{11}$, and remains as large as $3.3 \times 10^{5}$ under $N\!-\!1$ security constraints. 
Therefore, substation-level transmission topology control is considerably more challenging than topology control limited to line switching.

\section{State-of-the-Art Review}

The idea of steady-state transmission topology control, i.e, controlling the steady-state network topology to improve transmission network performance, was first proposed by \cite{4-b1-4} and \cite{4-b1-5}. Since then, researchers have primarily focused on exploring and expanding the application domains of steady-state transmission topology control, and have proposed various solution methods to address the associated control problems. In the existing literature, steady-state transmission topology control appears under various terminologies, including network switching, optimal switching configuration, corrective transmission switching, network topology optimization, optimal transmission switching, and transmission switching. While these terms emphasize different aspects, they all fundamentally refer to the control of steady-state topology of transmission networks.

\begin{figure}[h]
	\centering
    \includegraphics[width=1\linewidth]{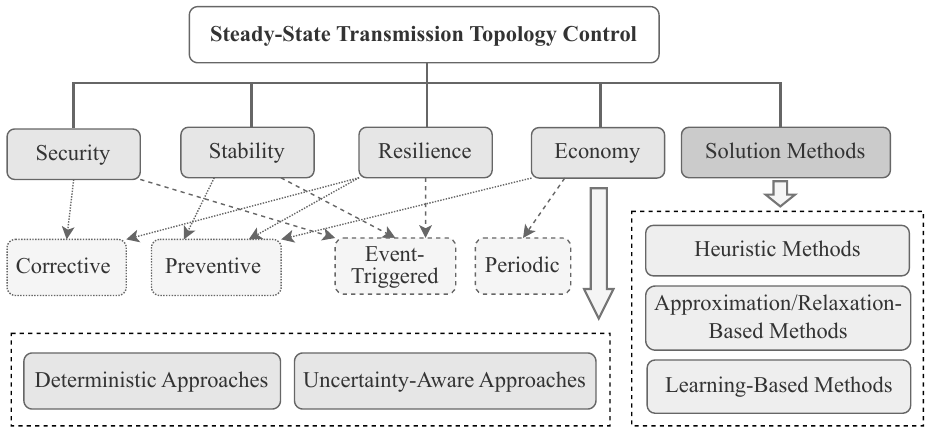}  
	\caption{Structure of the state-of-the-art review on steady-state transmission topology control.}
	\label{fig-b1-4-6}
\end{figure}

This section reviews the existing literature on steady-state transmission topology control according to control objectives, following the structure illustrated in Fig. \ref{fig-b1-4-6}. These objectives primarily include enhancing system security, stability, resilience, and economic efficiency. 
Based on the temporal relationship between topology control actions and system contingencies, steady-state transmission topology control may be implemented in either a preventive or corrective form. 
Preventive control involves topology switching actions that are implemented to improve system adequacy, stability, security, reliability, or resilience margins prior to contingencies, often incorporating anticipated corrective actions for potential future events. Corrective control, in contrast, executes topology switching actions after contingencies occur to alleviate operational violations such as overloads or instability. 
In the existing literature, the adopted control form depends on the specific control objective. For some objectives, only corrective or preventive topology control is addressed, whereas for others both preventive and corrective formulations are available. 

Moreover, the implementation paradigm may differ across objectives. For security-, stability-, and resilience-oriented objectives, steady-state transmission topology control is typically event-triggered, meaning that control actions are activated in response to specific system conditions or contingency occurrences. For economic objectives, it is generally periodic, where control actions are executed at fixed time intervals based on updated system operating points or forecasts. 

In particular, due to the extensive attention received by economy-oriented steady-state transmission topology control, it is further classified in the literature into deterministic and uncertainty-aware formulations for systematic review. 
In addition to addressing steady-state transmission topology control under different objectives, this review also examines representative solution methods, including heuristic methods, approximation/relaxation-based methods, and learning-based methods.

\subsection{Security}

Security enhancement is the major objective of corrective steady-state topology control in the existing literature. 
The idea of steady-state transmission topology control was originally proposed to ensure system security, more specifically, alleviate branch overloads \citep{4-b1-4, 4-b1-5}. 
Subsequent work has focused on the detail formulation and the development of corresponding solution methods for the control problem. 
\cite{4-229} reviewed very early studies on switching operation modeling and the associated search and optimization techniques, emphasizing that switching actions, including transmission lines and bus couplers, can serve as effective control measures for enhancing system security and limiting short-circuit currents. 
\cite{4-860} developed a fast algorithm based on linear sensitivity factors to select network switching actions to relieve system overloads. 
\cite{4-69} proposed a method for general network topology optimization problems, where the objective can be either line current or short-circuit current. The method identifies lines for switching sequentially by employing a current injection model combined with linear programming in each iteration. 
\cite{4-578} proposed a linear programming based methodology for the general corrective control problem that incorporates discrete controls including branch switching.  
\cite{4-b1-8} employed the existing Optimal Power Flow (\ac{opf}) framework to enhance system security by corrective switching, enabling contingencies to be addressed similarly to corrective actions. 
\cite{4-b1-10} presented a fast mathematical technique based on distribution factors to sequentially construct the network topology that eliminates branch overloads and nodal voltage violations. 
\cite{4-b1-7} was the first to consider line switching as a corrective control measure to ensure $N-1$ security in the security-constrained \ac{opf} problem, and developed a single step solution method, drawing on the integration of contingencies and switching actions proposed by \cite{4-b1-8}. 
It is noted that, while the security-constrained \ac{opf} problem is preventive in nature, line switching actions within the formulation are modeled as corrective, as they are executed in the post-contingency stage. 
\cite{4-b1-11} described a security enhancement package integrated within the real-time EMS of a large utility, where network switching is used as an control measure.  
\cite{4-b1-12} presented one of the first practical implementation of the network switching algorithm for overload relief. The algorithm is based on a linearized DC model and verification using AC load flows, with the unique ability to process complicated substation busbar switching scenarios. 
\cite{4-71} advanced earlier research on corrective switching for addressing overloads and voltage violations by proposing an efficient solution method that leverages sparse inverse techniques and fast decoupled power flow to accelerate convergence and reduce computational effort. The same corrective switching problem is further solved using a binary integer programming technique by \cite{4-b1-13}. The adoption of  mixed integer programming techniques for steady-state transmission topology control began during that period \citep{4-b1-14, 4-b1-16}. 
\cite{4-778} proposed a corrective steady-state topology control model for load shed recovery after contingencies to reduce the amount of load shed. 
\cite{4-1136} demonstrated the superiority of corrective switching solutions for the Pennsylvania-New Jersey-Maryland (\ac{pjm}) transmission network, obtained through state-of-the-art corrective topology control methods, compared to those based on operators' prior knowledge and offline studies. They concluded that corrective topology switching is ready for industry adoption. 
\cite{4-779, 4-780} proposed an Energy Management System (\ac{ems}) procedure that includes corrective topology switching in existing \ac{ems} with the concept of branch pseudo limit. It shows that corrective topology switching can achieve significant congestion cost reduction with the proposed \ac{ems} procedure. 
In contrast to previous studies that follow a deterministic control framework, \cite{4-538} proposed a robust corrective topology control approach based on a robust optimization framework, which ensures post-contingency system security under demand uncertainty. 
Recently, the first decision support tool for topological remedial actions, named GridOptions tool, is deployed in real-world \ac{tso} control rooms for day-ahead congestion management \citep{4-1934}.

\subsection{Stability}

Steady-state topology control can also contribute to enhancing the stability of transmission networks in a preventive form. 
\cite{4-841} made an early contribution by proposing the use of preventive steady-state transmission topology control for stability enhancement. 
They developed a three-stage online line switching methodology for increasing load margins to the static stability limit. 
For small-disturbance stability, \cite{4-361} proposed to utilize preventive line switching to enhance the stability margin. A similar three-stage approach, comprising screening, ranking, and detailed evaluation, was designed to enable fast online implementation. 
\cite{4-859} proposed the use of transmission switching to improve synchronism of low-inertia transmission networks. A synchronism metric based on $\mathcal{H}_2$ norm and its sensitivity to perturbation in network susceptance was employed for identifying candidate lines for switching. 
\cite{4-2004} utilized a similar $\mathcal{H}_2$-norm based stability metric to study the optimization of network topology. The task was reformulated as a mixed-integer linear program via McCormick linearization, with its graph-theoretic structure further exploited to improve computational efficiency. 
Compared to the aforementioned stability issues that rely on power flow models or linearized system dynamic models, where sensitivity information is readily available, large-disturbance stability poses significantly greater challenges in the context of steady-state topology control. Recently, \cite{4-1757} proposed a topology optimization method for improving transient stability, utilizing the genetic algorithm and a composite robustness metric as a non-simulation-based proxy transient stability measure.

\subsection{Resilience} 

Power system resilience has drawn growing attention in recent years due to the increasing frequency and severity of extreme weather events and cyber-physical threats. Steady-state transmission topology control can also contribute to enhancing system resilience in either preventive or corrective forms. 
\cite{4-75} proposed the utilization of preventive topology control to mitigate the effects of geomagnetic disturbances. An optimization framework was formulated to enhance power system resilience against the impacts of geomagnetically induced currents by coordinating line switching, generator redispatch, and load shedding. 
To enhance the system resilience, \cite{4-1768} developed an optimization model for corrective steady-state topology control following coordinated cyber-physical attacks on the transmission networks, with the objective of reducing load curtailment. 
To reduce the wildfire risk under high temperatures, some electricity utilities implement public safety power shutoff in which equipment in regions with high wildfire risk is de-energized \cite{4-b1-25}. Transmission lines are the primary components subject to de-energization; consequently, public safety power shutoffs can be regarded as a form of preventive steady-state topology control. \cite{4-b1-24} proposed the optimal power shutoff problem formulated as an optimization model that accounts for the impact of preventive wildfire risk measures on both wildfire risk and short-term power system reliability. \cite{4-b1-23} further improved the optimal power shutoff model by explicitly considering $N-1$ security constraints. 
\cite{4-1942} incorporated decision-dependent wildfire-driven failure probabilities into the public safety power shutoff problem. 
\cite{4-2170} proposed one of the first dynamic topology optimization models amidst wildfire risks as a multi-stage distributionally robust optimization formulation with decision-dependent uncertainty. 

\subsection{Economy}

Improving economic efficiency is the primary objective of preventive steady-state transmission topology control in the existing literature. 
Earlier studies demonstrated the potential of steady-state transmission topology control for operation cost (including network loss) reduction \citep{4-229, 4-69, 4-b1-6}. \cite{4-b1-22, 4-b1-20} proposed the concept of optimal reconfiguration of transmission network with transient stability constraints, wherein transmission systems is optimized by switching lines so that the power system has better performance in terms of both stability and economics. However, \cite{4-62} was the first to specifically utilize steady-state transmission topology control to improve the results of economic dispatch, building on the concept of dispatchable network for transmission markets proposed by \cite{4-b1-15}. This led to the concept of \ac{ots} and a corresponding Mixed-Integer Linear Programming (\ac{milp}) model based on DC power flow. 
Since then, the OTS method has been extended through integration with unit commitment and other operational coordination strategies, and has also been further refined by incorporating practical considerations. 
The following reviews existing work on economy-oriented steady-state transmission topology control by classifying it into deterministic and uncertainty-aware approaches.

\subsubsection{Deterministic approaches}

Most early studies formulated the economy-oriented steady-state transmission topology control problem as a deterministic one, where uncertainties are not explicitly modeled. 
After \cite{4-62} proposed the OTS concept and its preliminary model, \cite{4-64} studied  the impact of \ac{ots} on nodal prices, load payment, generation revenues, cost, and rents, congestion rents, and flowgate, finding that \ac{ots} typically results in lower load payments and higher generation rents. \cite{4-63} further considered $N$-1 security in \ac{ots} and developed a corresponding \ac{milp} formulation, showing that the percentage of cost savings from \ac{ots} was higher when $N$-1 security constraints were considered compared to the case without such considerations. \cite{4-b1-19} accounted for revenue adequacy constraints in \ac{ots}. Furthermore, \cite{4-670} are the first to extended the application of \ac{ots} to the unit commitment problem, where the network topology is adjusted on an hourly basis in coordination with the unit commitment schedule. Around the same time, \cite{4-627} incorporated \ac{ots} into the security-constrained unit commitment problem and applied the Benders decomposition method to enhance computational efficiency. \cite{4-b1-17} analyzed the economics (e.g., generation rent, congestion cost, and load payment) of the \ac{ots} problem integrated with unit commitment, based on the associated dual formulation with integer variables fixed at their optimal solutions. 
\cite{4-1003} addressed static switching security in multi-period transmission switching problems to ensure system security is maintained both prior to and following line switching operations of each time period.
\cite{4-735} considered both voltage security and $N-1$ security in economy-oriented steady-state topology control, with the resulting optimization model solved using Benders decomposition. 
\cite{4-761} addressed network connectedness requirement in transmission switching problems using a tailored branching algorithm, in contrast to earlier work that incorporated insufficient constraints into the optimization models.

While the aforementioned \ac{ots} methods accounted for $N$-1 security, the corrective steady-state transmission topology control was not incorporated as a response to contingencies. 
In fact, under economy-oriented steady-state topology control, corrective actions following a contingency may involve not only switching off but also switching on transmission lines, both of which can be leveraged by corrective steady-state transmission topology control to enhance security. 
\cite{4-b1-18} proposed the just-in-time transmission concept that enables operators to optimize network topology for economic efficiency by preventive steady-state topology control, while ensuring system security after contingencies via corrective steady-state topology control by switching lines back into service. 
Most steady-state topology control methods are designed as a centralized manner and thus more suitable for regional transmission networks. However, for interconnected transmission networks, the implementation of steady-state topology control may need to be carried out in a decentralized manner. 
\cite{4-1141} developed a decentralized method for steady-state transmission topology control, where augmented Lagrangian relaxation is used to realize decomposition. 


\subsubsection{Uncertainty-aware approaches}

For uncertainty-aware formulations of economy-oriented steady-state transmission topology control, uncertainty is explicitly represented within the optimization framework and handled through stochastic, robust, or distributionally robust approaches, depending on how it is modeled. 
In such formulations, uncertainties primarily arise from two sources: contingency-related uncertainties in $N-1$ security constraints, and the uncertainties of Variable Renewable Energy (\ac{vre}) in systems with high levels of renewable energy integration. 
A common approach to consider $N-1$ security is to impose the operational feasibility on a pre-defined contingency set \cite{4-63, 4-670, 4-79, 4-1142, 4-1248}, which however, becomes computationally intensive for large-scale power grids. More importantly, this approach can be prohibitive for $N-k$ security due to the massive number of possible contingencies. 
Accordingly, \cite{4-1127} and \cite{4-1140} formulated the steady-state topology control as a robust optimization model considering uncertain contingencies. Solving the model is computationally efficient but can result in conservative solutions. 

For renewable transmission networks, deterministic steady-state topology control probably fails to ensure cost efficiency and operational security since they optimize with only one possible scenario. Accordingly, most existing economy-oriented steady-state topology control methods tackle \ac{vre} uncertainty under the framework of stochastic optimization or robust optimization where the uncertainties are modelled in terms of probability distributions or uncertainty bounds respectively. Another approach of distributionally robust optimization attempts to combine the strengths of the robust and stochastic approaches using an ambiguity  set of distributions. 
In the stochastic methods, \ac{vre} outputs are modeled as uncertain parameters with known probability distributions. \cite{4-1137} developed a chance-constrained steady-state topology control method where the chance constraints are approximated by the sample average approach. The stochastic steady-state topology control methods with \ac{vre} uncertainty handled by the point estimation method are proposed by \cite{4-79, 4-1223, 4-1158}. 
Two-stage stochastic \ac{ots} methods are developed by \cite{4-1139} and \cite{4-1142}, which differing from the single-stage version \citep{4-1137, 4-79, 4-1223, 4-1158}, capture the corrective control reacting to \ac{vre} uncertainty. 
Unlike deterministic steady-state topology control \citep{4-62, 4-670}, the stochastic control can achieve optimal average performance. However, the computational challenge of stochastic optimization models becomes significant when considering the $N-k$ security or for high-level \ac{vre} penetrated power grids with numerous \ac{vre}-based generators. 
In the robust methods, \ac{vre} outputs are treated as uncertain parameters in an uncertainty set that ignores all distributional knowledge, retaining only the support (i.e., the set of all possible values of \ac{vre} outputs). 
\cite{4-1226} proposed an interval-based two-stage robust steady-state topology control method which maximizes the span of the possible wind variation interval. 
A two-stage robust steady-state topology control method for wind integrated hybrid AC/DC power grids was developed by \cite{4-1127}, where both uncertainties in wind generation and generator failure are considered. In contrast to the stochastic models, the robust models are computationally cheaper and the support information of uncertain parameters (i.e., the range of wind generation outputs and all feasible realizations of generator states) is easily accessible. However, they often raise concerns of over-conservatism, i.e., pursuing better performance in the rare worst-case scenario while sacrificing the average performance. 
More recently, \cite{4-b1-39} developed a two-stage distributionally robust chance-constrained method for steady-state topology control, assuring limited constraint violations for any uncertainty distribution of \ac{vre} outputs within an ambiguity set. 
\cite{4-2127} developed a two-level, three-stage formulation for robust AC \ac{ots}, considering injection uncertainties in the day-ahead operation of power networks. Unlike most formulations, this method avoids power flow approximations and thus guarantees AC-feasible solutions. 
\cite{4-1998} developed a multistep and stochastic model of day-ahead steady-state topology control, optimizing the grid topology hourly with \ac{vre} uncertainty. The model is applicable to both AC grids and hybrid AC/DC grids, while simultaneously accounting for line switching and bus splitting.

\subsection{Solution Methods}

Mathematically, steady-state transmission topology control problems come down to solving Mixed-Integer Nonlinear Programming (\ac{minlp}) models, whose non-convexity from AC power flow equations and combinatorial nature over network topology cause NP-hardness and thus the hindrance to high-performance solution approaches. Moreover, accounting for system dynamics (e.g., stability) and $N-1$ security constraints further exacerbates the solution difficulty due to increased non-convexity and problem scale.

\subsubsection{Approximation/relaxation-based methods}

Approximation/relaxation-based methods enable the resulting optimization problem to be solved directly by off-the-shelf optimization solvers. Approximation-based methods simply employ approximate power flow models. Relaxation-based methods replace the original problem by a related simpler problem such as removing constraints. 
A commonly used approximation is to use DC power flow equations that transform the original \ac{minlp} model into a mixed-integer linear programming formulation \citep{4-b1-16, 4-62, 4-63, 4-670}. 
\cite{4-1132} developed two shift factor-based Mixed-Integer Programming (\ac{mip}) formulations for security-constrained steady-state topology control with substation internal configurations, which starts with a topology with candidate breakers disconnected  and a fully closed network, respectively. 
\cite{4-478} derived strong relaxed formulations for steady-state transmission topology control problems, by utilizing a set of convex quadratic relaxations of power flow equations together with disjunctive programming. 
\cite{4-543} revealed that for the big-M technique widely utilized in the approximation/relaxation-based methods, finding the strongest variable bounds to be used on the big-M inequalities is NP-hard. A bound strengthening method was proposed to strengthen the convex relaxation of the optimization problem in certain cases, which can be treated as a preprocessing step independent of the solver to gain remarkable speedup. 
For the same issue, \cite{4-b1-30} recently proposed an iterative tightening strategy to strengthen the big-M inequalities by solving a series of bounding problems. 
\cite{4-2123} developed strengthened quadratic convex relaxations for AC \ac{ots}, where the relaxation is tightened with several new valid
inequalities by taking advantage of the network structure. This work is the first AC \ac{ots} relaxation-based approach to demonstrate near-optimal switching solutions on realistic large-scale power grid instances. 
For the approximation/relaxation-based methods, control solutions obtained by solving the convex mixed-integer programming model once may not be feasible with respect to the AC power flow equations. To address this issue, \cite{4-1128} proposed a two-level iterative framework to ensure AC feasibility of control solutions, where the upper level solves a mixed integer second-order cone programming model to provide candidate solutions, and the lower level check AC feasibility. \cite{4-61} proposed a new exact formulation of steady-state transmission topology control problems, and its mixed-integer second-order cone programming relaxation improved via strong valid inequalities. A practical algorithm was also proposed to obtain high-quality solutions with provably tight bounds and guaranteed AC feasibility. 
\cite{4-b1-34} proposed an enhanced DC \ac{ots} formulation that facilitates AC feasibility by optimizing DC power flow parameters using machine learning techniques. 
\cite{4-2124} utilized the method called modeling to generate alternatives, to find alternative solutions of DC \ac{ots} problems, thus reducing the chance of finding no AC feasible solutions.

To enhance computational efficiency, various techniques have been explored. To address redundant computation caused by the symmetry from identical transmission lines between two buses, \cite{4-1154} proposed a disjunction-based branching method and also found that adding symmetry-breaking constraint may actually be worse than ignoring symmetry. 
To reduce the number of binary variables in the mixed-integer programming model, \cite{4-1134} proposed a prescreening method to select a few switchable line candidates before solving the model. 
\cite{4-1135} proposed a decomposition approach to solve the large-scale seasonal transmission switching problem, in which the seasonal problem is decomposed into small-scale one-hour problems. 
\cite{4-1129} accelerated an existing mixed-integer linear programming heuristic for AC optimal transmission switching for load shed prevention, by modifying certain computationally-costly constraints in the linear programming model for AC power flows. 
\cite{4-1935} derived the formulas of bus split distribution factors to efficiently compute the effects of busbar splitting on the DC power flow, enabling much faster screening of topological remedial actions at the substation level for congestion management. 
\cite{4-1936} utilized GPUs to accelerate low-rank update based DC loadflow calculation, yielding several orders of magnitudes speed improvement in the frame of topology optimization. 
\cite{4-b1-32} proposed an efficient algorithm by iteratively solving the \ac{milp} model incorporating with a fast heuristic to improve the incumbent solution.

\subsubsection{Heuristic methods}

By exploiting the structural properties of optimization problems, more efficient heuristic algorithms can be designed but at the expense of global optimality. For steady-state transmission topology control problems, the common heuristic pattern begins with the initial topology in which all transmission lines are closed, and then sequentially selects lines to open in order to ultimately obtain the optimal topology. Following this heuristic pattern, \cite{4-85}, \cite{4-841}, \cite{4-361}, and \cite{4-859} selected lines to open based on different kinds of sensitivity analysis. In addition, \cite{4-736} selected lines to open using a line-ranking parameter that is computed by solving a sequence of linear programs or mixed-integer linear programs. 
Recently, \cite{4-2125} addressed the \ac{ots} problem by modeling the binary decisions as a continuous switch model. The continuous switch formulation offers physical insights to develop strong heuristics in the form of homotopy methods and Newton-Raphson dampening that provide scalable and robust convergence. 
\cite{4-2126} developed the theory of linearizing the power flow equations around changes in the complex network admittance parameters. The theory was further utilized to design a greedy algorithm based on continuous relaxation and rounding for solving AC \ac{ots} problems. 
\cite{4-2160} proposed a randomized switching algorithm that forms a configuration distribution concentrating congestion near its optimum. The congestion metric is shown to be a generalized self-concordant function over the switching probability space, enabling efficient conditional gradient optimization. This algorithm can efficiently computes nearly-optimal solutions to a class of combinatorial reconfiguration problems on weighted, undirected graphs, including some \ac{ots} problems. 
\cite{4-2003} utilized a set of metrics to identify and rank candidate buses for busbar splitting where topology optimization can effectively reduce generation costs. 
\cite{4-2178} developed a new \ac{milp} formulation for \ac{ots} with de-energization that represents post-contingency loss of connectivity without requiring additional binary variables. This formulation provides the foundation for a fast algorithm for finding feasible solutions by integrating a column-and-constraint algorithm and variable neighborhood local search strategies.

\subsubsection{Learning-based methods}

Machine learning techniques, especially reinforcement learning and graph neural networks, have been recently employed to solve steady-state transmission topology control problems. \cite{4-1156} applied machine learning algorithms to prioritize the possible line switching actions, and also developed learning-based algorithm selectors among existing transmission switching algorithms. 
\cite{4-b1-37} proposed a framework to learn the solution of steady-state transmission topology control through imitation and reinforcement learning. 
\cite{4-1743} proposed an AlphaZero-based grid topology optimization agent for congestion management. 
\cite{4-b1-36} and \cite{4-b1-38} both proposed an hierarchical reinforcement learning framework for steady-state transmission topology control, to address the challenge posed by the combinatorial nature of the action space. 
\cite{4-b1-27} trained a graph neural network offline to predict the best topology before entering the optimization stage, thereby accelerating the overall solution process. 
\cite{4-b1-26} developed a graph-enhanced model-free reinforcement learning agent for steady-state transmission topology control. 
\cite{4-b1-28} proposed a graph-based soft-label imitation leaning approach for steady-state transmission topology control. 
\cite{4-b1-29} proposed the first graph neural network model for steady-state topology control that uses only graph neural network layers. By identifying the busbar information asymmetry problem that the homogeneous graph representation suffers from, a heterogeneous graph representation was proposed to resolve it. 
\cite{4-b1-33} proposed a centrally coordinated multi-agent architecture for action space factorization of the reinforcement learning-based steady-state transmission topology control. 
\cite{4-2177} proposed a dispatch-aware deep neural network that accelerates DC-OTS without relying on pre-solved labels and embeds an OPF layer enabling generalization to untrained system configurations, such as changes in line flow limits. 
\cite{4-2171} developed a safe and intelligent transmission topology control framework that integrates large language models (LLMs) with a safety soft actor-critic architecture. In this architecture, a knowledge-based Safety-LLM module refines unsafe or suboptimal transitions through domain knowledge and state-informed reasoning, guiding the learning agent toward safer and more effective switching actions. 
Moreover, the French Transmission System Operator (\ac{tso}) RTE and collaborators launched a series of competitions, i.e., the Learning to Run a Power Network challenge, providing an open benchmark for applying learning-based methods to realistic power network operations with topology control \citep{4-1937}. 
Participants employed reinforcement learning and hybrid artificial intelligence approaches. While demonstrating the potential of artificial intelligence for topology control, these approaches still fail in a significant fraction of scenarios, leaving substantial room for improvement \citep{4-1941}. 
\cite{4-2007} provided a comprehensive and structured overview of the competitions and reinforcement learning applications for steady-state transmission topology control, along with comparative numerical study of commonly applied reinforcement-learning-based methods.

\section{Recent Advance: Topology Control with Multiple Uncertainties}

For steady-state transmission topology control in power systems with high renewable penetration, a major challenge arises from the coexisting uncertainties in renewable generation and contingencies related to security constraints. To address this issue, \cite{4-1308} proposed a three-stage \ac{ots} method considering both uncertainties and their associated corrective controls. 
This section presents an introduction and numerical examples on practical transmission networks of this method. 

\subsection{Overall Framework and Uncertainty Modeling}

The overall framework of the three-stage \ac{ots} model is illustrated in Fig. \ref{book-1-4-1}. Each stage is described as follows: 
\begin{itemize}
    \item The first stage has the primary purpose to schedule the power generation and network topology based on the forecast of \ac{vre}. Voltage magnitudes of voltage-controlled buses and reactive power outputs of \ac{vre}-based generators are optimized simultaneously.
    \item The second stage finds the corrective controls of the power generation and voltage magnitudes of voltage-controlled buses, reacting to the \ac{vre} forecast errors due to \ac{vre} uncertainty. Corrective line switching is not considered since the main target of this stage is to correct the mild power imbalance ensuing from the \ac{vre} forecast errors. 
    \item The third stage corrective control responds to an $N-k$ contingency by corrective controls including line switching, generation redispatch, voltage regulation, and load shedding. When no contingency occurs, the system operates with the topology scheduled at the first stage, and the power generation and voltage magnitudes corrected by the second stage, and the third stage is inactive. 
\end{itemize}

\begin{figure}[h]
	\centering
    \includegraphics[width=1\linewidth]{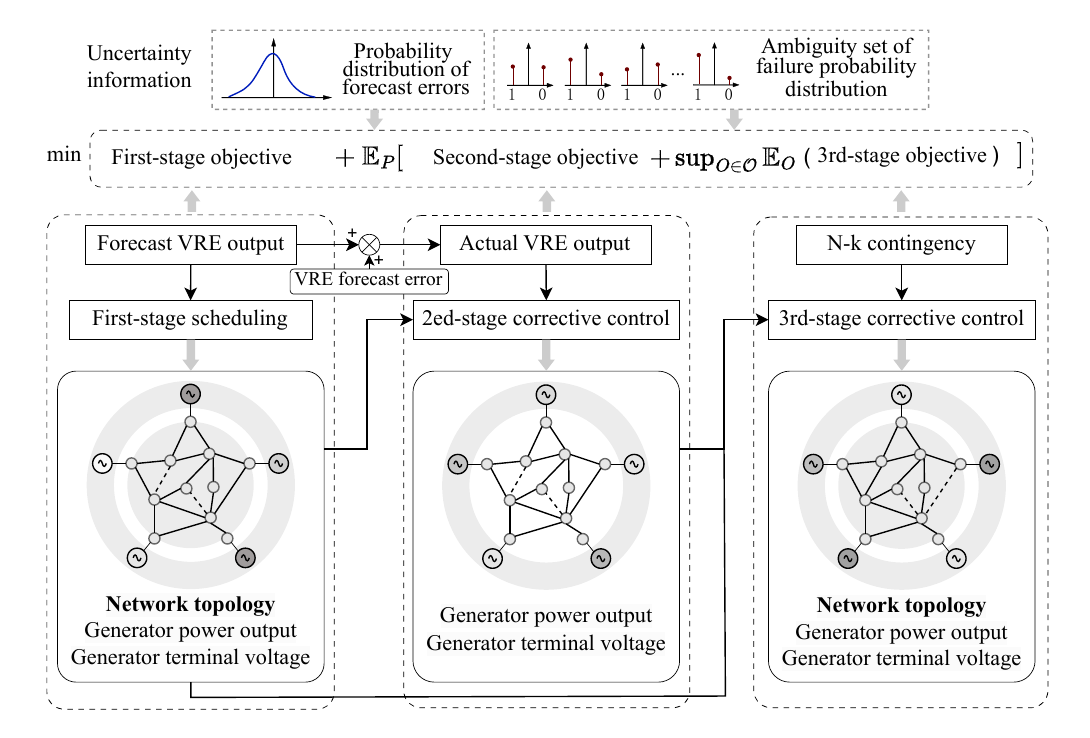}  
	\caption{Overall framework of the three-stage \ac{ots} model.}
	\label{book-1-4-1}
\end{figure}

Regarding the uncertainty information, the \ac{vre} forecast errors, denoted as $\bm{\varepsilon} \in \mathbb{R}^{n_{\rm gv}}$, are assumed to follow a probability distribution denoted as $P$, which can be obtained from historical data. Denote by $\bm{o} \in \mathbb{B}^{n_{\rm g} + n_{\rm e}}$ the parameterization vector of $N-k$ contingencies of generators and
branches, whose entry values of 1 and 0 indicate the normal and failure state of associated components. The contingency uncertainty $\bm{o}$ is assumed to follow a probability distribution denoted as $O$, which belongs to a certain ambiguity set $\mathcal{O}$. Specifically, the support of $\bm{o}$, denoted by $\Omega$, can be formulated as: 
\begin{equation} 
         \Omega = \{ \bm{o} \in \mathbb{B}^{n_{\rm e} + n_{\rm g}}: n_{\rm e} + n_{\rm g} - \bm{1}\T \bm{o} \leq k_{\max}  \} 
\end{equation} 
An accurate probability distribution of components failures is generally unattainable since the failure probabilities are close to 0 and most failures have not actually happened in historical data \citep{4-1214}. Thus, the following moment-based ambiguity set is adopted:
\begin{equation}\label{eq-3-6-ambiguity-set-1}
    \mathcal{O} = \{ O \in \mathcal{M}(\Omega, \mathcal{F}):  {\bm{o}}_{\rm min}  \leq   \mathbb{E}_{O}( \bm{1} - \bm{o} )   \leq {\bm{o}}_{\rm max}    \}
\end{equation}
where $\mathcal{F}$ is the $\xi$-field of $\Omega$; 
$\mathcal{M}(\Omega, \mathcal{F})$ is the set of all probability distributions defined on $(\Omega, \mathcal{F})$;
 ${\bm{o}}_{\rm min}$ and ${\bm{o}}_{\rm max}$ are the vectors of lower and upper bounds of probabilities of components failures, respectively. Concisely, (\ref{eq-3-6-ambiguity-set-1}) can be rewritten as
\begin{equation}\label{eq-3-6-ambiguity-set}
    \mathcal{O} = \{ O \in \mathcal{M}(\Omega, \mathcal{F}): \mathbb{E}_{O}( \bm{T} (\bm{1} - \bm{o}) )  \leq \tilde{\bm{o}}  \}
\end{equation}
where $\bm{T} = [ \bm{I}_{n_{\rm e} + n_{\rm g}}, -\bm{I}_{n_{\rm e} + n_{\rm g}} ]\T $ and $\tilde{\bm{o}} = [{\bm{o}}_{\max}\T, -{\bm{o}}_{\min}\T]\T$. 

With the aforementioned uncertainty modeling, the first-stage scheduling aims to optimize the overall performance integrating the objectives of all three stages, with $\bm{\varepsilon}$ and $\bm{o}$ handled by stochastic optimization and distributionally robust optimization respectively. 
Accordingly, the three-stage \ac{ots} model determines the optimal network topology and generation dispatch that jointly optimize normal operations under forecasted \ac{vre}, average performance under \ac{vre} uncertainty, and worst-case average performance under contingency uncertainty.

\subsection{Compact Form of the Three-Stage \ac{ots} Model}

The three-stage \ac{ots} model can be compactly formulated as follows:
\begin{subequations}\label{eq-3-6-final-stage-1}
    \begin{align}
        Q_1 := \min_{\bm{x} \in \mathbb{R}^{n_1} \times \mathbb{B}^{m_1}} ~&
        \bm{c}_{\rm 1}\T \bm{x} +   \mathbb{E}_{{P}}[ Q_{\rm 2}( \bm{x}, \bm{\varepsilon} ) ] \label{eq-3-6-final-stage-1:1}\\
        \text{s.t.} ~& 
        \bm{A} \bm{x} \leq \bm{b}  \label{eq-3-6-final-stage-1:2} 
    \end{align}
\end{subequations}
with the second-stage recourse function
\begin{subequations}\label{eq-3-6-final-stage-2}
    \begin{align}
         Q_2( \bm{x}, \bm{\varepsilon} ) := 
        \min_{\bm{x}_{\bm{\varepsilon}} \in \mathbb{R}^{n_2}} ~ & \bm{c}_{\rm 2}\T \bm{x}_{\bm{\varepsilon}} 
        + 
        \sup_{ O \in \mathcal{O} }   \mathbb{E}_{ O }   
            [ Q_3( \bm{x}_{\bm{\varepsilon}}, \bm{x}, \bm{o} ) ]  \label{eq-3-6-final-stage-2:1} \\
        \text{s.t.} ~ & 
        \bm{C} \bm{x}_{\bm{\varepsilon}} + \bm{D} \bm{x} \geq  \bm{d}( \bm{\varepsilon} )   \label{eq-3-6-final-stage-2:2} 
    \end{align}
\end{subequations}
with the third-stage recourse function
\begin{subequations}\label{eq-3-6-final-stage-3}
    \begin{align}
         Q_3( \bm{x}, \bm{x}_{\bm{\varepsilon}}, \bm{o} )   &  := 
        \min_{\bm{x}_{\bm{o}}^{\bm{\varepsilon}} \in \mathbb{R}^{n_3} \times \mathbb{B}^{m_3}} ~   \bm{c}_{\rm 3}\T \bm{x}_{\bm{o}}^{\bm{\varepsilon}}   \label{eq-3-6-final-stage-3:1}\\
        \text{s.t.} ~&  \bm{F} \bm{x}_{\bm{o}}^{\bm{\varepsilon}} + \bm{H} \bm{x}_{\bm{\varepsilon}} +  \bm{L} \bm{x} \geq \bm{f} \label{eq-3-6-final-stage-3:2} \\
        ~ &  \bm{M} \bm{x}_{\bm{o}}^{\bm{\varepsilon}} + \bm{N} \bm{o} \geq \bm{g}  \label{eq-3-6-final-stage-3:3} \\
        ~ & \bm{R} \bm{x}_{\bm{o}}^{\bm{\varepsilon}} + \bm{S} (\bm{o}) [\bm{x}\T ~\bm{x}_{\bm{\varepsilon}}\T ~(\bm{x}_{\bm{o}}^{\bm{\varepsilon}})\T ]\T   \geq \bm{h} \label{eq-3-6-final-stage-3:4}\\
        ~ & [\bm{x}_{\bm{o}}^{\bm{\varepsilon}}]_{\rm c} \in \argmin_{ [\bm{x}_{\bm{o}}^{\bm{\varepsilon}}]_{\rm c} }  \{ \bm{l}\T [\bm{x}_{\bm{o}}^{\bm{\varepsilon}}]_{\rm c}: {  \bm{V} [\bm{x}_{\bm{o}}^{\bm{\varepsilon}}]_{\rm z} + \bm{W} [\bm{x}_{\bm{o}}^{\bm{\varepsilon}}]_{\rm c} \leq \bm{u}} \} \label{eq-3-6-final-stage-3-5}  
    \end{align}
\end{subequations} 
where 
$\bm{x}$, $\bm{x}_{\bm{\varepsilon}}$ and $\bm{x}_{\bm{o}}^{\bm{\varepsilon}}$ are stage-wise optimization variable vectors, including all control and state variables related to each stage; 
$[\bm{x}_{\bm{o}}^{\bm{\varepsilon}}]_{\rm c}$ denotes the sub-vector of $\bm{x}_{\bm{\varepsilon}}$ corresponding to optimization variables of (\ref{eq-0-2-5}) used to enforce network connectedness, and $[\bm{x}_{\bm{o}}^{\bm{\varepsilon}}]_{\rm z}$ denotes the sub-vector of $\bm{x}_{\bm{o}}^{\bm{\varepsilon}}$ corresponding to $\tilde{\bm{z}} \in \mathbb{B}^{n_{\rm e}}$ that represents the post-contingency topology; 
$n_1$ (or $n_3$) and $m_1$ (or $m_3$) are dimensions of continuous variables and binary variables in $\bm{x}$ (or $\bm{x}_{\bm{o}}^{\bm{\varepsilon}}$), respectively, and $n_2$ is the dimension of $\bm{x}_{\bm{\varepsilon}}$; 
$\bm{A}$, $\bm{C}$, $\bm{D}$, $\bm{F}$, $\bm{H}$, $\bm{L}$, $\bm{M}$, $\bm{N}$, $\bm{R}$, $\bm{S}$, $\bm{V}$ and $\bm{W}$ are proper matrices; $\bm{b}$, $\bm{f}$, $\bm{g}$, $\bm{h}$, $\bm{u}$, $\bm{l}$, and $\bm{c}_i$ are proper vectors; 
$\bm{d}( \bm{\varepsilon} )$ and $\bm{S}(\bm{o})$ denote a proper variable vector and matrix linearly related to $\bm{\varepsilon}$ and $\bm{o}$, respectively; 
$\bm{c}_{\rm 1}\T \bm{x}_{\rm 1}$, $\bm{c}_{\rm 2}\T \bm{x}_{\bm{\varepsilon}}$, and $\bm{c}_{\rm 3}\T \bm{x}_{\bm{o}}^{\bm{\varepsilon}}$ represent the total generation cost of all conventional generators, the regulation cost of active power outputs of all generators, and the corrective control cost for contingency response, respectively. 
More specifically, Eq. (\ref{eq-3-6-final-stage-1:2}) accounts for the power flow equations represented using the linear-programming approximation-based power flow model with variable topology \citep{4-1163}; operational constraints on output power of generators, power factors of VRE-based generators, bus voltage magnitude, branch phase angle difference and branch power; topological constraints on network connectedness and switching actions. Eq. (\ref{eq-3-6-final-stage-2}) incorporates power flow and operational constraints analogously to those in (\ref{eq-3-6-final-stage-1:2}); regulation constraints on generator ramping and equality for the second-stage control variables. Eqs. (\ref{eq-3-6-final-stage-3:2})-(\ref{eq-3-6-final-stage-3-5}) represent power flow and operational constraints that are analogous to those in (\ref{eq-3-6-final-stage-1:2}); topology coupling constraints between the first-stage network topology, branch contingencies, post-contingency topology, corrective steady-state topology control, and third-stage topology; and network connectedness constraints as introduced in Section \ref{sec-b1-3-2}.

\subsection{Tractable Reformulation and Solution}\label{sec-4-2}

The three-stage \ac{ots} model cannot be solved directly due to its complex three-stage and bi-level structure. However, a traceable reformulation can be derived to enable the application of existing solution algorithms. 

First, for the recourse function $Q_3$ formulated as a bi-level optimization problem, the lower-level linear programming model (\ref{eq-3-6-final-stage-3-5}) can be replaced by its necessary and sufficient Karush-Kuhn-Tucker (\ac{kkt}) conditions as follows:
\begin{subequations}\label{eq-3-6-KKT}
    \begin{align}
        & \bm{W} \bm{x}_{\bm{o}}^{\bm{\varepsilon}} \leq \bm{u} - \bm{V} [\bm{x}_{\bm{o}}^{\bm{\varepsilon}}]_{\rm z}, \bm{\lambda}_{\bm{o}}^{\bm{\varepsilon}} \geq \bm{0}, \bm{W}\T \bm{\lambda}_{\bm{o}}^{\bm{\varepsilon}} = \bm{l}  \label{eq-3-6-KKT:1} \\
        & \bm{u} - \bm{V} [\bm{x}_{\bm{\varepsilon}}]_{\rm z} - \bm{W} [\bm{x}_{\bm{\varepsilon}}]_{\rm c}  \leq M (\bm{1} - \bm{\xi}_{\bm{o}}^{\bm{\varepsilon}} ), \bm{\lambda}_{\bm{o}}^{\bm{\varepsilon}} \leq M \bm{\xi}_{\bm{o}}^{\bm{\varepsilon}} \label{eq-3-6-KKT-MI}
    \end{align}
\end{subequations}
where $\bm{\lambda}_{\bm{o}}^{\bm{\varepsilon}} \in \mathbb{R}^{n_{\rm u}'}$ with $n_{\rm u}'$ being the dimension of $\bm{u}$, and $\bm{\xi}_{\bm{o}}^{\bm{\varepsilon}} \in \mathbb{B}^{n_{\rm u}'}$.

Second, considering that relatively complete recourse of the three-stage \ac{ots} model cannot always be ensured, (\ref{eq-3-6-final-stage-2}) and (\ref{eq-3-6-final-stage-3}) are reformulated as penalized problems. Specifically, the second-stage recourse function is reformulated as
\begin{subequations}\label{eq-3-6-reform-2}
    \begin{align}
         Q_2( \bm{x}, \bm{\varepsilon} ) := 
        \min_{\tilde{\bm{x}}_{\bm{\varepsilon}} \in \mathbb{R}^{\tilde{n}_2}} ~ & \tilde{\bm{c}}_{\rm 2}\T \tilde{\bm{x}}_{\bm{\varepsilon}} 
        + \sup_{ O \in \mathcal{O} }   \mathbb{E}_{ O }   
        [ Q_3( \bm{x}, \bm{x}_{\bm{\varepsilon}}, \bm{o} ) ]   \label{eq-3-6-reform-2:1} \\
        \text{s.t.} ~ & 
        \tilde{\bm{C}} \tilde{\bm{x}}_{\bm{\varepsilon}} + \tilde{\bm{D}} \bm{x} \geq  \tilde{\bm{d}}( \bm{\varepsilon} )   \label{eq-3-6-reform-2:2} 
    \end{align}
\end{subequations}
where $\tilde{\bm{x}}_{\bm{\varepsilon}}$ is the augmentation of $\bm{x}_{\bm{\varepsilon}}$ with $y_{\bm{\varepsilon}} \in \mathbb{R}$ denoting the slack variable for the equality between the forecast and actual maximal outputs of \ac{vre}-based generators, and $\tilde{\bm{c}}_{\rm 2}\T \tilde{\bm{x}}_{\bm{\varepsilon}} = \bm{c}_2\T \bm{x}_{\bm{\varepsilon}} + \sigma_2 y_{\bm{\varepsilon}}$ with $\sigma_2$ being the penalty coefficient. 
Similarly, the penalized reformulation of (\ref{eq-3-6-final-stage-3}) can be formed by adding $-y_{\bm{o}}^{\bm{\varepsilon}} \bm{1}$ to the RHS of constraints (\ref{eq-3-6-final-stage-3:2})-(\ref{eq-3-6-final-stage-3:4}), where $y_{\bm{o}}^{\bm{\varepsilon}} \in \mathbb{R}$ is the slack variable; and replacing the objective function in (\ref{eq-3-6-final-stage-3:1}) by $\tilde{\bm{c}}_{\rm 3}\T \tilde{\bm{x}}_{\bm{o}}^{\bm{\varepsilon}} = \bm{c}_3\T \bm{x}_{\bm{o}}^{\bm{\varepsilon}} + \sigma_3 y_{\bm{o}}^{\bm{\varepsilon}} $, with $\sigma_3$ being the penalty coefficient and $\tilde{\bm{x}}_{\bm{o}}^{\bm{\varepsilon}} = [(\bm{x}_{\bm{o}}^{\bm{\varepsilon}})\T ~ \bm{\lambda}\T ~ \bm{\xi}\T ~ y_{\bm{o}}^{\bm{\varepsilon}}]\T$.  
Then, substituting (\ref{eq-3-6-final-stage-3-5}) in the penalized reformulation by (\ref{eq-3-6-KKT:1}) and (\ref{eq-3-6-KKT-MI}) yields: 
\begin{subequations}\label{eq-3-6-reform}
    \begin{align}
         Q_3( \bm{x}, \bm{x}_{\bm{\varepsilon}}, \bm{o} )    := 
        \min_{\tilde{\bm{x}}_{\bm{o}}^{\bm{\varepsilon}} \in \mathbb{R}^{\tilde{n}_3} \times \mathbb{B}^{\tilde{m}_3}} ~ &   \tilde{\bm{c}}_{\rm 3}\T \tilde{\bm{x}}_{\bm{o}}^{\bm{\varepsilon}}  \label{eq-3-6-reform:1}\\
        \text{s.t.} ~&  \tilde{\bm{F}} \tilde{\bm{x}}_{\bm{o}}^{\bm{\varepsilon}} + \tilde{\bm{H}} {\bm{x}}_{\bm{\varepsilon}} +  \tilde{\bm{L}} \bm{x} \geq \tilde{\bm{f}} \label{eq-3-6-reform:2} \\  
        ~ &  \tilde{\bm{M}} \tilde{\bm{x}}_{\bm{o}}^{\bm{\varepsilon}} + \tilde{\bm{N}} \bm{o} \geq \tilde{\bm{g}}  \label{eq-3-6-reform:3} \\ 
        ~ & \tilde{\bm{R}} \tilde{\bm{x}}_{\bm{o}}^{\bm{\varepsilon}} + \tilde{\bm{S}}(\bm{o}) [\bm{x}\T ~\bm{x}_{\bm{\varepsilon}}\T ~(\bm{x}_{\bm{o}}^{\bm{\varepsilon}})\T ]\T  \geq \tilde{\bm{h}}   \label{eq-3-6-reform:4}
    \end{align}
\end{subequations}
where 
(\ref{eq-3-6-reform:2}) represents the slacked Eq. (\ref{eq-3-6-final-stage-3:2}), Eq. (\ref{eq-3-6-KKT:1}), and Eq. (\ref{eq-3-6-KKT-MI}); 
(\ref{eq-3-6-reform:3}) and (\ref{eq-3-6-reform:4}) denote the slacked Eq. (\ref{eq-3-6-final-stage-3:3}) and slacked Eq. (\ref{eq-3-6-final-stage-3:4}), respectively. With slack variables, the three-stage \ac{ots} model has relatively complete recourse, provided Eq. (\ref{eq-3-6-final-stage-1:2}) is feasible.

Third, let the problem $\sup\nolimits_{O \in \mathcal{O}}  \mathbb{E}_{O} [Q_3( \bm{x}, \bm{x}_{\bm{\varepsilon}}, \bm{o} )] $ be denoted as $\mathcal{Q}_3( \bm{x}, \bm{x}_{\bm{\varepsilon}})$. It can then be formulated as: 
\begin{subequations}\label{eq-3-6-before-duality}
    \begin{align}
        \mathcal{Q}_3(\bm{x}, \bm{x}_{\bm{\varepsilon}}) = \sup_{ O \in \mathcal{O} }  &  \sum\nolimits_{\bm{o} \in \Omega_{\bm{\varepsilon}} }  Q_3( \bm{x}, \bm{x}_{\bm{\varepsilon}}, \bm{o} ) O( \bm{o}) \label{eq-3-6-before-duality:1}\\
        \text{s.t.} &  \sum\nolimits_{\bm{o} \in \Omega_{\bm{\varepsilon}}} \bm{T} (\bm{1} - \bm{o}) O(\bm{o}) \leq \tilde{\bm{o}}, \sum\nolimits_{\bm{o} \in \Omega_{\bm{\varepsilon}}}   O(\bm{o}) = 1  \label{eq-3-6-before-duality:2}
    \end{align}
\end{subequations}
Here $\Omega_{\bm{\varepsilon}} = \Omega$ while $\Omega_{\bm{\varepsilon}}$ will be specified differently later. According to the strong duality theory, (\ref{eq-3-6-before-duality}) admits the following dual formulation:
\begin{subequations}\label{eq-3-6-duality}
    \begin{align}
          \mathcal{Q}_3(\bm{x}, \bm{x}_{\bm{\varepsilon}}) =  \min_{ \bm{\lambda}_{\bm{\varepsilon}}, \lambda_{\bm{\varepsilon}}' } ~&  \tilde{\bm{o}}\T \bm{\lambda}_{\bm{\varepsilon}}  + \lambda_{\bm{\varepsilon}}' \label{eq-3-6-duality:1}\\
          \text{s.t.} ~&  \bm{\lambda}_{\bm{\varepsilon}}\T \bm{T} (\bm{1} - \bm{o}) + \lambda_{\bm{\varepsilon}}' \geq Q_3( \bm{x}, \bm{x}_{\bm{\varepsilon}}, \bm{o} ) ~ \forall \bm{o} \in \Omega_{\bm{\varepsilon}} \label{eq-3-6-duality:2}\\
        & \bm{\lambda}_{\bm{\varepsilon}} \geq \bm{0} \label{eq-3-6-duality:3}
    \end{align}
\end{subequations}
where $\bm{\lambda}_{\bm{\varepsilon}}$ and $\lambda_{\bm{\varepsilon}}'$ are dual variables for the two constraints of (\ref{eq-3-6-before-duality:2}), respectively. By further replacing $Q_3( \bm{x}, \bm{x}_{\bm{\varepsilon}}, \bm{o} )$ in (\ref{eq-3-6-duality:1}) with the objective function in (\ref{eq-3-6-reform}) and augment constraints of (\ref{eq-3-6-duality}) with that of (\ref{eq-3-6-reform}), $\mathcal{Q}_3( \bm{x}, \bm{x}_{\bm{\varepsilon}})$ can be equivalently formulated as: 
\begin{subequations}\label{eq-3-6-duality-sub}
    \begin{align}
         \mathcal{Q}_3( \bm{x}, \bm{x}_{\bm{\varepsilon}}) & =  \min_{ \bm{\lambda}_{\bm{\varepsilon}}, \lambda_{\bm{\varepsilon}}', \{ \tilde{\bm{x}}_{\bm{o}}^{\bm{\varepsilon}} \} } ~ \tilde{\bm{o}}\T \bm{\lambda}_{\bm{\varepsilon}}  + \lambda_{\bm{\varepsilon}}' \label{eq-3-6-duality-sub:1}\\
        \text{s.t.} ~&  \bm{\lambda}_{\bm{\varepsilon}}\T \bm{T} (\bm{1} - \bm{o}) + \lambda_{\bm{\varepsilon}}' \geq  \tilde{\bm{c}}_{\rm 3}\T \tilde{\bm{x}}_{\bm{o}}^{\bm{\varepsilon}}, \text{(\ref{eq-3-6-reform:3}), (\ref{eq-3-6-reform:4})} ~ \forall \bm{o} \in \Omega_{\bm{\varepsilon}} \label{eq-3-6-duality-sub:2}\\
        & \text{(\ref{eq-3-6-duality:3}), (\ref{eq-3-6-reform:2}) }  \label{eq-3-6-duality-sub:3} 
    \end{align}
\end{subequations}
Furthermore, substituting (\ref{eq-3-6-duality-sub}) into (\ref{eq-3-6-reform-2}) yields the equivalent second-stage recourse function:
\begin{equation}\label{eq-3-6-stage2-final} 
    \begin{aligned}
        Q_2( \bm{x}, \bm{\varepsilon} ) :=
        \min_{ \tilde{\bm{x}}_{\bm{\varepsilon}}, \{\tilde{\bm{x}}_{\bm{o}}^{\bm{\varepsilon}} \}, \bm{\lambda}_{\bm{\varepsilon}}, \lambda_{\bm{\varepsilon}}' }  & \tilde{\bm{c}}_2\T \tilde{\bm{x}}_{\bm{\varepsilon}} + \tilde{\bm{o}}\T \bm{\lambda}_{\bm{\varepsilon}}  + \lambda_{\bm{\varepsilon}}' \\
          \text{s.t.} & 
        \text{ (\ref{eq-3-6-reform-2:2}), (\ref{eq-3-6-duality-sub:2}), (\ref{eq-3-6-duality-sub:3}) } 
    \end{aligned}
\end{equation}

Finally, by replacing the expectation in (\ref{eq-3-6-final-stage-1}) with its scenario-based approximation, the final tractable reformulation of the three-stage \ac{ots} model is obtained as follows:
\begin{equation}\label{eq-3-6-final-opt-model}
    \begin{aligned}
        & Q_1 = \min_{\bm{x}, \{ \tilde{\bm{x}}_{\bm{\varepsilon}} \}, \{ \tilde{\bm{x}}_{\bm{o}}^{\bm{\varepsilon}} \}, \{\bm{\lambda}_{\bm{\varepsilon}}\}, \{\lambda_{\bm{\varepsilon}}' \} } 
        \bm{c}_{\rm 1}\T \bm{x} + \sum_{ \bm{\varepsilon} \in \Xi }   {P}( \bm{\varepsilon} ) \left[ \tilde{\bm{c}}_2\T \tilde{\bm{x}}_{\bm{\varepsilon}} + \tilde{\bm{o}}\T \bm{\lambda}_{\bm{\varepsilon}}  + \lambda_{\bm{\varepsilon}}' \right]  \\
        &~~~~~~~~~~~~~~~~~~~ \text{s.t.}   \text{(\ref{eq-3-6-final-stage-1:2})}, \{ \text{(\ref{eq-3-6-reform-2:2}), (\ref{eq-3-6-duality-sub:2}), (\ref{eq-3-6-duality-sub:3}) } | \forall \bm{\varepsilon} \in \Xi  \}
    \end{aligned}
\end{equation}
where $\Xi$ is the set of finite number of scenarios of $\bm{\varepsilon}$.

The tractable reformulation (\ref{eq-3-6-final-opt-model}) can be solved by the nested Column-and-Constraint Generation (\ac{ccg}) algorithm proposed by \cite{4-1265}. To overcome the slow convergence issue of the \ac{ccg} algorithm \citep{4-1213}, Dantzig–Wolfe procedure can be employed to update the subset of contingency scenarios under which cutting planes are defined. Moreover, for the inner \ac{ccg} loop, instead of using the \ac{kkt} condition, strong duality theory can be employed to derive a tractable formulation of the master problem to avoid introducing massive auxiliary binary variables.

\subsection{Numerical Example}\label{sec-4-4}

The three-stage \ac{ots} method is tested on two real-world transmission networks: the 50Hertz control area of the German transmission network, and the transmission network of Liaoning province in China. The 50Hertz transmission network is the 380 and 220 kV transmission area in the north-eastern Germany (colored in red and blue in Fig. \ref{book-1-4-2}). The network contains 76 buses, 166 transmission lines, 16 transformers, 93 conventional generators, 87 wind farms and 85 solar farms. The Liaoning transmission network comprises the 220kV and 550kV transmission system located in Liaoning province, China,  as illustrated in Fig. \ref{book-1-4-2-add}. It contains 772 buses, 970 transmission lines, 143 two-phase transformers, 62 three-phase transformers, 118 conventional generators and 38 wind farms.

\begin{figure}[h]
	\centering
    \includegraphics[width=0.66\linewidth]{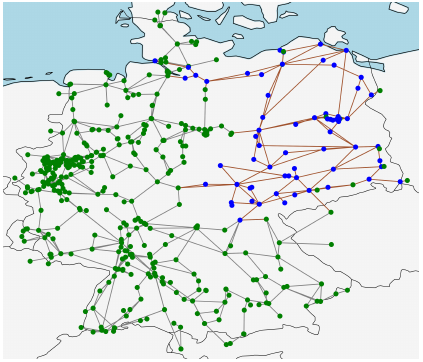}  
	\caption{Diagram of the 50Hertz transmission network.}
	\label{book-1-4-2}
\end{figure}

\begin{figure}[h]
	\centering
    \includegraphics[width=0.9\linewidth]{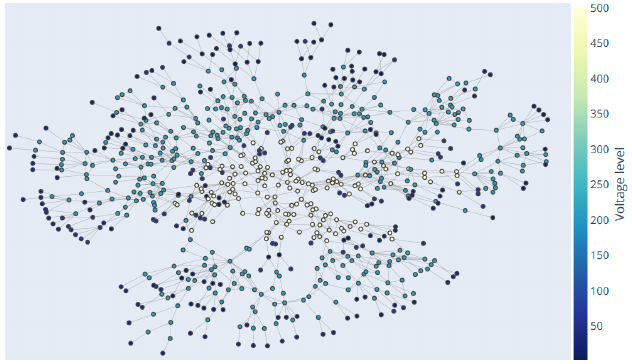}  
	\caption{Diagram of the Liaoning transmission network.}
	\label{book-1-4-2-add}
\end{figure}

For each transmission network, the three-stage \ac{ots} model and its variation where line switching is not allowed are both solved for 100 time-series scenarios with a 1-hour interval. 
The associated results are given by Fig. \ref{fig-3-6-r-3} and Fig. \ref{fig-3-6-r-4}, where $\Delta Q_1$ denotes the relative change from the value of $Q_1$ without lie switching to that with line switching. 
In Fig. \ref{fig-3-6-r-4}, the missing points in the blue curve indicate that values of $Q_1$ are abnormally large or the optimization model is infeasible, and the orange bars indicate invalid values of $\Delta Q_1$. 

\begin{figure}[t!]
	\centering
    \includegraphics[width=0.85\linewidth]{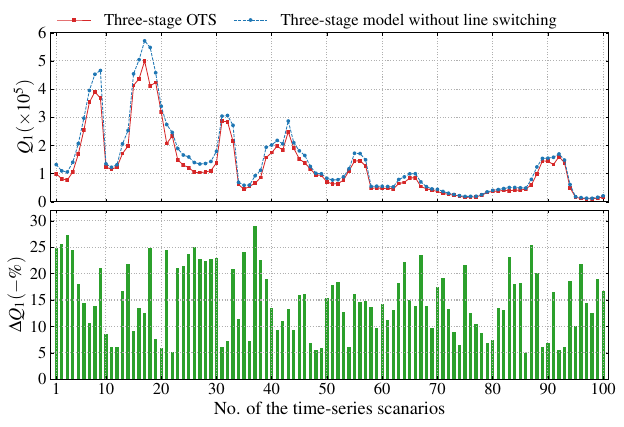}  
	\caption{Results of the three-stage \ac{ots} method applied to the 50Hertz network.}
	\label{fig-3-6-r-3}
\end{figure}
\begin{figure}[t!]
	\centering
    \includegraphics[width=0.85\linewidth]{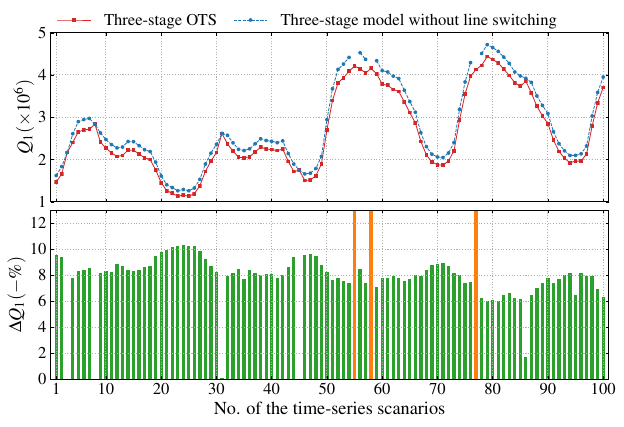}  
	\caption{Results of the three-stage \ac{ots} method applied to the Liaoning network.}
	\label{fig-3-6-r-4}
\end{figure}

According to Fig. \ref{fig-3-6-r-3}, for the 50Hertz transmission network, the three-stage \ac{ots} method reduces total operating cost $Q_1$ for all the scenarios compared with the case where the network topology is not optimized. The maximal rate of decrease is close to 30\%, while the minimum is still about 5\%. According to Fig. \ref{fig-3-6-r-4}, for the Liaoning transmission network,  the three-stage \ac{ots} method reduces $Q_1$ by 6\% to 10\% for the majority of the scenarios. In particular, for some scenarios, i.e., the 55th, 58th, and 77th ones, no feasible solution exists when the network topology is fixed, while the three-stage \ac{ots} method ensures the feasibility by leveraging the flexibility of network topology.  

Furthermore, the three-stage \ac{ots} method is compared to the two-stage stochastic \ac{ots} established following the framework developed by \cite{4-1139} and \cite{4-1142}. 
Under this framework, the first stage is the same as that in the three-stage \ac{ots} model; the second stage is the corrective control reacting to probabilistic scenarios representing the simultaneous occurrence of \ac{vre} forecast errors and contingencies. 
The second-stage corrective control in the two-stage model is set as the third-stage corrective control in the three-stage \ac{ots} model since otherwise the two-stage model is unobtainable. 
Ten $N-1$ contingencies with the probabilities of component failures being $(\bm{o}_{\min} + \bm{o}_{\max})/2$ are considered in the two-stage \ac{ots} model for computational tractability and the same contingencies plus the normal state are used as the support $\Omega$ in the three-stage \ac{ots} model for comparability. Meanwhile, the values of $\bm{o}_{\min}$ and $\bm{o}_{\max}$ are increased to compensate for the impact of reducing $\Omega$. 
Assume that $O^*: \mathbb{E}_{O}(\bm{1} - \bm{o}) = (\bm{o}_{\min} + \bm{o}_{\max})/2$ is the true probability distribution of $\bm{o}$. Considering that the \ac{vre} uncertainty and contingency happen separately in practical, the performance of each \ac{ots} model is evaluated as follows: solve the \ac{ots} model to obtain the first-stage scheduling solution with the associated objective value being $c_1$; with this scheduling solution, for each $\bm{\varepsilon} \in \Xi$, solve the second-stage corrective control problem to obtain the associated regulation cost denoted by $c_{2}^{\bm{\varepsilon}}$; following the second-stage corrective control induced by $\bm{\varepsilon}$, for each $\bm{o} \in \Omega$, solve the third-stage corrective control problem to obtain the associated corrective control cost $c_{3}^{\bm{\varepsilon}, \bm{o}}$; compute $Q_1^* = c_1 + \sum_{\bm{\varepsilon} \in \Xi} P(\bm{\varepsilon}) [ c_{2}^{\bm{\varepsilon}} + \sum_{\bm{o} \in \Omega} O^*(\bm{\varepsilon}) c_{3}^{\bm{\varepsilon}, \bm{o}} ]$, which is the total operational cost with the true probability distribution of $\bm{o}$ and practical sequence of uncertainty occurrence.

\begin{table}[h]
    \setlength{\tabcolsep}{5pt}
	\centering
    \caption{Values of $Q_1^*$ for different \ac{ots} models under some scenarios} 
    \begin{tabular}{|ccccccc|} 
        \hline\hline
        \parbox[t]{4mm}{\multirow{3}{*}{~\rotatebox[origin=c]{90}{{German~}}}} 
               &  Model   & No.7 & No.8 & No.22  & No.29 & No.47 \\ \cline{2-7}
               &  Three-stage  & 382314  & 348972  & 138533   & 131472   & 92456  \\
               &  Two-stage   & 397881  & 357411  & 144780   & 138409   & 96980   \\ \hline\hline
               \parbox[t]{4mm}{\multirow{3}{*}{~\rotatebox[origin=c]{90}{Liaoning~$\!$}}}
                &  Model   & No.4 & No.25 & No.42  & No.77 & No.81 \\ \cline{2-7}
                &  Three-stage & 2600620 & 1129983 & 1817773  & 4023560  & 4014464  \\
                &  Two-stage  & 2663538 & 1162542 & 1844586  & 4181332  & 4080097  \\[0.3mm] 
        \hline\hline
    \end{tabular}  
    \label{table-3-6-3}  
\end{table}

Table \ref{table-3-6-3} compares the values of $Q_1^*$ associated with different \ac{ots} models, for some time-series scenarios of the German and Liaoning networks. The true probability distribution of $\bm{o}$ is employed by the two-stage \ac{ots} model but is not necessarily the worst-case one associated with the optimal first-stage scheduling solution yield by the three-stage \ac{ots} model. Thus, $Q_1^*$ of the two-stage \ac{ots} model should not be larger than that of the three-stage \ac{ots} model for the same time-series scenario. However, for each time-series scenario in Table \ref{table-3-6-3}, the opposite of this inference is observed. This indicates that with the practical sequence of uncertainty occurrence and separate corrective controls, the two-stage \ac{ots} method can lead to inferior first-stage scheduling solutions compared to the three-stage \ac{ots} method.

\graphicspath{{chapter_5/Figs/}}
\chapter{Steady-State Distribution Topology Control}\label{chapter-5}

This chapter presents the steady-state topology control of distribution networks, commonly termed Distribution Network Reconfiguration (\ac{dnr}). First, a canonical formulation, a representative example model, and a classical heuristic solution method for steady-state distribution topology control are introduced. Then, a comprehensive state-of-the-art review of existing solution methodologies for steady-state distribution topology control is provided. Finally, a recent advance in solution methodologies, namely a hybrid learning-heuristic solution paradigm, is discussed.

\section{Fundamentals}

\subsection{Canonical Formulation and Representative Example Model}

We first present a canonical formulation of steady-state distribution topology control problems. Let $\mathcal{G}_{\rm v}(\mathcal{V}, \mathcal{E}_{\rm us} \cup \mathcal{S}, \mathcal{R})$ denote distribution network graph with variable topology, where $\mathcal{S} \subseteq \mathcal{E} \backslash \mathcal{E}_{\rm us}$ is the variable set containing the branches selected to close, and $\mathcal{R}$ is the set collecting all other parameters and variables associated with the problem. 
Set $\mathcal{R}$ is partitioned into the following four subsets: 
\begin{itemize}
    \item[(i)] The constant parameter set $\mathcal{R}_{\rm cp}$, containing the associated parameters invariant across executions of steady-state topology control; 
    \item[(ii)] The time-varying input parameter set $\mathcal{R}_{\rm vp}$, typically including load powers,  Variable Renewable Energy (\ac{vre}) generation power limits, etc.; 
    \item[(iii)] The state variable set $\mathcal{R}_{\rm sv}$, including bus voltage phase angles, load bus voltage magnitudes, and reactive power outputs of generators, which may vary depending on power flow models; 
    \item[(iv)] The coordinated variable set $\mathcal{R}_{\rm cv}$, containing control variables jointly optimized with the network topology. 
\end{itemize}

Then, steady-state distribution topology control problems can be represented in the following canonical form:
\begin{subequations}\label{eq-10-1-1}
    \begin{align}
        \min_{\mathcal{S}, \mathcal{R}_{\rm cv}, \mathcal{R}_{\rm sv}} ~~ &  f( \mathcal{S}, \mathcal{R}_{\rm cv}, \mathcal{R}_{\rm sv}, \mathcal{R}_{\rm vp}, \mathcal{R}_{\rm cp} ) \label{eq-10-1-1:0} \\
        \rm{s.t.} ~~~~~&   
        g( \mathcal{S}, \mathcal{R}_{\rm cv}, \mathcal{R}_{\rm sv}, \mathcal{R}_{\rm vp}, \mathcal{R}_{\rm cp} ) = 0   \label{eq-10-1-1:1} \\  
        ~~ & 
		h( \mathcal{S}, \mathcal{R}_{\rm cv}, \mathcal{R}_{\rm sv}, \mathcal{R}_{\rm vp}, \mathcal{R}_{\rm cp} ) \leq 0 \label{eq-10-1-1:2} \\ 
		~~ & 
		s( \mathcal{S}, \mathcal{R}_{\rm cp} ) \leq 0    \label{eq-10-1-1:3}
    \end{align}
\end{subequations} 
where $f$ is the objective function such as the network loss and operation cost; Eq. (\ref{eq-10-1-1:1}) describes system properties, including power flow equations and dynamic equations if dynamic performance is considered; 
Eq. (\ref{eq-10-1-1:2}) ensures operational feasibility; Eq. (\ref{eq-10-1-1:3}) is specially for topological feasibility. 
In addition to branch switchability where $\mathcal{S} \subseteq \mathcal{E} \backslash \mathcal{E}_{\rm us}$, topological feasibility also includes network radiality constraints and, if multiple substations are present, ensures that each bus is connected to a single substation. 
For subsequent expression, let $\Psi ( \mathcal{S}, \mathcal{R}_{\rm cp}, \mathcal{R}_{\rm vp} )$ denote the reduced form of problem (\ref{eq-10-1-1}), where $\mathcal{S}$ is a known parameter and Eq. (\ref{eq-10-1-1:3}) is omitted. The optimal value of $f$ in this reduced problem is then given by $f^* = \Psi ( \mathcal{S}, \mathcal{R}_{\rm cp}, \mathcal{R}_{\rm vp} )$. Also, let $\Phi ( \mathcal{S}, \mathcal{R}_{\rm cp}, \mathcal{R}_{\rm vp} )$ denote the corresponding optimizer of the reduced problem, such that $\{ \mathcal{R}_{\rm sv}^*, \mathcal{R}_{\rm cv}^*\} = \Phi ( \mathcal{S}, \mathcal{R}_{\rm cp}, \mathcal{R}_{\rm vp} )$, where $\{ \mathcal{R}_{\rm sv}^*, \mathcal{R}_{\rm cv}^*\}$ represents the optimal solution pair associated with $f^*$.

To obtain a representative example model of steady-state distribution topology control, consider a distribution network containing distributed generators whose active power outputs can be adjusted together with the network topology. For simplicity, assume that all distributed generators are photovoltaic-based with constant power factors, and the distribution network is supplied externally by a single substation. 
Moreover, the network is assumed to operate in a radial topology, as is typical for distribution systems. 
Then, the LinDistFlow-based \citep{4-811, 4-308} steady-state distribution topology control for loss reduction can be formulated as a mixed-integer quadratic program as follows: 
\begin{subequations}\label{eq-4-1-dnropt}
    \begin{align}
        \min_{ \bm{z} \in \mathbb{B}^{ n_{\rm e} }, \bm{p}_{\rm g} \in \mathbb{R}^{ n_{\rm g}  } } ~& \bm{r}_{\rm b}\T (\bm{p}_{\rm b}^{ 2} + \bm{q}_{\rm b}^{ 2} ) \label{eq-4-1-dnropt:1} \\
        \text{s.t.} 
        ~& \bm{E}_{\mathcal{G}}\T \bm{v}_{\rm s} - 2( \bm{p}_{\rm b} \circ \bm{r}_{\rm b} + \bm{q}_{\rm b} \circ \bm{x}_{\rm b}) \leq M(\bm{1} - \bm{z}) \label{eq-4-1-dnropt:2}\\
        ~& \bm{E}_{\mathcal{G}}\T \bm{v}_{\rm s} - 2( \bm{p}_{\rm b} \circ \bm{r}_{\rm b} + \bm{q}_{\rm b} \circ \bm{x}_{\rm b}) \geq M(\bm{z} - \bm{1}) \label{eq-4-1-dnropt:3} \\
        ~& \bm{E}_{\rm g} \bm{p}_{\rm g} + \bm{E}_{\rm sub} p_{\rm sub} - \bm{E}_{\rm d} \bm{p}_{\rm d} = \bm{E}_{\mathcal{G}} \bm{p}_{\rm b}  \label{eq-4-1-dnropt:4}\\
        ~& \bm{E}_{\rm g} [ \bm{p}_{\rm g} \circ \tan( \arccos \bm{\phi}_{\rm g} ) ] + \bm{E}_{\rm sub} q_{\rm sub} \!-\! \bm{E}_{\rm d} \bm{q}_{\rm d} \!=\! \bm{E}_{\mathcal{G}} \bm{q}_{\rm b} \label{eq-4-1-dnropt:5} \\
        ~& ( \bm{p}_{\rm b}\PW + \bm{q}_{\rm b}\PW )\HR  \leq \bm{z} \circ \bm{s}_{\rm b}\U \label{eq-4-1-dnropt:6}\\
        ~& (\bm{v}\B)\PW \leq \bm{v}_{\rm s} \leq (\bm{v}\U)\PW \label{eq-4-1-dnropt:7-pre}\\
        ~& \bm{E}_{\rm sub}\T \bm{v}_{\rm s} = v_{\rm sub}^2 \label{eq-4-1-dnropt:7}\\
        ~& \bm{p}_{\rm g}\B  \leq  \bm{p}_{\rm g} \leq \bm{p}_{\rm g}\U \label{eq-4-1-dnropt:8}\\
        ~& [ \bm{p}_{\rm g}\PW + [\bm{p}_{\rm g} \circ \tan( \arccos \bm{\phi}_{\rm g} ) ]\PW ]\HR \leq \bm{s}_{\rm g}\U  \label{eq-4-1-dnropt:9}\\ 
        ~& {p}_{\rm sub}\B \leq p_{\rm sub} \leq {p}_{\rm sub}\U, \label{eq-4-1-dnropt:10-pre}\\
        ~& {q}_{\rm sub}\B \leq q_{\rm sub} \leq {q}_{\rm sub}\U \label{eq-4-1-dnropt:10}\\
        ~& \bm{\zeta} + \tilde{\bm{\zeta}} = \bm{z}  \label{eq-4-1-dnropt:11-pree} \\
        ~& \bm{E}_{\mathcal{G},\rm f} \bm{\zeta} - \bm{E}_{\mathcal{G}, \rm t} \tilde{\bm{\zeta}} = \bm{1}^0_{n_{\rm n}} \label{eq-4-1-dnropt:11-pre} \\
        ~&  \bm{E}_{\rm us} \bm{z} = \bm{1} \label{eq-4-1-dnropt:11} 
    \end{align}
\end{subequations}
where the vector $\bm{r}_{\rm b} + j \bm{x}_{\rm b} \in \mathbb{C}^{n_{\rm e}}$ denotes the impedance of branches, $\bm{\phi}_{\rm g} \in \mathbb{R}^{n_{\rm g}}$ denotes the fixed power factors of generators, $v_{\rm sub}$ denotes the setpoint of voltage magnitude of the substation bus; 
$\bm{E}_{\rm sub} \in \mathbb{R}^{n_{\rm n}}$ is incidence matrix between all buses and the substation bus, $\bm{1}^0_{n_{\rm n}} \in \mathbb{R}^{n_{\rm n}}$ is the vector with one zero entry and the rest being one; 
$\bm{v}_{\rm s} \in \mathbb{R}^{n_{\rm n}}$ denotes the square of bus voltage magnitudes, 
$p_{\rm sub} + j q_{\rm sub} \!\in\! \mathbb{C}$ denotes the power injection from the substation;  
$\bm{\zeta} \in \mathbb{B}^{n_{\rm e}}$ and $\tilde{\bm{\zeta}} \in \mathbb{R}^{n_{\rm e}}$ are auxiliary variables introduced for ensuring network radiality. 
The objective function (\ref{eq-4-1-dnropt:1}) represents the approximated network loss; 
Eqs. (\ref{eq-4-1-dnropt:2}) and (\ref{eq-4-1-dnropt:3}) are the power flow constraints that enforce the voltage drop relationship along each line using a Big-M formulation, and the equality form of the LinDistFlow voltage equation \citep{4-811, 4-308} is activated when the line is in service, while it is relaxed when the line is disconnected; 
Eq. (\ref{eq-4-1-dnropt:4}) represents the active power balance at each bus, ensuring that total active power injections and withdrawals are balanced with the power flows in the network; 
Eq. (\ref{eq-4-1-dnropt:5}) represents the reactive power balance at each bus, where reactive power injections are linked to active power through a fixed power factor assumption; 
Eqs. (\ref{eq-4-1-dnropt:6})-(\ref{eq-4-1-dnropt:10}) are the operational feasibility constraints; 
Eqs. (\ref{eq-4-1-dnropt:11-pree}) and (\ref{eq-4-1-dnropt:11-pre}) ensure radiality of the network topology, which is the matrix form of the parent-child relation based network radiality formulation given by (\ref{eq-b1-3-9}); 
and Eq. (\ref{eq-4-1-dnropt:11}) ensures the unswitchable branches remain closed.

\subsection{Classical Heuristic Solution Method}

The special structure of the steady-state distribution topology control problem (\ref{eq-4-1-dnropt}) allows the solution to be directly addressed using off-the-shelf mixed-integer optimization solvers such as Gurobi and CPLEX. However, solving large-scale instances of this problem using these general-purpose solvers often becomes computationally inefficient, even when seeking a feasible solution. Moreover, most mixed-integer optimization solvers are limited to convex formulations, such as problem (\ref{eq-4-1-dnropt}) using the LinDisFlow model, and therefore are not suitable when the AC power flow model is incorporated into the optimization problems. 

To address these limitations, various heuristic methods have been developed by exploiting structural characteristics of the problem. Although these heuristics compromise optimality, they offer enhanced scalability and applicability. One of the classical heuristics for solving steady-state distribution topology control problems is the successive branch reduction method \citep{4-790, 4-847, 4-835, 4-819, 4-1677, 4-b1-1}. The common heuristic pattern of existing variations of this method can be given as Algorithm \ref{algo-4-1-0} based on the canonical form (\ref{eq-10-1-1}). 

\begin{algorithm}
    \KwInput{$\mathcal{R}_{\rm cp}$, $\mathcal{R}_{\rm vp}$}
    \KwOutput{$\mathcal{S}^*$, $\mathcal{R}_{\rm cv}^*$} 

    Define problem (\ref{eq-10-1-1}) with given $\mathcal{R}_{\rm cp}$ and $\mathcal{R}_{\rm vp}$

    Define evaluation function $Q(\mathcal{S}, s, \mathcal{R}_{\rm cp}, \mathcal{R}_{\rm vp} )$ and termination condition $\tau(\mathcal{S})$

    Initialize the redundant solution as $\mathcal{S} \gets \mathcal{E} \backslash \mathcal{E}_{\rm us}$

    \Repeat{$\tau(\mathcal{S})$ is satisfied}{
    $s^* \gets \argmin_{ s \in \mathcal{S}  }  Q( \mathcal{S}, s, \mathcal{R}_{\rm cp}, \mathcal{R}_{\rm vp} )$

    $\mathcal{S} \gets \mathcal{S} \backslash s^* $
    }

    $\mathcal{S}^* \gets \mathcal{S} $ 
    
    $\{ \mathcal{R}_{\rm sv}^*, \mathcal{R}_{\rm cv}^*\} \gets \Phi ( \mathcal{S}^*, \mathcal{R}_{\rm cp}, \mathcal{R}_{\rm vp} ) $ \label{algo-4-1-0:line-8} 
    \caption{The pattern of successive branch reduction methods}\label{algo-4-1-0}
\end{algorithm}

In Algorithm \ref{algo-4-1-0}, the solution of $\mathcal{S}$ is constructed by sequentially removing branches from the redundant topology where all branches are closed. 
At each iteration, the selection of the branch to be removed is determined by minimizing an evaluation function $Q$, which quantifies the impact of branch removal within the context of the current redundant topology. 
The iteration ends when certain termination condition $\tau$ is satisfied, which is typically defined as the graph $\mathcal{G}_{\rm v}(\mathcal{V}, \mathcal{E}_{\rm us} \cup \mathcal{S}, \cdot )$ becoming radial. 
The primary differences among various successive branch reduction methods lie in the evaluation functions $Q$. Taking the successive branch reduction method developed by \cite{4-819} as an example, the evaluation function $Q$ can be expressed as follows: 
\begin{equation}
    \!\!\! Q \!=\! \left\{\!\!
            \begin{aligned}
                & M &&  d( \mathcal{G}_{\rm v}(\mathcal{V}, \mathcal{E}_{\rm us} \!\cup\! \mathcal{S}, \cdot ) ) \!\leq\! d( \mathcal{G}_{\rm v}(\mathcal{V}, \mathcal{E}_{\rm us} \!\cup\! \mathcal{S} \backslash s, \cdot ) ) \\ 
                & \Psi(\mathcal{S}\backslash s, \mathcal{P}_{\rm cp}, \mathcal{P}_{\rm vp}) &&  d( * ) > d( \cdot ) \land s \in \mathcal{N}(v^*) \cap \mathcal{E} \backslash \mathcal{E}_{\rm us} \\ 
                & M &&  \text{otherwise}
            \end{aligned}
        \right. 
\end{equation}
where $d(\mathcal{G}_{\rm v})$ is the sum of loops and paths among {the} substation buses in graph $\mathcal{G}_{\rm v}$; $v^*$ is the start bus of branch $e^* = \argmin_{e \in \mathcal{S}} |p ( e , \mathcal{P}_{\rm cp}, \mathcal{P}_{\rm vp} )|$ where $p ( \cdot )$ denotes the real power through branch $e$ in $\mathcal{G}_{\rm v}(\mathcal{V}, \mathcal{E}_{\rm us} \cup \mathcal{S}, \mathcal{P}^* )$, with $\mathcal{P}^* = \{\mathcal{P}_{\rm cp}, \mathcal{P}_{\rm vp}, \mathcal{P}_{\rm sv}^*, \mathcal{P}_{\rm cv}^* \}$ and $\{ \mathcal{P}_{\rm sv}^*, \mathcal{P}_{\rm cv}^*\} = \Phi ( \mathcal{S} \backslash s, \mathcal{P}_{\rm cp}, \mathcal{P}_{\rm vp} ) $; the start bus refers to the one from which power flow $p ( e , \mathcal{P}_{\rm cp}, \mathcal{P}_{\rm vp} )$ is positive; and $\mathcal{N}(v^*)$ is the set of branches with end node $v^*$.

\section{State-of-the-Art Review} 

Compared with steady-state transmission topology control, steady-state distribution topology control has received more attention from both academic and industry. There are several comprehensive reviews such as those by \cite{4-808}, \cite{4-809}, \cite{4-810}, and \cite{4-b1-40}. Similar to steady-state transmission topology control problems, steady-state distribution topology control problems are also essentially characterized as highly non-convex Mixed-Integer Nonlinear Programming (\ac{minlp}) problems. 
Evolving from the basic steady-state distribution topology control problems involving only power flow-related terms to more recent formulations that incorporate uncertainties and system dynamics, it has become increasingly challenging to solve the corresponding optimization models efficiently. 
Thus, this section focuses on the state-of-the-art solution methodologies for steady-state distribution topology control. 

\subsubsection{Approximation/relaxation-based methods}

Similar to the approximation/relaxation-based methods for steady-state transmission topology control, this type of methods for steady-state distribution topology control also transforms the non-convex \ac{minlp} model into a convex formulation using approximate or relaxed distribution power flow models, enabling direct solution via off-the-shelf optimization solvers. However,  unlike transmission networks, resistance of distribution networks cannot be ignored, and accordingly, DistFlow equations proposed by \cite{4-2201} are commonly employed to transform the original \ac{minlp} model into a convex mixed-integer quadratic programming or mixed-integer quadratically constrained programming formulation \citep{4-818, 4-1366}. \cite{4-b1-56} proposed a cold-start linear branch flow model termed modified DistFlow, with reduced errors introduced by conventional branch flow linearization approaches and branch power flows and nodal voltage magnitudes expressed in succinct matrix forms. The \ac{dnr} problem with the modified DistFlow model becomes a much more computationally efficient mixed-integer quadratic program. 

Regarding the relaxation-based methods, \cite{4-817} proposed the first mixed-integer conic programming formulation of \ac{dnr} problems. Compared to the previous approximation-based methods, this formulation employs a convex representation of the power flow model ensuing global optimality of the obtained optimal solution. 
Nearly at the same time, \cite{4-818} derived mixed-integer quadratic, quadratically constrained, and second-order cone programming models of \ac{dnr}, also stated as the first formulations of the AC \ac{dnr} problem that have convex and continuous relaxations. 
The approximation/relaxation-based methods also enable the tractable solution of more complex \ac{dnr} problems, such as two-stage robust \ac{dnr} \citep{4-800} and multi-timescale coordinated \ac{dnr} \citep{4-b1-58}, by ensuring favorable computational properties. 
\cite{4-2010} established that the fundamental bottleneck of existing convex relaxation-based methods to solve \ac{dnr} problems lies in the static cone structure of standard second-order cone relaxation. They further proposed a hybrid perspective-disjunctive \ac{dnr} formulation that achieves the closure of the convex hull for branch flow models.

In the approximation/relaxation-based methods, Big-M and McCormick linearization techniques are commonly employed, which however may also cause issues related to tightness as in the approximation/ relaxation-based methods for steady-state transmission topology control. 
\cite{4-b1-57} proposed a disjunctive convex hull approach which treats continuous parent-child relationship variables in spanning tree constraints as disjunctive variables to represent disjunctive convex hull of DistFlow equations. This disjunctive convex hull yields a tighter relaxation than the Big-M and McCormick linearization techniques.

\subsubsection{Heuristic methods}

Heuristic algorithms designed by exploiting the structural properties of steady-state distribution topology control problems mainly include three types: iterative branch exchange methods \citep{4-838, 4-b1-53}, constructive methods \citep{4-848, 4-b1-54}, and successive branch reduction methods \citep{4-840, 4-790, 4-847, 4-835, 4-819, 4-b1-1}. 
The iterative branch exchange method was first proposed by \cite{4-838}, in which the optimal power flow pattern of a single loop formed by closing a normally open switch is identified, and a radial network topology is restored by opening a closed switch. This process is repeated iteratively until a distribution network topology minimizing the objective function is obtained. 
\cite{4-b1-53} further proposed an efficient algorithm to optimize the branch exchange step based on the convex relaxation of optimal power flow. Notably, a bound on the gap between the optimal objective function value and that of the proposed method was provided. 
The constructive method first proposed by \cite{4-848} starts with all switchable switches open, and at each step, closes the switch that results in the least increase in the objective function. A backtracking option is designed to mitigate the method's greedy search. 
\cite{4-b1-54} developed a uniform voltage distribution based constructive algorithm, which initiates by expanding a subnetwork through tracking the maximum bus voltage and concurrently performing a series of branch exchange operations. 

The idea of successive branch reduction methods was first proposed by \cite{4-840}, where all network switches are initially closed converting the radial distribution network into a meshed network, and network switches are then open at a time until the network topology becomes radial. \cite{4-790} further enhanced the original successive branch reduction method by addressing its major limitations, including the oversimplification of the power flow model and the omission of network constraints. 
\cite{4-847} made a refinement on the successive branch reduction method using the branch exchange technique involving neighboring open switches. 
\cite{4-835} employed an Optimal Power Flow (\ac{opf}) program to determine the sensitivity of the switches that have positions close to zero and determine in the successive branch reduction method. 
\cite{4-819} optimized the successive branch reduction method by using the convex relaxation-based \ac{opf} model and realizing constant complexity. 
\cite{4-b1-1} extended the successive branch reduction method to solve stochastic \ac{dnr} problems. 
Fundamentally, all three aforementioned types of heuristic methods follow the principle of greedy algorithms, making locally optimal decisions at each step with the aim of achieving a globally optimal solution \citep{4-842}. 

Different from the above mentioned common heuristics, \cite{4-2001} recently proposed a feasibility-oriented algorithm inspired by random walks, which incorporates several novel techniques, including strategic graph partitioning, dual-graph condensation, and capacity-aware edge swapping. Notably, this algorithm provides theoretical guarantees of feasibility and supports parallel processing while maintaining desirable optimality properties. 

In addition, a variety of meta-heuristics, i.e., high-level problem-independent algorithmic optimization frameworks, have been employed to solve the steady-state distribution topology control problems, including the particle swarm algorithm \citep{4-b1-52}, the equilibrium optimization algorithm \citep{4-1359}, the genetic algorithm \citep{4-1358}, etc. 
The generalized Benders decomposition, originally for solving convex \ac{minlp} problems, has been utilized by \cite{4-828} to solve the non-convex \ac{minlp} models of \ac{dnr} with the complex voltage volatility constraint. 
However, the generalized Benders decomposition is unable to guarantee the convergence to the optimal solution of non-convex \ac{minlp} model due to the loss of strong duality \citep{4-1360}. 

\subsubsection{Learning-based methods}

Similar to steady-state transmission topology control, machine learning techniques, particularly reinforcement learning, have recently been applied to steady-state distribution topology control problems, primarily to address complex operational constraints and meet the requirement for real-time solutions. 
\cite{4-b1-43} developed a data-driven batch-constrained reinforcement learning algorithm for the dynamic \ac{dnr} problem. This algorithm learns the distribution topology control policy from a finite historical operational dataset without interacting with the distribution network. 
\cite{4-b1-50} proposed a multi-objective Bayesian learning-based evolutionary algorithm to address the steady-state distribution topology control problem, aiming to simultaneously optimize the wind power absorption rate and voltage stability level. 
\cite{4-b1-51} proposed a deep reinforcement learning-assisted multi-objective bacterial foraging optimization algorithm for solving the many-objective \ac{dnr} problem. 
\cite{4-b1-47} utilized a long short-term memory network to learn the mapping mechanism between load distribution and optimal reconfiguration strategies, coordinated with a model-based module hierarchically to reduce the problem size. 
\cite{4-b1-42} developed a cloud-edge collaboration framework based on multi-agent deep reinforcement learning, where the reinforcement learning model can be trained centrally in the cloud center while executed in edge servers in a decentralized manner to reduce the training cost and execution latency.  
\cite{4-b1-44} proposed a bi-graph neural network-based deep reinforcement learning framework to enable efficient real-time coordination of \ac{dnr} and volt-var control. 

For the steady-state distribution topology control considering short-term voltage stability enhancement, \cite{4-b1-45} utilized deep convolution neural networks to learn the relationship between network topology and stability performance from historical data, which after off-line training, was employed to form the computationally efficient stability constraint in the optimization model. 
Similarly, \cite{4-b1-46} utilized a trained convolutional neural network model to predict the voltage stability index value in the stochastic \ac{dnr} problem. 
\cite{4-b1-49} developed a physics-informed end-to-end machine learning framework for steady-state distribution topology control, enabling simultaneous optimization of grid topology and generator dispatch with certified satisfiability of safety-critical constraints. 
\cite{4-b1-41} proposed a deep learning-based approach termed LLM4DistReconfig, utilizing a fine-tuned large language model to solve the steady-state distribution topology control problem.

\section{Recent Advance: Hybrid Learning-Heuristic Solution Paradigm}

\begin{figure}[h]
	\centering
	\includegraphics[width=0.95\linewidth]{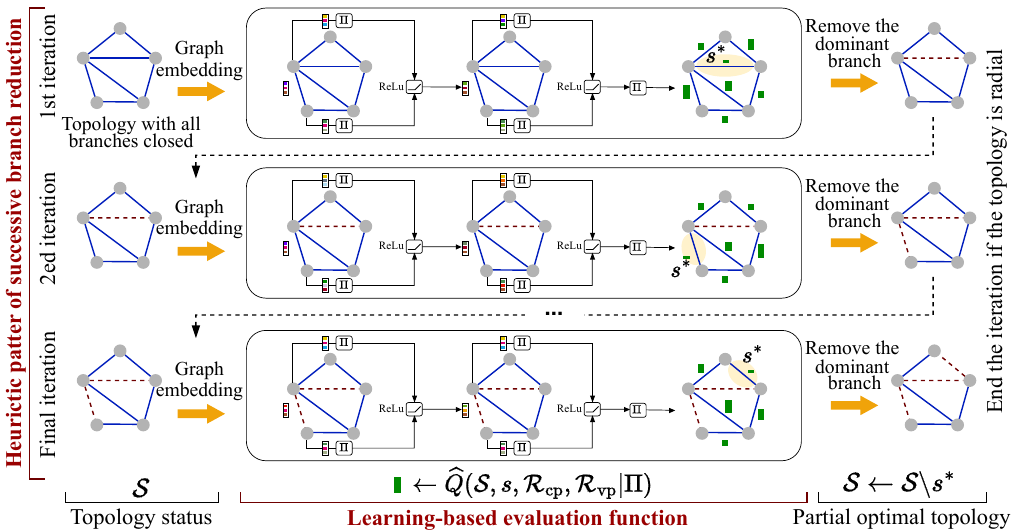}
	\caption{{Schematic diagram of the hybrid learning-heuristic solution paradigm.}}
    \label{fig-b1-5-4} 
\end{figure}

In addressing steady-state distribution topology control problems, the greatest challenge from a solution standpoint is to simultaneously achieve optimality, scalability, and applicability. Most heuristic solution methods generally meet the requirements of scalability and applicability; however, they can only be expected to get 'close to optimality' at best. 
For instance, in existing successive branch reduction methods, the hand-crafted evaluation function $Q$ typically considers only the immediate impact of branch removal within the current redundant topology. This myopic perspective lacks foresight regarding the influence of present switching decisions on future iterations, thereby leading to suboptimal solutions. 
Recently, \cite{4-1371} proposed a hybrid learning-heuristic solution paradigm as illustrated in Fig. \ref{fig-b1-5-4}. This paradigm combines the heuristic pattern of classical successive branch reduction methods with the reinforcement learning technique. The evaluation function $Q(\mathcal{S}, s, \mathcal{R}_{\rm cp}, \mathcal{R}_{\rm vp})$ in Algorithm \ref{algo-4-1-0} is replaced with a parameterized function $\widehat{Q}(\mathcal{S}, s, \mathcal{R}_{\rm cp}, \mathcal{R}_{\rm vp}| \Pi )$, with parameters $\Pi$ yielded by reinforcement learning. This solution paradigm preserves the applicability of successive branch reduction methods, while enhancing computational efficiency and particularly optimality.

\subsection{Parameterization Architecture of the Evaluation Function}\label{sec-5-1-2}

The architecture of $\widehat{Q}$ is first introduced. Intuitively, $\widehat{Q}$ should summarize the power network states with topology given by $\mathcal{S}$, and figure out the value of each branch in $\mathcal{S}$ to designate whether it is to be removed in the context of the current redundant topology. However, both power network states and the consequence of removing a branch are hard to describe in closed form and involve complex implicit computation. Accordingly, deep neural networks are naturally a promising option for $\widehat{Q}$ due to their strong representation power. Among them, the Gated Graph Neural Network (\ac{ggnn}), as a kind of graph neural networks operating on the graph domain, is suitable for processing graph-structured data \citep{4-961}, and more importantly, it can contain many more layers and allow for complex information about global graph structure to be propagated to all nodes \citep{4-985}. Therefore, for representing the complex power network states over combinatorial graph structures and involving global interactions, the architecture of $\widehat{Q}$ will be designed mainly employing the \ac{ggnn}.

\subsubsection{Pretreatments}
 
\textit{Features}: 
First, the time-varying inputs of $\widehat{Q}$ and topology status are associated with the features of each bus and each branch. Specifically, introduce $\bm{x}_v \in \mathbb{R}^{n_v}$, $\forall v \in \mathcal{V}$, to collect parameters in $\mathcal{R}_{\rm vp}$ which are associated with {bus $v$}; and $\bm{x}_e \in \mathbb{R}^{n_e}, \forall e \in \mathcal{E}$, containing parameters in $\mathcal{R}_{\rm vp}$ which are associated with branch $e$, and a binary feature (called \textit{topology feature}) only for $e \in \mathcal{E} \backslash \mathcal{E}_{\rm us}$, which equals to 1 if $e \in \mathcal{S}$ and 0 otherwise. For branches and nodes with neither parameters in $\mathcal{R}_{\rm vp}$ nor topology features, such as a bus without any load and generator, we pad a fictitious feature for each of them and set the feature to a constant to preserve the network structure of power networks.

\begin{figure}[h]
	\centering
	\includegraphics[width=0.85\linewidth]{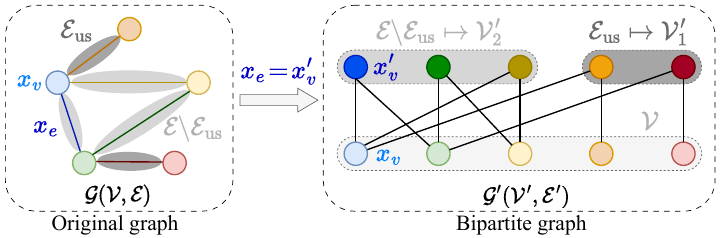}
	\caption{{Illustration of graph conversion.}}
    \label{fig-4-1-2} 
\end{figure}

\textit{Graph conversion}: 
Consider graph $\mathcal{G}(\mathcal{V}, \mathcal{E}) $ with node features $\{\bm{x}_v\}$ and edge features$\{\bm{x}_e\}$. Except for the topology feature, $\bm{x}_e$ may also contain branch impedances for transformer branches or branches with controllable impedance. However, the \ac{ggnn} is incapable of effectively handling graphs with informative edge features \citep{4-956}. 
To address this issue, we convert the original graph into a bipartite graph by treating edges as additional nodes, so that edge attributes are represented as node features. This reformulation allows the influence of edge features to be captured indirectly through their corresponding nodes. 
Accordingly, $\mathcal{G}(\mathcal{V}, \mathcal{E}) $ is converted into a bipartite graph $\mathcal{G}'(\mathcal{V}', \mathcal{E}') $, as illustrated in Fig. \ref{fig-4-1-2}. 
{In the original graph in Fig. \ref{fig-4-1-2}, each circle and line represent a node in $\mathcal{V}$ and an edge in $\mathcal{E}$, respectively.} 
In the bipartite graph with $\mathcal{V}' = \mathcal{V} \cup \mathcal{V}'_1 \cup \mathcal{V}'_2$, 
{the original nodes $\mathcal{V}$ are preserved;} 
{the original edges $\mathcal{E}$ is converted into nodes $\mathcal{V}'_1 \cup \mathcal{V}'_2$ called edge nodes;} 
and each original edge is split into two new edges between the edge node and terminal nodes of the original edge, {forming the edges $\mathcal{E}'$ represented by the black edges in the bipartite graph in Fig. \ref{fig-4-1-2}.} 
For simplicity, introduce the function $\psi$ that maps an edge in graph $\mathcal{G}$ to the associated node in graph $\mathcal{G}'$. In this way, removing edge $s$ from $\mathcal{G}$ is equivalent to removing nodes $\psi(s)$ in $\mathcal{G}'$ such that only nodes in $\mathcal{G}'$ need to be handled. Denote by $\{ \bm{x}_{v}|  v \in \mathcal{V}' \}$ node features of graph $\mathcal{G}'$, where $\forall v \in \mathcal{V}$, $\bm{x}_v$ is the corresponding node feature in $\mathcal{G}$, and $\forall v \in \mathcal{V}'_1 \cup \mathcal{V}'_2$, $\bm{x}_v = \bm{x}_e$ with $e$ being the corresponding edge of node $v$ in $\mathcal{G}$.

\textit{Feature projection}: 
Graph $\mathcal{G}'$ is heterogeneous as it not only inherits node and edge heterogeneity of graph $\mathcal{G}$ but also introduces new heterogeneity by graph conversion. Thus nodes of $\mathcal{G}'$ have different traits and their features may fall in different feature spaces. For example, for load buses, $\bm{x}_v$ generally only contains load powers while for generator buses whose real power is also optimized with network topology, $\bm{x}_v$ can contain load power and voltage magnitude. It is assumed that there are $n_{\rm t}$ types of nodes in $\mathcal{V}'$, and let $\phi_i(\mathcal{V}')$ denote the node set of the $i$-th type. Then for each type of nodes, introduce a type-specific transformation matrix $\bm{M}_{\phi_i}$ to project the features of different types of nodes into the same feature space. Specifically, this feature projection is formulated as follows:
\begin{equation}\label{eq-4-1-3}
    \bm{x}'_{v} = \text{rrelu}( \bm{M}_{\phi_i} \cdot \bm{x}_{v} ) ~~~ \forall v \in \phi_i(\mathcal{V}'), i \in \llbracket 1, n_{\rm t} \rrbracket 
\end{equation} 
where $\bm{x}'_v \in \mathbb{R}^{n_{\rm x}} $ is the projected feature of node $v$; $\bm{M}_{\phi_i}$ are with the same row dimension but type-specific column dimension; and $\rm{rrelu}(\cdot)$ is the element-wise randomized leaky rectified linear unit, defined as: 
\begin{equation}
    \text{rrelu}( [\bm{M}_{\phi_i} \cdot \bm{x}_{v} ]_j ) = \left\{
        \begin{aligned}
            & [\bm{M}_{\phi_i} \cdot \bm{x}_{v} ]_j    && \text{if}~ [\bm{M}_{\phi_i} \cdot \bm{x}_{v} ]_j \geq 0 \\
            & \eta [\bm{M}_{\phi_i} \cdot \bm{x}_{v} ]_j  && \text{otherwise} 
        \end{aligned}    
    \right.
\end{equation}
with $[\bm{M}_{\phi_i} \cdot \bm{x}_{v} ]_j$ being the $j$-th entry of vector $\bm{M}_{\phi_i} \cdot \bm{x}_{v}$, and $\eta$ randomly sampled from a uniform distribution.

\subsubsection{Gated graph neural network} 

The \ac{ggnn} encodes the graph $\mathcal{G}'$ into a vector space while maximally preserving its structural information and properties, such that useful information can be extracted more easily when learning $\widehat{Q}(\mathcal{S}, s, \mathcal{R}_{\rm cp}, \mathcal{R}_{\rm vp}; \Pi )$. Specifically, given parameters $\mathcal{R}_{\rm vp}$ and the current redundant solution $\mathcal{S}$, the \ac{ggnn} computes an $n_{\rm h}$-dimensional feature embedding for each node $v \in \mathcal{V}'$, which will be used as inputs of the later parameterized function $\widehat{Q}$. The architecture of the \ac{ggnn} is expressed as follows: 
\begin{subequations}\label{eq-4-1-4}
    \begin{align}
        & \bm{h}_v^1 = [ \bm{x}'_{v} , \bm{0} ]^T \label{eq-4-1-4:1} \\ 
        & \bm{a}_v^{t} = [ \bm{h}_1^{t-1} \cdots \bm{h}_{n_{\rm n} + n_{\rm e}}^{t-1} ] \bm{A}_{v} + \bm{b} \label{eq-4-1-4:2} \\  
        & \bm{z}_v^t = \sigma ( \bm{W}_z \bm{a}_v^{t} + \bm{U}_z \bm{h}_v^{t-1}  ) \label{eq-4-1-4:3} \\ 
        & \bm{r}_v^t = \sigma ( \bm{W}_r \bm{a}_v^{t} + \bm{U}_r \bm{h}_v^{t-1}  )  \label{eq-4-1-4:4} \\ 
        & \tilde{\bm{h}}_v^{t} = {\rm{tanh}} \big( \bm{W} \bm{a}_v^{t} + \bm{U} ( \bm{r}_v^t \circ \bm{h}_v^{t-1} )  \big) \label{eq-4-1-4:5} \\   
        & \bm{h}_v^{t} = ( 1 - \bm{z}_v^t ) \circ \bm{h}_v^{t-1} + \bm{z}_v^t \circ \tilde{\bm{h}}_v^{t} \label{eq-4-1-4:6}
    \end{align}
\end{subequations} 
where $v \in \mathcal{V}'$ and $t \in \llbracket 2, N_{\rm T}  \rrbracket$ with $N_{\rm T}$ being the maximal number of time steps; $\bm{h}_v \in \mathbb{R}^{n_{\rm h}} $ and  $\bm{a}_v \in \mathbb{R}^{n_{\rm h}} $ are the state vector and aggregated message vector of node $v$, respectively; $\bm{A}_{v} \in \mathbb{R}^{ n_{\rm n} + n_{\rm e} } $ is the column of the adjacency matrix of graph $\mathcal{G}'$ corresponding to node $v$; $\bm{z}_v$ and $\bm{r}_v$ are the update and reset gate, respectively; 
and $\sigma(\cdot)$ and $\tanh(\cdot)$ denote the logistic sigmoid function and $\tanh$ function, respectively. 
(\ref{eq-4-1-4:1}) is the initialization step, which copies node features into the first components of the state vector and pads the rest with zeros; (\ref{eq-4-1-4:2}) aggregates message from its neighbors; and given the aggregated message $\bm{a}_v^{t}$ and the previous state vectors $\bm{h}_v^{t-1}$, the state of the current time step $\bm{h}_v^{t}$ is computed by the recurrent cell function of the gated recurrent unit, as given by (\ref{eq-4-1-4:3}) to (\ref{eq-4-1-4:6}). The dynamics defined by (\ref{eq-4-1-4:2}) to (\ref{eq-4-1-4:6}) are repeated for $N_{\rm T}$ time steps. Then, the state vectors from the last time step, i.e., $\{\bm{h}_v^{N_{\rm T}} | v \in \mathcal{V}' \}$, are used as feature embeddings.

\subsubsection{Parameterization of the evaluation function}

With the feature embedding of each node, pooled maps of the feature embedding and projected features over the entire graph $\mathcal{G}'$ are used as the surrogate for the collection of $\mathcal{S}$ and $\mathcal{R}_{\rm vp}$, and the pooled maps of feature embedding and projected features over nodes $\psi(s)$ are used as the surrogate for $s$. Then the parameterization $\widehat{Q}$ is designed as follows: 
\begin{equation}\label{eq-4-1-5}  
    \begin{aligned}
        & \widehat{Q}(\mathcal{S}, s, \mathcal{R}_{\rm cp}, \mathcal{R}_{\rm vp}| \Pi ) = \\
        &~~~~~~~~ \bm{W}_{q} {\rm{rrelu}} \left(
            \begin{bmatrix}
                 \sum\limits_{v \in \mathcal{V}'} \sigma  \big( g_1( \bm{h}_v^{N_{\rm T}}, \bm{x}'_v ) \big) \circ \tanh \big( g_2( \bm{h}_v^{N_{\rm T}}, \bm{x}'_v ) \big) \\
                 \sum\limits_{v \in \psi(s) } \sigma  \big( g_3( \bm{h}_v^{N_{\rm T}}, \bm{x}'_v ) \big) \circ \tanh \big( g_4( \bm{h}_v^{N_{\rm T}}, \bm{x}'_v ) \big)
            \end{bmatrix}
            \right) 
    \end{aligned}
\end{equation} 
where $\{ g_i(\cdot) | i \in \llbracket 1, 4 \rrbracket \}$ denotes multilayer perceptrons with $N_{{\rm g}, i}$ hidden layers, among which $\{ g_i(\cdot) | i \!\in\! \{1, 3\} \}$ maps the concatenation of $\bm{h}_v^{N_{\rm T}}$ and $\bm{x}'_v$ to a scaler, and $\{ g_i(\cdot) | i \!\in\! \{2, 4\} \}$ maps the same concatenation to vectors of $n_{{\rm g}, i}$ dimension. The map from the original features $\bm{x}_v$ of graph $\mathcal{G}'$ to $\bm{h}_v^{N_{\rm T}}$ in (\ref{eq-4-1-5}) is based on the feature projection by (\ref{eq-4-1-3}) and the \ac{ggnn} computation by (\ref{eq-4-1-4}), thus $\Pi$ is expressed as follows:
\begin{equation}\label{eq-4-1-6}
    \begin{aligned}
        \Pi = & \{ \bm{M}_{\phi_i} | i \in \llbracket 1, {n_{\rm t}} \rrbracket \}  \cup \{ \bm{b}, \bm{W}_z, \bm{W}_r, \bm{W}, \bm{U}_z, \bm{U}_r, U \} \\ 
              & \cup \{\bm{W}_q\}   \cup \{ \bm{W}_{(i, j)}, \bm{b}_{(i, j)} | j \in \llbracket 1, N_{{\rm g}, i} + 1 \rrbracket, i \in \llbracket 1, 4 \rrbracket \}
    \end{aligned}
\end{equation} 
where the last set contains weight matrices $\bm{W}_{(i, j)}$ and bias vectors $\bm{b}_{(i, j)}$ for all multilayer perceptrons.

\subsection{Learning and Learning-Based Heuristic}\label{sec-5-1-3}

The focus now is on learning the function $\widehat{Q}$. Specifically, for any given power networks with constant parameters $\mathcal{R}_{\rm cp}$ and a set of realizations of $\mathcal{R}_{\rm vp}$ sampled from its uncertainty set$\mathbb{P}$, parameters $\Pi$ are trained following the framework of reinforcement learning.

\subsubsection{Reinforcement learning formulation}

\begin{figure}[h]
	\centering
	\includegraphics[width=0.9\linewidth]{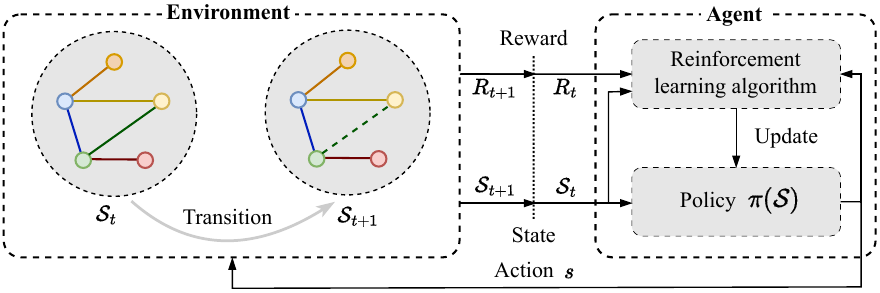}
	\caption{\small{General framework of reinforcement learning.}}
    \label{fig-4-1-3} 
\end{figure} 

As illustrated in Fig. \ref{fig-4-1-3}, the reinforcement learning problem can be formalized in terms of an agent that performs an action in an environment providing a state and a reward, to optimize a given notion of cumulative rewards. 
Transition predicts the next state after taking an action. The agent contains two components: a policy and a learning algorithm. The policy is a mapping that selects actions based on the state of the environment. The learning algorithm continuously updates the policy parameters based on the actions, states, and rewards \citep{4-954}. According to the definition of $\widehat{Q}$, reinforcement learning is a natural framework to train its parameters $\Pi$ with elements in reinforcement learning defined as follows: 
\begin{itemize}
    \item{\textit{States}}: A state is the redundant solution $\mathcal{S}$. It corresponds to the topology features of $\bm{h}_v, \forall v \in \psi(\mathcal{E} \backslash \mathcal{E}_{\rm us} )  $. The terminal state depends on the termination condition $\tau'(\mathcal{S})$. 
    \item{\textit{Actions}}: An action is defined by $s \in \mathcal{S}$, namely removing branch $s$ from $\mathcal{S}$. It is represented by the concatenation of the node embedding and project feature of each $v \in \psi(s)$ as shown in (\ref{eq-4-1-5}).
    \item{\textit{Transition}}: Transition $T(\mathcal{S}_{t}, s_t, \mathcal{S}_{t+1} )$ is deterministic. With a current state $\mathcal{S}_t$ and taking an action $s_t$, the next state $\mathcal{S}_{t+1} = \mathcal{S}_t \backslash s_t $.
    \item{\textit{Rewards}}: $R(\mathcal{S}, s) = q( \mathcal{S} \backslash s, \mathcal{R}_{\rm cp}, \mathcal{R}_{\rm vp} ) - q( \mathcal{S}, \mathcal{R}_{\rm cp}, \mathcal{R}_{\rm vp} )$.
    \item{\textit{Policy}}: The agent here aims to find a policy $\pi( \mathcal{S} )$ to minimize the Q-value function. Denoting the optimal Q-value function as $Q^*( \mathcal{S}, s )$, the optimal policy $\pi^*( \mathcal{S} ) =  \argmin_{ s \in \mathcal{S}  } Q^* ( \mathcal{S}, s )$.  
\end{itemize}
Recalling the selection of branch $s^*$ in Algorithm \ref{algo-4-1-0}, it can be found that $\widehat{Q}(\mathcal{S}, s, \mathcal{R}_{\rm cp}, \mathcal{R}_{\rm vp} | \Pi )$ is exactly a parameterized function approximator for $Q^*( \mathcal{S}, s)$ under different realizations of $\mathcal{R}_{\rm vp}$. Thus, $\widehat{Q}$ can be learned by value-based methods for reinforcement learning.

\subsubsection{Learning algorithm}

\begin{algorithm}[t!]
 
    \KwInput{$\mathcal{R}_{\rm cp}$, $\mathbb{P}, N_{\rm p}, N_{\rm r}, N_{\rm m}, N_{\rm c}, N_{\rm b}, \epsilon_{\rm s}, \epsilon_{\rm e}, \epsilon_{\rm d}, \varrho$}
    \KwOutput{$\Pi^*$} 

    Initialize experience replay memory $\mathcal{M}$ to capacity $N_{\rm r}$

    $e \gets 0$, $k \gets 0$, $\Pi_{k} \gets $ Initialize $\Pi_k$ randomly

    \Repeat{$e>200$ and $\tilde{\tau}(e, \varrho)$ is satisfied}{

        Generate $\mathcal{R}_{\rm vp}$ from uncertainty set $\mathbb{P}$ \label{algo-4-1-1-s1}

        \Repeat{$\tau(\mathcal{S}_t)$ is satisfied}{

            Initialize the state to $\mathcal{S}_1 = \mathcal{E} \backslash \mathcal{E}_{\rm us} $, $t \gets 1$

            \Repeat{$\tau'(\mathcal{S}_t)$ is satisfied}{ 

                $\rho \gets U(0, 1)$, $\epsilon \gets \epsilon_{\rm e} + (\epsilon_{\rm s} - \epsilon_{\rm e}) \exp( - \frac{k}{\epsilon_{\rm d}} )  $

                \eIf{$\rho < \epsilon$}{
                    $s_t  \gets \xi( \{ s \bm{|}  s \subset \mathcal{S}_t \}, 1)$ \label{alg-4-1-1-line-1}
                }
                {
                    $s_t  \gets \argmin~ _{  s \subset \mathcal{S}_t }  \widehat{Q}( \mathcal{S}_t, s, \mathcal{R}_{\rm cp}, \mathcal{R}_{\rm vp} | \Pi_k ) $
                }

                $\mathcal{S}_{t + 1} \gets \mathcal{S}_t \backslash s_t $, $t \gets t + 1$

                \If{$t > N_{\rm m}$}{
                    $\mathcal{M} \gets \mathcal{M} \cup \{( \mathcal{S}_{t - N_{\rm m}}, s_{t- N_{\rm m}}, R_{t - N_{\rm m}, t}^{\Sigma}, \mathcal{S}_t )\}  $ \label{alg-4-1-1-line-2}

                    $\mathcal{B} \gets \xi( \mathcal{M}, N_{\rm b}) $ \label{alg-4-1-1-line-3}

                    \If{$k ~{\rm{mod}}~ N_{\rm c} = 0 $}{
                        $\Pi_k^- \gets \Pi_k$
                    }

                    $\Pi_{k + 1} \gets$ Update $\Pi$ by minimizing $L$ over $\mathcal{B}$

                    $k \gets k + 1$
                }

            }

            $e \gets e + 1 $    
        } 
    }
    $\Pi^* \gets \Pi_k $ 

\caption{ The learning algorithm for {\small{$\widehat{Q}$}} }\label{algo-4-1-1} 
\end{algorithm}

The learning algorithm for $\widehat{Q}$ is designed by combining double deep Q-network and multi-step learning \citep{4-963, 4-965, 4-954}, as provided by Algorithm \ref{algo-4-1-1}. In lines \ref{alg-4-1-1-line-1} and \ref{alg-4-1-1-line-3}, $\xi(\mathcal{X}, n)$ randomly selects $\min\{ |\mathcal{X}|, n \}$ elements from $\mathcal{X}$; in line \ref{alg-4-1-1-line-2}, the cardinality of $\mathcal{M}$ is maintained equal to $N_{\rm r}$ by the last-in-first-out principle if it exceeds $N_r$ after addition. The algorithm is further explained as follows:
\begin{itemize}
    \item A complete sequence of branch removal from the redundant solution being $\mathcal{E} \backslash \mathcal{E}_{\rm us}$ to that satisfying $\tau'(\cdot)$ is called an episode. A step within an episode is a single action (branch removal). The termination condition $\tau'(\cdot)$ is formed by augmenting $\tau(\mathcal{S}_t)$ with the condition that $q(\mathcal{S}_t, \mathcal{R}_{\rm cp}, \mathcal{R}_{\rm vp}) = M$ with the OR operator, and $\tau''(\cdot)$ is formed by augmenting $\tau(\mathcal{S}_t)$ with the condition that $q(\mathcal{S}_t, \mathcal{R}_{\rm cp}, \mathcal{R}_{\rm vp}) \neq M$ with the AND operator. An episode with the terminal state satisfying $\tau''(\cdot)$ is termed a feasible episode, corresponding to a topology that is feasible to problem (\ref{eq-10-1-1}). For each sample of $\mathcal{R}_{\rm vp}$, one feasible episode is experienced.
    \item The double deep Q-network can reduce overestimation of the Q-learning values and improve stability \citep{4-954}. The multi-step learning helps to deal with the issue of delayed rewards and thus empower $\widehat{Q}$ with foresight that is lacking in $Q$. Since the objective value of an episode is only revealed after multiple branch removals, the multi-step learning waits $N_{\rm m}$ steps before updating $\Pi$ to better estimate future rewards and capture the impact of current branch switches on subsequent decisions.
    \item The parameters $\Pi_k$ are updated by the Adam algorithm (with learning rate $l_{\rm r}$) by minimizing the Huber loss: 
    \begin{equation}\label{eq-4-1-8}
        L = L_{\delta}\Big(  \widehat{Q} ( \mathcal{S}_t, s_t, \mathcal{R}_{\rm cp}, \mathcal{R}_{\rm vp} | \Pi_k ) -  Y_k \Big)
    \end{equation} 
    with $Y_k$ representing the $N_{\rm m}$-step target value given by, if $\mathcal{S}_{t + N_{\rm m}}$ is not a terminal state, 
    \begin{equation}\label{eq-4-1-7}
        \begin{aligned}
            Y_k = &  \sum\nolimits_{i = 0}^{N_{\rm m}-1}    R( \mathcal{S}_{t+i}, s_{t+i} ) + \gamma \widehat{Q} \big( \mathcal{S}_{t + N_{\rm m}},   \\  
            &  \argmin\nolimits_{ s'}  \widehat{Q} ( \mathcal{S}_{t + N_{\rm m}}, s', \mathcal{R}_{\rm cp}, \mathcal{R}_{\rm vp} | \Pi_k ), \mathcal{R}_{\rm cp}, \mathcal{R}_{\rm vp} | \Pi_k^- \big)
        \end{aligned}  
    \end{equation}
    or otherwise, $
        Y_k =  \sum\nolimits_{i = 0}^{N_{\rm m}-1}   R( \mathcal{S}_{t+i}, s_{t+i} ) $. 
    Here $\gamma \in [0, 1)$ is the discount factor; $s' \in \mathcal{S}_{t+N_{\rm m}}$; and $ \Pi_k^-$ are parameters of the target network and are updated only every $N_{\rm c}$ iterations by the assignment that $\Pi_k^- \gets \Pi_k$. It is pointed out that experience replay is used for updating $\Pi_k$ with a batch of samples $\mathcal{B}$ from the replay memory $\mathcal{M}$ that gathers experiences in the form of tuples $( \mathcal{S}_{t - N_{\rm m}}, s_{t- N_{\rm m}}, R_{t - N_{\rm m}, t}^{\Sigma}, \mathcal{S}_t )$. Here $R_{t - N_{\rm m}, t}^{\Sigma}$ is the cumulative reward from state $\mathcal{S}_{t - N_{\rm m}}$ to $\mathcal{S}_{t}$, corresponding to the first term of the right-hand side of (\ref{eq-4-1-7}). Specifically, $R_{t - N_{\rm m}, t}^{\Sigma} = \sum_{i = 0}^{N_{\rm m}-1} \!\! R( \mathcal{S}_{t - N_{\rm m} + i}, s_{t - N_{\rm m} + i} ) \!=\! q( \mathcal{S}_{t}, \mathcal{R}_{\rm cp}, \mathcal{R}_{\rm vp} ) - q( \mathcal{S}_{t - N_{\rm m}}, \mathcal{R}_{\rm cp}, \mathcal{R}_{\rm vp} )$.
    \item Stop condition $\tilde{\tau}(e, \varrho)$ is defined as $\min\{ \lambda_i |_{i=e-100}^e \} \geq \varrho$, where {$\varrho \in (0,1]$} is a threshold, and $\lambda_i$ denotes the proportion of episodes whose terminal state is not worse than the network topology computed by Algorithm \ref{algo-4-1-0} regarding objective function $q(\cdot)$, over the latest 100 episodes. More formally,
    \begin{equation}
         \lambda_i = \frac{1}{100} \sum\nolimits_{j = i}^{i - 100} H[ q(\mathcal{S}^{j*}, \mathcal{R}_{\rm cp}, \mathcal{R}_{\rm vp}^j )  - q(\mathcal{S}^j, \mathcal{R}_{\rm cp}, \mathcal{R}_{\rm vp}^j) ] 
    \end{equation}
    where $H[\cdot]$ is the Heaviside step function; $\mathcal{S}^{j*}$ is the solution of $\mathcal{S}$ given by Algorithm \ref{algo-4-1-0} with $\mathcal{R}_{\rm vp}$ being $\mathcal{R}_{\rm vp}^j$, i.e., that generated in episode $j$; and $\mathcal{S}^j$ is the terminal state of episode $j$.  
\end{itemize}

\subsubsection{Learning-based heuristic algorithm}

\begin{algorithm}[t!]
    \KwInput{$\mathcal{R}_{\rm cp}$, $\mathcal{R}_{\rm vp}$, $n_{\rm b}$, $\Pi^*$}
    \KwOutput{$\mathcal{S}^*$, $\mathcal{R}_{\rm cv}^*$} 

    Define problem (\ref{eq-10-1-1}) with given $\mathcal{R}_{\rm cp}$ and $\mathcal{R}_{\rm vp}$.

    $\mathcal{S} \gets \mathcal{E} \backslash \mathcal{E}_{\rm us}$, $\mathcal{S}_{\rm inf, 1} \gets \emptyset $, $K \gets 0$

    \While{True}{
        \Repeat{$\tau(\mathcal{S})$ is satisfied}{
            $\mathcal{S}_{\rm inf, 2} \gets \emptyset$, $\mathcal{S}_{\rm inf} \gets \mathcal{S}_{\rm inf, 1} \cup \mathcal{S}_{\rm inf, 2}$

            \Repeat{$\mathcal{G}_{\rm v}(\mathcal{V}, \mathcal{E}_{\rm us} \cup \mathcal{S} \backslash s^*, \mathcal{R})$ satisfies $\partial\{ \text{(\ref{eq-10-1-1:3})} \}$}{
                {\small{$s^* \gets \argmin_{ s \in \mathcal{S},  s \notin \mathcal{S}_{\rm inf} }  \widehat{Q}( \mathcal{S}, s, \mathcal{R}_{\rm cp}, \mathcal{R}_{\rm vp} | \Pi^* )$} \label{algo-4-1-3:1} }

                $\mathcal{S}_{\rm inf, 2} \gets \mathcal{S}_{\rm inf, 2} \cup \{ s^* \}$, $\mathcal{S}_{\rm inf} \gets \mathcal{S}_{\rm inf, 1} \cup \mathcal{S}_{\rm inf, 2}$
            }
        $\mathcal{S} \gets \mathcal{S} \backslash s^* $
        } 
        \eIf{ $\Phi ( \mathcal{S}, \mathcal{R}_{\rm cp}, \mathcal{R}_{\rm vp} )$ is feasible}{\textbf{Break}}{
            \eIf{$|\mathcal{S}_{\rm inf, 1}| > n_{\rm b}$ }{$K \gets K + 1$, $\mathcal{S}_{\rm inf, 1} \gets \emptyset$}{
                $\mathcal{S} \gets \mathcal{S} \cup \bigcup_{k=0}^{K} s^{*(-k)} $, 
                $\mathcal{S}_{\rm inf, 1} \gets \mathcal{S}_{\rm inf, 1} \cup \{s^{*(-K)}\} $
            }  
        }
    }
    $\mathcal{S}^* \gets \mathcal{S} $, $\{ \mathcal{R}_{\rm sv}^*, \mathcal{R}_{\rm cv}^*\} \gets \Phi ( \mathcal{S}^*, \mathcal{R}_{\rm cp}, \mathcal{R}_{\rm vp} ) $  
    \caption{Learning-based heuristic algorithm}\label{algo-4-1-3} 
\end{algorithm}

By substituting $Q$ in Algorithm \ref{algo-4-1-0} with the learned parameterization $\widehat{Q}( \cdot | \Pi^*)$, we can obtain the learning-based heuristic algorithm for steady-state distribution topology control. However, experiments find a small probability that this learning-based heuristic produces an infeasible topology, i.e., $\Phi ( \mathcal{S}^*, \mathcal{R}_{\rm cp}, \mathcal{R}_{\rm vp} ) $ in step \ref{algo-4-1-0:line-8} of Algorithm \ref{algo-4-1-0} is infeasible or (\ref{eq-10-1-1:3}) is violated. To prevent such infeasibility, the final learning-based heuristic algorithm contains a feasibility recovery procedure as given by Algorithm \ref{algo-4-1-3}. The basic idea of this feasibility recovery is to repeat step \ref{algo-4-1-3:1} with infeasible or suspiciously infeasible branch $s$, that are collected by $\mathcal{S}_{\rm inf}$, being excluded. Notably, when $\Phi$ is infeasible, branches $s^{*(0)} $ being removed from $\mathcal{S}$ most recently are added to $\mathcal{S}_{\rm inf}$, and then the algorithm returns to the point just before removing $s^{*(0)} $. This process is repeated at most $n_{\rm b}$ times, and if $\Phi$ is still infeasible, the branch $s^{*(-1)} $ that is removed from $\mathcal{S}$ just before $s^{*(0)} $ is added to $\mathcal{S}_{\rm inf}$ and the algorithm returns to the point just before removing $s^{*(-1)}$. This backtracking loop is broken until a feasible topology is found.

\subsection{Numerical Example}\label{sec-5-1-4}

The hybrid learning-heuristic solution  method is tested on the steady-state distribution topology control problem given by (\ref{eq-4-1-dnropt}). The modified IEEE 33-bus radial distribution network shown in Fig. \ref{fig-4-1-5} is employed. In this system, except for line 1-2, all lines are switchable; four buses, including buses 14, 22, 25 and 33, connect with distributed generators whose power outputs are optimized coordinating with network reconfiguration.  
Time-varying parameters $\mathcal{P}_{\rm vp}$ include:
$\bm{p}_{\rm d}$, $\bm{q}_{\rm d}$, $\bm{p}_{\rm g}\U$ and $v_{\rm sub}$. 
The uncertainty set $\mathbb{P}$ is defined as: 
$\bm{p}_{\rm d} \!\in\! [0.5 \tilde{\bm{p}}_{\rm d}, 2 \tilde{\bm{p}}_{\rm d} ] $, 
$\bm{q}_{\rm d} \!\in\! [0.5 \tilde{\bm{q}}_{\rm d}, 2 \tilde{\bm{q}}_{\rm d} ] $, 
$\bm{p}_{\rm g}\U \!\in\! [\bm{0}, \tilde{\bm{p}}_{\rm g}\U ] $, 
$v_{\rm sub} \!\in\! [0.975, 1.025]$, 
and the optimization model is feasible in terms of solving by both successive branch reduction methods proposed by \cite{4-847} and \cite{4-819}. Here 
$\tilde{\bm{p}}_{\rm d}$, $\tilde{\bm{q}}_{\rm d}$, $\tilde{\bm{p}}_{\rm g}\U$ denote the values of ${\bm{p}}_{\rm d}$, ${\bm{q}}_{\rm d}$, ${\bm{p}}_{\rm g}\U$ given in the dataset of the IEEE 33-bus system, respectively; and these two methods will be also used for comparison in the simulation. 

\begin{figure}[h]
	\centering
	\includegraphics[width=0.98\linewidth]{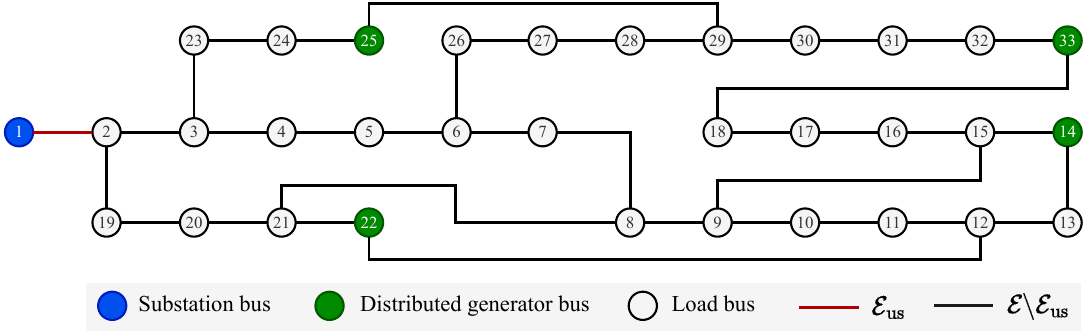}
	\caption[Diagram of the modified IEEE 33-bus radial distribution network.]{Diagram of the modified IEEE 33-bus radial distribution network. }
	\label{fig-4-1-5}
\end{figure}

\subsubsection{Hyperparameters and learning curves}\label{appendix-4-1-2}

The hyperparameter values for $\widehat{Q}$ and Algorithm \ref{algo-4-1-1} adopted in the simulation are given in Table \ref{table-0-1-1}. In the table, $\!n_{(i,j)}\!$ denotes the number of nodes in the $\!j$-th hidden layer of $g_{i}(\cdot)$. Also, $n_{\rm b} \!=\! 3 $ in Algorithm \ref{algo-4-1-3}. 

\begin{table}[h]
    \centering
    \caption{Hyperparameters of $\widehat{Q}$ and its learning algorithm.}
    \begin{tabular}{|cc|}
        \hline\hline
        Parameter & Value  \\ \hline\hline
        ($n_{\rm x}$,  $n_{\rm h}$, $N_{\rm T}$) & (3, 12, 8) \\
        $(N_{{\rm g}, i} | i \!\in\! \llbracket 1, 4 \rrbracket )$  & (5,5,3,3) \\
        $(n_{{\rm g},i}  | i \!\in\! \{2, 4\} )$ & (32,32) \\
        $(n_{(1,i)}| i \!\in\! \llbracket 1, N_{{\rm g}, 1} \rrbracket )$ & (15,12,15,6,3) \\
        $(n_{(2,i)}| i \!\in\! \llbracket 1, N_{{\rm g}, 2} \rrbracket )$ & (15,6,15,24,27) \\
        $(n_{(3,i)}| i \!\in\! \llbracket 1, N_{{\rm g}, 3} \rrbracket )$ & (12,15,3)  \\
        $(n_{(4,i)}| i \!\in\! \llbracket 1, N_{{\rm g}, 4} \rrbracket )$ & (15,15,27)  \\
        $(N_{\rm r}, N_{\rm m}, N_{\rm c}, N_{\rm b})$ & (128,3,50,32)   \\
        $(\epsilon_{\rm s}, \epsilon_{\rm e}, \epsilon_{\rm d})$ & (0.9,0.001,400) \\
        ($\gamma$, $l_{\rm r}$, $\varrho$) & (0.9, $3\!\times\! 10^{-2}$, 0.9) \\
        \hline\hline
    \end{tabular}
    \label{table-0-1-1}
\end{table}

\begin{figure}[h]
    \centering
        \includegraphics[width=0.53\linewidth]{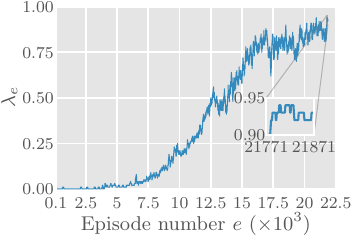} 
    \caption{Learning curves of the learning algorithm for $\widehat{Q}$.}\label{fig-4-1-curve}
\end{figure}

Fig. \ref{fig-4-1-curve} shows the learning curves, i.e., changes of $\lambda_i$, during the process of training in Algorithm \ref{algo-4-1-1}. It is observed that after about 7500 episodes, $\lambda_e$ starts to noticeably increase, which indicates that $\widehat{Q}$ has partially learned the evaluation of the impact of removing a branch in the context of a given topology, behaving like evaluation function $Q$ but potentially being more foresighted. At the 21871-th episode, $\lambda_e$ of the latest 100 episodes are all larger than $\varrho = 0.9$, and the stop condition is satisfied.

\subsubsection{Evaluation on the test dataset}

Performance of the learning-based heuristic method is evaluated on a test dataset which contains 2000 optimization model instances with different samples of $\mathcal{R}_{\rm vp}$. The solution quality of heuristic algorithms can be measured by the average relative optimality gap defined as 
\begin{equation}
    G_{\rm o} = \frac{1}{n_{\rm st}} \sum\nolimits_{i=1}^{n_{\rm st}} G_{{\rm o}, i} = \frac{1}{n_{\rm st}} \sum\nolimits_{i=1}^{n_{\rm st}} \frac{ f_{{\rm h}, i} - f_{{\rm go}, i}  }{ f_{{\rm go}, i} }  
\end{equation} 
where $n_{\rm st}$ is the size of test dataset; $G_{{\rm o}, i}$ represents the relative optimality gap of the $i$-th solution; $f_{{\rm go}, i}$ denotes the globally minimum value of $f(\cdot)$ that is obtained by Gurobi; and $f_{{\rm h}, i}$ denotes the minimum value of $f(\cdot)$ obtained by the learning-based heuristic or the successive branch reduction methods.

\begin{figure}[h]
    \centering
    \includegraphics[width=0.98\linewidth]{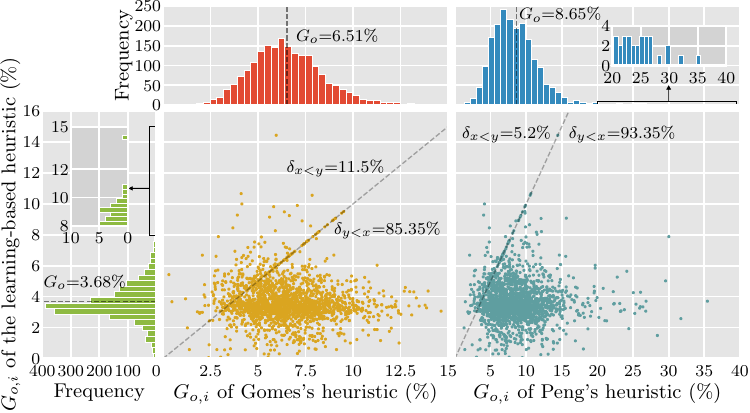} 
    \caption{Comparison of $G_{{\rm o}, i}$ of solutions found by different methods.}
    \label{fig-4-1-dnr-r21:evaluation}
\end{figure} 

Fig. \ref{fig-4-1-dnr-r21:evaluation} shows scatter plots and histograms of $G_{{\rm o}, i}$ of solutions found by the learning-based heuristic and the heuristic methods developed by \cite{4-847} and \cite{4-819}. In this figure, the dashed lines in the histograms indicate $G_{\rm o}$ for each heuristic; the dashed lines in the scatter plots locate the position where x-axis and y-axis values are equal; $\delta_{x<y}$ denotes the proportion of points whose x-axis values are smaller than y-axis values and $\delta_{y<x}$ is analogous. By Fig. \ref{fig-4-1-dnr-r21:evaluation}, the following conclusions for the evaluation on the test dataset can be drawn:
\begin{itemize}
    \item Overall, the learning-based heuristic outperforms the two successive branch reduction methods developed by \cite{4-847} and \cite{4-819} in terms of solution optimality, and {$G_o = 3.68\%$} for the learning-based heuristic which is noticeably smaller than that of the two successive branch reduction methods.
    \item Compared with the successive branch reduction methods, solution optimality of the learning-based heuristic is more robust with changes of time-varying parameters $\mathcal{R}_{\rm vp}$, which is indicated by the tall and skinny histogram for the learning-based heuristic.
    \item For some samples of $\mathcal{R}_{\rm vp}$, the learning-based heuristic even finds solutions which almost equal to the global optimum.
    \item The successive branch reduction methods may produce solutions with a significantly large optimality gap. The worse-case optimality gap of solutions found by the heuristic developed by \cite{4-819} is more than {35.00\%}.
    \item The learning-based heuristic underperforms compared with the two successive branch reduction methods for a small portion of samples, whose corresponding relative optimality gaps for the learning-based heuristic, however, are mostly less than {10.00\%}.
\end{itemize}

\subsubsection{Solution time}

\begin{figure}[h]
    \centering
    \includegraphics[width=0.85\linewidth]{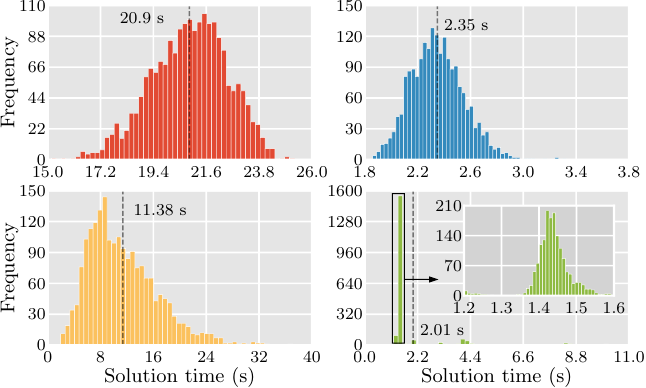} 
    \caption{Comparison of solution time of different heuristics and Gurobi.}
    \label{fig-4-1-dnr-r21:time}
\end{figure}

Fig. \ref{fig-4-1-dnr-r21:time} shows the histograms of solution time of the heuristics developed by \cite{4-847} (red plot) and \cite{4-819} (blue plot), Gurobi (yellow plot) and the learning-based heuristic (green plot), where the dashed lines indicate the average solution time. For the sake of fairness, the parallel algorithm of Gurobi is disabled since the three heuristics are not with a parallel design. The learning-based heuristic dominates Gurobi and the successive branch reduction method developed by \cite{4-847} regarding the average solution time. Although the heuristic developed by \cite{4-819} and the learning-based heuristic are with close average solution time and the worse-case solution time of the former is even much less than that of the other three algorithms, the heuristic developed by \cite{4-819} performs worst in solution optimality. In addition, the solution time of most samples solved by the learning-based heuristic concentrates in a narrow range, i.e., about 1.3 s to 1.6 s. This indicates that the learning-based heuristic is also more robust in terms of solution time with changes of $\mathcal{R}_{\rm vp}$.

\subsubsection{Solution construction}

\begin{figure}[h] 
	\centering
    \includegraphics[width=1\linewidth]{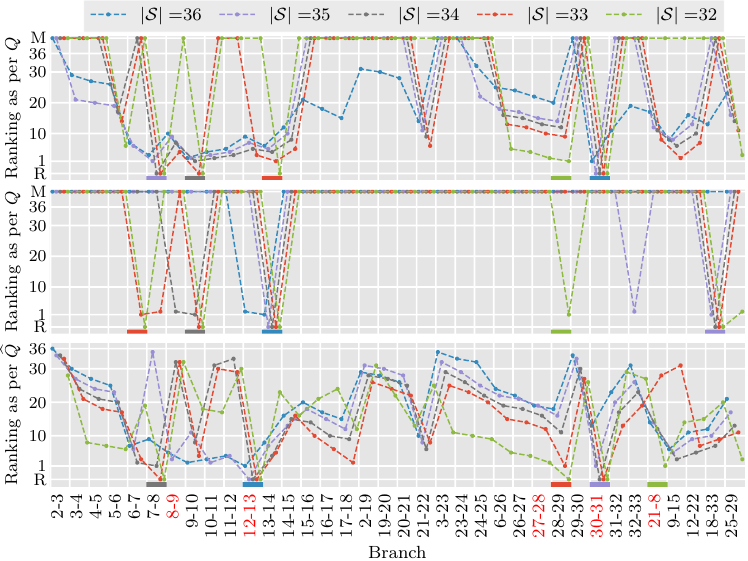} 
    \caption{Solution construction processes of different heuristics.}
    \label{fig-4-1-dnr-r4} 
\end{figure}

\begin{figure}[h]
	\centering
    \includegraphics[width=1\linewidth]{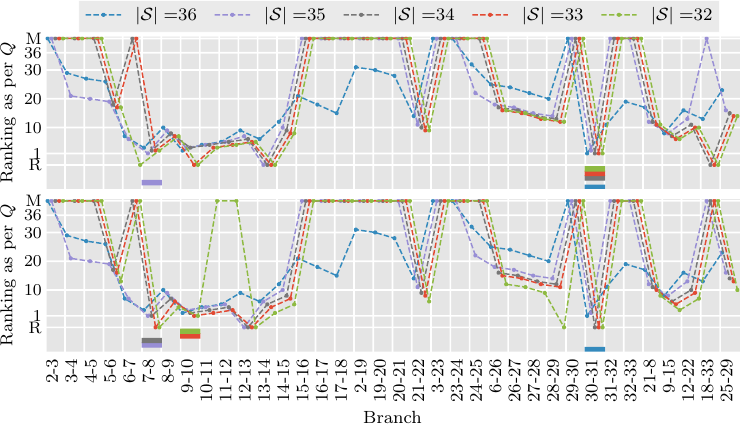} 
    \caption{Branch ranking according to $Q$ used in \cite{4-847}, during different solution construction processes.}
    \label{fig-4-1-dnr-r5} 
\end{figure}
 
Finally, the solution construction of the three heuristics are compared for one selected test sample where the minimum network losses given by the heuristics developed by \cite{4-819} and \cite{4-847}, the learning-based heuristic and Gurobi are 162.86 kW, 173.32 kW, 154.56 kW and 149.34 kW, respectively. 
Fig. \ref{fig-4-1-dnr-r4} shows the solution construction processes according to branch ranking evaluated by $Q$ of the heuristics developed by \cite{4-847} (top plot) and \cite{4-819} (middle plot), and by $\hat{Q}$ of the learning-based heuristic (bottom plot), for the selected test sample. In the figure, the ranking is in ascending order; the branch being selected to remove at each step is marked by a rectangular block with the same color as the ranking curve of that step; the ranking being ``R" indicates branches have already been removed, and that being ``M" indicate $Q=M$; branches to remove corresponding to the solution given by Gurobi are red-colored in the x-axis. 
In addition, Fig. \ref{fig-4-1-dnr-r5} gives the branch ranking during the solution construction processes of the heuristic developed by \cite{4-819} (top plot) and the learning-based heuristic (bottom plot), but evaluated by $Q$ used by \cite{4-847} which finds the exact locally optimum for each individual step. 

According to Fig. \ref{fig-4-1-dnr-r4} and Fig. \ref{fig-4-1-dnr-r5}, it can be found that: 
\begin{itemize}
    \item In the first step, the heuristic developed by \cite{4-847} selects branch 30-31 that is also one of the removed branches in the globally optimal solution. However, neither of the next 4 branches selected by the heuristic developed by \cite{4-847} are removed by the globally optimal solution. This verifies the short-sightedness of the hand-crafted heuristic, namely that locally optimal selection at each step results in a final solution being relatively far away from the global optimum.
    \item For the heuristic developed by \cite{4-819}, none of the branches being selected to remove is switched off in the globally optimal topology. The $Q$ function for this heuristic is computationally cheap compared with that used by \cite{4-847}, while at least for the selected sample, this cheap $Q$ leads to neither a locally optimal topology for each individual step nor a competitive final optimum.
    \item For the learning-based heuristic, except for the step where $|\mathcal{S}| = 34$, none of the branches selected to remove according to $\widehat{Q}$ are locally optimal for the current step. For example, the branch selected at the first step, i.e., branch 12-13, only ranks 9th in terms of local optimality. However, the final solution, where the three removed branches are also switched off in the global optimum, is with lower network loss compared with the other two hand-crafted heuristics. This verifies the foresight of the learning-based heuristic, namely that function $\widehat{Q}$ is embedded with consideration of the long-term impacts of removing a branch instead of merely a local impact.
    \item The learning-based heuristic possibly makes obvious mistakes at some step. For instance, for the last step where $|\mathcal{S}|= 32$, the learning-based heuristic selects branch 21-8 to remove, but the best branch to switch off is 9-10 by Fig. \ref{fig-4-1-dnr-r5}. However, this kind of mistake occurring at the last step can be overcome by a combination of the learning-based and hand-crafted heuristics, where only the last branch is selected based on $Q$ used by \cite{4-847}.
\end{itemize}

\graphicspath{{chapter_6/Figs/}}
\chapter{Network Topology Transition}\label{chapter-6}

After obtaining the optimal topology for steady-state topology control, the subsequent problem is \textit{how to transition from the initial topology to the target one (i.e., the optimal topology)}. This topology transition problem, which arises in the context of implementing steady-state topology control, has been largely overlooked in conventional power systems. However, it is becoming increasingly critical in future high-renewable power systems where steady-state topology control is expected to be deployed more broadly and actively. Recently, the concepts of \textit{optimal topology transition} and \textit{bumpless topology transition}, along with corresponding methodologies to address the topology transition problem, have been proposed by \cite{4-1309} and \cite{4-1372}. The former takes only static factors into consideration, whereas the latter incorporates both static and dynamic aspects of the topology transition problem. 
In this chapter, the topology transition problem is introduced, with particular emphasis on the bumpless topology transition method.

\section{Emerging Needs and Key Issues}

The sequence of line switching operations during a topology transition should be carefully planned rather than arbitrary. Although both the initial and target topologies may satisfy operational constraints, intermediate topologies during the transition might violate these constraints. Furthermore, the line-switching actions involved in the transition process constitute significant disturbances to the system, potentially leading to stability issues. However, the network topology transition problem has received little particular attention for conventional power systems due to the following reasons:
\begin{itemize}
    \item Steady-state topology control is traditionally executed infrequently. For example, the execution cycle of Distribution Network Reconfiguration (\ac{dnr}) typically spans several days and often extends to several weeks or even months. This low execution frequency implies a long time window with the feasible region for topology transition being sufficiently large, enabling operators to readily select an ad hoc transition strategy. Furthermore, given the infrequent execution and the short duration of the transition process, slight violations of operational constraints during topology transitions may be considered acceptable. 
    \item Conventional power systems are commonly unstressed with regular power flow patterns, which diminishes the necessity of particular attention to topology transition. Specifically, for these common power flow patterns, operators can rely on preplanned topology control schemes and  their corresponding transition strategies. 
    
    \item For steady-state distribution topology control, topology transition has little impact on the stability of conventional distribution networks with only passive loads. For steady-state transmission topology control, electromechanical stability is usually unaffected by the transition process, provided that sufficient time separation is maintained between consecutive switching actions to allow the system to settle.
\end{itemize}

With the renewable energy transition that leads to high penetration of variable renewable energy, the network  topology transition problem is becoming increasingly critical. This growing importance can be attributed to the following two emerging aspects: 
\begin{itemize}
    \item Constant and dramatic changes in variable renewable energy promote more frequent and even dynamic real-time steady-state topology control to maximize profits \citep{4-583, 4-905}. This complicates the selection of an ad hoc feasible transition strategy and motivates optimization of the transition process. Also, power flow patterns become highly diverse due to uncertainties in variable renewable energy together with changes in demand, and power systems are stressed both locally and globally. Accordingly, the feasible region of topology transition is potentially highly time-varying. 
    
    \item Dynamic properties of converter-dominant power systems complicate topology transition. For transmission networks, converter dynamics enlarge the timescale of structure-dependent stability phenomena to a period from several milliseconds to tens of seconds \citep{4-1302}. Hence, fast transitions can induce a wider range of instability phenomena to prevent. In medium/low voltage networks, inverter-based distributed generators also introduce new dynamics that complicate the overall system dynamic behavior \citep{4-990}. For instance, in AC microgrids, it has been observed that fast line dynamics also have a major influence on the stability of slower modes \citep{4-891}, and network loops can reduce the stability margin \citep{4-620}. These structure-related observations contrast with the conventional view, which highlights the necessity of re-examining stability concerns associated with topology transition. 
\end{itemize}

Although topology transition serves as the implementation phase of the optimal topology determined by steady-state topology control, the transition problem can be even more challenging than the steady-state topology control problem itself. First, both static and transient performance during the topology transition process must be considered. In particular, as a fundamental requirement, the transition methodology should incorporate mechanisms to ensure system stability throughout the transition. 
Second, given that line switching constitutes a significant system disturbance, optimizing only the switching sequence may be insufficient to mitigate the disturbance to an acceptable level. The incorporation of additional control resources, such as generation control as discussed by \cite{4-1285}, can further elevate the complexity of the topology transition problem.

\section{Preliminaries}\label{sec-7-1}

This section provides the necessary preliminaries for the network topology transition problem, including a detailed problem description, an analysis of the transition process, and associated system models. Although topology transition issues arise in both transmission and distribution networks, the remainder of this chapter focuses on transmission systems, where topology transitions pose greater challenges.
 
\subsection{Problem Description}\label{sec-pd-a}

Given that optimizing only the line switching sequence may be insufficient to ensure a satisfactory topology transition process, the Auxiliary Control Variables (\ac{acv}s) are considered, which refer to the electrical properties and dynamic parameters of the transmission system that are allowed to be adjusted to assist the topology transition. 
Considering the practical difficulties of implementation, it is reasonable to assume that \ac{acv} adjustment and line switching are executed asynchronously. Then the topology transition process is determined by a sequence of executions of \ac{acv} adjustment and line switching.  
The concepts of \textit{complete transition episode} and \textit{transition episode} given by Definition \ref{def-6-2-0} can be used to represent any transition process as a unique standard form. These two concepts are illustrated in Fig. \ref{fig-6-2-6} with a transition process. 
For ease of description and modeling, some fictitious executions of \ac{acv} adjustment and line switching are introduced to convert all transition episodes to complete transition episodes. 

\begin{figure}[h]
	\centering
	\includegraphics[width=0.9\linewidth]{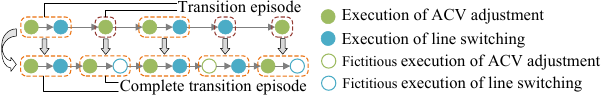}  
	\caption[Transition episodes and complete transition episodes.]{Illustration of transition episodes and complete transition episodes.
    }  
	\label{fig-6-2-6}  
\end{figure}
 
\begin{definition}[Complete transition episode, transition episode]\label{def-6-2-0}
    A complete transition episode is an execution of \ac{acv} adjustment and a following execution of line switching. A transition episode is a complete transition episode, or an execution of \ac{acv} adjustment or line switching in the transition process excluding all complete transition episodes.
\end{definition}

With the above conversion, all transition episodes in the following analysis are treated as complete transition episodes. Next, denote the state of topology and \ac{acv}s by $(\bm{z}, \bm{a})$. 
Then the $i$-th transition episode can be denoted as $(\bm{z}^{i-1}, \bm{a}^{i-1}) \to (\bm{z}^{i-1}, \bm{a}^{i}) \to (\bm{z}^{i}, \bm{a}^{i})$, where $\bm{z}^{i-1} \in \mathbb{B}^{n_{\rm e}}$, $\bm{z}^i \in \mathbb{B}^{n_{\rm e}}$, $\bm{a}^{i-1} \in \mathbb{R}^{n_{\rm a}}$, $\bm{a}^i \in \mathbb{R}^{n_{\rm a}}$, $(\bm{z}^{i-1}, \bm{a}^{i-1}) \to (\bm{z}^{i-1}, \bm{z}^{i})$ means adjusting $\bm{a}$ from $\bm{a}^{i-1}$ to $\bm{a}^{i}$, and $(\bm{z}^{i-1}, \bm{a}^{i}) \to (\bm{z}^{i}, \bm{a}^{i})$ means changing $\bm{z}$ from $\bm{z}^{i-1}$ to $\bm{z}^i$. 
When $\bm{a}^{i-1} = \bm{a}^i$ or $\bm{z}^{i-1} = \bm{z}^i$, the $i$-th transition episode contains a fictitious execution of \ac{acv} adjustment or line switching. Accordingly, the transition process can be represented by the transition trajectory of topology and \ac{acv}s, i.e., $ \cdots \to (\bm{z}^{i-1}, \bm{a}^{i-1}) \to (\bm{z}^{i-1}, \bm{a}^{i}) \to (\bm{z}^{i}, \bm{a}^{i}) \to \cdots$. With the above notations, a formal description of the network topology transition problem is given as follows: 
\begin{problem}[Network topology transition]\label{def-6-2-1}
    Given the initial value of $\bm{a} = \bm{a}^0$, an initial topology $\bm{z}^0$, and a final topology $\bm{z}^T$, under which the systems are operationally feasible, find a feasible transition trajectory of topology and \ac{acv}, i.e., $(\bm{z}^0, \bm{a}^0) \to (\bm{z}^0, \bm{a}^1) \to (\bm{z}^1, \bm{a}^1) \to (\bm{z}^1, \bm{a}^2) \to ... \to (\bm{z}^{T-1}, \bm{a}^{T-1}) \to (\bm{z}^{T-1}, \bm{a}^{T}) \to (\bm{z}^{T}, \bm{a}^{T})$ with $\bm{a}^T = \bm{a}^0$, such that the transition process is as bumpless as possible.
\end{problem}

The following four criteria can be employed to select \ac{acv}s: 
\begin{itemize}
    \item \textit{Fast response time}: To minimize overall system disruption, the topology transition process should be as short as possible. This requires \ac{acv}s to have fast response time.
    \item \textit{Bumpless}: The purpose of adjusting the \ac{acv}s is to facilitate a bumpless topology transition; therefore, the adjustment process of \ac{acv}s should be as bumpless as possible, which generally necessitates the continuity of the \ac{acv}s.
    \item \textit{Negligible or low cost}: Negligible or low cost of adjusting \ac{acv} can maximize the profits from steady-state topology control. 
    \item \textit{Negligible impacts on reliability}: Topology transitions are potentially frequent in future power systems, necessitating frequent adjustments of \ac{acv}s. These adjustments must have negligible impact on asset reliability and overall system operation.
    \item \textit{System-wide impacts}: Given that line switching actions across the network may be involved in numerous topology transition scenarios, the selected \ac{acv}s should be capable of contributing to system-wide impacts.
\end{itemize}

\begin{table}[h]
    \caption{Common control resources in transmission networks}\label{tab-6-2-1}
    \centering 
    \small{
    \begin{tabularx}{1.0\linewidth}{|p{2.5cm}<{\centering}p{2.6cm}<{\centering}p{1.0cm}<{\centering}p{1.0cm}<{\centering}p{0.4cm}<{\centering}p{1.6cm}<{\centering}|}
    \hline\hline   
    Resource & Response time & Bumpless & Cost & RI & Competent 
    \\  
    \hline\hline
    \ac{sg}'s output                &  2-30\%/min$^{\textcolor{blue}{\rm a}}$  & \checkmark & High & Neg. & $\bm{\times}$   \\
    Generator TV     & $<$1s & \checkmark & Negl. &   Negl. & \checkmark   \\
    \ac{ess}'s output               & $>$200\%/s$^{\textcolor{blue}{\rm a}}$      & \checkmark & Low  &  Negl. & \checkmark    \\
    \ac{cig}'s I\&D   & $<$1s & \checkmark & Negl. & Negl. & \checkmark  \\ \hline
    Load demands                       & $<$30s  & \checkmark  & High  & Negl. & $\bm{\times}$  \\
    Shunt capacitors                   & $<$1s                & $\bm{\times}$  & Negl. & Neg. & $\bm{\times}$   \\
    \ac{dvc}                          & $<$5ms, 20-40ms $^{\textcolor{blue}{\rm b}}$  & \checkmark & Negl. & Negl. & \checkmark   \\ \hline
    OLTC                          &  3s-10s  & \checkmark/$\bm{\times}$$^{\textcolor{blue}{\rm c}}$  & Negl.  & Neg. & $\bm{\times}$  \\
    \ac{tcsc}                          & 15-20ms  & \checkmark & Negl. & Negl.   & \checkmark     \\
    Line switching                     & $<$1s  & $\bm{\times}$ &  Negl. & Negl.  & \checkmark    \\
    \hline\hline
    \multicolumn{6}{p{0.965\textwidth}}{
        \footnotesize{  
            Abbreviation: RI (reliability impact), TV (terminal voltage), I\&D (inertia and damping), Negl. (negligible), Neg. (negative).
    }}\\ 
    \multicolumn{6}{p{0.965\linewidth}}{
        \footnotesize{ 
            $^{\textcolor{blue}{\rm a}}$ The value is the ramp rate, expressed in the percent of maximum capacity per minute  or second, that a generator or \ac{ess} changes its output. 
            $^{\textcolor{blue}{\rm b}}$ The two ranges respectively  correspond to \ac{statcom} and \ac{svc}. 
            $^{\textcolor{blue}{\rm c}}$ This depends on the type of \ac{oltc}, namely that the tap adjustment is continuous or discrete. 
     }}
    \end{tabularx} }
\end{table}

The common control resources in transmission networks are listed in Table \ref{tab-6-2-1}. According to their performances regarding the above mentioned criteria, at the generation side, terminal voltages of generators, outputs of Energy Storage Systems (\ac{ess}s), and virtual inertia and damping of Converter-Interfaced Generators (\ac{cig}s) can be selected as \ac{acv}s. 
At the load side, Dynamic VAR Compensators (\ac{dvc}s), mainly including Static VAR Compensator (\ac{svc}) and Static Synchronous Compensator (\ac{statcom}), can be selected. \ac{dvc}s are commonly used to maintain constant bus voltage and therefore the voltage setpoints of \ac{dvc}s can be chosen as \ac{acv}s.

At the network side, Thyristor-Controlled Series Compensator (\ac{tcsc}) and line switching can be selected for auxiliary control. Taking the open-loop impedance control of \ac{tcsc} as an example, \ac{tcsc} is equivalent to a constant series reactance operating at its setpoint \citep{4-1303}, and thus the reactance setpoints of \ac{tcsc} can be chosen as \ac{acv}s. 
Line switching used as auxiliary control for topology transition refers to the interim line switching. Specifically, let $\mathcal{E}_i \subseteq \mathcal{E}$ be the set containing all closed lines indicated by $\bm{z}^i$. To transition from topology $\mathcal{E}_0$ to $\mathcal{E}_T$, $|\mathcal{E}_0 \backslash \mathcal{E}_T|$ openings of every line in $\mathcal{E}_0 \backslash \mathcal{E}_T$ and $|\mathcal{E}_T \backslash \mathcal{E}_0|$ closures of every line in $\mathcal{E}_T \backslash \mathcal{E}_0$ are necessary. The others are just interim line switchings, which include switching of lines not in $(\mathcal{E}_T \backslash \mathcal{E}_0) \cup (\mathcal{E}_0 \backslash \mathcal{E}_T)$, switching of lines in $\mathcal{E}_0 \backslash \mathcal{E}_T$ which is not the first open, and switching of lines in $\mathcal{E}_T \backslash \mathcal{E}_0$ which is not the first closure. Notably, 
interim line switching can be incorporated into $\bm{z}$ instead of $\bm{a}$, such that all \ac{acv}s are continuous variables. 

In addition, the following assumptions on network topology transition can be adopted: 
\begin{itemize} 
    \item \text{A.1}: In each transition episode, \ac{acv} adjustment is fast enough and the induced dynamic response is smooth enough, such that the associated transient process of system state is negligible for evaluating the transition performance.
    \item \text{A.2}: Line switching is performed asynchronously, with at most one line switched per transition episode; both line switching and \ac{acv} adjustment occur after the system reaches steady state.
    \item \text{A.3}: All control variables and system parameters, except for \ac{acv}s and network topology, remain constant during the transition.
    \item \text{A.4}: Let $n_{\rm ad}$ and $n_{\rm is}$ be the number of executions of \ac{acv} adjustment and interim line switching actions performed, respectively. They are bounded as follows:  
    \begin{equation}\label{eq-6-2-a4} 
        n_{\rm ad} T_{\rm ad}  +  n_{\rm is}  T_{\rm is} \leq T_{\max}
    \end{equation} 
    where $T_{\rm ad}$ and $T_{\rm is}$ are the estimated increases of transition time due to an execution of \ac{acv} adjustment and an interim line switching action, respectively; and $T_{\max}$ is the maximal allowable increase of transition time from both actions.
\end{itemize} 

\subsection{Phases of the Topology Transition Process}

\begin{figure}[h]
	\centering
	\includegraphics[width=1\linewidth]{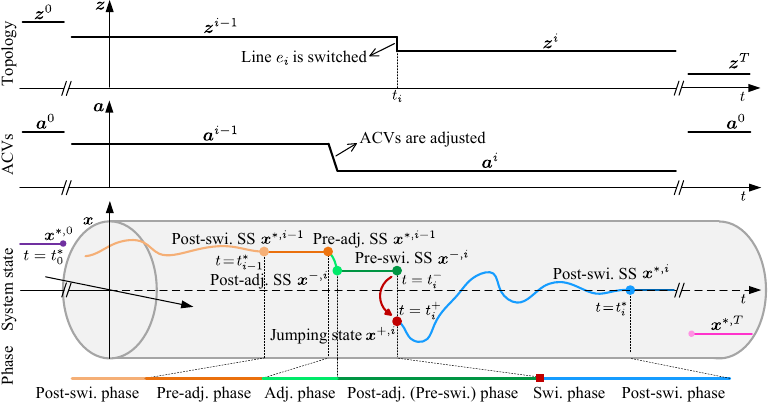}  
	\caption{Illustration of the process of network topology transition (SS: steady state; adj.: adjustment; swi.: switching). }  
	\label{fig-6-2-1}  
\end{figure}

The process of topology transition is illustrated in Fig. \ref{fig-6-2-1} that takes the $i$-th transition episode as an example. This process, starting at $t = t_0^*$ with $(\bm{z}, \bm{a})=(\bm{z}^0, \bm{a}^0)$ and $\bm{x} = \bm{x}^{*,0}$ consists of the following phases: 
\begin{itemize}
    \item \textit{Pre-adjustment phase}: This phase refers to the period after the system reaches steady state $\bm{x}^{*,i-1}$, called post-switching steady state, at $t=t_{i-1}^*$, and before \ac{acv} adjustment. 
    \item \textit{Adjustment phase}: This phase refers to the period where \ac{acv}s are adjusted, during which the system state changes from the pre-adjustment steady state, identical to the post-switching steady state by assumption A.3, to a new steady state $\bm{x}^{-, i}$, termed the post-adjustment steady state. The transient process of this phase is neglected per assumption A.1. 
    \item \textit{Post-adjustment or pre-switching phase}: This phase spans from the post-adjustment steady state and the switching of any line $e_i$.
    \item \textit{Switching phase}: This phase refers to the instant $t_i$ where the line $e_i$ is switched. In this phase, the system state jumps from the pre-switching steady state at $t=t_i^-$, identical to the post-adjustment steady state, to the state $\bm{x}^{+, i}$ at $t = t_i^+$, called the jumping state. 
    \item \textit{Post-switching phase}: This phase is the period after the system jumps to the jumping state and if the system is stable, ending with reaching the next post-switching steady state $\bm{x}^{*, i}$ at $t = t_{i}^*$.
\end{itemize}

The topology transition process ends with the network topology switched to $\bm{z}^T$ and \ac{acv} reset to its initial value $\bm{a}^0$. The system in each phase henceforth is named after that phase.

\subsection{System Models}\label{sec-7-2}

Consider a transmission network containing 
$n_{\rm ng}$ generator buses, 
$n_{\rm ce}$ \ac{cig}s with \ac{ess}s, 
$n_{\rm tcsc}$ \ac{tcsc}s, 
and $n_{\rm dvc}$ \ac{dvc}s that include $n_{\rm svc}$ \ac{svc}s and $n_{\rm svg}$ \ac{statcom}s. 
Let $\bm{a} \in \mathbb{R}^{n_{\rm a}}$ denote the vector of all $n_{\rm a}$ \ac{acv}s, which is further divided into two subvectors, i.e., $\bm{a}_{\rm s}$ consisting of the \ac{acv}s which impact steady states, and $\bm{a}_{\rm t}$ consisting of the \ac{acv}s which have an impact on transients but no impact on steady states. Specifically, 
\begin{equation}
    \bm{a}_{\rm s} = [\bm{v}_{\rm g}\T, \bm{p}_{\rm g, ess}\T, \bm{v}_{\rm dvc}\T, \bm{b}_{\rm b, tcsc}\T]\T, 
    \bm{a}_{\rm t} = [\bm{m}_{\rm cg}\T, \bm{d}_{\rm cg}\T ]\T
\end{equation}
where 
$\bm{v}_{\rm g} \in \mathbb{R}^{n_{\rm ng}}$ is the subvector of $\bm{v}$ associated with the generator buses, 
$\bm{p}_{\rm g, ess} \in \mathbb{R}^{n_{\rm ce}}$ is the subvector of $\bm{p}_{\rm g}$ associated with \ac{cig}s with \ac{ess}s, 
$\bm{v}_{\rm dvc} \in \mathbb{R}^{n_{\rm dvc}}$ is the subvector of $\bm{v}$ associated with buses with \ac{dvc}s, 
$\bm{b}_{\rm b, tcsc} \in \mathbb{R}^{n_{\rm tcsc}}$ is the subvector of $\bm{b}_{\rm b}$ associated with the lines with \ac{tcsc}, 
$\bm{m}_{\rm cg} \in \mathbb{R}^{n_{\rm gv}}$ and $\bm{d}_{\rm cg} \in \mathbb{R}^{n_{\rm gv}}$ are respectively the inertia and damping coefficients of \ac{cig}s. Also, let $\bm{a}_{\rm s}^i$ and $\bm{a}_{\rm t}^i$ be the values of $\bm{a}_{\rm s}$ and $\bm{a}_{\rm t}$ corresponding to $\bm{a}^i$, respectively.

\textit{Power flow models}: 
The AC power flow model of the entire system, consisting of that for the network, generators (both \ac{cig}s and Synchronous Generators (\ac{sg}s)), loads, \ac{svc} and \ac{statcom} can be represented in a descriptor form as follows: 
\begin{equation}\label{eq-6-2-cfnonlinear}
    f_{\rm p}( \bm{x}_{\rm p} | \bm{z}, \bm{a}_{\rm s} | \bm{y}_{\rm p}) = \bm{0}  
\end{equation}
where $\bm{x}_{\rm p}$ is the vector of all voltage variables, and $\bm{y}_{\rm p}$ is the vector of observed variables which are needed in the formulation of the topology transition problem. Specifically, 
\begin{subequations}
    \begin{align}
        & \bm{x}_{\rm p} = [\bm{v}\T, \bm{\theta}\T, \bm{e}\T, \bm{\delta}\T, \bm{e}_{\rm s}\T, \bm{\delta}_{\rm s}\T, \bm{v}_{\rm m}\T, \bm{\theta}_{\rm m}\T, \bm{e}_{\rm m}\T, \bm{\delta}_{\rm m}\T, \bm{v}_{\rm svg}\T, \bm{\theta}_{\rm svg}\T]\T \\ 
        & \bm{y}_{\rm p} = [ \bm{p}_{\rm fb}\T, \bm{q}_{\rm fb}\T, \bm{p}_{\rm tb}\T, \bm{q}_{\rm tb}\T, \bm{p}_{\rm go}\T, \bm{q}_{\rm g}\T, \bm{\epsilon}_{\rm p}\T, \bm{\epsilon}_{\rm q}\T, \bm{q}_{\rm c}\T, \bm{b}_{\rm svc}, \bm{b}_{\rm b, tcsc}\T ]\T
    \end{align}
\end{subequations} 
where $\bm{e} \angle \bm{\delta} \in \mathbb{C}^{n_{\rm gs}}$ is the electromotive force (emf) of \ac{sg}s, 
$\bm{e}_{\rm s} \angle \bm{\delta}_{\rm s} \in \mathbb{C}^{n_{\rm gs}}$ is the subtransient emf of \ac{sg}s, 
$\bm{v}_{\rm m} \angle \bm{\theta}_{\rm m} \in \mathbb{C}^{n_{\rm gv}}$ is the modulation voltages at the outputs of \ac{cig}s; 
$\bm{\epsilon}_{\rm p} \in \mathbb{R}^{n_{\rm d}}$ and $\bm{\epsilon}_{\rm q} \in \mathbb{R}^{n_{\rm d}}$ are the weights of induction motor components of active and reactive loads, respectively; 
$\bm{p}_{\rm go} \in \mathbb{R}^{n_{\rm g} - n_{\rm ce}}$ is formed by removing $\bm{p}_{\rm g, ess}$ from $\bm{p}_{\rm g}$,     
$\bm{b}_{\rm svc} \in \mathbb{R}^{n_{\rm svc}}$ is the susceptances of \ac{svc}s, 
$\bm{q}_{\rm c} \in \mathbb{R}^{n_{\rm dvc}}$ is the reactive power outputs of \ac{dvc}s, 
$\bm{e}_{\rm m}\angle \bm{\delta}_{\rm m} \in \mathbb{C}^{n_{\rm d}}$ is the internal voltages behind the subtransient impedances of the induction motor component of loads, 
$\bm{v}_{\rm svg} \angle \bm{\theta}_{\rm svg} \in \mathbb{C}^{n_{\rm svg}}$ is the modulation voltages at the outputs of \ac{statcom}s. 

To enhance computational tractability, linearized formulations of the aforementioned AC power flow model may be employed. 
Without loss of generality, the following descriptor form can be used to represent any such linearized version of the AC power flow model: 
\begin{equation}\label{eq-6-2-cflinear}
    \tilde{f}_{\rm p}( \bm{x}_{\rm p} | \bm{z}, \bm{a}_{\rm s} | \bm{y}_{\rm p}) \leq \bm{0}
\end{equation}
where $\tilde{f}(\cdot)$ is linear in terms of $\bm{x}_{\rm p}$, $\bm{z}$, $\bm{a}_{\rm s}$ and $\bm{y}_{\rm p}$, and the inequality arises from the inequality constraints introduced in the linearization. 

\textit{Dynamic models}:
The system dynamics can be represented by a state-space descriptor form as follows: 
\begin{subequations}\label{eq-6-2-13}
    \begin{align}
        & \begin{bmatrix}
            \dot{\bm{x}} \\
            \bm{0}
        \end{bmatrix}
        = 
        \begin{bmatrix}
            f(\bm{x}, \bm{\xi}, \bm{z}, \bm{a}) \\
            g(\bm{x}, \bm{\xi}, \bm{z}, \bm{a})
        \end{bmatrix}
        \\
        & \bm{y} = h(\bm{x}, \bm{\xi}, \bm{z}, \bm{a})
    \end{align}
\end{subequations}
where $\bm{x} \in \mathbb{R}^{n_{\rm x}}$ and $\bm{\xi} \in \mathbb{R}^{n_{\rm xi}}$ are the vectors of $n_{\rm x}$ state variables and $n_{\rm xi}$ algebraic variables, respectively; 
and $\bm{y} \in \mathbb{R}^{n_{\rm y}}$ is the vector of $n_{\rm y}$ performance outputs.
The linearized model of (\ref{eq-6-2-13}) around a given steady state $(\bm{x}, \bm{\xi}) = (\bm{x}^*, \bm{\xi}^*)$, with $\bm{\xi}$ eliminated and an input term added, is written as 
\begin{subequations}\label{eq-6-2-14} 
    \begin{align}
        & \Delta \dot{\bm{x}}  = \bm{A}(\bm{z}, \bm{a}, \bm{x}^*, \bm{\xi}^*) \Delta \bm{x} + \bm{B}(\bm{x}^{\Delta}) \bm{u} \\
        & \Delta \bm{y} = \bm{C}(\bm{z}, \bm{a}, \bm{x}^*, \bm{\xi}^*)  \Delta \bm{x}   
    \end{align}
\end{subequations}
where $\Delta\bm{x} = \bm{x} - \bm{x}^*$, $\Delta\bm{y} = \bm{y} - h(\bm{x}^*, \bm{\xi}^*, \bm{z}, \bm{a})$, $\bm{B}(\bm{x}^{\Delta}) \in \mathbb{R}^{n_{\rm x} \times n_{\rm x}}$ is the input matrix determined by $\bm{x}^{\Delta} \in \mathbb{R}^{n_{\rm x}}$, $\bm{u} \in \mathbb{R}^{n_{\rm x}}$ is the input vector; and
\begin{subequations}
    \begin{align}
        & \bm{A} = \frac{\partial f}{\partial \bm{x}} - \frac{\partial f}{\partial \bm{\xi}}  \left( \frac{\partial g}{\partial \bm{\xi}} \right)^{-1} \frac{\partial g}{\partial \bm{x}} \\
        & \bm{C} =  
        \begin{bmatrix}
            \frac{\partial h}{\partial \bm{x}}- \frac{\partial h}{\partial \bm{\xi}}  \left( \frac{\partial g}{\partial \bm{\xi}} \right)^{-1} \frac{\partial g}{\partial \bm{x}} &
            - \frac{\partial h}{\partial \bm{\xi}} \left( \frac{\partial g}{\partial \bm{\xi}} \right)^{-1} \frac{\partial g}{\partial \bm{u}} 
        \end{bmatrix}
    \end{align}
\end{subequations}
all at $(\bm{x}, \bm{\xi}) = (\bm{x}^*, \bm{\xi}^*)$. Hereinafter, $\Psi(\bm{z}, \bm{a})$ refers to the system given by (\ref{eq-6-2-13}), and the transfer function of (\ref{eq-6-2-14}), denote as $G(s, \bm{z}, \bm{a}, \bm{x}^*, \bm{\xi}^*, \bm{x}^{\Delta})$, refers to the system (\ref{eq-6-2-14}). Also, denote by $\bm{y}^{*, i}$ and $\bm{\xi}^{*, i}$ the points of $\bm{y}$ and $\bm{\xi}$ for $\bm{x}^{*, i}$, respectively; and $\bm{y}^{-, i}$, $\bm{y}^{+, i}$, and $\bm{\xi}^{+, i}$ are analogous.

\section{Bumpiness of Network Topology Transition}\label{sec-7-3}

This section introduces the metrics used to evaluate the bumpiness of the topology transition process, along with their tractable surrogates that can be integrated into mathematical models of the transition problem. 

\subsection{Bumpiness Metric}

\begin{figure}[t]
	\centering
	\includegraphics[width=0.55\linewidth]{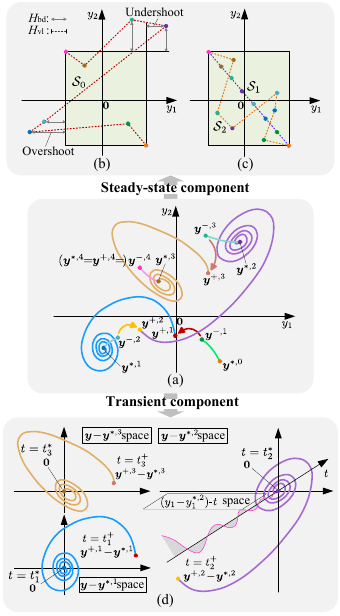}  
	\caption{Illustration of the bumpiness metric. }
	\label{fig-6-2-23} 
\end{figure}

To provide a more intuitive understanding, the bumpiness metric is illustrated using an example of topology transition with four transition episodes, i.e., $T=4$. The schematic of the trajectory of $\bm{y}$ during the transition is shown in Fig. \ref{fig-6-2-23}-(a). 
The trajectory of $\bm{y}$ is decomposed into two components, i.e., the steady-state component as shown in Fig. \ref{fig-6-2-23}-(b), and the transient-state component as shown in Fig. \ref{fig-6-2-23}-(d). Bumpiness of the trajectory can be analyzed following this decomposition. 

\textit{Steady-state component}:
Bumpiness of the steady-state component contains two aspects, called boundedness and volatility. Define the \textit{optimal region} as the intersection of 
all axis-aligned  surfaces that pass through $\bm{y}^{*, 0}$ and $\bm{y}^{*,T}$. 
The boundedness refers to how well the steady-state component of $\bm{y}$ is bounded by the optimal region. This property is associated with the concepts of overshoot and undershoot. In control theory, overshoot is the occurrence of an output exceeding its target and undershoot is the same phenomenon in the opposite direction. Here the target for undershoot is relaxed such that it occurs only when an output falls behind the starting point. 
Overshoot and undershoot (with a slight abuse of terminology, these two terms are used to refer to the similar phenomenons here) of the steady-state component of $\bm{y}$ of $\mathcal{S}_0$ is illustrated in Fig. \ref{fig-6-2-23}-(b), and the green rectangle is the optimal region to bound the steady-state component of $\bm{y}$. It can be seen that better boundedness indicates smaller overshoot and undershoot. Comparatively, the steady-state components of $\bm{y}$ of $\mathcal{S}_1$ and $\mathcal{S}_2$ in Fig. \ref{fig-6-2-23}-(c) are bounded by the green rectangle, with neither overshoot nor undershoot, and thus being more bumpless.

The volatility refers to the degree of the number of changes of the steady-state component of $\bm{y}$ during the topology transition. As shown in Fig. \ref{fig-6-2-23}-(c), the most bumpless topology transition is $\mathcal{S}_1$ where the path of the steady-state points of $\bm{y}$ is the shortest one between $\bm{y}^{*, 0}$ and $\bm{y}^{*, T}$, and all changes of $\bm{y}$ are necessary to realize the topology transition. For the topology transition where the length of this path is longer, such as $\mathcal{S}_0$ and $\mathcal{S}_1$, more and unnecessary changes of the steady-state component of $\bm{y}$ are involved and thus reducing bumpiness of the topology transition. Thus, the deviation between the length of the path of steady-state points of $\bm{y}$ and the shortest one can be an indicator of the volatility.

\textit{Transient-state component}:
Boundedness and volatility of the steady-state component capture the global bumpiness, while the transient-state component contains information about the local bumpiness after each line switching. Taking the projection of the transient-state component of $\bm{y}$ from $t=t_2^+$ to $t=t_2^*$, i.e., the pink trajectory in Fig. \ref{fig-6-2-23}-(d) for example, a smaller gray area indicates that the line switching at $t = t_2$ is more bumpless. Therefore, bumpiness of the transient-state component can be characterized by analogous area-type measures in the $(\bm{y}-\bm{y}^{*, i})$-$t$ space, which generalize the gray-area metric.

By integrating the steady-state and transient-state components, the bumpiness metric, denoted by $H$, is formulated as follows: 
\begin{equation}\label{eq-6-2-16}
    H((\bm{z}^i, \bm{a}^i | i \in \llbracket 1, T \rrbracket )) = H_{\rm bd} + H_{\rm vl} + H_{\rm ts} 
\end{equation}
with 
\begin{subequations}\label{eq-6-2-17}
    \begin{align}
        &
        \begin{aligned}
            & H_{\rm bd} =
            {\sum\limits_{i=1}^{T-1}}  \big[ \Vert \phi_i^*(\bm{y}\B \!-\! \bm{y}^{*,i}) \Vert_{2, \bm{w}_{\rm bd}\B }^2 \!+\! \Vert \phi_i^*(\bm{y}^{*,i} \!-\! \bm{y}\U)  \Vert_{2, \bm{w}_{\rm bd}\U }^2 \big]
            \\
            &
             + {\sum\limits_{i=1}^{T}}  \big[ \Vert \phi_i^-(\bm{y}\B - \bm{y}^{-,i}) \Vert_{2, \bm{w}_{\rm bd}\B }^2 + \Vert \phi_i^-(\bm{y}^{-,i} - \bm{y}\U) \Vert_{2, \bm{w}_{\rm bd}\U}^2 \big] 
        \end{aligned}
        \\
        & 
        \begin{aligned}
            H_{\rm vl} = &  {\sum\limits_{i=1}^{T}} \big[ \Vert \bm{y}^{*,i-1} - \bm{y}^{-,i} \Vert_{2, \bm{w}_{\rm vl}} + \Vert \bm{y}^{-,i} - \bm{y}^{*,i} \Vert_{2, \bm{w}_{\rm vl}} \big]
            \\
            & - \Vert \bm{y}^{*,0} - \bm{y}^{*,T} \Vert_{2, \bm{w}_{\rm vl}} 
        \end{aligned}
        \\
        & H_{\rm ts} = {\sum\limits_{i=1}^{T}} \int_{t_{i}^+}^{t_{i}^*} \Vert \bm{y} - \bm{y}^{*, i} \Vert_{2, \bm{w}_{\rm ts}}^2 \text{d} t
    \end{align}
\end{subequations} 
where ${H}_{\rm bd}$ and ${H}_{\rm vl}$ denote boundedness and volatility of the steady-state components of $\bm{y}$, respectively; ${H}_{\rm ts}$ represents bumpiness of the transient-state component; $\bm{w}_{\rm bd}\B$, $\bm{w}_{\rm bd}\U$, $\bm{w}_{\rm vl}$, and $\bm{w}_{\rm ts}$ are weight vectors for associated components of $H$; $\bm{y}\B = \min( \bm{y}^{*,0}, \bm{y}^{*, T} )$, $\bm{y}\U = \max( \bm{y}^{*,0}, \bm{y}^{*, T})$, where $\max(\cdot, \cdot)$ and $\min(\cdot, \cdot)$ are the entry-wise maximum and minimum of the two vectors, respectively; $\phi_i^*(\cdot) = \bm{0}$ if the $i$-th transition episode contains a fictitious execution of line switching, and $\phi_i^*(\cdot) = \max(\cdot, \bm{0})$ otherwise; $\phi_i^-(\cdot)$ is analogous to $\phi_i^*(\cdot)$ but depends on if the $i$-th transition episode contains a fictitious execution of \ac{acv} adjustment.

\subsection{Tractable Surrogates}

Two surrogates for $H_{\rm ts}$, termed $\mathcal{H}_2$-norm surrogate and jumping-state-based surrogate, can be employed to ensure  tractable modeling of the topology transition problem.

\textit{$\mathcal{H}_2$-norm surrogate}:
The transient-state component $\bm{y} - \bm{y}^{*, i}$ during the $i$-th post-switching phase is the free output response of system $\Psi(\bm{z}^{i}, \bm{a}^{i})$ to the initial state $(\bm{x}(t_i^+), \bm{\xi}(t_i^+)) = (\bm{x}^{+, i}, \bm{\xi}^{+, i})$ with a shift $- \bm{y}^{*, i}$. 
Denote this transient-state component by $\Delta \bm{y}_{\rm nl}^i (t)$ with $t \in [t_{i}^+, t_i^*]$. 
Let $\Delta \bm{y}_{\rm fr}^i(t)$ with $t \in [t_{i}^+, t_{i}^*]$ be the free output response of 
$G(s, \bm{z}^{i}, \bm{a}^{i}, \bm{x}^{*, i}, \bm{\xi}^{*, i}, \cdot)$ to the initial state $\Delta \bm{x}(t_{i}^+) = \bm{x}^{+, i} - \bm{x}^{*, i}$. 
If the jumping state is sufficiently close to the following post-switching steady state, $G$ can approximate $\Psi$ regarding the free output response. Formally, the following assumption will be used: 
\begin{itemize} 
    \item \textit{A.5}: $\forall i \in \llbracket 1, T \rrbracket $, $\Vert \bm{x}^{+, i} - \bm{x}^{*, i} \Vert_2$ is sufficiently small such that $\Delta \bm{y}_{\rm nl}^i (t) = \Delta \bm{y}_{\rm fr}^i(t) $ with $t \in [t_{i}^+, t_{i}^*]$. 
\end{itemize}

Let $\bm{B}(\bm{x}^{\Delta}) = (\bm{x}^{\Delta}) \D $, $\bm{u} = \bm{1}_{n_{\rm x}} \tilde{\delta}(t - t_{i}^+) $ 
with $\tilde{\delta}(t)$ being the unit impulse function, and $\Delta \bm{y}_{\rm im}^i(t)$ with $t \in [t_i^+, t_{i}^*]$ be the response of $G(s,\bm{z}^{i}, \bm{a}^{i}, \bm{x}^{*, i}, \bm{\xi}^{*, i}, \bm{x}^{+, i} - \bm{x}^{*, i} )$ with $\bm{x}(t_i^+) = \bm{0}$ to the inputs $\bm{u}$. 
Then the following proposition holds:
\begin{proposition}\label{prop-6-2-1}
    $\forall i \in \llbracket 1, T \rrbracket$, $\Delta \bm{y}_{\rm fr}^i (t) = \Delta \bm{y}_{\rm im}^i(t) $ with $t \in [t_i^+, t_{i}^*]$.
\end{proposition}

Under assumption A.5, Proposition \ref{prop-6-2-1} yields 
\begin{equation}\label{eq-6-2-21} 
        H_{\rm ts}   = {\sum\limits_{i=1}^{T}}  \Vert  \tilde{G}(s,\bm{z}^{i}, \bm{a}^{i}, \bm{x}^{*, i}, \bm{\xi}^{*, i}, \bm{x}^{+, i} - \bm{x}^{*, i} )  \Vert_{\mathcal{H}_2}^2 
\end{equation}
where $\tilde{G}$ is the transfer function formed by replacing $\bm{C}$ by $\tilde{\bm{C}} = (\bm{w}_{\rm ts}\D)^{\frac{1}{2}} \bm{C}$ in $G$. When $G$ is asymptotically stable, $\mathcal{H}_2$ norm of the transfer function can be computed using the observability Gramian, yielding the tractable $\mathcal{H}_2$-norm surrogate for $H_{\rm ts}$, expressed as:  
\begin{equation}\label{6-2-22} 
     H_{\rm ts}'' = {\sum\limits_{i=1}^{T}} \Tr( ( \bm{x}^{+, i} - \bm{x}^{*, i} )\D \bm{Q}_i ( \bm{x}^{+, i} - \bm{x}^{*, i} )\D)  
\end{equation}
with $\bm{Q}_i \in \mathbb{R}^{n_{\rm x} \times n_{\rm x}}$ satisfying the Lyapunov equation: 
\begin{equation}\label{eq-6-2-lya}
    \bm{A}(\bm{z}^{i}, \bm{a}^{i}, \bm{x}^{*, i}, \bm{\xi}^{*, i})\T \bm{Q}_i + \bm{Q}_i \bm{A} = - \bm{C}\T \bm{C}
\end{equation}

\textit{Jumping-state-based surrogate}:
This surrogate is motivated by observations of the transient-state behavior in post-switching trajectories. Consider the  4-bus system shown in Fig. \ref{fig-6-2-4}, where line 1-4 is to be opened. Five cases with different network topology before opening line 1-4 are examined. 

\begin{figure}[h]
	\centering
	\includegraphics[width=0.95\linewidth]{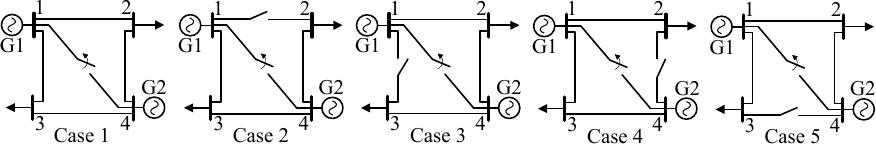} 
	\caption[The 4-bus system with different pre-switching network topology.]{Illustration of the 4-bus system with different pre-switching topology.}
	\label{fig-6-2-4} 
\end{figure}
\begin{figure}[h]
	\centering
	\includegraphics[width=0.75\linewidth]{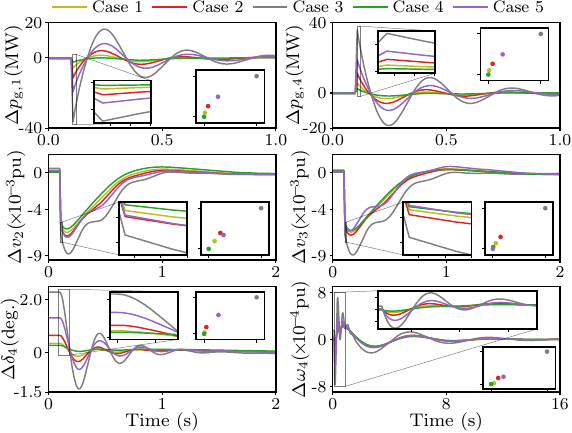} 
	\caption{Simulation results for the five switching cases of the 4-bus system.}
	\label{fig-6-2-r1} 
\end{figure}

Fig. \ref{fig-6-2-r1} gives the post-switching trajectories of transient-state components of different performance outputs for cases 1 to 5; the performance outputs here include active power outputs of G1 and G2, voltages at bus 2 and 3, rotor angle of G2 with reference to the rotor angle of G1, and rotor angle speed of G2, whose transient-state components are denoted by $\Delta p_{\rm g, 1}$, $\Delta p_{\rm g, 4}$, $\Delta v_2$, $\Delta v_4$, $\Delta \delta_4$, and $\Delta \omega_4$, respectively; the left-side small figure of each subfigure shows the zoom of the trajectories around the jumping state; in the right-side small figure of each subfigure, taking the subfigure for $\Delta p_{\rm g, i}$ as an example, the scatter plot shows the relationship between bumpiness $H_{\rm ts}$ for $\Delta p_{\rm g, 1}$ and the value of $|\Delta p_{\rm g, 1}|$ under the jumping state, and each marker is associated with the case with the same trajectory color.

It can be observed that bumpiness of each transient-state component is proportional to the absolute value of the transient-state component under the jumping state, except for the bumpiness of $\Delta v_2$ of cases 2 and 5. However, $H_{\rm ts}$ for $\Delta v_2$ of case 2 and that for case 5 are very close. Motivated by this observation, the minimization of bumpiness of the transient-state component can be approximated by minimizing its absolute value under the jumping state. 

Accordingly, the jumping-state-based surrogate for $H_{\rm ts}$ is given as: 
\begin{equation}
    H_{\rm ts}' = {\sum\limits_{i=1}^{T}} \Vert \bm{y}^{+, i} - \bm{y}^{*, i} \Vert_{2, \bm{w}_{\rm ts}\D \bm{w}'_{\rm ts} }^2 
\end{equation}
where $\bm{w}'_{\rm ts}$ is the vector of estimated scale factors between $\int_{t_i^+}^{t_{i}^*} (\bm{y}_j - \bm{y}^{*, i}_j)^2  \text{d} t$ and $(\bm{y}^{+, i}_j - \bm{y}^{*, i}_j)^2$.

\section{Modeling of Network Topology Transition}\label{sec-7-4}

\begin{figure}[h]
    \vspace{-8pt}
	\centering
	\includegraphics[width=0.95\linewidth]{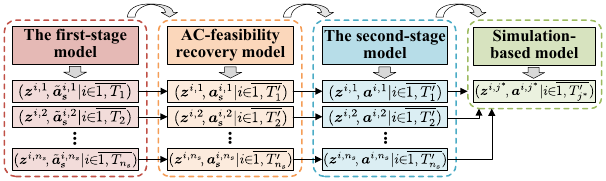} 
	\caption{High-level formulation of the network topology transition model.}
	\label{fig-6-2-7} 
\end{figure}

This section introduces the mathematical model of the network topology transition problem. First, the high-level formulation of the model is illustrated in Fig. \ref{fig-6-2-7}. It contains four submodels as follows:

\begin{itemize}
    \item \textit{The first-stage model}: This is a Mixed-Integer Second-Order Cone Programming (\ac{misocp}) model to find the optimal and suboptimal transition trajectories of topology and \ac{acv}s in $\bm{a}_{\rm s}$, using the linearized power flow model and the surrogate $H_{\rm ts}'$ for $H_{\rm ts}$. In Fig. \ref{fig-6-2-7}, $n_{\rm s}$ is the number of optimal and suboptimal solutions, $(\bm{z}^{i,j}, \tilde{\bm{a}}_{\rm s}^{i, j}| i \in \llbracket 1, T_j \rrbracket )$ with $j \in \llbracket 1, n_{\rm s} \rrbracket$ denotes the $j$-th solution of $(\bm{z}^{i}, \bm{a}_{\rm s}^{i}| i \in \llbracket 1, T \rrbracket )$ produced by the model and $T_j$ is the value of $T$ for this solution.

    \item \textit{AC-feasibility recovery model}: The solutions yielded by the first-stage topology transition model may be AC infeasible due to the linearization of power flow models \citep{4-1301}. The model here is a Nonlinear Programming (\ac{nlp}) model to find the AC feasible solution of $(\bm{z}^{i}, \bm{a}_{\rm s}^{i}| i \in \llbracket 1, T \rrbracket )$, denoted by $(\bm{z}^{i,j}, \bm{a}_{\rm s}^{i,j}| i \in \llbracket 1, T_j' \rrbracket )$ with $T_j'$ being the value of $T$ for it, which is closest to $(\bm{z}^{i,j}, \tilde{\bm{a}}_{\rm s}^{i, j}| i \in \llbracket 1, T_j \rrbracket )$, for $j \in \llbracket 1, n_{\rm s} \rrbracket $. 

    \item \textit{The second-stage model}: Given $(\bm{z}^{i,j}, \bm{a}_{\rm s}^{i,j}| i \in \llbracket 1, T_j' \rrbracket )$, this model finds the optimal transition trajectory of \ac{acv}s in $\bm{a}_{\rm t}$, denoted by $(\bm{a}_{\rm t}^{i, j} | i \in \llbracket 1, T_j' \rrbracket )$. Then the transition trajectory of topology and all \ac{acv} is $(\bm{z}^{i,j}, \bm{a}^{i, j} | i \in \llbracket 1, T_j' \rrbracket )$, with $\bm{a}^{i, j} = [(\bm{a}_{\rm s}^{i, j})\T, (\bm{a}_{\rm t}^{i, j})\T]\T$. This \ac{nlp}-based model uses the $\mathcal{H}_2$-norm surrogate $H_{\rm ts}''$ for $H_{\rm ts}$ and low-fidelity dynamic models to ensure computational efficiency.

    \item \textit{Simulation-based model}: Based on high-fidelity time-domain simulations, this model accurately evaluates the $n_{\rm s}$ transition trajectories of topology and \ac{acv}s, and finds the optimal one, denoted as $(\bm{z}^{i, j^*}, \bm{a}^{i, j^*} | i \in \llbracket 1, T_{j^*}' \rrbracket )$ with $j^* \in \llbracket 1, n_{\rm s} \rrbracket$. 
    This model captures previously overlooked complexities, including overall transition stability and the presence of multiple post-switching equilibria. 
\end{itemize}

\subsection{The First-Stage Model}

\textit{Pre-treatment}: The number of transition episodes $T$ is unknown beforehand. Without assuming any pattern for each transition episode, 
i.e., each transition episode can be a complete one or one with only \ac{acv} adjustment or line switching, 
$T$ is restricted only by (\ref{eq-6-2-a4}). Thus,
the transition trajectory should be modeled with the maximal possible value of $T$ determined by (\ref{eq-6-2-a4}), i.e., 
\begin{equation} 
    T\U = \max\{  \Vert \bm{z}^0  - \bm{z}^T \Vert_1 + \lfloor \frac{T_{\max}}{T_{\rm is}} \rfloor, \Vert \bm{z}^0  - \bm{z}^T \Vert_1 + \lfloor \frac{T_{\max}}{T_{\rm ad}} \rfloor - 1 \}
\end{equation}

\textit{Objective}: The objective is to minimize the bumpiness metric with the jumping-state-based surrogate, i.e.,
\begin{equation}
    \min\nolimits_{(\bm{z}^i, \bm{a}_{\rm s}^i|i \in \llbracket 1, T\U \rrbracket )}   ~ H' = H_{\rm bd} + H_{\rm vl} + H'_{\rm ts}
\end{equation}
with $T$ substituted by $T\U$.

\textit{Constraints}: 
Power flow constraints for the steady and jumping states are written as follows: 
\begin{subequations} 
    \begin{align}
        & \tilde{f}_{\rm p}( \bm{x}_{\rm p}^{-, i} | \bm{z}^{i-1}, \bm{a}_{\rm s}^i | \bm{y}_{\rm p}^{-,i}  ) \leq \bm{0} \\
        & \tilde{f}_{\rm p}( \bm{x}_{\rm p}^{+, i} | \bm{z}^i, \bm{a}_{\rm s}^i | \bm{y}_{\rm p}^{+,i} ) \leq \bm{0} \\
        & \tilde{f}_{\rm p}( \bm{x}_{\rm p}^{*, i} | \bm{z}^i, \bm{a}_{\rm s}^i | \bm{y}_{\rm p}^{*, i} ) \leq \bm{0}
    \end{align}
\end{subequations}
for all $i \in \llbracket 1, T\U \rrbracket$. Unless otherwise specified, all constraints in the first-stage model are defined for all $i \in \llbracket 1, T\U \rrbracket$.

The \ac{acv}s are required to be adjusted to their initial values and topology to $\bm{z}^{T}$ at the end of the transition process. Thus
\begin{equation}\label{eq-6-2-terminal}
    \bm{a}_{\rm s}^{T\U} = \hat{\bm{a}}^0_{\rm s}, \bm{z}^{T\U} = \bm{z}^{T}
\end{equation}
where $\hat{\bm{a}}^0_{\rm s}$ is a given value of $\bm{a}_{\rm s}$. 
If the system with $(\bm{z}, \bm{a}_{\rm s}) = (\bm{z}^{T}, \bm{a}_{\rm s}^0)$ is operationally feasible with the linearized power flow, $\hat{\bm{a}}^0_{\rm s} = \bm{a}^0_{\rm s}$; and otherwise, $\hat{\bm{a}}^0_{\rm s}$ is the closest one to $\bm{a}^0_{\rm s}$ which makes the system operationally feasible with the linearized power flow. 

For the post-adjustment and post-switching steady states, the \ac{acv} values match their setpoints, $\bm{p}_{\rm go}$ are remains at its initial value, the induction motor components of load powers stay constant. Therefore, 
\begin{subequations}\label{eq-6-2-maintain}
    \begin{align}
        & \bm{a}_{\rm s}^{-, i} = \bm{a}_{\rm s}^i, \bm{p}_{\rm go}^{-, i} = \bm{p}_{\rm go}^{*, 0}, \bm{\epsilon}_{\rm p}^{-,i} = \bm{\epsilon}_{\rm p}^{*,0}, \bm{\epsilon}_{\rm q}^{-,i} = \bm{\epsilon}_{\rm q}^{*,0} \label{eq-6-2-maintain:1} \\
        & \bm{a}_{\rm s}^{*, i} = \bm{a}_{\rm s}^i, \bm{p}_{\rm go}^{*, i} = \bm{p}_{\rm go}^{*, 0}, \bm{\epsilon}_{\rm p}^{*,i} = \bm{\epsilon}_{\rm p}^{*,0}, \bm{\epsilon}_{\rm q}^{*,i} = \bm{\epsilon}_{\rm q}^{*,0}  \label{eq-6-2-maintain:2}
    \end{align}
\end{subequations}

For the jumping states, certain variables, denoted by the vector $\bm{x}_{\rm c}$, retain their pre-switching steady-state values, yielding: 
\begin{equation}\label{eq-6-2-inst-same}
    \bm{x}_{\rm c}^{+, i} = \bm{x}_{\rm c}^{-, i} 
\end{equation}
with $\bm{x}_{\rm c} = [\bm{e}\T, \bm{\delta}\T, \bm{e}_{\rm s}\T, \bm{\delta}_{\rm s}\T, \bm{v}_{\rm m}\T, \bm{\theta}_{\rm m}\T, \bm{e}_{\rm m}\T, \bm{\delta}_{\rm m}\T, \bm{b}_{\rm svc}\T, \bm{v}_{\rm svg}\T, \bm{\theta}_{\rm svg}\T, \bm{b}_{\rm b, tcsc}\T]\T$.

Operational constraints for the steady states include: 
\begin{subequations}\label{eq-6-2-oprconst}
    \begin{align}
        & \bm{q}_{\rm gs}\B  \leq \bm{q}_{\rm gs}^{-, i} \leq   \bm{q}_{\rm gs}\U, [(\bm{p}_{\rm gv}^{-,i})\PW + (\bm{q}_{\rm gv}^{-,i})\PW]\HR \leq \bm{s}_{\rm gv}\U, \langle * \rangle  \label{eq-6-2-oprconst:1} \\
        & \bm{q}_{\rm gv}^{-, i} \leq \bm{p}_{\rm gv}^{-, i} \circ \tan(\arccos( \bm{\phi}_{\rm gv}\B )), \bm{v}\B \leq \bm{v}^{-, i} \leq \bm{v}\U, \langle * \rangle \label{eq-6-2-oprconst:2}\\
        & - \bm{\theta}\U - (\bm{1} - \bm{z}^{i}) M \leq \bm{E}_{\mathcal{G}}\T \bm{\theta}^{-, i} \leq \bm{\theta}\U + (\bm{1} - \bm{z}^{i}) M, \langle * \rangle \label{eq-6-2-oprconst:3}\\
        & [ (\bm{p}^{-, i}_{\rm fb}){}\PW \!+\! (\bm{q}^{-, i}_{\rm fb}){}\PW ]\HR \leq \bm{s}_{\rm b}\U, 
          [ (\bm{p}^{-, i}_{\rm tb})\PW \!+\! (\bm{q}^{-, i}_{\rm tb})\PW ]\HR \leq \bm{s}_{\rm b}\U, \langle * \rangle \label{eq-6-2-oprconst:4} \\
        & \bm{b}_{\rm b, tcsc}\B \!\leq \bm{b}_{\rm b, tcsc}^{-,i} \!\leq \bm{b}_{\rm b, tcsc}\U, 
        \bm{b}_{\rm svc}\B \!\leq\! \bm{b}_{\rm svc}^{-,i} \!\leq\! \bm{b}_{\rm svc}\U, 
        \bm{q}_{\rm c}\B \!\leq\! \bm{q}_{\rm c}^{-,i} \!\leq\! \bm{q}_{\rm c}\U\!, \langle * \rangle  \label{eq-6-2-oprconst:5}
    \end{align}
\end{subequations}
where (\ref{eq-6-2-oprconst:1}) are output power constraints of generators; (\ref{eq-6-2-oprconst:2}) are constraints for power factors of \ac{cig}s and bus voltage magnitudes, respectively; (\ref{eq-6-2-oprconst:3}) and (\ref{eq-6-2-oprconst:4}) bound branch phase angle differences and branch powers, respectively; (\ref{eq-6-2-oprconst:5}) are constraints for equivalent susceptances of lines with \ac{tcsc}s, output susceptances of \ac{svc}s, and reactive power outputs of \ac{dvc}s, respectively; and $\langle * \rangle$ denotes the same constraints as the left-side but for the post-switching steady states.

The abrupt changes in generator electric power during the switching phase can induce rotor shaft impacts, which should be kept within safe levels \citep{4-1285}. The same problem exists for induction motor loads. Thus the following constraints are considered: 
\begin{subequations}\label{eq-6-2-rsi}
    \begin{align}
        & (\bm{p}_{\rm gs}^{+, i} - \bm{p}_{\rm gs}^{-, i})\A \leq  \bm{\varepsilon}_{\rm gs} \circ \bm{p}_{\rm gs}^{\scriptscriptstyle\rm N} \\
        & \bm{p}_{{\rm{d_0}}} \circ ( \bm{\epsilon}_{{\rm p}}^{-, i} - \bm{\epsilon}_{{\rm p}}^{+, i} )\A \leq 
        \bm{\varepsilon}_{\rm im} \circ \bm{p}_{\rm im}^{\scriptscriptstyle\rm N}
    \end{align}
\end{subequations}
where $\bm{\varepsilon}_{\rm gs}$ and $\bm{\varepsilon}_{\rm gs}$ denote the proportions of rated power associated with critical rotor shaft impact levels for \ac{sg}s and induction motors, respectively; $\bm{p}_{\rm gs}^{\scriptscriptstyle\rm N}$ and $\bm{p}_{\rm im}^{\scriptscriptstyle\rm N}$ are their corresponding rated power vectors.

The setpoint of each \ac{acv} in $\bm{a}_{\rm s}$ is typically bounded, and the total adjustment amount of \ac{acv}s in each transition episode is also bounded to ensure fast and seamless \ac{acv} adjustment. Thus
\begin{subequations}\label{eq-6-2-adj-bound}
    \begin{align}
        & \bm{a}_{\rm s}\B \leq \bm{a}_{\rm s}^i \leq \bm{a}_{\rm s}\U \\
        & \Vert \bm{a}_{\rm s}^i - \bm{a}_{\rm s}^{i-1} \Vert_{2, \bm{w}_{\rm as} } \leq  \sigma_{\rm as} 
    \end{align}
\end{subequations}
where $\bm{w}_{\rm as}$ is the weight vector, and $\sigma_{\rm as}$ is the maximal allowed total adjustment amount of $\bm{a}_{\rm s}$ in one transition episode.

Network topology connectedness during the transition process is ensured by the following constraints: 
\begin{subequations}\label{eq-6-2-connectedness}  
     \!\!\!\!\!\begin{align} 
        & M (\bm{z}^i  - \bm{1} ) \leq \bm{E}_{\mathcal{G}}^T \bm{o}^i - \bm{\rho}^i  \leq M (\bm{1} - \bm{z}^i )  ~~ i \in \llbracket 1, T\U - 1 \rrbracket \\  
        & - M \bm{z}^i \!\leq\! \bm{\rho}^i \!\leq\! M \bm{z}^i   ~~ i \!\in\! \llbracket 1, T\U \!\!-\! 1 \rrbracket \\
        & \bm{E}_{\mathcal{G}} \bm{\rho}^i \!=\! \bm{c}_{\rm nc}  ~~ i \!\in\! \llbracket 1, T\U \!\!-\! 1 \rrbracket \\
        & \bm{o}^i \!\in\! \mathbb{R}^{n_{\rm n}}, \bm{\rho}^i \!\in\! \mathbb{R}^{n_{\rm e}} ~~ i \!\in\! \llbracket 1, T\U \!\!-\! 1 \rrbracket
    \end{align} 
\end{subequations}
with $\bm{c}_{\rm nc}$ being an $n_{\rm n}$-dimensional constant uniquely-balanced vector as defined in Definition \ref{cond-0-1-1}. 

Furthermore, at most one line can be switched per transition episode by assumption A.2, and certain branches, denoted as $\mathcal{E}_{\rm np}$, do not participate in the auxiliary control for topology transition. Thus, 
\begin{subequations}
    \begin{align}
        & \Vert \bm{z}^i - \bm{z}^{i-1} \Vert_1 \leq 1\\
        & \bm{E}_{\rm np}  [ ( \bm{z}^0 - \bm{z}^T )\A  -  {\sum\nolimits_{i=1}^{T\U}} (\bm{z}^i - \bm{z}^{i-1})\A ] = \bm{0} 
    \end{align}
\end{subequations}
where $\bm{E}_{\rm np}$ is the adjacent matrix between $\mathcal{E}_{\rm np}$ and $\mathcal{E}$.

To formulate constraint (\ref{eq-6-2-a4}), introduce $\bm{\zeta} = [\zeta_{i}]$ and $\tilde{\bm{\zeta}} = [\tilde{\zeta_{i}}]$, 
with $i \in \llbracket 1, {T\U} \rrbracket $, $\zeta_{i} \in \mathbb{B}$, and $\tilde{\zeta_{i}} \in \mathbb{R}$, to indicate the type of each transition episode. Let $\bm{\zeta}$ and $\tilde{\bm{\zeta}}$ satisfy the following constraints: 
\begin{subequations}
    \begin{align}
        & \Vert \bm{a}^i_{\rm s} -  \bm{a}^{i-1}_{\rm s} \Vert_{2, \bm{w}_{\rm as} }  \leq \sigma_{\rm as} \zeta_i \\
        & \Vert \bm{z}^{i} - \bm{z}^{i-1} \Vert_1 = \tilde{\zeta}_i
    \end{align}
\end{subequations}
with $\delta_{\rm pen} \bm{1}\T \bm{\zeta}$ being penalized in the objective function. Here $\delta_{\rm pen} >0 $ is a small penalty coefficient to prioritize minimizing $H'$ over $\bm{1}\T \bm{\zeta}$. Then, $\zeta_{i}= 1$ indicates \ac{acv} adjustment, and $\tilde{\zeta_{i}}= 1$ indicates line switching in the $i$-th transition episode. Accordingly, (\ref{eq-6-2-a4}) is equivalent to: 
\begin{equation}
    T_{\rm ad}  \bm{1}\T \bm{\zeta} + T_{\rm is} ( \bm{1}\T \tilde{\bm{\zeta}} - \Vert \bm{z}^0  - \bm{z}^T \Vert_1 )   \leq T_{\max}
\end{equation}

Lastly, the structure of the transition episode sequence is considered. 
Introduce variables $\bm{\zeta}' = [\zeta'_{i}]$ with $i \in \llbracket 1, T\U \rrbracket$ and $\zeta'_{i} \in \mathbb{R}$ satisfying
\begin{equation}
    \bm{\zeta}' \geq \bm{\zeta}, \bm{\zeta}' \geq \tilde{\bm{\zeta}}, \bm{\zeta}' \geq \bm{\zeta} + \tilde{\bm{\zeta}}, \bm{\zeta}' \geq  \bm{1} 
\end{equation}
such that $\zeta'_i = 0 \Leftrightarrow$ the $i$-th transition episode is with only fictitious executions, and $\zeta'_i = 1 \Leftrightarrow$ otherwise. 
Transition episodes with only fictitious actions should occur only at the end of the sequence. This indicates that $\forall i \in \llbracket 1, T\U - 2 \rrbracket$, $(\zeta'_i, \zeta'_{i+1}, \zeta'_{i+2}) \notin \{(0,1,1), (0,1,0), (0,0,1), (1,0,1) \}$, which is equivalent to: 
\begin{equation}
    \begin{bmatrix}
        1 & 1 & -1 \\
        1 & -1 & 1\\
        1 & -1 & -1\\
        -1 & 1 & -1
    \end{bmatrix}
    \begin{bmatrix}
        \zeta'_i \\
        \zeta'_{i+1} \\
        \zeta'_{i+2}
    \end{bmatrix}
    \geq 
    \begin{bmatrix}
        0 \\
        0 \\
        -1 \\
        -1
    \end{bmatrix}
    ~~\forall i \in \llbracket 1, T\U - 2 \rrbracket
\end{equation}
Moreover, adjacent episodes with \ac{acv} adjustment followed by line switching are disallowed, as they should be combined into a single transition episode. This indicates that $\forall i \in \llbracket 1, T\U - 1 \rrbracket, (\zeta_i, \tilde{\zeta}_i, \zeta_{i+1}, \tilde{\zeta}_{i+1}) \neq (1,0,0,1)$, which can be ensured by: 
\begin{equation}
    \zeta_i  + \tilde{\zeta}_{i+1} - (\tilde{\zeta}_i  + \zeta_{i+1}) \leq 1 ~~ \forall i \in \llbracket 1, T\U - 1 \rrbracket
\end{equation}

\textit{Post-treatment}: 
For the any $j$-th optimal solution given by the above optimization model, removing all its invalid transition episodes whose associated $\zeta'_i$ equals to 0 yields the associated solution of the first-stage model, i.e., $(\bm{z}^{i,j}, \tilde{\bm{a}}_{\rm s}^{i, j} | i \in \llbracket 1, T_j  \rrbracket)$.

\subsection{AC-Feasibility Recovery Model}

\textit{Pre-treatment}: 
As shown in Fig. \ref{fig-6-2-8}, the AC-feasibility of $(\bm{z}^{i,j}, \tilde{\bm{a}}_{\rm s}^{i, j} | i \in \llbracket 1, T_j  \rrbracket)$ is recovered by altering the values of \ac{acv}s in $\bm{a}_{\rm s}$ for each execution of \ac{acv} adjustment. The structure of the transition process and topology trajectory remains unchanged, except in two cases. The first case is that if the first episode involves only line switching, a potential \ac{acv} adjustment is added. The second case is related to the AC-feasibility recovery model, detailed in the later post-treatment. 

\begin{figure}[b!] 
	\centering
	\includegraphics[width=0.8\linewidth]{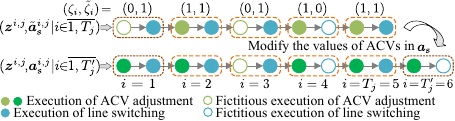} 
	\caption{Illustration of the pre-treatment and post-treatment.}  
	\label{fig-6-2-8}  
\end{figure}
 
Based on the values of $\bm{\zeta}$ and $\tilde{\bm{\zeta}}$ for the $j$-th solution, the set $\llbracket 1, T_j  \rrbracket$ is partitioned into three subsets: $\mathbb{T}_{\rm as}^j$ for transition episodes with both \ac{acv} adjustment and line switching, $\mathbb{T}_{\rm a}^j$ for those with only \ac{acv} adjustment, and $\mathbb{T}_{\rm s}^j$ for those with only line switching. For example, for the $j$-th solution in Fig. \ref{fig-6-2-8}, we have $\mathbb{T}_{\rm as}^j = \{1,2,5\}$, $\mathbb{T}_{\rm a}^j = \{4\}$, and $\mathbb{T}_{\rm s}^j = \{3\}$.

\textit{Formulation}:
The AC-feasibility recovery model is formulated as follows:  
\begin{subequations}\label{eq-6-2-recovery-model}
    \begin{align}
        \min_{ (\bm{a}_{\rm s}^i | i \in \llbracket 1, T_j  \rrbracket)} & ~~  \sum\limits_{i\in \mathbb{T}_{\rm as}^j \cup \mathbb{T}_{\rm a}^j}  w_{{\rm ac}, i}  \Vert  \bm{a}_{\rm s}^i - \tilde{\bm{a}}_{\rm s}^{i, j}  \Vert_{2} 
        +  w'_{\rm ac} \Vert  \bm{a}_{\rm s}^{T_j} - {\bm{a}}_{\rm s}^0  \Vert_{2}   \\
        \text{s.t.} ~~~&  {f}_{\rm p}( \bm{x}_{\rm p}^{-, i} | \bm{z}^{i-1, j}, \bm{a}_{\rm s}^i | \bm{y}_{\rm p}^{-,i}  ) \leq \bm{0}  ~~ \forall i \in \mathbb{T}_{\rm as}^j \cup \mathbb{T}_{\rm a}^j   \\
        & {f}_{\rm p}( \bm{x}_{\rm p}^{+, i} | \bm{z}^{i, j}, \bm{a}_{\rm s}^i | \bm{y}_{\rm p}^{+,i} ) \leq \bm{0} ~~ \forall i \in \mathbb{T}_{\rm as}^j \cup \mathbb{T}_{\rm s}^j  \\
        & {f}_{\rm p}( \bm{x}_{\rm p}^{*, i} | \bm{z}^{i, j}, \bm{a}_{\rm s}^i | \bm{y}_{\rm p}^{*, i} ) \leq \bm{0} ~~\forall  i \in \mathbb{T}_{\rm as}^j \cup \mathbb{T}_{\rm s}^j  \\
        & \{ \text{(\ref{eq-6-2-maintain:1})}, \text{(\ref{eq-6-2-oprconst}) excluding~} \langle * \rangle, \text{(\ref{eq-6-2-adj-bound})} \} ~~ \forall i \!\in\! \mathbb{T}_{\rm as}^j \cup \mathbb{T}_{\rm a}^j    \\
        & \{\text{(\ref{eq-6-2-maintain:2})}, \langle * \rangle ~\text{in (\ref{eq-6-2-oprconst})}\}  ~~ \forall i \in \mathbb{T}_{\rm as}^j \cup \mathbb{T}_{\rm s}^j  \\
        & \bm{a}_{\rm s}^i = \bm{a}_{\rm s}^{i-1}  ~~ \forall i \in \mathbb{T}_{\rm s}^j \\
        & \{\text{(\ref{eq-6-2-inst-same}), (\ref{eq-6-2-rsi})} \} ~~ \forall i \in \mathbb{T}_{\rm as}^j \\
        &\!\! \begin{bmatrix}
            \bm{x}_{\rm c}^{+, i} = \bm{x}_{\rm c}^{-, i-1} \\
             (\bm{p}_{\rm gs}^{+, i} - \bm{p}_{\rm gs}^{-, i-1})\A \leq  \bm{\varepsilon}_{\rm gs} \circ   \bm{p}_{\rm gs}^{\scriptscriptstyle\rm N} \\
             \bm{p}_{{\rm{d_0}}} \circ ( \bm{\epsilon}_{{\rm p}}^{-, i-1} - \bm{\epsilon}_{{\rm p}}^{+, i} )\A \leq \bm{\varepsilon}_{\rm im} \circ \bm{p}_{\rm im}^{\scriptscriptstyle\rm N}  
        \end{bmatrix} 
        \begin{aligned}
            & \forall (i, i-1) \in \\[-1mm] & \mathbb{T}_{\rm s}^j \times \mathbb{T}_{\rm a}^j
        \end{aligned}
          \\
        &\!\! \begin{bmatrix}
            \bm{x}_{\rm c}^{+, i} = \bm{x}_{\rm c}^{*, i-1} \\
            (\bm{p}_{\rm gs}^{+, i} - \bm{p}_{\rm gs}^{*, i-1})\A \leq  \bm{\varepsilon}_{\rm gs} \circ  \bm{p}_{\rm gs}^{\scriptscriptstyle\rm N} \\
             \bm{p}_{{\rm{d_0}}} \!\circ\! ( \bm{\epsilon}_{{\rm p}}^{*, i-1} - \bm{\epsilon}_{{\rm p}}^{+, i} )\A \leq \bm{\varepsilon}_{\rm im} \!\circ\! \bm{p}_{\rm im}^{\scriptscriptstyle\rm N}  
        \end{bmatrix}  
        \begin{aligned}
            & \forall (i, i-1) \in \\[-1mm] & \mathbb{T}_{\rm s}^j  \!\times\! (\mathbb{T}_{\rm as}^j \!\cup\! \mathbb{T}_{\rm s}^j)
        \end{aligned}
    \end{align}
\end{subequations}
where if $i=1$ and its associated \ac{acv} adjustment is added in the pre-treatment, $w_{{\rm ac}, i} \gg 1$, and otherwise, $w_{{\rm ac}, i} = 1$; $w'_{\rm ac} \gg 1$.

\textit{Post-treatment}: 
Let $(\bm{a}_{\rm s}^{i, j} | i \in \llbracket 1, T_j  \rrbracket)$ be the optimal solution of (\ref{eq-6-2-recovery-model}). When $\bm{a}_{\rm s}^{T_j, j} \neq \bm{a}_{\rm s}^0$, as illustrated in Fig. \ref{fig-6-2-8}, we add a transition episode with only execution of \ac{acv} adjustment following the $T_j$-th transition episode, which adjusts $\bm{a}_{\rm s}$ from $\bm{a}_{\rm s}^{T_j, j}$ to $\bm{a}_{\rm s}^{T_j + 1, j} = \bm{a}_{\rm s}^{0}$. Let $\bm{z}^{T_j + 1,j} = \bm{z}^{T_j,j}$, and $T_j' = T_j + 1$ if $\bm{a}_{\rm s}^{T_j, j} \neq \bm{a}_{\rm s}^0$ and $T_j' = T_j$ otherwise. Then the AC-feasible solution for $(\bm{z}^{i,j}, \tilde{\bm{a}}_{\rm s}^{i, j} | i \in \llbracket 1, T_j  \rrbracket)$ is $(\bm{z}^{i,j}, \bm{a}_{\rm s}^{i, j} | i \in \llbracket 1, T_j' \rrbracket)$.

\subsection{The Second-Stage Model}

For each $i \in \mathbb{T}_{\rm as}^j \cup \mathbb{T}_{\rm s}^j$, denote the values of $\bm{x}_{\rm p}^{+, i}$ and $\bm{y}_{\rm p}^{+, i}$ associated with the AC-feasible solution $(\bm{z}^{i,j}, \bm{a}_{\rm s}^{i, j} | i \in \llbracket 1, T_j' \rrbracket)$ as $\bm{x}_{\rm p}^{+, i, j}$ and $\bm{y}_{\rm p}^{+, i, j}$, respectively. They are by-products of solving (\ref{eq-6-2-recovery-model}). Then the jumping state associated the $i$-th transition episode and $j$-th solution, denoted as $\bm{x}^{+, i, j}$, can be obtained by solving
\begin{equation}
    \begin{aligned}
        & \bm{0} = f_{\rm g}( \bm{x}^{+, i, j}, \bm{\xi}^{+, i, j}, \bm{z}^{i, j}, \cdot )  \\
        & \Xi(\bm{x}^{+, i, j}, \bm{\xi}^{+, i, j}) =  \Xi_{\rm p} (\bm{x}_{\rm p}^{+, i, j}, \bm{y}_{\rm p}^{+, i, j}) 
    \end{aligned}
\end{equation}
where $\Xi(\bm{x}^{+}, \bm{\xi}^{+})$ and $\Xi_{\rm p}(\bm{x}_{\rm p}^{+}, \bm{y}_{\rm p}^{+})$ return the shared variables among $(\bm{x}^{+}, \bm{\xi}^{+})$ and $(\bm{x}_{\rm p}^{+}, \bm{y}_{\rm p}^{+})$; $f_{\rm g}(\cdot)$ is formed by removing components of $f(\cdot)$ and $g(\cdot)$ in (\ref{eq-6-2-13}) which only depend on $\Xi(\bm{x}^{+}, \bm{\xi}^{+})$. Analogously, the pre-switching steady state, denoted as $\bm{x}^{*, i, j}$ can be yielded.

Then, the second-stage model is formulated as follows:
\begin{subequations}\label{eq-6-2-second-stage-model}
    \begin{align}
        \min_{(\bm{a}_{\rm t}^i | i \in \llbracket 1, T_j' \rrbracket)} &  \sum\limits_{i \in \mathbb{T}_{\rm as}^j \cup \mathbb{T}_{\rm s}^j }  \Tr( ( \bm{x}^{+, i, j} - \bm{x}^{*, i, j} )\D \bm{Q}_i ( \bm{x}^{+, i, j} - \bm{x}^{*, i, j} )\D) 
        \label{eq-6-2-second-stage-model:1} \\[-1mm]
         \text{s.t.} ~&  
         \begin{bmatrix}
             \bm{A}(\bm{z}^{i, j}, [(\bm{a}_{\rm s}^{i, j}){}\T , (\bm{a}_{\rm t}^i){}\T]\T, \bm{x}^{*, i, j}, \bm{\xi}^{*, i, j})\T \bm{Q}_i  \\
            + \bm{Q}_i \bm{A} = - \bm{C}\T \bm{C} 
         \end{bmatrix} 
         \begin{aligned}
            &\forall i \in \\[-1mm] &\mathbb{T}_{\rm as}^j \cup \mathbb{T}_{\rm s}^j
         \end{aligned}
         \label{eq-6-2-second-stage-model:2} \\
         & \bm{a}_{\rm t}\B \leq \bm{a}_{\rm t}^i \leq \bm{a}_{\rm t}\U ~~ \forall i \in \mathbb{T}_{\rm as}^j \cup \mathbb{T}_{\rm s}^j \label{eq-6-2-second-stage-model:3} \\
         & \bm{a}_{\rm t}^{i} = \bm{a}_{\rm t}^{i-1}  ~ \forall i \in \mathbb{T}_{\rm a}^j \label{eq-6-2-second-stage-model:3-a} \\
         & \bm{Q}_{i} \succ 0, \bm{Q}_{i} \in \mathbb{R}^{n_{\rm x} \times n_{\rm x}} ~~ \forall i \in \mathbb{T}_{\rm as}^j \cup \mathbb{T}_{\rm s}^j  \label{eq-6-2-second-stage-model:4}
    \end{align}
\end{subequations}
where the objective function is $H_{\rm ts}''$ with zero terms removed; 
(\ref{eq-6-2-second-stage-model:2}) are the Lyapunov equations as (\ref{eq-6-2-lya}); 
(\ref{eq-6-2-second-stage-model:3}) are the bound constraints for $\bm{a}_{\rm t}$; (\ref{eq-6-2-second-stage-model:3-a}) is the equality constraint for $\bm{a}_{\rm t}$ in transition episodes without line switching; and in (\ref{eq-6-2-second-stage-model:4}), $\bm{Q}_i \succ 0$ ensures the asymptotically stability of $G$. Solving (\ref{eq-6-2-second-stage-model}) yields $(\bm{a}_{\rm t}^{i, j} | i \in \llbracket 1, T_j' \rrbracket)$.

\subsection{Simulation-Based Model}

The simulation-based model produces the final optimal solution, indexed by $j^*$, as follows: 
\begin{subequations}
    \begin{align}
        j^* = \argmin\nolimits_{j \in \llbracket 1, n_{\rm s} \rrbracket } ~&  H((\bm{a}^{i, j}, \bm{z}^{i, j} | i \in \llbracket 1, T_j' \rrbracket))  \\[-1mm]
        \text{s.t.}  ~& \text{The transition process is stable} \\
        & \{\text{(\ref{eq-6-2-oprconst}), (\ref{eq-6-2-rsi})} \} ~~ \forall i \in \llbracket 1, T_j' \rrbracket
        \end{align}
\end{subequations}
where the objective function and constraints are all evaluated by high-fidelity time-domain simulations.

\section{Numerical Example}\label{sec-7-5}

The methodology of network topology transition is demonstrated using the modified IEEE 9-bus system shown in Fig. \ref{fig-6-2-9}. The system includes 1 \ac{sg}, 2 \ac{cig}s both with \ac{ess}s, an \ac{svc} at bus 7, and a \ac{tcsc} on line 5-9. The initial and target topologies correspond to the left and right networks in Fig. \ref{fig-6-2-9}. For clarity, only two performance outputs are considered: voltage magnitudes at bus 2 ($v_5$) and modulation voltage angles of G2 relative to G1's rotor angle ($\tilde{\theta}_{\rm m, 2}$). 

\begin{figure}[h]
	\centering
	\includegraphics[width=1\linewidth]{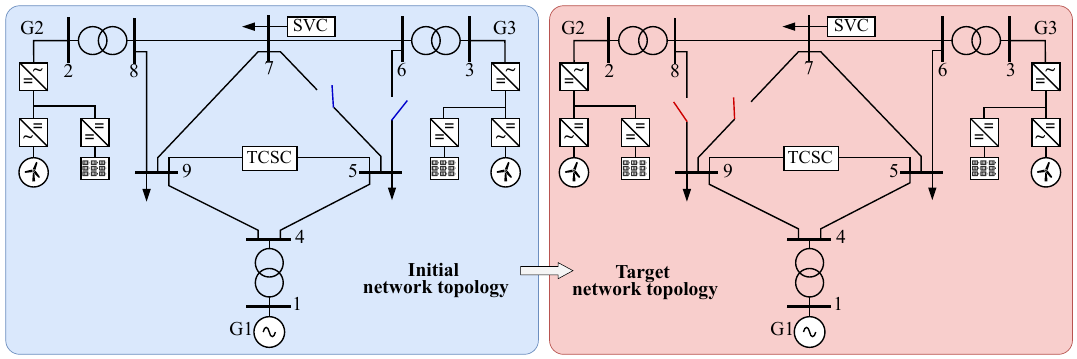} 
	\caption{Diagram of the modified IEEE 9-bus system.
    }  
	\label{fig-6-2-9}
\end{figure}

\begin{table}[h]
    \caption{The optimal transition scheme S1}\label{tab-6-2-2}
    \centering 
    \small{
    \begin{tabular*}{0.887\linewidth}{|cccccccc|}
    \hline\hline   
    TE & $v_1$  & $v_2$  & $v_3$  & $p_{\rm g, 2}$ &  $p_{\rm g, 3}$ & $v_7$  & $b_{\rm b, 5\text{-}9}$ \hspace{19pt}\\ \hline
    0  & 1.0400 & 1.0250 & 1.0250 &  163.00        &  85.00          & 1.0250 & -9.050  \\
    1  & 1.0782 & 1.0231 & 0.9575 &  154.52        &  93.48          & 1.0156 & -39.65  \\
    2  & 1.0790 & 1.0236 & 0.9614 &  158.97        &  89.03          & 1.0119 & -25.21  \\
    3  & 1.0644 & 1.0171 & 0.9636 &  153.28        &  94.72          & 1.0321 & -39.65  \\
    4  & 1.0434 & 1.0090 & 1.0175 &  149.30        &  98.70          & 1.0367 & -6.511  \\
    5  & 1.0400 & 1.0250 & 1.0250 &  163.00        &  85.00          & 1.0250 & -9.050     
    \end{tabular*}
    \\
    \hspace{2pt}
    \begin{tabular*}{0.887\linewidth}{|ccccccc|}
        \hline\hline 
    TE & $m_{\rm cg, 2}$ & $m_{\rm cg, 3}$ & $d_{\rm cg, 2}$ & $d_{\rm cg, 3}$ &  Close line & Open line \hspace{19pt}\\ \hline
    0  &  3.0000         &  4.0000         & 10.000          &  20.000         & ---  & --- \\
    1  &  0.4244         &  0.6056         & 7.2692          &  30.000         & ---  & 8-9  \\
    2  &  0.3758         &  0.7332         & 9.1833          &  24.638         & 5-6  & ---  \\
    3  &  0.3499         &  0.3741         & 6.2941          &  12.764         & 5-7  & ---  \\
    4  &  1.1746         &  0.8563         & 13.539          &  30.000         & ---  & 7-9  \\
    5  &  3.0000         &  4.0000         & 10.000          &  20.000         & ---  & ---  \\ 
    \hline\hline
    \end{tabular*}
    }
\end{table}

\begin{table}[h]
    \caption{Transition schemes S2, S3 and S4}\label{tab-6-2-3}
    \centering
    \setlength{\tabcolsep}{4pt}
    \small{
        \resizebox{0.98\linewidth}{!}{
    \begin{tabular}{|c|cc|cc|cc|}
    \hline\hline   
    \multirow{2}{*}{TE} & \multicolumn{2}{c|}{S2}  & \multicolumn{2}{c|}{S3} & \multicolumn{2}{c|}{S4} \\ 
    \cline{2-7}
      &  Close line & Open line & Close line & Open line & Close line & Open line   \\ \hline
    1  &  --- & 8-9 &  5-7 & ---  &  --- & 7-9   \\ 
    2  & 5-6  & --- &  --- & 7-9  &  5-6 & ---   \\
    3  & 5-7  & --- &  --- & 8-9  &  --- & 8-9   \\
    4  & ---  & 7-9 &  5-6 & ---  &  5-7 & ---   \\  \hline\hline
    \end{tabular} }
    }
\end{table}

\begin{table}[h]
    \caption{Bumpiness metric values for S1 to S4}\label{tab-6-2-4}
    \centering
    \small{
    \begin{tabular}{|ccccc|}
    \hline\hline   
    Scheme &   $H$  & $H_{\rm bd}$ & $H_{\rm vl}$ & $H_{\rm ts}$  \\ \hline
    S1  & 0.55 & 0.000895 &  0.148 & 0.397       \\ 
    S2  & 1.11 & 0.0299   &  0.174 & 0.908    \\
    S3  & 1.14 & 0.00959  &  0.160 & 0.973     \\ 
    S4  & 10.6 & 0.3977   &  0.653 & 9.51    \\  \hline\hline
    \end{tabular}
    }
\end{table}

\begin{figure}[h!]
	\centering
	\includegraphics[width=0.75\linewidth]{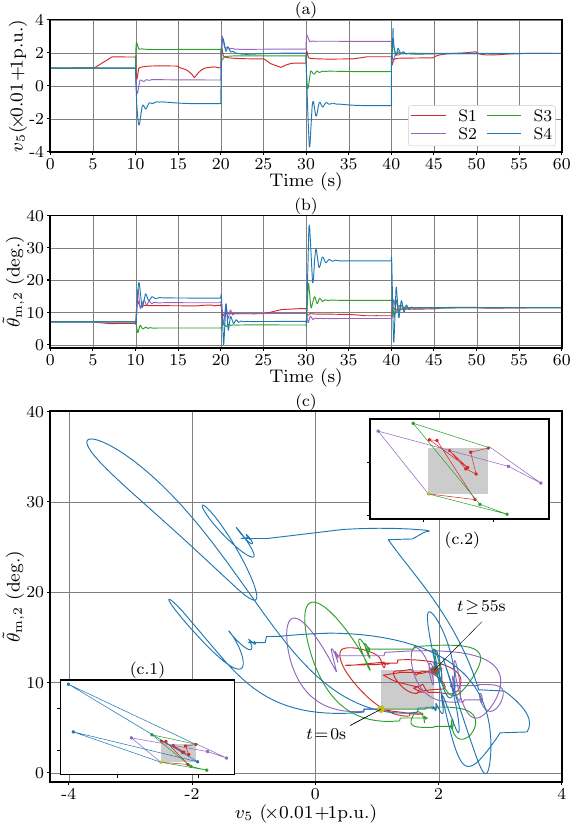} 
	\caption{Time-domain trajectories and steady-state components with S1 to S4.} 
	\label{fig-6-2-r2} 
\end{figure}

The transition scheme yielded by the introduced network topology transition methodology is shown in Table \ref{tab-6-2-2}. This transition scheme contains 5 transition episodes among which the first four are complete transition episodes containing both \ac{acv} adjustment and line switching and the 5th  transition episode contains only \ac{acv} adjustment. In the 1st  transition episode, all \ac{acv}s are together adjusted from the initial values to the target values smoothly. Taking \ac{acv} $v_1$ for example, it is adjusted from the initial value 1.04 to 1.0782. After the system reaches a SS, line 8-9 is switched off. When the system reaches the post-switching SS, the 2ed transition episode is performed similarly. In the last transition episode, adjusting all \ac{acv}s to their initial values completes the topology transition. For example, $v_1$ is adjusted from 1.0434 to its initial value 1.04 in the 5th transition episode.

To gain more insight into the transition process, the following four transition schemes are examined:
\begin{itemize}[itemsep=1pt, parsep=0pt, topsep=2pt]
    \item S1: the transition scheme yielded by the introduced methodology;
    \item S2: the feasible one without \ac{acv} adjustment and minimizing $H$;
    \item S3: the one given by the topology transition model which ignores transients during transition processes and \ac{acv} adjustment;
    \item S4: the feasible one without \ac{acv} adjustment and maximizing $H$.  
\end{itemize}
Table \ref{tab-6-2-3} shows the transition schemes S2 to S4, and Table \ref{tab-6-2-4} gives the values of $H$, $H_{\rm bd}$, $H_{\rm vl}$ and $H_{\rm ts}$ for them. The following observations and conclusions are drawn from Table \ref{tab-6-2-4}: 

\textit{(1)} The values of the bumpiness metric and its components of S4 are much larger than the other schemes. This indicates that ensuring only operational feasibility can lead to poor bumpiness performance, highlighting the need for the introduced topology transition mechanism. 

\textit{(2)} The values of $H_{\rm bd}$ and $H_{\rm vl}$ of S3 are smaller than that of S2, while the values of $H_{\rm ts}$ and $H$ are the opposite. Hence, neglecting transients during topology transitions can lead to suboptimal transient and overall bumpiness performance.

\textit{(3)} Transition scheme S1 outperforms all other schemes in terms of bumpiness, with over 90\% improvement in boundedness and 50\% in transient bumpiness. Thus, optimally adjusting \ac{acv}s during topology transitions ensures a significantly bumpless topology transition.

Fig. \ref{fig-6-2-r2} shows the trajectories of the performance outputs during the transition process with the 4 schemes, and the associated steady-state components. 
According to Fig. \ref{fig-6-2-r2} and Table \ref{tab-6-2-2}, with the transition scheme S1, in the 1st transition episode, \ac{acv}s are adjusted starting at $t=5$s and reach the target values by $t=7$s. The performance outputs $v_5$ and $\tilde{\theta}_{\rm m, 2}$ are changed during this period and reach steady states before $t=10$s. At $t=10$s, line 8-9 is switched off, inducing fast changes of $v_5$ and $\tilde{\theta}_{\rm m, 2}$. The next transition episode is performed similarly at $t = 15$s where $v_5$ and $\tilde{\theta}_{\rm m, 2}$ reach steady states. The transition process concludes at about $t = 55$s, when the system reaches steady state after \ac{acv}s are reset to their initial values in the 5th episode.

In addition, in Fig. \ref{fig-6-2-r2}-(c), the trajectory of $v_5$-$\tilde{\theta}_{\rm m, 2}$ and its steady-state component for S1 are better bounded by the optimal region compared to the other schemes. For the transient-state components of $v_5$ or $\tilde{\theta}_{\rm m, 2}$, Fig. \ref{fig-6-2-r2}-(a) and Fig. \ref{fig-6-2-r2}-(b) show clear post-switching oscillations for S2 to S4. In contrast, S1 results in minimal oscillations and the smallest overshoots in the transient-state components. 
\graphicspath{{chapter_7/Figs/}}
\chapter{Transient Topology Control}\label{chapter-7}

This chapter introduces the transient topology control of power grids. Unlike steady-state topology control, which changes the network topology under steady states, transient topology control alters the network topology during the transient period following disturbances, typically with the objective of stabilizing the system. 
In contrast to the bumpless topology transition, which aims to mitigate transients caused by line switching actions for realizing topology transitions, transient topology control actively utilizes line switching as a control action to enhance transient stability.

\section{Introduction}

The primary strategy of transient topology control for transmission networks is Intentional Controlled Islanding (\ac{ici}), an adaptive wide-area protection scheme employed as the last resort to prevent system collapse. The core idea of \ac{ici} is to partition the existing network into predefined islands, ensuring that the initial disturbance, which could otherwise cause system collapse, remains contained within a single island \citep{4-1881}.

The \ac{ici} problem consists of two key subproblems: determining when to island and where to island. The first stage identifies the ``point of no return'' beyond which only \ac{ici} can prevent grid collapse. The second stage focuses on partitioning the network to ensure stable islanded operation while meeting all restoration constraints. This section introduces the fundamental method for addressing the second subproblem \citep{4-1615}.

\subsection{Slow-Coherency Grouping Method}

One requirement of \ac{ici} is to ensure the generators in each island synchronize after the network is separated. The most common approach to meet this requirement is to use the slow coherency method for grouping all generators \citep{4-1882}. 

Consider an $n_{\rm g}$-machine transmission network, which is reduced by Kron reduction, preserving only the generator internal buses. 
Some assumptions are adopted, including the constant mechanical power input, the classical generator model, no damping and no phase shifters. 
As the coherent groups are independent of the detail level in modeling the generators, it is usually reasonable to neglect the damping in the study of generator coherency. 
Then the linearized state-space dynamic model of the network is given as follows:
\begin{subequations}\label{eq-b1-7-1}
    \begin{align}
        \ddot{\Delta \bm{\delta}} = \bm{A} \Delta \bm{\delta} \\
        \bm{A} = \bm{M}^{-1} \bm{K}
    \end{align}
\end{subequations}
where $\Delta \bm{\delta} = [\Delta \delta_1, ..., \Delta \delta_{n_{\rm g}}]$ is the state vector with $\Delta \delta_i$ denoting the rotor angle deviation from an equilibrium point; the Jacobian matrix $\bm{K}$ represents electrical coupling; $\bm{M} = \diag(2M_1/\omega_0, ..., 2M_{}/\omega_0)$ is the inertia matrix, with $M_i$ being the inertia of generator $i$ and $\omega_0$ being the base frequency.  

The slow-coherency grouping method using matrix $\bm{A}$ is summarized as Algorithm \ref{alg-b1-7-1}. The basic idea is that the eigenvectors of the matrix $\bm{A}$ provide insights into the mode shapes of electromechanical modes. Generators that exhibit similar entries in the eigenvector corresponding to a particular mode are considered coherent with respect to that mode. Thus, for a group of generators to be classified as slow coherent, their mode shapes must align closely in relation to the low-frequency interarea modes \citep{4-1905}.

\begin{algorithm}[t]
    \DontPrintSemicolon  
    \SetAlgoLined  
  
    \KwIn{Matrix $\bm{A}$}
    \KwOut{Generator grouping result $\mathcal{M}_1, \mathcal{M}_2, ..., \mathcal{M}_{k}$}
    
    Initialization: $k \gets 1$\;

    $\{\lambda_1, \lambda_2, ..., \lambda_{n_{\rm g}} \vert~ |\lambda_i| \leq |\lambda_{i+1}|  \}$ $\gets$ Eigenvalues of matrix $\bm{A}$ \;
    $\bm{v}_{1}, \bm{v}_2, ..., \bm{v}_{k}$ $\gets$ Eigenvectors for $\lambda_1, \lambda_2, ..., \lambda_{n_{\rm g}}$ \;

    \Repeat{$| \cup_{i=1}^k \mathcal{M}_i | = n_{\rm g}$}{

    $k \gets k + 1$\;

    \For{$\alpha = 1$ \KwTo $k$}{
        $\mathcal{M}_{\alpha} \gets \emptyset$ \;
    }

    $\bm{V}_s \gets [\bm{v}_{1}, \bm{v}_2, ..., \bm{v}_{k}]$ \;

    $\bm{V}_s \gets$ Gaussian elimination with complete pivoting to $\bm{V}_s$ \;

    \For{$\alpha = 1$ \KwTo $k$}{
        $\mathcal{M}_{\alpha} \gets$ Add the generator for the $\alpha$-th pivot to $\mathcal{M}_{\alpha}$ \;
    }

    $\begin{bmatrix} \bm{V}_{s1} \\[-1mm] \bm{V}_{s2} \end{bmatrix}$ $\gets$ $\begin{aligned}
        & \text{Permute the rows of} ~\bm{V}_s~ \text{such that} ~\bm{V}_{s1}~ \text{contains the} \\[-1.5mm]
        & \text{rows associated with the generators in} ~\cup_{i=1}^k \mathcal{M}_i
    \end{aligned}$

    $\bm{L} \!\in\! \mathbb{R}^{(n_{\rm g} - k) \times k}$ $\!\gets\!$ Solve $\bm{V}_{s1}\T \bm{L}\T \!\!=\! \bm{V}_{s2}\T$  \;

    \For{$i = 1$ \KwTo $n_{\rm g} - k$}{
    $\alpha \gets \argmax_{j = 1, ..., k} L_{ij} ~\text{s.t.}~ L_{ij} >0$ \;

    $\mathcal{M}_{\alpha}$ $\gets$ Add the generator for the $i$-th row of $\bm{L}$ to $\mathcal{M}_{\alpha}$

    } 
    }
    \caption{Slow-coherency grouping algorithm}\label{alg-b1-7-1}
\end{algorithm}

\subsection{Optimization Model for Network Separation}

Following the determination of the generator grouping scheme, the network separation scheme can be derived by solving optimization models that typically minimize the total load shedding of all islands while satisfying the associated constraints. Following the notation introduced in Section \ref{sec-3-3}, consider a generator grouping scheme that partitions the network into $n_{\rm is}$ islands, represented by the set $\mathcal{K} = \llbracket 1, n_{\rm is} \rrbracket$. Let $\mathcal{V}_k^{\rm gen}$ denote the set of generator buses that are required to be located in island $k$. Then by adopting the DC power flow model, a fundamental optimization model for network separation can be expressed in matrix form as follows: 
\begin{subequations}\label{eq-b1-7-2}
    \begin{align}
        \min ~& \bm{1}\T (\bm{p}_{\rm d}^* - \bm{p}_{\rm d}) \label{eq-b1-7-2:0}\\
        {\rm s.t.} ~& - (\bm{1} - \bm{z}) M \leq \bm{b}_{\rm b}\D \bm{E}_{\mathcal{G}}\T \bm{\theta} - \bm{p}_{\rm b} \leq  (\bm{1} - \bm{z}) M \label{eq-b1-7-2:1}\\
        ~& \bm{E}_{\rm g} \bm{p}_{\rm g} - \bm{E}_{\rm d} \bm{p}_{\rm d} = \bm{E}_{\mathcal{G}} \bm{p}_{\rm b} \label{eq-b1-7-2:2}\\
        ~&  \bm{\theta}_{\rm b}\B - (1 - \bm{z}) M \leq \bm{E}_{\mathcal{G}}\T \bm{\theta} \leq \bm{\theta}_{\rm b}\U + (1 - \bm{z}) M \label{eq-b1-7-2:3}\\
        ~& \bm{p}_{\rm b}\B \circ \bm{z} \leq \bm{p}_{\rm b} \leq  \bm{p}_{\rm b}\U  \circ \bm{z} \label{eq-b1-7-2:4}\\
        ~& \bm{p}_{\rm d}\B \leq \bm{p}_{\rm d} \leq \bm{p}_{\rm d}^* \label{eq-b1-7-2:5}\\
        ~& \bm{p}_{\rm g}\B \leq \bm{p}_{\rm g} \leq \bm{p}_{\rm g}\U \label{eq-b1-7-2:6}\\
        ~& \bm{E}_{\rm us}\T \bm{z} = \bm{1} \label{eq-b1-7-2:7}\\
        ~& \bm{E}_{\rm ref}\T \bm{\theta} = \bm{0}  \label{eq-b1-7-2:8}\\
        ~& \bm{W} \leq \bm{E}_{\mathcal{G}, \rm f}\T \bm{X} \label{eq-b1-7-2:9}\\
        ~& \bm{W} \leq - \bm{E}_{\mathcal{G}, \rm t}\T \bm{X} \label{eq-b1-7-2:10}\\
        ~& \bm{z} = \bm{W} \cdot \bm{1} \label{eq-b1-7-2:11}\\
        ~& \bm{X} \cdot \bm{1} = \bm{1} \label{eq-b1-7-2:12}\\ 
        ~& \frac{1}{n_{\rm n}} \cdot \bm{L}_{\rm lt} \bm{X} \leq \bm{Y} \leq \bm{L}_{\rm lt} \bm{X} \label{eq-b1-7-2:13}\\
        ~& \bm{Y} \geq \bm{X} \label{eq-b1-7-2:14}\\
        ~& \bm{U} = \bm{L}_{\rm lb} \bm{Y} \label{eq-b1-7-2:15}\\
        ~& \bm{U}\T \bm{1} = \bm{1} \label{eq-b1-7-2:16} \\
        ~& \bm{U} \circ (\bm{J}_{n_{\rm n}} \bm{X}) - \bm{X} = \bm{E}_{\mathcal{G}} \bm{F} \label{eq-b1-7-2:17}\\
        ~& \bm{0} \leq \frac{1}{n_{\rm n}} \cdot \bm{F} \leq \bm{z} \cdot \bm{1}\T \label{eq-b1-7-2:18}\\
        ~& \bm{E}_{\rm gp} \circ \bm{X} = \bm{E}_{\rm gp} \label{eq-b1-7-2:19}
    \end{align}
\end{subequations}
where the objective function (\ref{eq-b1-7-2:0}) represents the total amount of load shedding resulting from the islanding process, with $\bm{p}_{\rm d}^* \in \mathbb{R}^{n_{\rm d}}$ being the vector of active load powers before islanding; analogous to the constraints in the Optimal Transmission Switching (\ac{ots}) model (\ref{eq-0-1-0}) in Section \ref{sec-4-1a}, Eqs. (\ref{eq-b1-7-2:1})-(\ref{eq-b1-7-2:7}) form the physical and operational constraints of the system, including the DC power flow constraint (\ref{eq-b1-7-2:1}), the power balance constraint (\ref{eq-b1-7-2:2}), the phase angle difference constraint (\ref{eq-b1-7-2:3}), the branch power flow constraint (\ref{eq-b1-7-2:4}), the load power constraint (\ref{eq-b1-7-2:5}), the generation power constraint (\ref{eq-b1-7-2:6}), and the unswitchable branch constraint (\ref{eq-b1-7-2:7}); the constraint (\ref{eq-b1-7-2:8}) sets up an angle reference bus for each island to ensure a unique solution for the phase angles; 
the constraints (\ref{eq-b1-7-2:9})-(\ref{eq-b1-7-2:19}) are the matrix form of the constraints (\ref{eq-b1-3-5})-(\ref{eq-b1-3-7}) introduced in Section \ref{sec-3-3}, ensuring a multi-islanding topological structure according to the specified generator grouping scheme. 
In the optimization model, $\bm{E}_{\rm ref} \in \mathbb{R}^{n_{\rm n} \times n_{\rm is}}$ is the incidence matrix between all buses in $\mathcal{V}$ and the phase angle reference buses determined by arbitrarily selecting one generator bus from each island; 
$\bm{L}_{\rm lt} \in \mathbb{R}^{n_{\rm n} \times n_{\rm n}}$ is the lower triangular matrix with all elements in the lower triangle (including the diagonal) equal to 1, and all other elements equal to 0; 
$\bm{E}_{\rm lb} \in \mathbb{R}^{n_{\rm n} \times n_{\rm n}}$ is the lower bidiagonal matrix with unit diagonal and -1 subdiagonal; 
$\bm{E}_{\rm gp} \in \mathbb{R}^{n_{\rm n} \times n_{\rm is}}$ is the incidence matrix between all buses in $\mathcal{V}$ and the set of generator buses located in each islanded, i.e., $\mathcal{V}_{k}^{\rm gen}$ for all $k \in \mathcal{K}$. The matrices $\bm{X} \in \mathbb{R}^{n_{\rm n} \times n_{\rm is}}$, $\bm{W} \in \mathbb{R}^{n_{\rm e} \times n_{\rm is}}$, $\bm{Y} \in \mathbb{R}^{n_{\rm n} \times n_{\rm is}}$, $\bm{U} \in \mathbb{R}^{n_{\rm n} \times n_{\rm is}}$, $\bm{F} \in \mathbb{R}^{n_{\rm e} \times n_{\rm is}}$ correspond to the variables $X_{i, h}$,  $W_{ij, h}$,  $Y_{i, h}$,  $U_{i, h}$,  and $F_{ij, h}$ introduced in Eq. (\ref{eq-b1-3-5}) and (\ref{eq-b1-3-6}) of Section \ref{sec-3-3}.

\subsection{Numerical Example}

\begin{figure}[t]
	\centering
 	\includegraphics[scale=0.75]{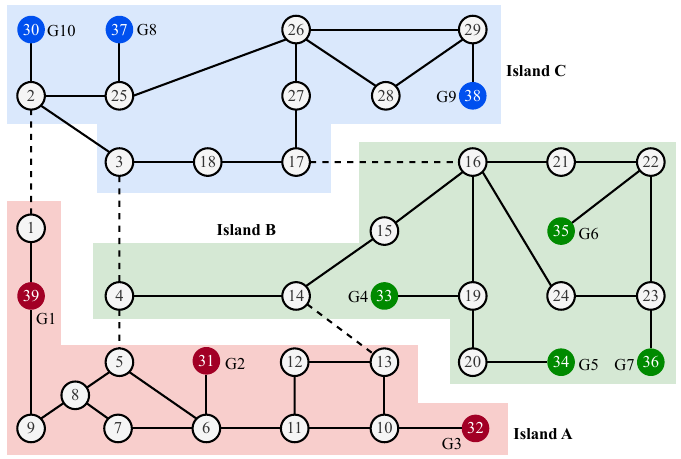}  
	\caption{Diagram of the IEEE 39-bus system and the network separation scheme.}
	\label{book1-7-1}
\end{figure}

The \ac{ici} of transmission networks is demonstrated using the IEEE 39-bus system, as depicted in Fig. \ref{book1-7-1}. A three-phase short-circuit fault is applied at bus 16 at $t=0.1$ s and cleared at $t=0.5$ s. Fig. \ref{book-1-7-r1} (the left two plots) presents the time-domain responses of rotor angle $\delta_i$ and frequency $f_i$ for each generator. It can be seen that the system becomes unstable following the fault clearance, as the generators lose synchronism. 

Fig. \ref{book1-7-1} also illustrates the generator grouping scheme obtained using the slow-coherency grouping method, and the network separation scheme obtained by solving the optimization model (\ref{eq-b1-7-2}). 
All generators are classified into three groups: group 1 consists of G1, G2, and G3; group 2 includes G4, G5, G6, and G7; and group 3 includes G8, G9, and G10. 
Given the generator grouping scheme, switching off line 1-2, line 3-4, line 4-5, line 13-14, line 16-17 separate the network into three islands, which minimizes the total amount of load shedding and satisfies all constraints of the optimization model (\ref{eq-b1-7-2}).

\begin{figure}[t]
	\centering
 	\includegraphics[scale=1]{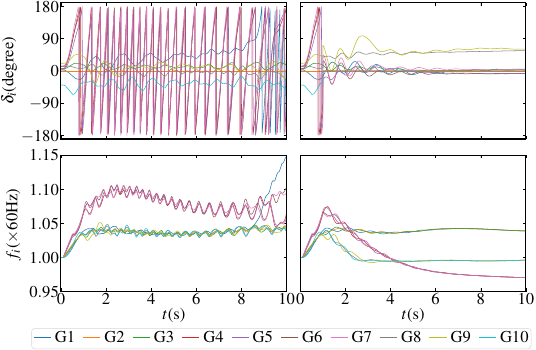}  
	\caption{Time-domain responses of generator rotor angle $\delta_i$ and frequency $f_i$ without \ac{ici} (the left two plots) and with \ac{ici} (the right two plots).}
	\label{book-1-7-r1}
\end{figure}

Fig. \ref{book-1-7-r1} shows the time-domain responses of $\delta_i$ and $f_i$ for each generator, with the network separation scheme implemented at $t=1$ s. The generators with each group achieve synchronization following the network separation. This demonstrates the effectiveness of \ac{ici} in stabilizing the system and thus preventing system collapse.

\section{State-of-the-Art Review}

This section reviews the existing literature on transient topology control, following the structure illustrated in Fig. \ref{book1-7-2}. 
According to the preservation of network connectedness within the control process, transient topology control can be broadly classified into two categories: \ac{ici}, which partitions the network into multiple islands, and connectedness-preserving transient topology control, in which the network remains connected throughout control. For \ac{ici}, six classes of methods for determining where to island are reviewed, and existing approaches for deciding when to island are further examined.
For connectedness-preserving transient topology control, this review discusses existing methods, including preliminary conceptual methods, tree-partitioning, phase-sequence exchange, and topology switching control. 

\begin{figure}[h]
	\centering
 	\includegraphics[scale=0.65]{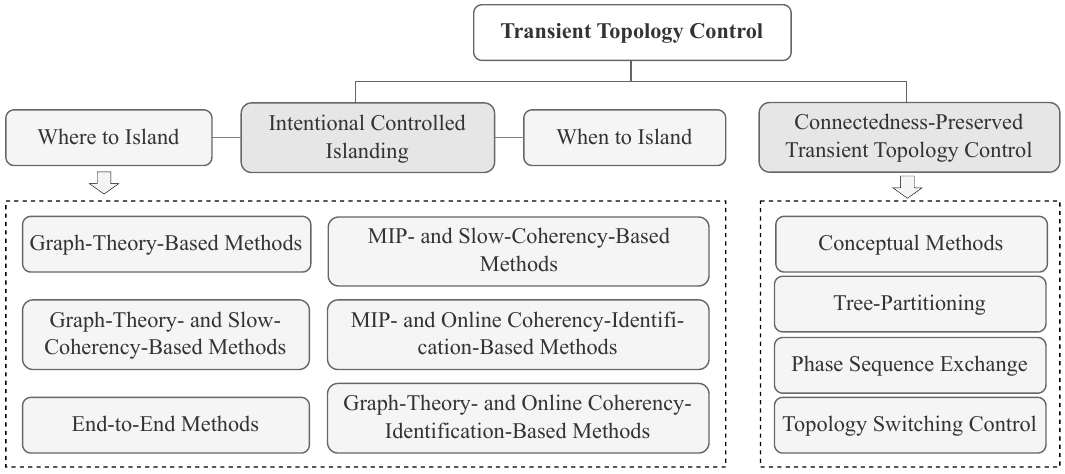}  
	\caption{Structure of the state-of-the-art review on transient topology control.}
	\label{book1-7-2}
\end{figure}

\subsection{Intentional Controlled Islanding}

Most existing studies on transient topology control primarily focus on \ac{ici} of transmission networks. The first key subproblem of \ac{ici} is where to island, for which various network separation methods have been developed. These methods can be divided into the following six categories: 
\begin{itemize}
\item \textit{Graph-theory-based methods}. 
\cite{4-1607} and \cite{4-1606} developed the Ordered Binary Decision Diagram (OBDD) method for network separation. This method represents the power network as a node-weighted undirected graph, where irrelevant nodes and edges are reduced to avoid the curse of dimensionality, and islanding strategies are verified by transient simulations \citep{4-1608}. The central feature of OBDD is its explicit Boolean representation of constraints, allowing the computationally expensive formation process to be handled offline, thereby enhancing online processing speed. Additionally, \cite{4-1895} proposed a multilevel method for optimizing weighted graph clustering, eliminating the need for eigenvector computation. This method was adopted by \cite{4-1625} for network separation with some characteristics of system islanding considered. 
\cite{4-1887} accounted for the real and reactive imbalance in islands by setting the absolute values of power flow as the edge weights of the graph, thus improving the voltage stability of each island. 
However, these graph theory-based methods may fail to maintain synchronization within each island due to the absence of generator coherency information.

\item \textit{Graph-theory- and slow-coherency-based methods}. 
\cite{4-1878} and \cite{4-1877} combined the graph theory and slow coherency technique \citep{4-1882} to address the network separation issue. In this method, coherent generator groups are first identified through singular perturbation analysis, followed by the selection of the cutset with the minimum power flow as the optimal cutset. 
Building upon this method, \cite{4-1879} introduced a two-step islanding strategy. In the first stage, normalized spectral clustering is used to identify coherent generators, while in the second stage, constrained spectral clustering is applied to determine the optimal lines for islanding. 
The slow coherency technique employed in these methods provides only static coherency information, which may vary in real time. Additionally, this technique necessitates detailed power system models.

\item \textit{Graph-theory- and online coherency-identification-based method}. 
\cite{4-1609} integrated the OBDD with online coherency identification using synchronized phasor measurements, yielding superior results compared to the graph-theory-based methods. Similarly, \cite{4-1896} developed a data-driven method incorporating Koopman mode analysis for network partitioning. This method share common features with the graph-theory- and slow-coherency-based methods, given the similarities between Koopman modes and Laplacian eigenvectors \citep{4-1897}. 
Compared to the slow coherency technique, the online coherency identification used in these methods relies solely on synchronized phasor measurements and can identify coherent generators in real time.

\item \textit{MIP- and slow-coherency-based methods}.
\cite{4-1610} developed the first comprehensive MILP-based method for network separation. This established MILP model incorporates connectedness constraints and utilizes DC power flow equations to enable efficient solutions. Subsequently, \cite{4-1889} utilized the second-order cone relaxation of the AC power flow equations to establish a mixed-integer second order cone programming model for network separation. This model ensures better voltage stability within each formed island given the incorporation of reactive power balance. 
The MILP-based method iwas improved by \cite{4-1891} through the use of linearized real/reactive power flow equations and the inclusion of frequency stability constraints in each formed island. 
Compared to determining the separation boundary using graph theory, utilizing MIP offers greater flexibility in incorporating system operational constraints.

\item \textit{MIP- and online coherency-identification-based methods}.
\cite{4-1876}, \cite{4-1888} and \cite{4-1886} developed network separation methods by combining online wide-area coherency identification and MIP optimization. In particular, \cite{4-1876} employed adjustable robust optimization programming to obtain network separation solutions that remain operationally feasible despite the uncertainty of renewable energy sources. \cite{4-1886} established a transient-stability-constrained MIP optimization model to ensure both steady-state and transient stability of the formed islands. 
For power systems with high renewable penetration, this category of network separation methods is more suitable due to the flexibility of the MIP framework in handling renewable uncertainty and the highly time-varying nature of system dynamics \citep{4-999-90}.

\item \textit{End-to-end methods}. 
\cite{4-1623} proposed a weak-submodularity-based method, which jointly captures the minimal generator non-coherency and minimal load-generation imbalance in one objective function.
\cite{4-1750} formulated the \ac{ici} process as a Markov decision process, where the optimal network separation is directly learned using the reinforcement learning approach. The learning-based method offers unique advantages in providing a real-time separation scheme and adapting to varying system conditions. 
Unlike other methods that address the network separation problem in two steps, these end-to-end methods jointly optimize both generator grouping and islanding boundaries, thereby achieving superior optimality. 
\end{itemize}

The second key subproblem in the \ac{ici} process is determining when to island. 
\cite{4-1898} first addressed this subproblem using a decision tree based tool for recognizing conditions existing in the system that warrant \ac{ici}. 
\cite{4-1609} utilized a synchrophasor-based separation risk index for each separation boundary to predict the time to perform \ac{ici}. 
\cite{4-1612} used a controlling unstable equilibrium point-based method to determine the time window within which controlled islanding ensures generator synchronism. 
These three methods however, all need offline estimations of the coherent groups of generators, which may cause inadequate coherent groups as the impact of actual disturbances is neglected. 
To this end, \cite{4-1617} proposed a unified methodology to determine the network separation scheme and time of islanding both in real time. The concept of area-based center of inertia-referred rotor angle index is adopted to determine the actual time for islanding. 
To eliminate the need for wide area measurement system communication and the complex computation, \cite{4-1890} proposed a simplified approach which uses local out-of-step protection relays of generators to identify the overall system instability.

\subsection{Connectedness-Preserved Transient Topology Control}

Different from \ac{ici} which separates the network into multiple islands, some earlier studies proposed transient topology control for transmission networks to maintain network connectedness while stabilizing the system following a disturbance. 
For example, \cite{4-1306} developed a line switching strategy for stabilization regarding transient instability, based on the variable structure systems theory.  \cite{4-88} proposed a switching-over control strategy of lines based on the constructed energy function, where certain lines are switched off in the normal state, and two topologies are switched alternately after fault clearance to stabilize the system. A similar idea of transient topology control was outlined conceptually and illustrated numerically by \cite{4-93}, which however, only needs to simultaneously switch multiple lines once after fault clearance. 
Switching control of tie lines or interconnections can also effectively stabilize weakly interconnected transmission networks \citep{4-1320}. 
More recently, \cite{4-94} developed an algorithm based on the discretization of system's differential algebraic equations and the optimization method of Lagrangian multipliers. The algorithm can identify the switched lines during transients and associated switching time to guarantee system stability. 
\cite{4-1060} proposed a structural emergency control paradigm for system stabilization by rendering post-fault dynamics, which is promisingly able to incorporate transmission line switching actions.  
Basically, these studies are conceptual and initial, with no rigorous theoretical foundation established. 

Recently, \cite{4-1893} proposed tree-partitioning as an emergency measure against cascading failures. Similar to \ac{ici}, this approach prevents fault spread; however, it maintains connectedness between clusters through some of the tie-lines. In practice, a two-step emergency control strategy can leverage both tree partitioning to maintain system integrity and limit load shedding initially, and \ac{ici} to separate clusters if stability is still threatened. \cite{4-1892} further developed an MILP formulation to optimally solve tree-partitioning problems.  
However, tree-partitioning still lacks comprehensive validation through time-domain simulations, and a theoretical foundation is needed to determine the conditions under which it is effective.

The Phase Sequence Exchange (PSE) technology proposed by \cite{4-1054} can also be viewed as a form of transient topology control when detailed three-phase system models are taken into account. 
This technology switches the three-phase sequence from (A, B, C) to (C, A, B), thereby reducing the power angle of the generator by 120$^\circ$ and stabilizing the system. Further refinements of this technology include the design of power electronic devices and the selection criteria for its components to facilitate implementation \citep{4-1900, 4-1901}, its extension to multi-machine systems \citep{4-1902}, the development of PSE strategies based on local transmission line measurements \citep{4-1903}, and its coordination with \ac{ici} \citep{4-1884}. Unlike other transient topology control methods that switch entire transmission lines, PSE technology modifies the phase status of specific lines, thereby preserving network connectedness while causing smaller disturbances in the system. 

For low-voltage power networks, \cite{4-1581, 4-1582, 4-1862} developed topology switching control methods for microgrids and multimicrogrids under emergency state. These methods stabilize the system by switching the system between different connected topologies according to the real-time system state. The topology switching control laws are designed based on rigorous switched system theory. According to \citep{4-1581}, In certain cases, stability of an islanded microgrid can be restored by repeatedly switching the same tie switch following disturbances. 
More recently, \cite{4-b1-48} proposed safe reinforcement learning-based emergency transient stability control approach for islanded microgrids, with transient topology reconfiguration included in the emergency control actions. 
The transient topology control methods for low-voltage networks have the potential to be extended to transmission networks, and the switched system theory can also help to establish rigorous theoretical foundation for early initial studies on transient line switching control of transmission networks in \citep{4-1306, 4-88, 4-93, 4-1320}. 
However, unlike microgrids, transmission networks are generally much larger, and the action time for line switching may be relatively long. Therefore, the scalability and dwell time of different topologies need to be considered.

\section{Recent Advance: Topology Control to Stabilize MMGs}

Topology control for Multi-microgrid (\ac{mmg}) stabilization is recently proposed by \cite{4-1582}. This approach involves switching the interconnection topology of an \ac{mmg} system to stabilize the system following a disturbance. Unlike \ac{ici} of transmission networks, this kind of transient topology control preserves network connectedness and may switch the topology multiple times.

\subsection{System Model and Problem Statement}\label{sec-5-4-2}

\subsubsection{Multimicrogrid with switchable topology}

Consider an \ac{mmg} system consisting of $m$ microgrids interconnected by $l$ lines. 
The topology of the interconnection network of the $m$ microgrids can change by switching on/off the $l$ lines. We refer to this interconnection network topology simply as topology, where the interconnection within each microgrid is not considered. 
Each microgrid is connected to the system at a Point of Common Coupling (\ac{pcc}) via a converter interface. Suppose, without loss of generality, that the $p$-$\omega$ and $q$-$V$ droop control law is deployed in the converter interface of each microgrid for real and reactive power sharing.

Let $\mathcal{M} \coloneqq \{1,2,\cdots, m\}$ denote the set of all $m$ microgrids, and $\mathcal{L}\coloneqq \{1,2,\cdots, l\}$ the set of all $l$ lines. With a slight abuse of notation, let $\mathcal{M}$ denote the set of all $m$ \ac{pcc}s. 
Without loss of generality, the \ac{pcc} of microgrid $1$ is taken as the angle reference. 
Let $\theta_i$ be the relative voltage phase angle w.r.t. the angle reference frame with frequency $\omega_1$, $\omega_i$ and $v_i$ be the angular frequency and voltage magnitude, respectively, and $p_i$ and $q_i$ be the real and reactive power injections, respectively, all at the \ac{pcc} of microgrid $i$.  Note that $\theta_1 = 0$. 
Then, the interface dynamics of microgrid $i$ are
\begin{subequations}\label{eq-5-4-1}
    \begin{align}
        & \dot{\theta}_i = \omega_{\rm b} (\omega_i - \omega_{\rm 1}) ~ \text{~if~} i \neq 1\\
        & \tau_{{\rm p}, i} \dot{\omega}_i = (\omega_i^{\rm s} - \omega_i) + k_{{\rm p}, i} (p_i^{\rm s} - p_i) \\
        & \tau_{{\rm q}, i} \dot{v}_i =  (v_i^{\rm s} - v_i) + k_{{\rm q}, i} (q_i^{\rm s} - q_i)
    \end{align}
\end{subequations}
where $\omega_{\rm b}$ is the base angular frequency; 
$\omega_i^{\rm s}$, $v_i^{\rm s}$, $p_i^{\rm s}$, and $q_i^{\rm s}$ are the setpoints of $\omega_i$, $v_i$, $p_i$, and $q_i$, respectively, which can be designed based on economic dispatch and remain constant during a dispatch interval; 
$\tau_{{\rm p}, i}$ and $\tau_{{\rm q}, i}$ are the time constants of low-pass filters for measuring $p_i$ and $q_i$, respectively; 
and $k_{{\rm p}, i}$ and $k_{{\rm q}, i}$ are the frequency and voltage droop gains, respectively. 
Using the power flow equations, the power injections at the \ac{pcc} of microgrid $i$ are formulated as 
\begin{subequations}\label{eq-5-4-2}
    \begin{align}
        & p_{i} = G_{ii} v_i^2 + \sum\limits_{j \in \mathcal{M}, j\neq i} Y_{ij} v_i v_j  \sin ( \theta_{ij} + \pi/2 - \varphi_{ij} ) 
        \\
        & q_{i} = -B_{ii} v_i^2 + \sum\limits_{j \in \mathcal{M}, j\neq i} Y_{ij} v_i v_j  \sin ( \theta_{ij}  - \varphi_{ij} )
    \end{align}
\end{subequations}
where $G_{ii} + jB_{ii}$ is the self-admittance at the \ac{pcc} of microgrid $i$; $Y_{ij}$ and $\varphi_{ij}$ are the modulus and phase angle, respectively, of the mutual admittance between the \ac{pcc}s of microgrids $i$ and $j$; and $\theta_{ij} = \theta_i - \theta_j$. Here $G_{ii}$, $B_{ii}$, $Y_{ij}$, and $\varphi_{ij}$ are all depend on the network topology. 

Next, denote by $\mathcal{Q} \coloneqq \{1, 2, \cdots, n\}$ the set of $n$ eligible network topologies, and refer to the \ac{mmg} system with topology $k \in \mathcal{Q}$ as subsystem $k$. 
The \ac{mmg} system with switchable network topology can be described in a compact form by the continuous-time switched nonlinear system
\begin{equation}\label{eq-5-4-3}
        \dot{\bm{x}}(t) =  f_{\sigma(t)}(\bm{x}(t)), \sigma(t) \in \mathcal{Q}  
\end{equation}
where $\bm{x} \coloneqq [[\theta_i]_{i \in \mathcal{M} \backslash \{1\} }, [\omega_i]_{i \in \mathcal{M}}, [v_i]_{i \in \mathcal{M}}]\T \in \mathbb{X} \subset \mathbb{R}^{3m -1}$ is the vector of {continuous-time} states of the \ac{mmg}, with $\mathbb{X}$ being its domain; 
$f_k: \mathbb{X} \!\to\! \mathbb{R}^{3m-1}$ is the vector field of {continuous-time} subsystem $k$, given by (\ref{eq-5-4-1}) with $p_i$ and $q_i$ substituted by the right side of (\ref{eq-5-4-2}), for all $i \!\in\! \mathcal{M}$; 
and $\sigma(t)$ denotes the topology switching control signal that determines the topology at time $t \!\in\! \mathbb{R}_+$. In addition, let $\mathcal{X}^* \coloneqq \{ \bm{x}^*_k | k \in \mathcal{Q} \}$ with ${x}^*_k$ being the equilibrium of interest for subsystem $k$.

\subsubsection{Principle and stabilization mechanism of topology control}

Assume that there is a control center for the \ac{mmg} system 
which monitors the state of the interface of each microgrid and the breakers of each interconnection lines, 
and can remotely switch the line breakers to change the network topology in real time. 
Considering the \ac{mmg} system after a disturbance, let $\bm{x}_0$ denote the value of $\bm{x}$ at $t = t_0$, i.e., the instant following disturbance clearance, and $\sigma(t_0) = \sigma_0$ denote the value of $\sigma(t)$ at $t = t_0$. The state $\bm{x}(t)$ is available for the control center for all time $t \in [0, +\infty)$. 

\begin{figure}[h]
	\centering
 	\includegraphics[scale=0.8]{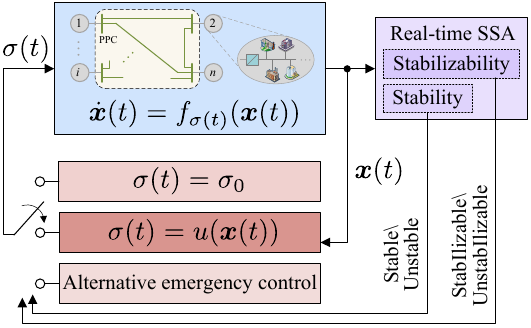}  
	\caption{Principle of using topology control for \ac{mmg} stabilization.}
	\label{fig-5-4-1}
\end{figure}

The principle of using topology control to stabilize the \ac{mmg} system after a disturbance is illustrated in Fig. \ref{fig-5-4-1}. 
After the disturbance clearance, the control center, through the real-time Stability and Stabilizability Assessment (\ac{ssa}), first determines whether the topology control needs to be activated. Given the initial state $\bm{x}_0$ and network topology $\sigma_0$, the stability assessment concludes whether the \ac{mmg} system is asymptotically stable w.r.t. the equilibrium $\bm{x}_{\sigma_0}^*$; the stabilizability assessment concludes whether the \ac{mmg} system can be asymptotically stabilized to an equilibrium $\bm{x}_k^*$ with $k \in \mathcal{Q}$ by the topology control. 

With the results of the real-time \ac{ssa}, 
if the post-disturbance \ac{mmg} system at $t= t_0$ is asymptotically stable, the topology control will not be activated, and the network topology will remain unchanged, i.e., $\sigma(t) = \sigma_0$ until the next disturbance or the following periodic network reconfiguration; 
if the system is not asymptotically stable but is stabilizable, the topology control determined through a state-dependent topology switching law $u: \mathbb{R}^{3m-1} \!\to\!  \mathcal{Q}$, will be activated to stabilize the system; 
and otherwise, alternative emergency control schemes such as load shedding and generator tripping will be activated. 

\begin{figure}[h]
	\centering
 	\includegraphics[scale=0.98]{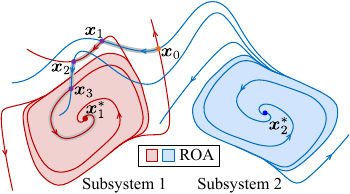}  
	\caption{Illustration of the stabilization mechanism of topology control.}
	\label{fig-5-4-2}
\end{figure}

The stabilization mechanism of the above topology control is to drive the system to the Region of Attraction (\ac{roa}) of one subsystem by switching between different subsystems. Intuitively, consider two subsystems with the phase portrait shown in Fig. \ref{fig-5-4-2}. Subsystem 1 is also the post-disturbance system, i.e., $\sigma(t_0) = 1$. The state $\bm{x}_0$ is outside the \ac{roa} of equilibrium $\bm{x}_1^*$ and thus the \ac{mmg} system is unstable without additional controls. By switching between subsystem 1 and 2 as $(\bm{x}_0, 2) \to (\bm{x}_1, 1) \to (\bm{x}_2, 2) \to (\bm{x}_3, 1)$ where $(\bm{x}_j, i)$ represents switching to subsystem $i$ when $\bm{x}(t) = \bm{x}_j$, the state trajectory, as marked by the gray curves, converges to $\bm{x}_1^*$.

\subsubsection{Topological switching stabilization problem}\label{sec-5-4-tssp}

When the topology control is activated, the closed-loop \ac{mmg} system can be written as 
\begin{subequations}\label{eq-5-4-5}
    \begin{align}
        & \dot{\bm{x}}(t) =  f_{\sigma(t)}(\bm{x}(t)) \\
        & \sigma(t) = u(\bm{x}(t)) \\
        & \bm{x}(t_0) = \bm{x}_0 
    \end{align}
\end{subequations}
where $u$, given the measured system state, decides which network topology in $\mathcal{Q}$ is activated for the \ac{mmg} system at any $t \geq t_0$. 
Accordingly, topology control for \ac{mmg} stabilization boils down to a topological switching stabilization problem. 
Before stating this problem, we first introduce the following assumptions and definitions. 

\begin{assumption}\label{asp-5-4-1}
Let $\{(\sigma_0, t_0), (\sigma_1, t_1), \cdots, (\sigma_j, t_j), \cdots |  j \in \mathbb{N}, \sigma_j \neq \sigma_{j+1} \} $, 
with $\mathbb{N}$ being the nonnegative integer set, 
be the switching sequence characterizing $\sigma(t)$, such that when $t \in [t_j, t_{j+1})$, $\sigma(t) = \sigma_j$. For any $j \in \mathbb{N}$, $\sigma_{j+1}$ can be arbitrarily selected from $\mathcal{Q} \backslash \{ \sigma_j \}$, namely that the topologies in $\mathcal{Q}$ can be directly switched to each other. 
\end{assumption}

\begin{assumption}\label{asp-5-4-2}
    The equilibria in $\mathcal{X}^*$ are isolated and distinct.   
\end{assumption}

\begin{definition}[\textit{Topological state-dependent switching stabilizability, stabilizable region}]
    System (\ref{eq-5-4-3}) is said to be topologically state-dependent switching stabilizable for an initial state $\bm{x}_0$ if there exists a state-dependent topology switching law $u$ under which the closed-loop system (\ref{eq-5-4-5}) is asymptotically stable w.r.t. $\mathcal{X}^*$ (i.e., asymptotically stable w.r.t. an arbitrary equilibrium in $\mathcal{X}^*$). The stabilizable region is the set of all initial state $\bm{x}_0$ for which the system (\ref{eq-5-4-3}) is topologically state-dependent switching stabilizable.
\end{definition}

\begin{definition}[\textit{Finitely piecewise constant topology switching control signal}]
    Topology switching control signal $\sigma(t)$ is said to be finitely piecewise constant if it exhibits a finite number of topology switching events in any finite time interval, and $\exists t'> t_0$, it exhibits no topology switching event for all $t \geq t'$.
\end{definition}

\begin{remark}
    From the switched system viewpoint, signal $\sigma(t)$ can theoretically exhibit infinitely fast switching at the final steady state of the switched system or along the state trajectory due to the existence of sliding modes, and also exhibit infinite number of topology switching events in a finite time interval due to Zeno behavior \citep{4-1400}. Moreover, topology switching events can occur continuously before the switched system reaches the final steady state at $t = \infty$. 
    Clearly, topology switching associated with these cases is practically infeasible. Thus, $\sigma(t)$ needs to be finitely piecewise constant. In the sequel, we assume that sliding modes and Zeno behavior do not exist in system (\ref{eq-5-4-5}) such that $\sigma(t)$ is always piecewise constant for $t \in [t_0, +\infty)$. 
\end{remark}

Moreover, let $\mathcal{X} \subseteq \mathbb{X}$ be a connected and compact domain that covers all possible values of $\bm{x}_0$, or all values of $\bm{x}_0$ being considered to potentially activate topology control. 
Set $\mathcal{X}$ in general is not completely contained in the stabilizable region, namely that for some $\bm{x}_0 \in \mathcal{X}$, system (\ref{eq-5-4-3}) is not topologically state-dependent switching stabilizable. However, a subset of $\mathcal{X}$, denoted as $\mathcal{O}$, may often be a subset of the stabilizable region. 
At this point, the topological switching stabilization problem to be addressed hereinafter can be formally stated as follows. 

\begin{problem}[\textit{Topological switching stabilization problem}]\label{problem-5-4-1}
    Given system (\ref{eq-5-4-3}), $\mathcal{X}$, and $\mathcal{X}^*$,  
    find a finitely piecewise constant state-dependent topology switching law $u$ and a region $\mathcal{O} \subseteq \mathcal{X}$ under which the closed-loop system (\ref{eq-5-4-5}) with $\bm{x}_0 \in \mathcal{O}$ is asymptotically stable w.r.t $\mathcal{X}^*$.
\end{problem}

\subsection{Theoretical Foundations}\label{sec-5-4-3}

This subsection establishes the main theoretical foundations for the topological switching stabilization problem. The stabilization theorem provides a framework for solving Problem \ref{problem-5-4-1}. 
Some properties of a specially-structured Feedforward Neural Network (\ac{fnn}) are used when implementing the solution framework.

\subsubsection{Multiple Lyapunov functions based stabilization theorem}

Multiple Lyapunov Functions (\ac{mlf}s) are a powerful tool for analyzing stability and designing switching laws of switched systems. The existing results, however, cannot be directly implemented to address the topological switching stabilization problem, as they are generally restricted to switched systems that have a common equilibrium of all subsystems. Therefore, the result in \citep{4-1486} to extended to solve Problem \ref{problem-5-4-1}.


\begin{definition}[\text{The class of Metzler matrices $\mathcal{P}$}]
    The class of Metzler matrices, denoted by $\mathcal{P}$, is the set of all matrices $\bm{{\Pi}} \in \mathbb{R}^{n \times n}$ with entry $\pi_{kl}$, $k, l \in \mathcal{Q}$ satisfying
    \begin{subequations}\label{eq-5-4-6} 
        \begin{align}
            & \pi_{kl} \geq 0 ~~ \forall k, l \in \mathcal{Q}, k \neq l \\
            & \sum\limits_{k \in \mathcal{Q}} \pi_{kl} = 0 ~~ \forall l \in \mathcal{Q} 
        \end{align}
    \end{subequations} 
\end{definition}

\begin{definition}[\textit{Lyapunov-like function}]
    Given system (\ref{eq-5-4-3}), $\mathcal{X}$, and $\mathcal{X}^*$, a differentiable function $V_k: \mathcal{X} \to \mathbb{R}$, is said to be a Lyapunov-like function for subsystem $k$ for any $k \in \mathcal{Q}$ if the following conditions hold:
    \begin{subequations}\label{eq-5-4-7}
        \begin{align}
            & V_k(\bm{x}) > 0 ~~ \forall \bm{x} \in \mathcal{X} \backslash \{ \bm{x}_k^* \}   \label{eq-5-4-7:1}\\
            & V_k( \bm{x}_k^* ) = 0 \label{eq-5-4-7:2}\\
            &  \mathcal{L}_{f_k} V_k(\bm{x}) + \sum\limits_{l \in \mathcal{Q}} \pi_{lk} V_l(\bm{x}) < 0  ~~ \forall \bm{x} \in \mathcal{X} \backslash \{ \bm{x}_k^* \} \label{eq-5-4-7:3} 
        \end{align}
    \end{subequations}
    where $\pi_{lk}$ is the element of matrix $\bm{{\Pi}} \in \mathcal{P}$; and 
    $\mathcal{L}_{f_k} V_k(\bm{x}) \coloneqq \left(\frac{ \partial V_k }{ \partial \bm{x} }(\bm{x}) \right)\T \! f_{k}(\bm{x})$ is the Lie derivative of $V_k(\bm{x})$ along $f_k(\bm{x})$, by which the nonlinear system dynamics are incorporated. The conditions (\ref{eq-5-4-7}) are called Lyapunov-Metzler inequalities in stability theory \citep{4-1486}. 
\end{definition}

\begin{lemma}\label{lemma-5-4-1}
    Given system (\ref{eq-5-4-3}), $\mathcal{X}$, $\mathcal{X}^*$, and 
    the Lyapunov-like functions $V_k$, $\forall k \in \mathcal{Q}$, 
    there exists $\delta_k > 0$ s.t. 
    $V_k(\bm{x}) \!\!<\!\!  V_l(\bm{x})$, 
    $\forall \bm{x} \!\in\!  \mathcal{B}_{\delta_k} \!( \bm{x}_k^* )$ and $\forall l \in \mathcal{Q} \backslash \{k\}$, 
    where $\mathcal{B}_{\delta_k}\!( \bm{x}_k^* ) \!\!\coloneqq\!\! \{\bm{x} \!\in\! \mathcal{X} \bm{|} \Vert \bm{x} \!-\! \bm{x}_k^* \Vert_2 \leq \delta_k \}$ is the closed ball of radial $\delta_k$ centered at $\bm{x}_k^*$. 
\end{lemma}

\begin{proof}
    Since the equilibria $\bm{x}^*_k$ and $\bm{x}^*_l$ are isolated and distinct, there exist $\delta_k >0$ and $\delta_l > 0$ such that $\mathcal{B}_{\delta_k}( \bm{x}_k^* ) \subset  \mathcal{X} \backslash \mathcal{B}_{\delta_l}( \bm{x}_l^* )$, $\forall k, l \in \mathcal{Q}$ and $k \neq l$. Considering conditions (\ref{eq-5-4-7:1}) and (\ref{eq-5-4-7:2}) of the Lyapunov-like functions, $\inf_{\bm{x} \in \mathcal{X} \backslash \mathcal{B}_{\delta_l}( \bm{x}_l^* )} V_l(\bm{x})$ with $l \in \mathcal{Q} \backslash \{k\}$, denoted as $b_l$, is larger than 0. Since the Lyapunov-like functions are also continuous, $\max_{\bm{x} \in \mathcal{B}_{\delta_k}( \bm{x}_k^* )} V_k(\bm{x})$ can be arbitrarily close to 0 and thus by reducing the value of $\delta_k$ and thus there exists $\delta_k >0$ such that $0 < \max_{\bm{x} \in \mathcal{B}_{\delta_k}( \bm{x}_k^* )} V_k(\bm{x}) < b_l$ . Therefore, in this case, $\forall \bm{x} \in  \mathcal{B}_{\delta_k} \!( \bm{x}_k^* ) \subset \mathcal{X} \backslash \mathcal{B}_{\delta_l}( \bm{x}_l^* )$, $V_k(\bm{x}) \!<\!  V_l(\bm{x})$. 
\end{proof}


\begin{theorem}\label{theorem-5-4-1}
    Given system (\ref{eq-5-4-3}), $\mathcal{X}$, and $\mathcal{X}^*$, if the Lyapunov-like functions $V_k(\bm{x})$ exist for all subsystems, then system (\ref{eq-5-4-5}) with $\bm{x}_0 \in \mathcal{O}$ is asymptotically stable w.r.t. $\mathcal{X}^*$ under the state-dependent topology switching law: 
    \begin{equation}\label{eq-5-4-8}
        \sigma(t) = u(\bm{x}(t)) = \argmin\limits_{k \in \mathcal{Q}} \{ V_k(\bm{x}(t)) \}
    \end{equation}
    and the region $\mathcal{O}$ is given as
    \begin{equation}\label{eq-5-4-9}
        \!\!\!\! \mathcal{O} \!=\! \{ \bm{x} \!\in\! \mathcal{X} | \tilde{V}(\bm{x}) \leq \min\limits_{\tilde{\bm{x}} \in \partial \mathcal{X}} \min\limits_{ k \in \mathcal{Q}} V_k(\tilde{\bm{x}})  \}
    \end{equation}
    where $\tilde{V}(\bm{x}) \coloneqq \min\nolimits_{k \in \mathcal{Q}} V_k(\bm{x})$, 
    and $\partial \mathcal{X}$ denotes the boundary of region $\mathcal{X}$. 
    Moreover, $\sigma(t) = u(\bm{x}(t))$ is finitely piecewise constant. 
\end{theorem}

\begin{proof}
     
     The Lyapunov function of system (\ref{eq-5-4-5}) under switching law (\ref{eq-5-4-8}) is $\min_{k \in \mathcal{Q}} V_k(\bm{x})$, i.e., $\tilde{V}(\bm{x})$. From (\ref{eq-5-4-7:1}), we have $\forall \bm{x} \in \mathcal{X} \backslash \mathcal{X}^*$, $\tilde{V}(\bm{x}) > 0$. 
     Similar to the proof of \cite[Theorem~3]{4-1486}, we have the Dini derivative of $\tilde{V}(\bm{x})$, $D^+ \tilde{V}(\bm{x}(t)) \leq \mathcal{L}_{f_k} V_k(\bm{x})$, 
     where 
     $k \in \mathcal{I}(\bm{x}(t)) \coloneqq \{ l | l \in \mathcal{Q}, \tilde{V}(\bm{x}) = V_{l}(\bm{x}) \}$. 
     This inequality together with (\ref{eq-5-4-7:3}) gives
     \begin{equation}
        D^+ \tilde{V}(\bm{x}(t)) < - \sum_{l \in \mathcal{Q}} \pi_{lk} V_l(\bm{x})  ~~ \forall \bm{x} \in \mathcal{X} \backslash \mathcal{X}^*.
     \end{equation}
     Moreover, since $V_l(\bm{x}) \geq V_k(\bm{x})$ for all $l \in \mathcal{Q} \backslash \{k\}$ with $k \in \mathcal{I}(\bm{x}(t))$, 
     we have
     \begin{equation}
        D^+ \tilde{V}(\bm{x}(t)) < - V_k(\bm{x}) \sum_{l \in \mathcal{Q}} \pi_{lk} = 0 ~~ \forall \bm{x} \in \mathcal{X} \backslash \mathcal{X}^*
     \end{equation} 

    Suppose that the trajectory $\bm{x}(t)$ of system (\ref{eq-5-4-5}) under switching law (\ref{eq-5-4-8}) with $\bm{x}(t_0) \in \mathcal{O}$ does not converge to any equilibrium in $\mathcal{X}^*$. Note that region $\mathcal{O}$ given by (\ref{eq-5-4-9}) can be rewritten as $\mathcal{O} \!=\! \{ \bm{z} \in \mathcal{X} | \tilde{V}(\bm{z}) \!\leq\! \min_{\tilde{\bm{z}} \in \partial \mathcal{X}} \tilde{V}(\tilde{\bm{z}}) \}$, which is a compact positively invariant set of the switched system (\ref{eq-5-4-5}). Since $\tilde{V}( \bm{x}(t) )$ is decreasing and positive, it converges to a positive value, denoted as $\beta > 0$, as $t \to + \infty$. Thus, $\forall t \in [t_0, +\infty]$, we have $\beta \leq \tilde{V}(\bm{x}(t)) \leq \tilde{V}(\bm{x}(t_0))$.
    Let 
    $\mathcal{C} \coloneqq \{ \bm{z} \in \mathcal{X} | \beta \leq \tilde{V}(\bm{z}) \leq \tilde{V}(\bm{x}(t_0)) \}$, 
    which is a compact set. 
    We have 
    $\sup_{\bm{z} \in \mathcal{C}} D^+ \tilde{V}(\bm{z}) = -\lambda < 0$, 
    and thus 
    $\forall t \in [t_0, + \infty]$, $D^+ \tilde{V}(\bm{x}(t)) \leq - \lambda$. 
    Also given that $\tilde{V}(\bm{x}(t))$ is continuous, 
    we have 
    \begin{equation}
        \tilde{V}(\bm{x}({t_0 + \Delta t })) = \tilde{V}(\bm{x}(t_0)) + \int_{t_0}^{t_0 + \Delta t} D^+ \tilde{V}(\bm{x}(t)) d t \leq \tilde{V}(\bm{x}(t_0)) - \lambda \Delta t
    \end{equation}
    This inequality, when $\Delta t > \frac{\tilde{V}(t_0)}{\lambda}$, gives $\tilde{V}(\bm{x}({t_0 + \Delta t })) < 0$, contradicting $\tilde{V}(\bm{x}) \geq 0$ for any $\bm{x} \in \mathcal{X}$. 
    Since $\bm{x}^*_k, \forall k \in \mathcal{Q}$, is an isolated equilibrium, 
    the trajectory $\bm{x}(t)$ converges to some equilibrium in $\mathcal{X}^*$, i.e., the asymptotic stability w.r.t. $\mathcal{X}^*$ holds.

    Finally, we shall show that $u(\bm{x}(t))$ is finitely piecewise constant. Without loss of generality, assume $\bm{x}(t)$, as $t \to + \infty$, converges to $\bm{x}^*_{\kappa}$ with $\kappa \in \mathcal{Q}$. According to Lemma \ref{lemma-5-4-1}, we have 
    \begin{equation}\label{eq-5-4-7:4}
        \exists \delta_{\kappa} > 0: \big( V_{\kappa}(\bm{x}) \!<\!\! \min\limits_{l \in \mathcal{Q} \backslash \{{\kappa}\} } \!\! V_l(\bm{x}) ~ \forall \bm{x} \in  \mathcal{B}_{\delta_{\kappa}} \!( \bm{x}_{\kappa}^* ) \big) 
    \end{equation} 
    Let  
    $\mathcal{O}_{\kappa} \coloneqq \{ \bm{x} \in \mathcal{B}_{\delta_{\kappa}}(\bm{x}_{\kappa}^*) | V_{\kappa}(\bm{x}) \leq \min_{\tilde{\bm{x}} \in \partial \mathcal{B}_{\delta_{\kappa}}(\bm{x}_{\kappa}^*)} V_{\kappa}(\tilde{\bm{x}}) \}$. 
    Combining (\ref{eq-5-4-7:3}) and (\ref{eq-5-4-7:4}),  
    we have 
    \begin{equation}
        \mathcal{L}_{f_{\kappa}} V_{\kappa}(\bm{x}) \!<\! - \sum\limits_{l \in \mathcal{Q}} \pi_{l{\kappa}} \!V_l(\!\bm{x}) \!<\! - V_{\kappa}(\bm{x}) \sum_{l \in \mathcal{Q}} \pi_{l{\kappa}} \!=\! 0 ~~ \forall \bm{x} \!\in\!  \mathcal{O}_{\kappa} \backslash \{\bm{x}_1^* \}
    \end{equation}
    Hence $\mathcal{O}_{\kappa}$ is a compact positively invariant set for subsystem $\kappa$. 
    Let $t_{\kappa}$ the instant when trajectory $\bm{x}(t)$ intersects with $\partial \mathcal{O}_{\kappa}$. 
    Further by (\ref{eq-5-4-8}) and (\ref{eq-5-4-7:4}), $\forall t \in [t_{\kappa}, + \infty)$, $u(\bm{x}(t)) \!=\! \kappa$. 
    Given that $\sigma(t)$ is piecewise constant for $t \in [t_0, +\infty)$, it is finitely piecewise constant. 
\end{proof}

\begin{proposition}\label{prop-5-4-0}
    Given system (\ref{eq-5-4-3}), $\mathcal{X}$, and $\mathcal{X}^*$, if any equilibrium in $\mathcal{X}^*$ is unstable, the Lyapunov-like functions $V_k(\bm{x})$ of some subsystems do not exist.
\end{proposition}

\begin{proof} 
    Without loss of generality, suppose the equilibrium $\bm{x}^*_1$ is unstable and the Lyapunov-like functions $V_k(\bm{x})$, $\forall k \in \mathcal{Q}$, exist. 
    Then by the proof of Theorem \ref{theorem-5-4-1}, $\forall \bm{x} \in  \mathcal{O}_1 \backslash \{\bm{x}_1^* \}$, $\mathcal{L}_{f_1} V_1(\bm{x}) < 0$, which together with (\ref{eq-5-4-7:1}) and (\ref{eq-5-4-7:2}) implies $\bm{x}^*_1$ is stable. This contradiction indicate the Lyapunov-like functions $V_k(\bm{x})$ of some subsystems do not exist. 
\end{proof}


According to Theorem \ref{theorem-5-4-1} and Proposition \ref{prop-5-4-0}, Problem \ref{problem-5-4-1} amounts to constructing the Lyapunov-like functions satisfying conditions (\ref{eq-5-4-7}) together with conditions (\ref{eq-5-4-6}) for Metzler matrices for all subsystems, provided that each equilibrium in $\mathcal{X}^*$ is stable. Then the topology switching law $u$ and region $\mathcal{O}$ are directly given by (\ref{eq-5-4-8}) and (\ref{eq-5-4-9}), respectively. 
Moreover, compare with the work by \cite{4-1486}, the present result extends the framework to the more general setting where subsystems may admit distinct equilibria, rather than a common equilibrium.

\subsubsection{Unique-zero \ac{fnn}}

\ac{fnn}s are formed by a sequence of function layers and can approximate any continuous function on compact subsets of the input domain with a finite number of parameters. Formally, taking the continuous-time state vector $\bm{x} \in \mathbb{R}^{3m-1}$ as the input for example, \ac{fnn}s can be defined as follows:     

\begin{definition}[\ac{fnn}]\label{definition-5-4-4}
    An \ac{fnn} with $n_{\rm f} + 1$ layers is with the form: 
    $\bm{y}_{0} = \bm{x}$ which represents the input layer;  
    $\bm{y}_{\ell} = \gamma_{\ell}( \bm{W}_{\ell} \bm{y}_{\ell - 1} + \bm{b}_{\ell})$, $\forall \ell \in \{1, 2, \cdots, n_{\rm f}\}$, $\bm{y}_{\ell} \in \mathbb{R}^{d_{\ell}}$ is the output of the $\ell$-th layer, $\gamma_{\ell}$ is a fixed element-wise activation function, $\bm{W}_{\ell} \in \mathbb{R}^{d_{\ell} \times d_{\ell} - 1}$ and $\bm{b}_{\ell} \in \mathbb{R}^{d_{\ell}}$ are respectively the weight matrix and bias vector of the linear transformation at the $\ell$-th layer, $d_{\ell}$ is an integer representing the width of the $\ell$-th layer, and $d_0 = 3m-1$.   
\end{definition}

\begin{definition}[Unique-zero \ac{fnn}]\label{definition-5-4-5}
    Given a constant vector $\bm{x}^* \in \mathbb{R}^{3m-1}$, a unique-zero \ac{fnn} w.r.t. $\bm{x}^*$ is an \ac{fnn} satisfying the following conditions
    \begin{subequations}\label{eq-5-4-10}
        \begin{align} 
            & d_1 = 3m-1 \label{eq-5-4-10:1} \\
            & d_{\ell} \geq d_{\ell - 1} ~~ \forall \ell \in \{2, 3, \cdots, n_{\rm f} \} \label{eq-5-4-10:1a} \\
            & \bm{W}_1 = \bm{I}_{d_1}  
            \label{eq-5-4-10:2} \\ 
            & 
            \bm{W}_{\ell} = 
            \begin{bmatrix}
                \bm{M}_{\ell}\T \bm{M}_{\ell} + \varepsilon \bm{I}_{d_{\ell - 1}} \\
                \bm{K}_{\ell}
            \end{bmatrix}
            ~~ \forall \ell \in \{2, 3, \cdots, n_{\rm f} \} 
            \label{eq-5-4-10:3} \\
            &  \bm{b}_1 = - \bm{x}^* \label{eq-5-4-10:3a} \\ 
            & \bm{b}_{\ell} = \bm{0} ~~ \forall \ell \in \{2, 3, \cdots, n_{\rm f} \} 
            \label{eq-5-4-10:4} \\
            & \gamma_{\ell} \text{~has a trivial nullspace}, \forall \ell \in \{1, 2, \cdots, n_{\rm f} \}  \label{eq-5-4-10:5}
        \end{align}
    \end{subequations}
    where $\bm{I}_{d_1} \!\in\! \mathbb{R}^{d_{1} \times d_{1}}$ and $\bm{I}_{d_{\ell - 1}} \!\in\! \mathbb{R}^{d_{\ell - 1} \times d_{\ell - 1}}$ are the identity matrices, 
    $\bm{M}_{\ell} \!\in\! \mathbb{R}^{r_{\ell} \times d_{\ell - 1}}$ for some $r_{\ell} \!\in\!  \mathbb{N}_{\geq 1}$, $\bm{K}_{\ell} \!\in\! \mathbb{R}^{(d_{\ell} - d_{\ell - 1}) \times d_{\ell - 1}}$, and $\varepsilon \in \mathbb{R}_{>0}$ is a constant. 
\end{definition}

\begin{theorem}\label{theorem-5-4-2}
    Let $\phi(\bm{x}| \bm{\xi}, \bm{x}^*): \mathbb{R}^{d_0} \!\to\! \mathbb{R}^{d_{n_{\rm f}}}$ denote the unique-zero \ac{fnn} defined in Definition \ref{definition-5-4-5} where $\bm{\xi} \coloneqq \{ \bm{M}_{\ell}, \bm{K}_{\ell} | \ell = 2, 3, \cdots, n_{\rm f} \}$ is the parameter vector. 
    Define $\nu(\bm{x}| \bm{\xi}, \bm{x}^*) \coloneqq \phi(\bm{x}| \bm{\xi}, \bm{x}^*)\T \phi(\bm{x}| \bm{\xi}, \bm{x}^*)$. 
    Then $\nu(\bm{x}^*| \bm{\xi}, \bm{x}^*) = 0$ and $\nu(\bm{x}| \bm{\xi}, \bm{x}^*) > 0$, $\forall \bm{x} \in \mathcal{X} \backslash \{ \bm{x}^* \}$.
\end{theorem}

\begin{proof} 
    For the matrix $\bm{W}_{\ell}$ structured as (\ref{eq-5-4-10:2}) and (\ref{eq-5-4-10:3}), the top partition $\bm{M}_{\ell}\T \bm{M}_{\ell} + \varepsilon \bm{I}_{d_{\ell - 1}}$ with $\varepsilon > 0$ is positive-definite since $\bm{M}_{\ell}\T \bm{M}_{\ell}$ is positive semi-definite and $\varepsilon \bm{I}_{d_{\ell - 1}}$ is positive-definite. Thus, $\bm{W}_{\ell}$ has full rank, and also a trivial nullspace as peer the rank-nullity theorem \citep{4-673}. 
    Recall that for $\phi(\bm{x}| \bm{\xi}, \bm{x}^*)$, $\bm{y}_0 = \bm{x}$, and $\bm{y}_{\ell} = \gamma_{\ell}( \bm{W}_{\ell} \bm{y}_{\ell - 1} + \bm{b}_{\ell})$, $\forall \ell \in \{1, 2, \cdots, n_{\rm f}\}$. Given (\ref{eq-5-4-10:3a}), we have $\bm{y}_1 = \gamma_1(\bm{x} - \bm{x}^*)$, which together with (\ref{eq-5-4-10:5}), yields
    \begin{equation}\label{eq-5-4-11}
        \bm{x} = \bm{x}^*  \Leftrightarrow \gamma_1(\bm{x} - \bm{x}^*) = \bm{0} \Leftrightarrow \bm{y}_1 = \bm{0}.
    \end{equation}
    Given (\ref{eq-5-4-10:4}), (\ref{eq-5-4-10:5}), and the fact that $\bm{W}_{\ell}$ has a trivial nullspace, the following sequence of logical statements holds:
    \begin{equation}\label{eq-5-4-12}
        \bm{y}_{\ell - 1} \!=\! \bm{0} \Leftrightarrow\! \bm{W}_{\ell} \bm{y}_{\ell - 1} \!=\! \bm{0} \Leftrightarrow\! \gamma_{\ell}( \bm{W}_{\ell} \bm{y}_{\ell - 1} ) \!=\! 0  \Leftrightarrow \bm{y}_{\ell} \!=\! \bm{0}
    \end{equation}
    with $\ell \in \{2,3, \cdots, n_{\rm f}\}$. Combining (\ref{eq-5-4-11}) and (\ref{eq-5-4-12}), we have 
    \begin{equation}\label{eq-5-4-13}
        \bm{x} = \bm{x}^* \Leftrightarrow \phi(\bm{x}| \bm{\xi}, \bm{x}^*) = \bm{0}
    \end{equation}
    Moreover, since $\phi(\bm{x}| \bm{\xi}, \bm{x}^*)\T \phi(\bm{x}| \bm{\xi}, \bm{x}^*)$ is positive-definite w.r.t. $\phi(\bm{x}| \bm{\xi}, \bm{x}^*)$, $\nu(\bm{x}| \bm{\xi}, \bm{x}^*) = 0 \Leftrightarrow \phi(\bm{x}| \bm{\xi}, \bm{x}^*) = \bm{0}$ and $\nu(\bm{x}| \bm{\xi}, \bm{x}^*) > 0$ otherwise. This together with (\ref{eq-5-4-13}) shows that $\nu(\bm{x}^*| \bm{\xi}, \bm{x}^*) = 0$ and $\nu(\bm{x}| \bm{\xi}, \bm{x}^*) > 0 $, $\forall \bm{x} \in \mathcal{X}\backslash \{ \bm{x}^* \}$.
\end{proof}

\begin{proposition}\label{proposition-5-4-1}
    Consider the same function $\nu(\bm{x}| \bm{\xi}, \bm{x}^*)$ as in Theorem \ref{theorem-5-4-2}.
     If $\gamma_{\ell}$ is differentiable everywhere, $\nu(\bm{x}| \bm{\xi}, \bm{x}^*)$ is differentiable on $\mathbb{R}^{d_0} \backslash \{ \bm{x}^* \}$.
\end{proposition}

\begin{proof}
    Let $\bm{\eta}_{\ell} = \bm{W}_{\ell} \bm{y}_{\ell-1} + \bm{b}_{\ell}$. 
    By the chain rule, $\frac{\partial \nu(\bm{x}| \bm{\xi}, \bm{x}^*)}{\partial \bm{x}} = 2 ( \prod_{\ell = 1}^{n_{\rm f}}$ $\frac{\partial \gamma_{\ell}(\bm{\eta}_{\ell}) }{\partial \bm{\eta}_{\ell}} \bm{W}_{\ell} )\T \cdot \phi( \bm{x}| \bm{\xi}, \bm{x}^* ) $. Thus, $\frac{\partial \nu(\bm{x}| \bm{\xi}, \bm{x}^*)}{\partial \bm{x}}$ exists at all $\bm{x} \in \mathbb{R}^{d_0} \backslash \{ \bm{x}^* \}$, if $\frac{\partial \gamma_{\ell}(\bm{\eta}_{\ell}) }{\partial \bm{\eta}_{\ell}}$ exists at all $\bm{\eta}_{\ell} \in \mathbb{R}^{n_{\ell}}$, i.e., $\gamma_{\ell}$ is differentiable everywhere, $\forall \ell \in \{1,2, \cdots, n_{\rm f}\}$. 
\end{proof}

\subsection{Neural Multiple Lyapunov Functions Method}\label{sec-5-4-4}

Based on the established theoretical foundations, the neural \ac{mlf}s method can be used to to solve Problem \ref{problem-5-4-1}. This method constructs the Lyapunov-like functions in Theorem \ref{theorem-5-4-1} utilizing the unique-zero \ac{fnn} and search their parameters following the framework of counter-example guided inductive synthesis \citep{4-1045, 4-1650}. 

\subsubsection{Structural construction of neural \ac{mlf}s}

The candidates of the Lyapunov-like functions in Theorem \ref{theorem-5-4-1} is first designed as 
\begin{equation}\label{eq-5-4-14}
    V_k(\bm{x}| \bm{\xi}_k, \bm{x}_k^*) = \phi_k(\bm{x} | \bm{\xi}_k, \bm{x}_k^*)\T \phi_k(\bm{x} | \bm{\xi}_k, \bm{x}_k^*) ~~ \forall k \in \mathcal{Q}
\end{equation}
where $\phi_k(\bm{x} | \bm{\xi}_k, \bm{x}_k^*)\!:\! \mathcal{X} \!\!\to\!\! \mathbb{R}^{D}$ denotes a unique-zero \ac{fnn} w.r.t. $\bm{x}_k^*$ introduced for subsystem $k$, {formulated by Definition \ref{definition-5-4-4} together with (\ref{eq-5-4-10});} 
and $\bm{\xi}_k$ is the parameter vector of the \ac{fnn}. 
It can be found that these candidate functions automatically satisfy conditions (\ref{eq-5-4-7:1}) and (\ref{eq-5-4-7:2}) for Lyapunov-like functions. 

\begin{remark}[Selection of $\gamma_{\ell}$]
    The activation functions $\gamma_{\ell}$ need to satisfy two requirements: having a trivial nullspace as stated by (\ref{eq-5-4-10:5}), and being differentiable everywhere. Accordingly, activation functions such as Tanh and Swish can be used \citep{4-1666}.
    It is noted that the second requirement comes from condition (\ref{eq-5-4-7:3}) which requires the existence of the Lie derivative of $V_k(\bm{x}| \bm{\xi}_k, \bm{x}_k^*)$ at $\bm{x} \in \mathcal{X} \backslash \{ \bm{x}_k^* \}$. As per Proposition \ref{proposition-5-4-1}, the existence can be guaranteed by the second requirement. 
\end{remark}

\begin{remark}[Selection of $r_{\ell}$]
    In condition (\ref{eq-5-4-10:3}), 
    the matrix $\bm{M}_{\ell} \in \mathbb{R}^{r_{\ell} \times d_{\ell - 1}}$ has $r_{\ell} d_{\ell -1}$ free parameters, and the symmetric matrix $\bm{M}_{\ell}\T \bm{M}_{\ell} \in \mathbb{R}^{d_{\ell -1} \times d_{\ell -1}}$ has at most $\sum_{j=1}^{d_{\ell -1}} j = d_{\ell -1}(d_{\ell -1} + 1)/2$ free parameters. The independence of the $d_{\ell -1}(d_{\ell -1} + 1)/2$ free parameters of $\bm{M}_{\ell}\T \bm{M}_{\ell}$ requires $r_{\ell} d_{\ell -1} \geq d_{\ell -1}(d_{\ell -1} + 1)/2$. Thus to minimize the number of free parameters required by the Lyapunov-like candidate functions, the value of $r_{\ell}$ can be selected as $\lceil (d_{\ell -1} + 1)/2 \rceil$. 
\end{remark}


\subsubsection{Parameter searching for neural \ac{mlf}s}

Given the Lyapunov-like candidate functions $V_k(\bm{x}| \bm{\xi}_k, \bm{x}_k^*)$, 
Problem \ref{problem-5-4-1} now amounts to searching the value of $\{ \bm{\xi}_k | k \in \mathcal{Q}\}$ with which conditions (\ref{eq-5-4-7:3}) together with conditions (\ref{eq-5-4-6}) for Metzler matrices for all subsystems holds. For notational simplicity, let $\bm{\Xi}$ denote the vector formed by all entries of $\{ \bm{\xi}_k | k \in \mathcal{Q}\}$; and recall the Metzler matrix $\bm{\Pi} =  [\pi_{lk}] \in \mathbb{R}^{n \times n}$. Then the desired value of $\bm{\Xi}$ can be yielded by solving the following problem:

\begin{problem}[Parameter searching problem]\label{problem-5-4-2}
    Given system (\ref{eq-5-4-3}), $\mathcal{X}$, $\mathcal{X}^*$, and $V_k(\cdot|\cdot, \cdot)$, find $(\bm{{\Xi}}^*, \bm{{\Pi}}^*) \in $
    \begin{equation*}
        \begin{aligned}
            & \arginf_{\bm{{\Xi}}, \bm{{\Pi}} } 
            \sum_{k\in \mathcal{Q}} 
            \Big[ 
            \sum_{l \in \mathcal{Q}\backslash \{k\} } \max(0, -\pi_{lk}) 
            + \Big( \sum_{l \in \mathcal{Q}} \pi_{lk} \Big)^{\!2}
            +
            \\
            & \sup_{\bm{x} \in \mathcal{X}_k }  \Big(  
                \max\limits_{\rho}  \big(
                    0,
                    \mathcal{L}_{f_k} V_k(\bm{x}| \bm{\xi}_k, \bm{x}_k^*) 
                    + \sum_{l \in \mathcal{Q}} \pi_{lk} V_l(\bm{x}| \bm{\xi}_l, \bm{x}_l^*)
                    \big)  
            \Big)  
            \Big]
        \end{aligned}
    \end{equation*}
    where $\mathcal{X}_k \coloneqq \mathcal{X} \backslash \{\bm{x}^*_k\}$, $\rho$ is a positive constant introduced for controlling the numerical error of converting a strict inequality to a nonstrict one, $\max\nolimits_{\rho}(\cdot, \cdot) \coloneqq \max(\cdot, \cdot +\rho)$, and the value of $\rho$ is orders of magnitude smaller than the range of the second term in $\max\nolimits_{\rho}(\cdot, \cdot)$.
\end{problem}

\begin{proposition}
    Given system (\ref{eq-5-4-3}), $\mathcal{X}$, and $\mathcal{X}^*$, if the infimum in Problem \ref{problem-5-4-2} equals to 0, the Lyapunov-like functions exist for all subsystems; and denoting by $\bm{\xi}_k^*$ the subvector of $\bm{\Xi}^*$ for $\bm{\xi}_k$, for $k \in \mathcal{Q}$, $V_k(\bm{x}|\bm{\xi}_k^*, \bm{x}_k^*)$ is a Lyapunov-like function for $f_k$ and trajectory $\bm{x}(t)$ over $\{ t | \sigma(t) = k \}$, with $\bm{\Pi}^*$ being the associated value of the Metzler matrix.
\end{proposition}

Following the framework of counter-example guided inductive synthesis in \citep{4-1045}, Problem \ref{problem-5-4-2} can be solved by executing two modules interactively, called \textit{learner} and \textit{verifier}, respectively. The learner first solves a surrogate parameter searching problem which is similar to Problem \ref{problem-5-4-2} but can be tackled by the global optimizers used for training neural networks, e.g., the stochastic gradient descent (SGD) algorithm and Adam algorithm \citep{4-1507}. The verifier verifies whether the surrogate problem and Problem \ref{problem-5-4-2} are equivalent. If the equivalence holds, then Problem \ref{problem-5-4-2} is solved; otherwise, the verifier modifies the surrogate parameter searching problem and returns to the learner module. This loop is repeated until the equivalence is verified. 
The topology control will shut off temporarily in response to possible failures of the parameter searching, which will be further clarified in the context of the application framework provided later. 


More specifically, the surrogate parameter searching problem solved by the learner is defined as follows:
\begin{problem}[Surrogate parameter searching problem]\label{problem-5-4-3}
    Given system (\ref{eq-5-4-3}), $\mathcal{X}^*$, $V_k(\cdot|\cdot)$, and a finite set of samples $\mathcal{S} \subset \mathcal{X}$, 
    find $(\bm{\Xi}^{\star}, \bm{\Pi}^{\star}) \in $
    \begin{equation*}
        \begin{aligned}
            & 
            \arginf_{\bm{{\Xi}}, \bm{{\Pi}} }
            L(\bm{\Xi}, \bm{\Pi}) 
            \coloneqq \!\! 
            \sum_{k\in \mathcal{Q}} 
                \Big[ 
                \sum_{l \in \mathcal{Q}\backslash \{k\} } \!\!\! \max(0, -\pi_{lk}) 
                \!+\! \Big( \sum_{l \in \mathcal{Q}} \pi_{lk} \Big)^{\!2} 
            +
            \\
            &
            \frac{1}{ |\mathcal{S}| } \sum_{\bm{x} \in \mathcal{S}} 
            \max\limits_{\rho} \big(
                0,
                \mathcal{L}_{f_k} V_k(\bm{x}| \bm{\xi}_k, \bm{x}_k^*) 
                + \sum_{l \in \mathcal{Q}} \pi_{lk} V_l(\bm{x}| \bm{\xi}_l, \bm{x}_l^*)
            \big)       
            \Big].
        \end{aligned}
    \end{equation*}
\end{problem}

\begin{proposition}\label{proposition-5-4-2}
    Given system (\ref{eq-5-4-3}), $\mathcal{X}^*$, $V_k(\cdot|\cdot)$, and $\mathcal{S}$, denoting by $\bm{\xi}_k^{\star}$ the subvector of $\bm{\Xi}^{\star}$ for $\bm{\xi}_k$ and $\pi_{lk}^{\star}$ the entry of $\bm{\Pi}^{\star}$ for $\pi_{lk}$, 
    if the infimum in Problem \ref{problem-5-4-3} equals to 0 and the equality
    \begin{equation}\label{eq-5-4-a} 
            \sup_{\bm{x} \in \mathcal{X}_k }  \!\!\Big(\!\!  
            \max\limits_{\rho} \big(
                0,
                \mathcal{L}_{f_k} V_k(\bm{x}| \bm{\xi}_k^{\star}, \bm{x}_k^*) 
                +\! \sum_{l \in \mathcal{Q}} \!\pi_{lk}^{\star} V_l(\bm{x}| \bm{\xi}_l^{\star}, \bm{x}_l^*)
                \big)  
        \Big)  
            \!=\! 0 
    \end{equation}
    holds for all $k \in \mathcal{Q}$, then Problem \ref{problem-5-4-2} and Problem \ref{problem-5-4-3} are equivalent, namely that the infimum in Problem \ref{problem-5-4-2} equals to 0 and there exists $(\bm{{\Xi}}^*, \bm{{\Pi}}^*)$ equal to $(\bm{{\Xi}}^{\star}, \bm{{\Pi}}^{\star})$.
\end{proposition}

As per Proposition \ref{proposition-5-4-2}, the verifier can verify the equivalence of Problem \ref{problem-5-4-2} and Problem \ref{problem-5-4-3} by checking the validness of equality (\ref{eq-5-4-a}). Let $\bm{\pi}_k^{\star} \coloneqq [\pi_{lk}]_{l \in \mathcal{Q}}$. 
For computational tractability, the following first-order logic formulae is considered:
\begin{equation*}\!\!
    \begin{aligned}
         & F(\bm{x}| \bm{\Xi}^{\star}, \bm{\Pi}^{\star}) \coloneqq \exists \bm{x}. \bigvee_{k \in \mathcal{Q}} F_k(\bm{x} | \bm{\xi}_k^{\star}, \bm{\pi}_k^{\star}) \coloneqq  
         \Big[ \big[ (\bm{x} \in \mathcal{X}) \wedge 
         \\
         &  (\bm{x} \neq\! \bm{x}_k^*) \wedge \Big( \mathcal{L}_{f_k} \!V_k(\bm{x}| \bm{\xi}_k^{\star}, \bm{x}_k^*) 
         + \sum_{l \in \mathcal{Q}} \pi_{lk}^{\star} V_l(\bm{x}| \bm{\xi}_l^{\star}, \bm{x}_l^*) \!\geq\! 0 \!\Big) \big] 
         \Big],
    \end{aligned}
\end{equation*}
which clearly is false iff equality (\ref{eq-5-4-a}) holds for all $k \in \mathcal{Q}$. 
$F(\bm{x}| \bm{\Xi}^{\star}, \bm{\Pi}^{\star})$ can be solved by Satisfiability Modulo Theories (\ac{smt}) solvers \citep{4-1509}. When $F(\bm{x}|\bm{\Xi}^{\star}, \bm{\Pi}^{\star})$ is true, \ac{smt} solvers will yield a satisfying assignment denoted as $\bm{x}_{\rm c}$. The equality (\ref{eq-5-4-a}) does not hold for some $k \in \mathcal{Q}$ since $\mathcal{X}_k$ includes $\bm{x}_{\rm c}$ while the sample set $\mathcal{S}$ does not when solving Problem \ref{problem-5-4-3}. Accordingly, adding $\bm{x}_{\rm c}$ to $\mathcal{S}$ can promote the validness of equality (\ref{eq-5-4-a}) and thus the equivalence of Problem \ref{problem-5-4-2} and Problem \ref{problem-5-4-3}. 

The pseudocode of solving the parameter searching problem is given in Algorithm \ref{alg-5-4-1}. Here $n_{\rm s}$ is the initial size of the sample set $\mathcal{S}$, $n_{c}$ is the number of additional points chosen from the neighborhood of each $\bm{x}_{{\rm c}, k}$ to augment $\mathcal{S}$, and $\mathcal{B}_{\delta}(\bm{x}_{{\rm c}, k}) \coloneqq \{\bm{x} \in \mathcal{X} \bm{|} \Vert \bm{x} - \bm{x}_{{\rm c}, k} \Vert_2 \leq \delta \}$ with $\delta$ being a positive constant that controls the neighborhood size. 
The value of $\delta$ can be an order of magnitude smaller than the range of $\bm{x}$. 

\begin{algorithm}\label{alg-5-4-1}



    \DontPrintSemicolon
    \KwInput{$f_k$, $\mathcal{Q}$, $\mathcal{X}$, $\mathcal{X}^*$, $V_k(\cdot|\cdot)$, $n_{\rm s}$, $n_{\rm c}$, $\rho$, $\delta$}

    $\mathcal{S} \gets$ randomly generate $n_{\rm s}$ samples of $\bm{x} \in \mathcal{X}$; $\mathcal{A} \gets \emptyset$
    
    \Repeat{$F_k(\bm{x}| \bm{\xi}_k^{\star}, \bm{\pi}_k^{\star})$ is false for all $k \in \mathcal{Q}$}
    {   
        $\mathcal{S} \gets \mathcal{S} \cup \mathcal{A}$

        $(\bm{\Xi}^{\star}, \bm{\Pi}^{\star}) \gets $ solve Problem \ref{problem-5-4-3} by the global optimizer

        \lIf{$L(\bm{\Xi}^{\star}, \bm{\Pi}^{\star}) > 0$}{\textbf{Break}}

        \For{$k \in \mathcal{Q}$}    
        { 
            Solve $F_k(\bm{x}| \bm{\xi}_k^{\star}, \bm{\pi}_k^{\star})$ by the \ac{smt} solver

            \If{$F_k(\bm{x}| \bm{\xi}_k^{\star}, \bm{\pi}_k^{\star}) $ is true}
            {   
                $\!\!\!\bm{x}_{{\rm c}, k} \!\gets\!$ a satisfying assignment of $F_k(\bm{x}| \bm{\xi}_k^{\star}, \bm{\pi}_k^{\star}) $

                $\!\!\!\mathcal{A}_k \!\gets\!$ randomly generate $n_{\rm c}$ point from $\mathcal{B}_{\delta}(\bm{x}_{{\rm c}, k})$

                $\!\!\!\mathcal{A} \gets \mathcal{A} \cup \mathcal{A}_k \cup \{ \bm{x}_{{\rm c}, k} \} $
            }
        }
    }
    \lIf{$L(\bm{\Xi}^{\star}, \bm{\Pi}^{\star}) = 0$}{\Return{$\bm{\Xi}^* \gets \bm{\Xi}^{\star}$}}
    \lElse{\Return{Parameter searching fails}}
    \caption{Parameter searching for neural \ac{mlf}s}
    \label{alg-5-4-1}
\end{algorithm}

\subsubsection{Solution of the topological switching stabilization problem}

Given $\bm{\Xi}^*$ returned by Algorithm \ref{alg-5-4-1}, following (\ref{eq-5-4-8}) and (\ref{eq-5-4-9}), the state-dependent topology switching law $u$ and the subset of the stabilizable region $\mathcal{O}$ respectively are 
\begin{subequations}\label{eq-5-4-sol}
    \begin{align}
        &\!\!\! u(\bm{x}(t)) = \argmin_{k \in \mathcal{Q}} \{ V_k(\bm{x}(t) | \bm{\xi}_k^*, \bm{x}_k^* ) \} \label{eq-5-4-sol:1}
        \\
        &\!\!\! 
        \mathcal{O} = \{ \bm{z} \in \mathcal{X} | \min_{k \in \mathcal{Q}}\! V_k(\bm{z} | \bm{\xi}_k^*, \bm{x}_k^* ) \leq \min_{\tilde{\bm{z}} \in \partial \mathcal{X}} \min_{ k \in \mathcal{Q}}  V_k(\tilde{\bm{z}} | \bm{\xi}_k^*, \bm{x}_k^* ) \} \label{eq-5-4-sol:2}
    \end{align}
\end{subequations}
with
\begin{equation}
    V_k(\bm{x}(t) | \bm{\xi}_k^*, \bm{x}_k^* ) = \phi_k(\bm{x}(t) | \bm{\xi}_k^*, \bm{x}_k^* )\T \phi_k(\bm{x}(t) | \bm{\xi}_k^*, \bm{x}_k^*).     
\end{equation}

\subsection{Applicability and Application Framework}\label{sec-5-4-5}

We further analyze the applicability of transient topology control for stabilizing MMGs with the neural MLFs method, and also provide the application framework. 

\subsubsection{Applicability}

The proposed topology control follows the state-feedback control paradigm and the state $\bm{x}(t)$ fed to the controller is in practice a discrete-time signal. Thus, 
the computation time of ${u}(\bm{x}(t))$ given by (\ref{eq-5-4-sol:1}) for a sampling point of $\bm{x}(t)$ is required to be strictly less than the sampling period. This requirement can be satisfied given that the cardinality of $\mathcal{Q}$ is generally small and the Lyapunov-like function $V_k$ is an algebraic function. 
However, $V_k$ for each subsystem should be obtained before computing ${u}(\bm{x}(t))$, which requires Algorithm \ref{alg-5-4-1} to be completed before a disturbance occurs. 
Moreover, according to the basic principle of the topology control for MMG stabilization, the stabilizability assessment needs to be completed almost instantaneously after the disturbance clearance. This requires obtaining the value of the RHS of (\ref{eq-5-4-sol:2}) before a disturbance occurs, with which if the state $\bm{x}_0$ is in $\mathcal{O}$ can be determined instantaneously. 
After a disturbance, to ensure the pre-obtained $V_k$ and the value of the RHS of (\ref{eq-5-4-sol:2}) to be valid, it is required that all network topologies in $\mathcal{Q}$ are still reachable and all parameters in (\ref{eq-5-4-1}) and (\ref{eq-5-4-2}) remain unchanged.  

The above considerations naturally yield the following conditions to employ the proposed transient topology control: 

\begin{condition}\label{cond-5-4-1}
    The setpoints $\omega_i^{\rm s}$, $v_i^{\rm s}$, $p_i^{\rm s}$ and $q_i^{\rm s}$ are adjusted periodically. Assume the adjustment time interval is $T_1$, and the total time of computing $\mathcal{X}^*$ and the RHS of (\ref{eq-5-4-sol:2}) and performing Algorithm \ref{alg-5-4-1} is $T_2$. Then it is required that $T_2 \leq T_1$. 
\end{condition}

\begin{condition}\label{cond-5-4-2}
    The topology control is allowed to be activated only after the disturbance that is cleared without changing the network topology of the MMG system.  
\end{condition}

\begin{remark}
    Let $T_3$ be the time interval between implementing a setpoint value and completing computation of the next setpoint value. 
    In Condition \ref{cond-5-4-1}, $T_2 \leq T_1$ ensures that for each time period of $T_1$ where the setpoints remain constant, there exists a subperiod of $\min \{ 2T_1 - T_2 - T_3, T_1\}$, where the topology control is able to be activated. Clearly, the topology control is always able to be activated if $T_2 \leq T_1 - T_3$. This stronger inequality can still be satisfied given the small number of microgrids in a MMG system and the typical value of $T_1$, e.g., 30 min. Moreover, Condition \ref{cond-5-4-2} holds for disturbances occurring within individual microgrids, which account for the majority of disturbances of the entire MMG system. 
\end{remark}

\subsubsection{Application framework}

\begin{figure}[t]
	\centering
 	\includegraphics[scale=0.6]{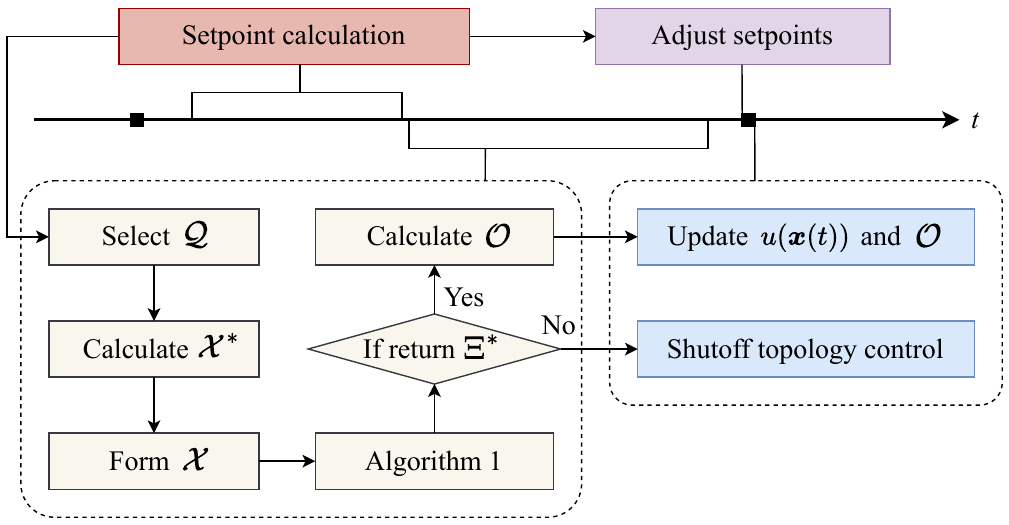}  
	\caption{Application framework of transient topology control for MMG stabilization.}
	\label{fig-5-4-3}
\end{figure}

The application framework of the proposed transient topology control for MMG stabilization is illustrated in Fig. \ref{fig-5-4-3}. Within each operating cycle of the MMG system, the setpoint calculation is executed first. Then the eligible subsystem set $\mathcal{Q}$ is determined based on the current status of lines, and the equilibrium of interest for each subsystem is calculate subsequently to form $\mathcal{X}^*$. The stability of the equilibria needs to be checked. 
{Specifically, the eligible network topologies in $\mathcal{Q}$ are selected according to the following five criteria: 
{(a)} $\mathcal{Q}$ includes the current network topology; 
{(b)} all topologies in $\mathcal{Q}$ are connected; 
{(c)} during the transition between any two topologies in $\mathcal{Q}$, asynchronous line-switching actions do not cause network disconnection; 
{(d)} the equilibrium associated with each topology in $\mathcal{Q}$ is stable; 
and {(e)} the number of topologies in $\mathcal{Q}$ should be maximized. 
For these criteria, 
criterion {(c)} ensures that during the topology transition, when both line opening and line closing are needed, performing only the opening of lines due to the uncertainty of real-world line-switching delays does not disconnect the MMG system. 
Criteria {(b)} and {(c)} together ensure that the network topology remains connected during the topology control. 
Criterion {(d)} is necessary to guarantee the existence of $V_k(\bm{x})$ for each subsystem in $\mathcal{Q}$, according to Proposition \ref{prop-5-4-0}. 
Criterion (e) maximizes the stabilizing capability of the topology control w.r.t. all possible disturbance scenarios. 
It is worth noting that the maximal number of topologies in $\mathcal{Q}$ is generally small given criterion (c). 
} 

Given the expression of $\mathcal{X}$, Algorithm \ref{alg-5-4-1} is executed. If Algorithm \ref{alg-5-4-1} fails to search the parameters $\bm{\Xi}$, the topology control will shut off at the end of the current operating cycle; or otherwise, the region $\mathcal{O}$ is calculated and the topology control will be activated with $u(\bm{x}(t))$ and $\mathcal{O}$ updated and the setpoints adjusted at the end of the current operating cycle.

\subsection{Numerical Example: 4-microgrid System}\label{sec-5-4-6}

The proposed topology control for \ac{mmg} stabilization is demonstrated using a simple 4-microgrid system. Fig. \ref{fig-5-4-4} shows the diagram of the system that consists of four microgrids and five tie lines. 
According to the criteria for determining $\mathcal{Q}$, four subsystems shown in Fig. \ref{fig-5-4-4} are used for topology control. Subsystem 1 corresponds to the topology with all lines closed, and the other three subsystems are formed by switching off a line from subsystem 1. {Detailed data of this test \ac{mmg} system is provided in Table \ref{tab-5-4-ap-1}. 

\begin{figure}[h]
	\centering
 	\includegraphics[scale=1.1]{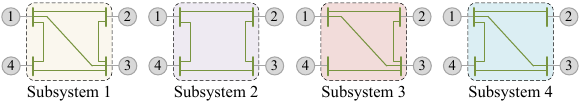}  
	\caption{Diagram of the 4-microgrid system and its subsystems.}
	\label{fig-5-4-4}
\end{figure}

\begin{table}[h!]
    \centering
    \setlength{\tabcolsep}{2.93pt}  
    \setlength\extrarowheight{4pt} 

    \caption{Parameters of the 4-microgrid system}
    \small{
    \begin{tabular*}{0.917\hsize}{|lllllllll|} \hline
    \multicolumn{9}{|c|}{Parameters of the converter interface of each microgrid} \\ \hline
    MG & $\tau_{{\rm p}, i}$(s)  & $k_{{\rm p}, i}$ &  $\tau_{{\rm q}, i}$(s) &  $k_{{\rm q}, i}$ & $\omega_i^{\rm s}$(p.u.) & $v_i^{\rm s}$(p.u.) & $p_i^{\rm s}$(MW) & $q_i^{\rm s}$(MW) \\ \hline 
    MG1 & 0.5 & 0.025 & 0.5 & 0.02 &  1.0 & 1.0400  & 1.3467   &  1.1317  \\
    MG2 & 0.4 & 0.020 & 0.4 & 0.02 &  1.0 & 0.9300 & -0.9000   &  -0.8000 \\
    MG3 & 0.5 & 0.015 & 0.5 & 0.02 &  1.0 & 0.9712 & -0.5000   &  -0.6000 \\
    MG4 & 0.2 & 0.012 & 0.2 & 0.02 &  1.0 & 1.0293 & 0.2000    &  0.4000  \\
    \end{tabular*}  
    \\ 
    \hspace{3pt}
    \begin{tabular*}{0.917\hsize}{|llllll|} \hline\hline
        \multicolumn{6}{|c|}{Parameters of the tie lines} \\ \hline
        Tie line  ~~~~~~~~~~~~~ & 1-2 ~~~~~~ & 1-3 ~~~~~~ & 1-4 ~~~~~~~ & 2-3 ~~~~~~~ & 3-4 ~~~~~~~ \hspace{-4.4pt} \\ \hline
        Resistance (p.u.) & 0.13 & 0.03 & 0.20 & 0.03 & 0.15 \\
        Reactance (p.u.)  & 0.05 & 0.02 & 0.08 & 0.02 & 0.06 \\
        \hline 
        \end{tabular*} 
    }
    \label{tab-5-4-ap-1}  
\end{table}

\subsubsection{{Computation results of neural \ac{mlf}s and region $\mathcal{O}$}}

The obtained value of matrix $\bm{\Pi}$ is given by (\ref{eq-5-4-r1}), which satisfies the conditions in (\ref{eq-5-4-6}). 
Moreover, Fig. \ref{fig-5-4-rr34} illustrates  the projections of the obtained Lyapunov-like functions and the left-hand side of inequality (\ref{eq-5-4-7:3}) onto some planes. It can be observed that for each of these projections, the conditions in (\ref{eq-5-4-7}) for the Lyapunov-like functions and their Lie derivatives are all satisfied. 
\begin{equation}\label{eq-5-4-r1}
    {\bm{\Pi}^{\star} = 
    \begin{bmatrix}
        -0.5017&  0.1058 &  1.7162 &  1.3881  \\
        0.1651 &  -0.2734 &  0.3134 &  0.1931 \\
        0.1358 &   0.0305 & -2.0954 &  0.0414 \\
        0.2008 &   0.1370 &  0.0658 & -1.6225
    \end{bmatrix}
    }
\end{equation}

Moreover, solving the RHS of the inequality in (\ref{eq-5-4-sol:2}) yields that $\min_{\tilde{\bm{z}} \in \partial \mathcal{X}} \min_{ k \in \mathcal{Q}} \! V_k(\tilde{\bm{z}} | \bm{\xi}_k^*, \bm{x}_k^* ) = 0.0265$. Thus the stabilizable region $\mathcal{O} \!=\! \{ \bm{z} \!\in\! \mathcal{X} | \min_{k \in \mathcal{Q}} V_k(\bm{z} | \bm{\xi}_k^*, \bm{x}_k^* ) \!\leq\! 0.0265 \}$. {For greater insight, Fig. \ref{fig-5-4-rr5} shows the contour graphs of projections of function $\min_{k \in \mathcal{Q}} V_k(\bm{x} | \bm{\xi}_k^*, \bm{x}_k^* )$ onto some planes. In the figure, the contour-filled regions at the level of 0.0265, i.e., the gray areas, are the projections of region $\mathcal{O}$ onto the associated planes.} 

\begin{figure}[H]
    \centering
    \begin{subcaptionbox}{Lyapunov-like functions \vspace{5pt}}[0.72\textwidth]
        {\includegraphics[scale=0.79]{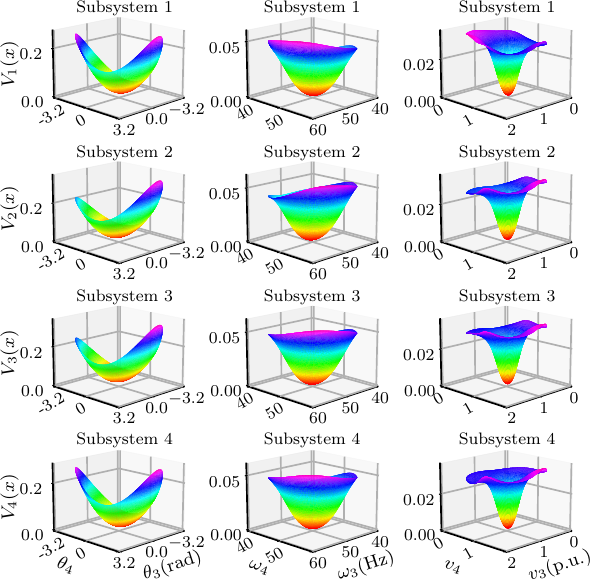}}
    \end{subcaptionbox}
    \hfill
    \begin{subcaptionbox}{$\mathcal{L}_k(\bm{x})$}[0.72\textwidth]
        {\includegraphics[scale=0.79]{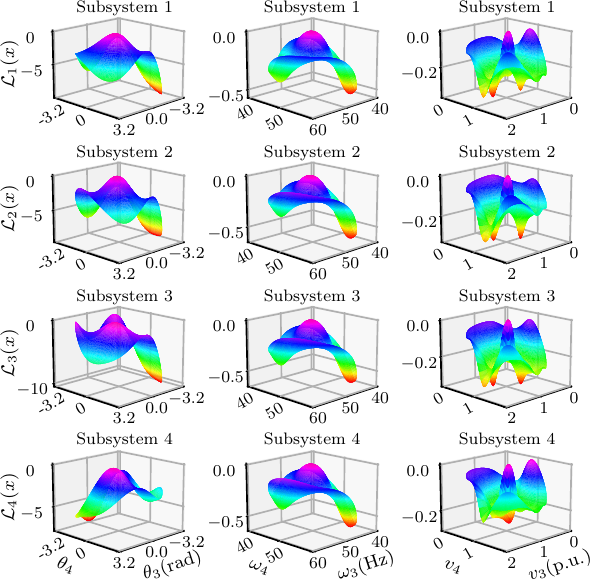}}
    \end{subcaptionbox}
    \caption{Projections of the obtained Lyapunov-like functions and $\mathcal{L}_k(\bm{x})$ (denoting the left-hand side of inequality (\ref{eq-5-4-7:3})) onto the $\theta_3$-$\theta_4$ plane (left side of each subfigure), $\omega_3$-$\omega_4$ plane (middle of each subfigure), and $v_3$-$v_4$ plane (right side of each subfigure), with the other state variables equal to their values in $\bm{x}_k^*$.}
    \label{fig-5-4-rr34}
\end{figure}

\begin{figure}[H]
	\centering
 	\includegraphics[scale=1]{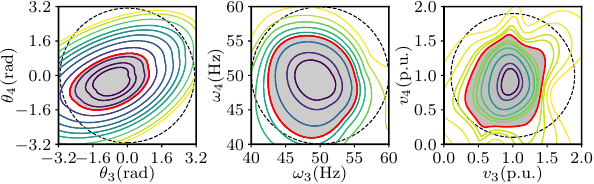}  
	\caption{{ 
    Contour graphs of the projections of function $\min_{k \in \mathcal{Q}} V_k(\bm{x} | \bm{\xi}_k^*, \bm{x}_k^* )$ onto the same planes as in Fig. \ref{fig-5-4-rr34}. The dash black circles are the boundary of region $\mathcal{X}$. The red lines are the contour lines with the level of 0.0265. 
     }}
	\label{fig-5-4-rr5}
\end{figure}

\subsubsection{Stabilization results}
Consider the scenario where the \ac{mmg} operates with the topology of subsystem 1 under the normal state and a 3-phase short-circuit fault occurs within MG 2 and is cleared at $t = 0$. Fig. \ref{fig-5-4-r1} (the left) shows the time-domain simulation results of $\theta_i$, $\omega_i$, and $v_i$ of the \ac{pcc} of each microgrid, without the proposed topology control. The \ac{mmg} is unstable with noticeable oscillations of voltage angles and frequencies in the \ac{pcc}s. 

\begin{figure}[h]
	\centering
 	\includegraphics[scale=1]{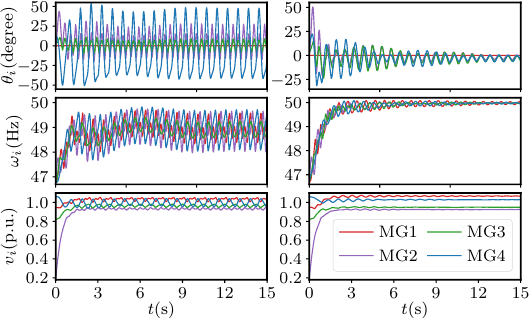}  
	\caption{Dynamic response of the \ac{mmg} without (the left) and with (the right) the topology control.}
	\label{fig-5-4-r1}
\end{figure}

\begin{figure}[h]
	\centering
 	\includegraphics[scale=1]{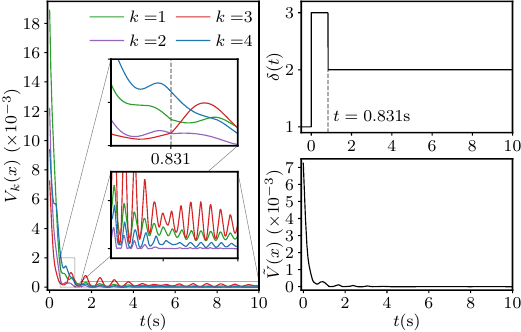}  
	\caption{Curves of $V_k(\bm{x})$-$t$ and $\tilde{V}(\bm{x})$-$t$, and $\sigma(t)$.}
	\label{fig-5-4-r2}
\end{figure}

Fig. \ref{fig-5-4-r2} (the left) shows the changes of the Lyapunov-like functions of all the subsystems after the fault is cleared. We can find that at $t = 0$, the value of $V_3(\bm{x})$ is the smallest and less than 0.0265, and accordingly the topology control is activated to first switch the \ac{mmg} from subsystem 1 to subsystem 3, by opening line 1-4. At $t=0.831$ s, the value of $V_2(\bm{x})$ exceeds that of $V_3(\bm{x})$. Thus the topology control further switches the \ac{mmg} from subsystem 3 to subsystem 2, by closing line 1-4 and opening line 1-3. When $t>0.831$ s, the value of $V_2(\bm{x})$ remains the smallest, indicating no further topology switching actions. 

The corresponding topology switching control signal $\sigma(t)$ is given by Fig. \ref{fig-5-4-r2} (the upper-right). 
Fig. \ref{fig-5-4-r2} (the lower-right) shows the change of the Lyapunov function value of the system with the corresponding topology switching actions. We can find that $\tilde{V}(\bm{x})$ will decay to zero as $t \to \infty$, i.e., the \ac{mmg} is stable. Clearly, according to the dynamic response of the \ac{mmg} with the topology control given by Fig. \ref{fig-5-4-r1} (the right), the system state will converge to an equilibrium as $t \!\to\! \infty$.

\subsection{Numerical Example: 10-Microgrid System}\label{sec-5-4-7}

The proposed topology control for MMG stabilization is further demonstrated using a 10-microgrid system designed based on the IEEE 123-node test feeders. Fig. \ref{fig-5-4-5} shows the system partition of the IEEE 123-node test feeder to form the MMG system, and the diagram at the interconnection level where 10 microgrids are connected through 11 tie lines. 
Detailed data of this test MMG system is given in Table \ref{tab-5-4-ap-2}. 
According to the criteria for determining $\mathcal{Q}$, subsystems 1 to 4 as shown in Fig. \ref{fig-5-4-5} are used for topology control. We denote the eligible topology set $\mathcal{Q}$ that contains all these four subsystems as $\mathcal{Q}_1$. For comparison, we also consider two additional cases of $\mathcal{Q}$, i.e., $\mathcal{Q}_2$ comprising subsystems 1, 2 and 3, and $\mathcal{Q}_3$ consisting of subsystems 1 and 2. Each of them involves fewer subsystems, thereby failing to meet the criterion (e) for determining $\mathcal{Q}$.

\begin{figure}[t]
	\centering
 	\includegraphics[scale=0.9]{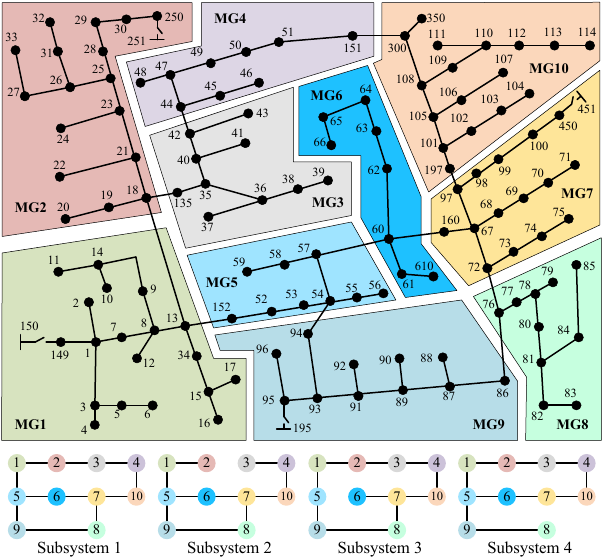}  
	\caption{{Diagram of the 10-microgrid system designed based on the IEEE 123-node test feeder, and its four subsystems.}}
	\label{fig-5-4-5}
\end{figure}

\begin{table}[h!]
    \centering
    \setlength{\tabcolsep}{2.7pt}  
    \setlength\extrarowheight{4pt}

    \caption{{Parameters of the 10-microgrid system}}
    \small{
    \begin{tabular*}{0.92\hsize}{|lllllllll|} \hline
    \multicolumn{9}{|c|}{Parameters of the converter interface of each microgrid} \\ \hline
    MG & $\tau_{{\rm p}, i}$(s)  & $k_{{\rm p}, i}$ &  $\tau_{{\rm q}, i}$(s) &  $k_{{\rm q}, i}$ & $\omega_i^{\rm s}$(p.u.) & $v_i^{\rm s}$(p.u.) & $p_i^{\rm s}$(MW) & $q_i^{\rm s}$(MW) \\ \hline 
    MG1 & 0.2 & 0.015 & 0.2 &  0.02 &  1.0 & 1.0400  & 0.6245    &  0.3068  \\
    MG2 & 0.5 & 0.025 & 0.5 &  0.02 &  1.0 & 0.9590  & -0.2600   &  -0.0800 \\
    MG3 & 0.4 & 0.020 & 0.4 &  0.02 &  1.0 & 0.9699  & -0.2000   &  -0.1000 \\
    MG4 & 0.5 & 0.020 & 0.5 &  0.02 &  1.0 & 1.0297  & 0.5000    &  0.3000  \\
    MG5 & 0.3 & 0.015 & 0.3 &  0.02 &  1.0 & 1.0027  & -0.1600   &  -0.0800 \\
    MG6 & 0.5 & 0.018 & 0.5 &  0.02 &  1.0 & 0.9761  & -0.4100   &  -0.2300 \\
    MG7 & 0.2 & 0.020 & 0.2 &  0.02 &  1.0 & 1.0246  & 0.8000    &  0.4500  \\
    MG8 & 0.5 & 0.024 & 0.5 &  0.02 &  1.0 & 0.9618  & -0.2850   &  -0.2000 \\
    MG9 & 0.3 & 0.018 & 0.3 &  0.02 &  1.0 & 0.9687  & -0.2600   &  -0.1300 \\
    MG10 &0.4 & 0.016 & 0.4 &  0.02  & 1.0 & 1.0088 & -0.3000   &  -0.1500 \\
    \end{tabular*} 
    \\
    \hspace{3pt}
    \begin{tabular*}{0.92\hsize}{|lllllll|} \hline\hline
        \multicolumn{7}{|c|}{Parameters of the tie lines} \\ \hline
        Tie line  ~~~~~~~~~~ & 1-2 ~~~~ & 1-5 ~~~~ &  2-3 ~~~~ & 3-4 ~~~~ & 4-10 ~~~~~ & 5-6 ~~~~   \hspace{5pt}  \\ \hline
        Resistance (p.u.) & 0.3914 &  0.0595 &  0.0695 &  0.0931 &  0.0839 &  0.3519  \\
        Reactance (p.u.)  & 0.8794 & 0.0428  & 0.0638  &  0.2192  & 0.0756  & 0.1231    \\ \hline
        Tie line ~~~~~~~~~~ &  5-9 & 6-7 & 7-8 & 7-10 & 8-9 &  \\ \hline
        Resistance (p.u.) & 0.0869 &  0.0869 &  0.0939 &  0.0848 &  0.3285  &    \\
        Reactance (p.u.)  & 0.0783  & 0.0783  & 0.2166  & 0.0667  & 0.7582  &     \\        
        \hline 
        \end{tabular*} 
    }
    \label{tab-5-4-ap-2}  
\end{table}

\subsubsection{Computation time and stabilizability with different $\mathcal{Q}$} 
For the topology control with $\mathcal{Q} = \mathcal{Q}_1$ which utilizes the largest number of subsystems, computing $\mathcal{X}^*$ takes 0.293 s, computing the RHS of (\ref{eq-5-4-sol:2}) takes 14.257 s, and performing Algorithm \ref{alg-5-4-1} takes 1623.667 s that includes 
314.728 s and 1303.590 s separately for line 4 and line 7 of Algorithm 1. Thus Condition \ref*{cond-5-4-1} is satisfied given $T_1 = 30$ min $>T_2 = 1638.217$ s.

To analyze the stabilizing capability of the topology control with different $\mathcal{Q}$, we consider five different disturbances, named D1 to D5. 
For all the disturbances, the MMG operates with the topology of subsystem 1 under the normal state; and a 3-phase short-circuit fault occurs within a microgrid and is cleared at $t = 0$. 
For D1 to D5, the faulty microgrids are MG1, MG2, MG6, MG8, and MG9, respectively. 
Table \ref{tab-5-4-1} gives the values of $\min_{k \in \mathcal{Q}}\! V_k(\bm{x}_0 | \bm{\xi}_k^*, \bm{x}_k^* )$ for different $\mathcal{Q}$ and disturbances. 
It is found that when $\mathcal{Q} = \mathcal{Q}_1$, for all the disturbances, the value of $\min_{k \in \mathcal{Q}}\! V_k(\bm{x}_0 | \bm{\xi}_k^*, \bm{x}_k^* )$ is smaller than the RHS value of (\ref{eq-5-4-sol:2}), indicating $\bm{x}_0 \in \mathcal{O}$. 
In contrast, when $\mathcal{Q} = \mathcal{Q}_2$, $\bm{x}_0 \in \mathcal{O}$ only for disturbances D1, D2 and D5; and when $\mathcal{Q} = \mathcal{Q}_3$, $\bm{x}_0 \in \mathcal{O}$ only for disturbances D1 and D2. 
Thus the topology control utilizing all the four subsystems is able to stabilize the MMG system after any of the five disturbances, while reducing the number of subsystems utilized makes the system after some disturbances not stabilizable. 
Therefore, it is necessary to include as many subsystems as possible in $\mathcal{Q}$ following criterion (e) such that the stabilizing capability of the topology control can be leveraged maximally.

\begin{table}[t!]
    \centering
    \setlength{\tabcolsep}{2.93pt}  
    \setlength\extrarowheight{3pt} 

    \caption{{Results of region $\mathcal{O}$ and stabilizability with different $\mathcal{Q}$.}}
    \small{
    \begin{tabular*}{0.78\hsize}{ccccccc} \hline\hline
    \multirow{2}{*}{$\mathcal{Q}$} & \multirow{2}{*}{\shortstack{The RHS value \\of (\ref{eq-5-4-sol:2})}}  & \multicolumn{5}{c}{The value of $\min_{k \in \mathcal{Q}}\! V_k(\bm{x}_0 | \bm{\xi}_k^*, \bm{x}_k^* )$    } \vspace*{0.4pt}  \\ \cline{3-7}
    & & D1 & D2 & D3 & D4 & D5   \\ \hline 
    $\mathcal{Q}_1$ & 0.0662 & 0.0451 & 0.0378 & 0.0594 &  0.0601 & 0.0515       \\
    $\mathcal{Q}_2$ & 0.0834 & 0.0627 & 0.0691 & 0.0965 &  0.1137 & 0.0741        \\
    $\mathcal{Q}_3$ & 0.1170 & 0.0993 & 0.0805 & 0.1640 &  0.1509 & 0.1282         \\ \hline\hline
    \end{tabular*} 
    }
    \label{tab-5-4-1}   
\end{table}

\subsubsection{Stabilization results with $\mathcal{Q} = \mathcal{Q}_1$ for different disturbances}

Fig. \ref{fig-5-4-rr1-r2} (the left) shows the time-domain simulation results of $\theta_i$ of the PPC of each microgrid for disturbance D1 to D5, all without the proposed topology control. 
For each disturbance, the MMG system is unstable with voltage angle divergence. 
The switching control signal $\sigma(t)$ of the topology control with $\mathcal{Q} = \mathcal{Q}_1$ for each disturbance is given in Fig. \ref{fig-5-4-rr1-r2} (the right). 
Fig. \ref{fig-5-4-rr1-r2} (the middle) shows the dynamic response of the MMG system with the topology control. It can be found that for all disturbances D1 to D5, the system state converges to an equilibrium point before $t = 9$ s, indicating the proposed topology control effectively stabilizes the MMG system. 
 
Moreover, Fig. \ref{fig-5-4-rr1-r2-vol} provides the dynamic response of $v_i$ of each microgrid's PPC for disturbance D1 to D5. It can be found that without the proposed topology control, all voltage magnitudes exhibit continuous oscillations, and some of them consistently remain below 0.9 p.u.. In this case, the low voltage ride-through control of the converter interface will typically be activated, causing tripping of the associated converters \citep{4-1689}. In contrast, under the proposed topology control, all voltage magnitudes rebound rapidly to above 0.9 p.u. using about a second following disturbance clearance, and further converge to their steady-state values all between 0.9 p.u. and 1.1 p.u.. As a result, converter tripping can be prevented.

\begin{figure}[t!]
	\centering
 	\includegraphics[scale=0.95]{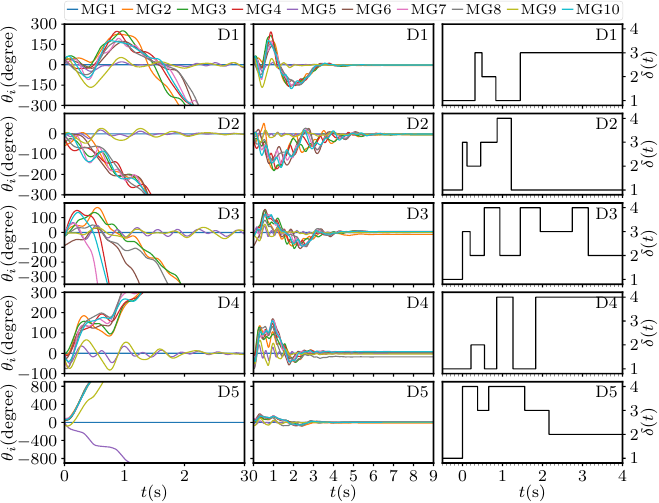} 
	\caption{Dynamic response of $\theta_i$ in the MMG without (left side) and with (middle) the topology control, and $\sigma(t)$ (right side) for disturbance D1 to D5.}
	\label{fig-5-4-rr1-r2}
\end{figure}

\begin{figure}[t!]
	\centering
 	\includegraphics[scale=0.95]{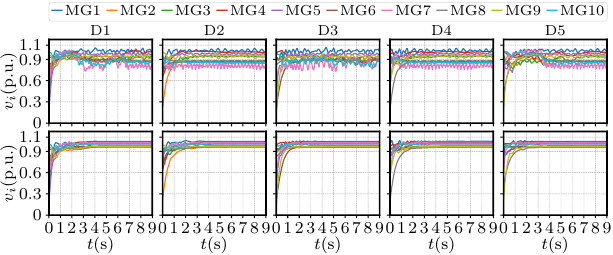} 
	\caption{{Dynamic response of $v_i$ in the MMG without (top) and with (bottom) the topology control, for disturbance D1 to D5.}}
	\label{fig-5-4-rr1-r2-vol}
\end{figure} 
\graphicspath{{chapter_8/Figs/}}
\chapter{Summary}\label{chapter-8}



This monograph has comprehensively introduced the fundamentals and recent progress in the field of power grid topology control. A basic prerequisite of the topology control process is the satisfaction of specific network topology constraints, which varies depending on the particular topology control problem. 
In Chapter \ref{chapter-3}, we have presented the mathematical formulations of four common types of network topology constraints, including network connectedness required by transmission networks, radiality required by distribution networks, network connectedness with contingencies and multi-island structure under emergency conditions. In particular, we have introduced, with detailed theoretical foundations, the electric flow-based connectedness formulation developed by the authors for both contingency-free and contingency cases. 

Topology control in current power grids is mainly performed under steady-state conditions. We have discussed steady-state topology control for both transmission networks and distribution networks in Chapter \ref{chapter-4} and Chapter \ref{chapter-5}, respectively. 
Regarding steady-state transmission topology control, in Chapter \ref{chapter-4}, we have first described the Braess’s paradox-type phenomena in transmission networks using two examples related to economic dispatch performance and system stability. This counterintuitive behaviour provides motivation for the effectiveness of steady-state transmission topology control. 
After introducing two of the most basic forms of steady-state transmission topology control, namely optimal transmission switching, and substation-level transmission topology control, we have reviewed the state-of-the-art developments in the field. Specifically, steady-state topology control has been examined for its potential to enhance system security, stability, resilience and economic efficiency. In particular, the deterministic and uncertainty-aware approaches for steady-state topology control to improve economic efficiency has been reviewed, respectively. Furthermore, various solution methodologies have been surveyed, including approximation/relaxation-based methods, heuristic methods, and emerging learning-based methods. 
A major challenge for steady-state transmission topology control with high renewable penetration arises from the multiple uncertainties in renewable generation and contingencies. Accordingly, we have introduced the three-stage OTS method developed recently by the authors to address this challenge. 

Furthermore, for steady-state distribution topology control, in Chapter \ref{chapter-5}, we have first introduced the basic distribution network reconfiguration model along with the classical heuristic solution method known as the successive branch reduction method. Following that, we have comprehensively reviewed the existing solution methods for steady-state distribution topology control problems, including approximation/relaxation-based methods, heuristic methods, and learning-based methods. 
In view of the challenge in simultaneously achieving optimality, scalability, and applicability in solving steady-state distribution topology control problems, we have presented the hybrid learning-heuristic solution paradigm recently proposed by the authors. By combining the heuristic pattern of classical successive branch reduction methods with the reinforcement learning, this paradigm offers enhanced performance across multiple aspects. 

The practical implementation of steady-state topology control gives rise to the issue of network topology transition, which is discussed in Chapter \ref{chapter-6}. We have emphasized the emerging needs to address the network topology transition problem. Following the formal description of the network topology transition problem and an analysis of the transition process, we have introduced the bumpless topology transition method developed recently by the authors. 
This method achieves superior topology transition performance by considering both static and dynamic system performance during the transition process. 

Compared to steady-state topology control, transient topology control is still at an early stage of research and development. Intentional Controlled Islanding (\ac{ici}) of transmission network is one of the few transient topology control methods that have been extensively developed and applied in practical systems. 
In Chapter \ref{chapter-7}, we have first introduced the fundamental methods for \ac{ici}, including the slow-coherency grouping method and optimization model for network separation. Then, we have reviewed the state-of-the-art methods of \ac{ici} and the development of connectedness-preserved transient topology control techniques. Most connectedness-preserving transient topology control methods remain conceptual and lack a rigorous theoretical foundation. We have provided a detailed introduction of the authors' recent advancements in connectedness-preserved transient topology control. 
Through the integration of switched system theory and machine learning techniques, connectedness-preserving transient topology control for the stabilization of multi-microgrid systems has been achieved, supported by a solid theoretical foundation. 

\begin{acknowledgements}
    The work is largely based on the first author's PhD and postdoctoral work supervised by Prof. David J. Hill and Prof. Yan Xu respectively. 
    We gratefully acknowledge Prof. Yue Song and Prof. Tao Liu, who were collaborators on some of the papers on which parts of this monograph are based. 
    We would also like to thank Prof. Marija D. Ilić (Editor-in-Chief) and anonymous reviewers for their helpful suggestions, as well as Mr. Mark de Jongh and Mr. Danny Kielty for managing the publication of this monograph.
    This work was supported in part by Ministry of Education (MOE), Republic of Singapore, under grant AcRF TIER-1 RT9/22. 
\end{acknowledgements}

\backmatter  

\printbibliography

\end{document}